\documentclass[11pt]{article}
\pdfoutput=1 
\usepackage{amsfonts,amsmath,amssymb,amsthm,bm,bbm}
\usepackage[noadjust]{cite}
\usepackage{enumitem}
\usepackage{etoolbox}
\usepackage{graphicx}
\usepackage{myOSAmeet}
\usepackage{mySupportCaption}
\usepackage{relsize}
\usepackage{rotating}
\usepackage{stmaryrd}
\usepackage[lofdepth,lotdepth]{subfig}
\usepackage{tensor}
\usepackage{tikz}
\usetikzlibrary{calc}
\usepackage[normalem]{ulem}
\mathchardef\myhyphen="2D
\newcommand{\Iver}[1]{\left[#1\right]_{\mathsmaller{\rm Iver}}}
\newcommand{\defeq}{\;{\raisebox{-.1\height}{$\stackrel{\raisebox{-.1\height}{\tiny\rm def}}{=}$}}\;}

\newcommand{\simX}[1]{\stackrel{\raisebox{-.3\height}{\tiny\it #1}}{\sim}}
\newcommand{\stkout}[1]{\ifmmode\text{\sout{\ensuremath{#1}}}\else\sout{#1}\fi} 
\newcommand{\tightoversetone}[2]{\mathop{#2}\limits^{\vbox to -.05ex{\kern-0.75ex\hbox{$#1$}\vss}}}
\newcommand{\tightoversettwo}[2]{\mathop{#2}\limits^{\vbox to -.25ex{\kern-0.75ex\hbox{$#1$}\vss}}}

\makeatletter 
\let\@@pmod\pmod
\DeclareRobustCommand{\pmod}{\@ifstar\@pmods\@@pmod}
\def\@pmods#1{\mkern4mu({\operator@font mod}\mkern 6mu#1)}
\makeatother 
\newtoggle{ForUSPTO}
\togglefalse{ForUSPTO}
\iftoggle{ForUSPTO} {
\def\myCorollary{Derived Utility}
\def\myLemma{Auxiliary Utility}
\def\myTheorem{Major Utility}
} {
\def\myCorollary{Corollary}
\def\myLemma{Lemma}
\def\myTheorem{Theorem} }
\newtheorem{algorithm}{Algorithm}
\newtheorem{corollary}{\myCorollary}
\newtheorem{definition}{Definition}
\newtheorem{lemma}{\myLemma}
\newtheorem{theorem}{\myTheorem}
\def\bfA{{\bf A}}
\def\bfB{{\bf B}}
\def\bfG{{\bf G}}

\def\bfg{{\bf g}}
\def\bfi{{\bf i}}

\def\bfq{{\bf q}}
\def\bfS{{\bf S}}

\def\calA{{\cal A}}
\def\calB{{\cal B}}
\def\calC{{\cal C}}
\def\calD{{\cal D}}
\def\calE{{\cal E}}
\def\calF{{\cal F}}
\def\calG{{\cal G}}
\def\calH{{\cal H}}
\def\calI{{\cal I}}
\def\calJ{{\cal J}}

\def\calL{{\cal L}}
\def\calM{{\cal M}}
\def\calN{{\cal N}}
\def\calP{{\cal P}}
\def\calQ{{\cal Q}}
\def\calR{{\cal R}}
\def\calS{{\cal S}}
\def\calT{{\cal T}}
\def\calU{{\cal U}}
\def\calV{{\cal V}}
\def\calX{{\cal X}}
\def\calY{{\cal Y}}
\def\sfE{{\sf E}}
\def\sft{{\sf t}}

\def\abs{{\rm abs}}
\def\Base{{\rm Base}}
\def\Bor{{\rm Bor}}
\def\Bun{{\rm Bun}}
\def\cl{{\rm cl}}
\def\Cost{{\rm Cost}}
\def\coni{{\rm coni}}
\def\Diri{{\rm Diri}}
\def\deg{{\rm deg}}
\def\diag{{\rm diag}}
\def\diam{{\rm diam}}
\def\dim{{\rm dim}}
\def\dist{{\rm dist}}
\def\dom{{\rm dom}}
\def\Gr{{\rm Gr}}
\def\gmul{{\rm gmul}}
\def\half{\frac{1}{2}}
\def\hom{{\rm hom}}
\def\HomoPhysL{\,{\stackrel{\;\mathsmaller{\mathfrak{M}}}{\Longleftarrow}}\,}
\def\HomoPhysR{\,{\stackrel{\mathsmaller{\mathfrak{M}}\;}{\Longrightarrow}}\,}
\def\IsoPhysLR{\,{\stackrel{\mathsmaller{\mathfrak{M}}}{\Longleftrightarrow}}\,}
\def\itGamma{\mathit{\Gamma}}
\def\itPi{\mathit{\Pi}}
\def\itTheta{\mathit{\Theta}}

\def\lab{{\rm lab}}
\def\LTK{{\rm LTK}}
\def\orb{{\rm orb}}
\def\poly{{\rm poly}}
\def\Pr{{\rm Pr}}
\def\Range{{\rm Range}}
\def\Slab{{\rm Slab}}
\def\sent{{\rm sent}}
\def\sign{{\rm sign}}
\def\size{{\rm size}}
\def\stab{{\rm stab}}
\def\smallperp{{\mathsmaller{\perp}}}
\def\supp{{\rm supp}}
\def\Tr{{\rm Tr}}
\begin{document}
\title{Monte Carlo Quantum Computing}
\iftoggle{ForUSPTO} {
\author{Haiqing Wei}
\vspace{0.75ex}
\address{Quantica Computing, LLC, San Jose, California, USA}
\vspace{0.75ex}
\newcommand{\PatentAbstract} {Monte Carlo methods are described for efficiently simulating a separately frustration-free Hamiltonian of a many-body quantum system on a classical computer. Also disclosed are methods for designing a separately frustration-free Hamiltonian to simulate a prescribed quantum system. Further described are methods for solving a prescribed computational problem by designing a quantum system having a separately frustration-free Hamiltonian and simulating the designed quantum system via Monte Carlo on a classical computer.}
} {
\author{David H. Wei$^{\,*}$}
\vspace{0.75ex}
\address{Quantica Computing, LLC, San Jose, California, USA}
\address{*\,Email: david.hq.wei@gmail.com}
\begin{abstract}
It is shown that a class of separately frustration-free (SFF) Hamiltonians can be Monte Carlo simulated efficiently on a classical computing machine, because such an SFF Hamiltonian corresponds to a Gibbs wavefunction whose nodal structure is efficiently computable by solving a small subsystem associated with a low-dimensional configuration subspace. It is further demonstrated that SFF Hamiltonians can be designed to implement universal ground state quantum computation. The two results combined have effectively solved the notorious sign problem in Monte Carlo simulations, and proved that all bounded-error quantum polynomial time algorithms admit bounded-error probabilistic polynomial time simulations.
\end{abstract}
}
\pagestyle{plain}


\iftoggle{ForUSPTO} {
%
\section*{CROSS-REFERENCE TO RELATED APPLICATIONS}
The present application claims the benefit under 35 U.S.C. \S119(e) of U.S. Provisional Application Ser. No. 63/130,847, filed December 27, 2020. Said U.S. Provisional Application Ser. No. 63/130,847 is hereby incorporated by reference in its entirety.

\section*{FIELD}
This invention relates to classical and quantum computations, specifically to implementations of quantum computing and simulations via Monte Carlo on classical computers.
} {}

\iftoggle{ForUSPTO} {
\section*{BACKGROUND}
\subsection*{\bf General Background}
} {
\section{Introduction}
}
Quantum mechanics and quantum field theories provide the most accurate description of almost everything in the known physical world, with the only exception of extremely strong gravitation. Vastly many questions in physics, chemistry, materials science, and even molecular biology could be answered definitively by solving a set of well established quantum equations governing a {\em quantum system}, which refers to a generic physical system consisting of particles and fields that mediate electromagnetic, weak and strong interactions, even non-extreme gravitational forces, so long as said particles and fields conform to the laws of quantum mechanics. Such a quantum system usually involves a large number of particles and modes of fields, and the noun quantum system is often premodified by an adjective {\em many-particle} or {\em many-body} \cite{Greiner01,Gustafson11,Coleman15} to stress that the sum of the number of particles and the number of modes of fields is large, usually much greater than $1$, although such adjective many-particle or many-body is sometimes omitted but implied and understood from the context. A formal definition of quantum physics or system will be given \iftoggle{ForUSPTO} {in the detailed description} {later}, whose mathematical rigor is needed to clearly present and discuss the problem of solving or simulating a given quantum system, as well as to measure the computational complexity incurred by a solution. Computational applications solving quantum equations include but are not limited to two general classes of problems, where one class consists of static problems that need to find the eigenvalues and eigenstates or the stationary distribution or the expectation value of a quantum observable, such as a time-independent Hamiltonian or the associated (thermal) density matrix, which is also known as the (quantum) statistical operator or the Gibbs operator \cite{Feynman72,Ceperley91,Ceperley95,Zinn-Justin05,Kozlov12}, whose position (that is, coordinate in a {\em configuration space} \cite{Greiner01}) representation is called the the Gibbs kernel, whereas the other class consists of dynamic problems that need to solve for the time evolution of a quantum state in the Schr\"odinger representation or a quantum observable in the Heisenberg representation \cite{Messiah99,Greiner01} when dealing with a closed quantum system, or of a reduced density matrix governed by the Liouville-von Neumann equation or the so-called master equations \cite{Prigogine62,Mandel95} for an open quantum system. However, the categorization of computational problems into static and dynamic classes is only relative, and the two classes of problems are often interchanged during formulations and solutions. Examples for solving a static problem using dynamic evolution include the homotopy method for eigenvalue problems \cite{Chu85,Chu88}, the equation-of-motion method for computing the eigenvalue distribution of large matrices on a classical computer \cite{Alben75,Hams00}, and Kitaev's quantum phase (eigenvalue) estimation algorithm as a quantum counterpart \cite{Kitaev95}. Conversely, for examples of mapping the dynamic execution of a computational algorithm into finding the minimum energy and state of a static system, there are the so-called quantum-dot cellular automata \cite{Lent93,Lent97} still in the realm of classical computing, and the celebrated paradigm of ground state quantum computation (GSQC) \cite{Feynman85,Kitaev02,Mizel04,Kempe06}, where a given quantum algorithm is mapped into a designer Hamiltonian and a simulation of its ground state, which encodes the history of the quantum evolution prescribed by said quantum algorithm.

\iftoggle{ForUSPTO} {
\subsection*{\bf Prior Art}
} {}
Ostensibly due to the exponentially exploding dimension of the Hilbert space that is needed to describe the state of a quantum system, it can be exceedingly hard to solve the quantum equations or simulate an eigenstate or dynamics of even a moderately sized quantum system on a classical computer \cite{Feynman82}. Various quantum Monte Carlo (QMC) approaches \cite{Ceperley79Binder,Schmidt84,Ceperley86,Hammond94,Grotendorst02} have the potential to break the curse of dimensionality, by mapping a non-negative ground state wavefunction or Gibbs kernel of a quantum system into a classical probability density, and simulating a random walk that embodies importance sampling of such a probability distribution. Given positivity of a concerned wavefunction or Gibbs kernel, QMC is arguably the only general and exact numerical method that is free of uncontrollable systematic errors due to modeling approximations, providing reliable and rigorous simulation results upon numerical convergence. Unfortunately, previous QMC procedures for many quantum systems, especially those involving multiple indistinguishable fermions which represent the vast majority of atomic, molecular, condensed matter, and nuclear systems, suffer from the notorious sign problem \cite{Feynman82,Grotendorst02,Loh90,Ortiz02,SignProblemWiki} that leads to an exponential slowdown of numerical convergence, due to the presence and cancellation of positive and negative amplitudes, when no computational basis is known to represent the concerned ground or thermal state by a non-negative wavefunction or Gibbs kernel. At the fundamental and theoretical level, as Feynman keenly noted \cite{Feynman87}, a defining characteristic, possibly the single most important aspect, setting quantum mechanics and computing apart from classical mechanics and computing, seemed to be the presence and necessity of a sort of ``negative probability'' in the quantum universe, endowing the power to represent and manipulate ``negative probabilities''. Indeed, their perceived inability to deal with ``negative probabilities'' efficiently was believed to fundamentally limit the power of classical mechanics and computers in terms of simulating their quantum counterparts and solving computational problems that quantum computers are predicted and believed to excel. The persistently unsolved status of the sign problem in the past, compounded by the piling of ``evidence problems'' that had an efficient quantum solution but no good classical algorithm being known or even thought possible \cite{Simon97,Shor97,Harrow09,Aaronson10,Aaronson15,Raz19,Gilyen19}, had fueled a pervasive belief that quantum computers were inherently more powerful than classical machines, and there existed certain hard computational problems which were amenable to polynomial quantum algorithms but could not be solved efficiently on classical computers, or in the terminology of quantum complexity theory \cite{Bernstein97,Watrous09,QuantumComplexityWiki}, that the computational complexity classes of bounded-error quantum polynomial time (BQP) and quantum Merlin Arthur (QMA) are strictly proper ({\it i.e.}, larger) supersets of the classical classes of bounded-error probabilistic polynomial time (BPP) and 1-message Arthur-Merlin interactive proof (AM[1]$\defeq$MA) \cite{Babai85,Goldwasser86,Arora09}.

\iftoggle{ForUSPTO} {
\subsection*{\bf Objects and Advantages}
} {}
Bucking the common and popular belief, here I shall \iftoggle{ForUSPTO} {disclose} {present} general and systematic solutions to the dreaded sign problem, thus establish BPP$\,=\,$BQP, by identifying and characterizing a class of quantum Hamiltonians/systems called separately frustration-free (SFF), whose Gibbs wavefunctions have a nodal structure that can be efficiently computed by solving a small subsystem involving a small number of dynamical variables associated with a constituent few-body interaction. When an SFF Hamiltonian is also fermionic stoquastic and spatially local, it is dubbed SFF fermionic Schr\"odinger (SFF-FS). Alternatively, when a general SFF Hamiltonian has each of its constituent few-body interactions essentially bounded, it is called SFF essentially bounded (SFF-EB). As a first \iftoggle{ForUSPTO} {Major Utility} {theorem}, it will be \iftoggle{ForUSPTO} {demonstrated} {proved} that any SFF-FS or SFF-EB Hamiltonian can be efficiently simulated on a classical probabilistic machine, in the sense that a Gibbs wavefunction of such a Hamiltonian can be efficiently sampled via a Monte Carlo procedure on a classical probabilistic machine, subject to a polynomially small error in a properly defined sense. Next, I will construct a special class of Hamiltonians/systems called SFF doubly universal (SFF-DU), that are both SFF-FS and SFF-EB consisting of distinguishable and interacting bi-fermions, each of which comprises two non-interacting identical spinless fermions moving in a three-well potential on a circle. Then as a second \iftoggle{ForUSPTO} {Major Utility} {theorem}, it will be \iftoggle{ForUSPTO} {demonstrated} {proved} that SFF-DU Hamiltonians/systems consisting of only bi-fermions are universal for quantum circuits and computations, as such, any many-body quantum system can be mapped onto and simulated by an SFF-DU system of bi-fermions.

The combination of the two \iftoggle{ForUSPTO} {major utilities} {theorems} effectively asserts a third \iftoggle{ForUSPTO} {Major Utility} {theorem} that the two computational complexity classes, bounded-error probabilistic polynomial time (BPP) \cite{Arora09} and bounded-error quantum polynomial time (BQP) \cite{Arora09}, are actually one and the same. That quantum computing and mechanics are just classical computing and probability up to polynomial reduction, is too good and true. Any quantum algorithm/circuit in BQP can be efficiently simulated by a Monte Carlo algorithm running on a classical computer. Such simulation, indeed implementation of quantum computing, is called Monte Carlo quantum computing (MCQC), which is not to be confused with the still sign-problem-prone, conventional quantum Monte Carlo simulation of quantum computation \cite{Cerf98}. The \iftoggle{ForUSPTO} {methods and algorithms} {theory and algorithms} presented here not only solve the sign problem that has plagued Monte Carlo simulations in many areas of science and technology, but also open up new avenues for developing and identifying efficient probabilistic algorithms from the vantage point of quantum computing. All known and to be discovered BQP algorithms reduce to BPP solutions. It should be noted that the BPP or BQP class of computational problems as referenced here is not to be understood as limited strictly to the family of decision or promise problems \cite{Even84,Goldreich06} on a classical or quantum computer. Rather, the BPP or BQP class should be broadly interpreted as to represent general types of computational problems that are efficiently solvable on a classical or quantum machine. Indeed, it has been well established and widely known that a great number of computational problems for function evaluation, objective optimization, and matching or solution search, {\it etc}., are reducible or polynomially equivalent to BPP or BQP problems, in that, the answer to a function/optimization/search problem can be obtained efficiently by solving one or a polynomially bounded number of BPP or BQP problem(s). Moreover, it is often quite straightforward in practice to modify and adapt only slightly a BPP or BQP algorithm to an efficient procedure for solving a function/optimization/search problem.

\iftoggle{ForUSPTO} {
\section*{SUMMARY}
Methods, software systems and related computer program products, involved computer-readable source and machine codes as well as numerical data, hardware systems and related physical devices storing, processing, and executing involved computer-readable source and machine codes as well as numerical data, are described in relation to processes of simulating a multivariable system on a classical computer, constructing a multivariable system for quantum computing, and solving a computational problem on a classical computer.

A first process of simulating a multivariable system comprises a step of receiving a multivariable system and a step of generating a plurality of result sampling points. Said multivariable system in turn comprises a configuration space containing configuration points, and a separately frustration-free (SFF) partial Hamiltonian as a combination of a plurality of directly frustration-free (DFF) or ground state frustration-free (GFF) partial Hamiltonians, as well as an objective operator whose expectation value is of interest. Said step of generating a plurality of result sampling points further comprises one step of initializing a current sampling point (CSP) variable, and another step of repeating a first iterative step and a second iterative step for a predetermined number of times, where the first iterative step executes a transition from a CSP to a new sampling point, while the second iterative step records and resets the CSP variable.

A second process of constructing a multivariable system comprises a step of receiving a description of a quantum algorithm, a step of constructing a Feynman-Kitaev (FK) register, and a step of constructing an SFF partial Hamiltonian. Said description of a quantum algorithm in turn comprises a plurality of qubits, an initial state, a linearly ordered sequence of a plurality of quantum gates, and a quantum measurement operator. Said step of constructing a Feynman-Kitaev (FK) register further comprises a sub-step of allocating a plurality of FK logic bits and a plurality of FK clock bits, a sub-step of creating a plurality of FK time projectors, and a sub-step of creating a plurality of FK time propagators, where said plurality of FK time propagators are divided into a first group of first FK time propagators and a second group of second FK time propagators. Said step of constructing an SFF partial Hamiltonian further comprises a sub-step of creating an FK state initializer, a sub-step of creating a first FK state operator as a direct sum of a plurality of first FK state propagators, with each said first FK state propagator being a tensor product between a first FK time propagator and a quantum gate acting on said FK logic bits, and a sub-step of creating a second FK state operator as a direct sum of a plurality of second FK state propagators, with each said second FK state propagator being a tensor product between a second FK time propagator and a quantum gate acting on said FK logic bits.

A third process of solving a computational problem comprises a step of generating data or signals describing a multivariable system and a step of simulating said multivariable system. Said step of generating data or signals describing a multivariable system further comprises a sub-step of allocating a configuration space, a sub-step of creating an SFF partial Hamiltonian, and a sub-step of generating an objective operator. Said step of simulating the multivariable system uses said first process of simulating a multivariable system.
} {}

\iftoggle{ForUSPTO} {
\newcommand{\ListOfDraw} {
Fig. \ref{ThreeWellVx} shows a three-well potential on a circle.

Fig. \ref{ThreeWellPsiEvenOdd} shows the eigenstates $\psi_{\mathsmaller{P}}$, $\psi_+$, and $\psi_-$ in dotted, solid, and dashed lines respectively.

Fig. \ref{ThreeWellPsiPsiLPsiR} shows the eigenstates $\psi_{\mathsmaller{P}}$, $\psi_{\mathsmaller{L}}$, and $\psi_{\mathsmaller{R}}$ in dotted, solid, and dashed lines respectively.

Fig. \ref{FourPhisAndNodalCurves} shows the bi-fermion wavefunctions $\Phi_+(x_1,x_2)$, $\Phi_-(x_1,x_2)$, $\Phi_{\mathsmaller{L}}(x_1,x_2)$, $\Phi_{\mathsmaller{R}}(x_1,x_2)$.

Fig. \ref{FourXsinZcosNodalCurves} shows the ground states of $H_{\mathsmaller{\rm BF},0}+
V_{\mathsmaller{X},\eta\sin\theta}+V_{\mathsmaller{Z},\eta\cos\theta}$ for $\theta=-3\pi/4,-\pi/4,\pi/4,3\pi/4$.

Fig. \ref{SixVertexSegment} shows a $6$-vertex segment of a ring-shaped graph.

Fig. \ref{MethodSimuSyst} illustrates one process of simulating a multivariable system.

Fig. \ref{MethodSimuDens} illustrates another process of simulating a multivariable system.

Fig. \ref{MethodConsFK} illustrates a process of constructing a multivariable system.

Fig. \ref{MethodSolvSyst} illustrates one process of solving a computational problem.

Fig. \ref{MethodSolvDens} illustrates another process of solving a computational problem.
} 

\section*{BRIEF DESCRIPTION OF THE DRAWINGS}
\ListOfDraw

\section*{DETAILED DESCRIPTION}
\subsection*{\bf Separately Frustration-Free Partial Hamiltonians and Systems}
} {
\section{Separately Frustration-Free Hamiltonians/Systems}
}
Although not being stated explicitly, it should be understood that the physics and properties of any quantum system involve such physical constants as the Planck constant, the speed of light in vacuum, the elementary electric charge, and the masses of elementary particles. For succinctness of mathematical expressions, it is good to work with a ``natural'' unit system where the relevant physical constants take the numerical value $1$, or to have the Schr\"odinger and other physical equations nondimensionalized, so that the physical constants do not manifest in the Hamiltonians and other physical quantities (operators).

A theory on a physical system, whether classical or quantum, needs to start with identifying a list of {\em dynamical variables} \cite{Dirac30,Dittrich20} associated with the particles and/or fields comprising the physical system. Commonly used dynamical variables include spatial positions and/or spins of particles and spatial distributions of field intensities or potentials, although other physical quantities may be used as well, and could be particularly convenient for specific problems of interest. Different choices of a list of dynamical variables associated with the same physical system are mathematically and physically equivalent, when they are related through either a canonical transformation \cite{Arnold06,Dittrich20} for a classical theory or a unitary transformation \cite{Dirac30,Weyl31,Greiner01} for a quantum theory, and in both the classical and quantum cases an extension by/to the so-called gauge transformations \cite{Schiff68,Itzykson80,Jackson98,Greiner01,Weinberg05,Sakurai20}.

In the non-relativistic regime, a general many-body quantum system may consist of a fixed number of particles belonging to multiple species of boltzmannons, bosons, and fermions, that interact with and/or through certain classical fields as well as quantized modes of fields that are bosonic in nature, whereas in the relativistic regime, {\it i.e.}, the realm of quantum field theory, both particles and fields could move in and out of quantized modes, and the number of particles or field quanta in any particular quantized mode may not be conserved. For considerations of QMC and quantum computing, if quantum systems consisting of many fermionic particles could be well handled, then bosonic particles and field modes would not impose additional challenges, because indistinguishable bosons can be artificially labeled and a bosonic system/subsystem can be simulated by a boltzmannonic random walk in conjunction with Monte Carlo sampling of particle label permutations. Permutation-averaging $N$-particle Gibbs wavefunctions, $N\in\mathbb{N}$ amounts to calculating the permanent of an $N\times N$ non-negative matrix, which is known to admit a fully polynomial randomized approximation scheme \cite{Jerrum89,Jerrum04}, that, with a specified approximation error that is bounded by a polynomial of the problem input size, can be cast into a decision version that belongs to the class of BPP promise problems \cite{Even84,Goldreich06}. Another excuse for not having to emphasize and work explicitly on bosons in the present considerations of QMC and quantum computing is that a quantum system consists only of fermions can be configured into a universal quantum computer/simulator, which can be so programmed that its effective physics is a homomorphic image, or called a homophysical image (the notion of homophysics will be precisely defined shortly), of any general quantum system containing particles of any type of statistics. Furthermore, an individual boltzmannon can be regarded as a single fermion of a unique species. Therefore, it is without loss of generality to focus on many-body quantum systems consisting of multiple species of fermions. A quantum system and its Hamiltonian are said to be {\em many-fermion} when they involve at least two identical fermions of the same particle species.

In a canonically quantized theory \cite{Dirac30,Itzykson80,CanoQuanWiki} with the position (that is, the {\em canonical coordinate} \cite{Dirac30}) representation, a quantum system could involve both continuous and discrete dynamical variables being represented by continuous and discrete canonical coordinates \cite{Dirac30}, thus move in a Cartesian product between a connected continuous topological space \cite{Munkres00} and a finite discrete point set. Physical motions and quantum mechanics in a finite discrete space can be treated nicely and simply using the theory of finite-dimensional vector spaces \cite{Halmos93} and matrices, as epitomized by the familiar physics and theory of the electron spin and Pauli matrices. By contrast, quite a bit of advanced and sophisticated mathematical machinery is required to render a quantum theory of continuous motions in a topological space with sufficient generality and necessary rigor, that involves topological manifolds, Riemannian differential geometry, and functional analysis, in particular, Wiener measures and path integrals, operator analysis, and spectral theory.

\begin{definition}{(Configuration Space \cite{Greiner01})}\label{defiConfigSpace}\\
In general, a configuration space $\calC$ refers to a Cartesian product $\calC\defeq\calM\times\calP$ between a connected continuous topological space $\calM$ and a finite discrete (topological) space $\calP$, that is a finite point set $\calP$ endowed with the discrete topology. A point in such a configuration space $\calC$ is meant to represent the values of a chosen list of dynamical variables associated with a plurality of constituent particles and/or fields comprising a quantum system.
\vspace{-1.5ex}
\end{definition}

A configuration space $\calC$ defined as such includes two special cases, one of which being with $\calM=\emptyset$ or $\calM$ containing but a single isolated point, when it reduces to $\calC\cong\calP$, while the other with $\calP=\emptyset$ or $\calP$ containing just one discrete point, when it reduces to $\calC\cong\calM$. Notwithstanding the adoption of a generic Cartesian product $\calC=\calM\times\calP$ representing an all-encompassing continuous-discrete product configuration space, and accordingly, our generality-striving mathematical formulations in this \iftoggle{ForUSPTO} {specification} {presentation}, it is worth noting that an actual or model physical system with either all discrete or all continuous dynamical variables can be made universal, in terms of both numerical simulations for or based on it and the computational power of it as a computing machine.

In practices of numerical computations as well as complexity analyses, every continuous coordinate can be discretized into a lattice at a finite resolution, with the entailed discretization error being controlled to decrease at least polynomially as the resolution gets finer. Certain required accuracy in approximating the quantum wave distribution or the expectation value of an observable sets an upper bound for the grid size of discretization. Such discretization is widely adopted in various fields like digital signal processing, computational physics, and lattice (quantum) field theory, where a finite grid resolution effects an upper limit to the frequency or energy scale, that, if set sufficiently high, should not affect the low-frequency dynamics or low-lying energy states of interests.

Conversely, any discrete dynamical variable, such as a spin-$\half$, can be homophysically simulated (in a sense to be precisely defined shortly) by a subspace of low-energy states of a particle moving in a continuous space with a multi-well potential, where the well depths and tunneling strengths between wells are suitably designed, such that well-localized wave packets resemble discretely valued states of the discrete dynamical variable, and tunnelings of wave packets between potential wells simulate couplings between the discretely valued states of the discrete dynamical variable in question. Indeed, one interesting development later in this \iftoggle{ForUSPTO} {specification} {presentation} will be to construct a system of two fermions moving continuously on a circle that simulates/realizes a qubit, or to be exact, a rebit, which is a quantum two-state system with the wavefunction restricted to be real-valued.

Consider a physical system consisting of $S\in\mathbb{N}$ fermion species, each of which indexed by $s\in[1,S]$ has $n_s\in\mathbb{N}$ identical particles moving in a configuration space $\calC_s\defeq\calM_s\times\calP_s$, with $\calP_s$ being a finite set of discrete points having a cardinality $d^{\mathsmaller{D}}_s\defeq|\calP_s|\in\mathbb{N}$, and $\calM_s\defeq(M_s,\calT_s,\calA_s)$ being a $(d^{\mathsmaller{C}}_s\in\mathbb{N})$-dimensional $C^k$-differentiable manifold, $k\ge 1$, that is a connected second-countable Hausdorff topological space of a set $M_s$ with a topology $\calT_s$, equipped additionally with a $C^k$-smooth atlas $\calA_s$, which is a collection of charts (also called coordinate patches) $\calA_s\defeq\{(U_{\alpha},\varphi_{\alpha})\}_{\alpha\in A_s}$, where $A_s$ is an index set, $\{U_{\alpha}\}_{\alpha\in A_s}$ is an open cover of $M_s$, and $\forall\alpha\in A_s$, $\varphi_{\alpha}:U_{\alpha}\mapsto\varphi_{\alpha}(U_{\alpha})\subseteq\mathbb{R}^{d_s}$ is a homeomorphism from $U_{\alpha}$ into $\varphi_{\alpha}(U_{\alpha})$; furthermore, $\forall\alpha,\beta\in A_s$ such that $U_{\alpha}\cap U_{\beta}\neq\emptyset$, the charts $(U_{\alpha},\varphi_{\alpha})$ and $(U_{\beta},\varphi_{\beta})$ are $C^k$-compatible, namely, the chart transition (also called a change of coordinates) $\varphi_{\beta}\circ\varphi_{\alpha}^{\mathsmaller{-}1}:\varphi_{\alpha}(U_{\alpha}\cap U_{\beta})\mapsto\varphi_{\beta}(U_{\alpha}\cap U_{\beta})$ is $C^k$-smooth; still further, the atlas $\calA_s$ is assumed maximal, in the sense that, if $(U,\varphi)$ is any chart $C^k$-compatible with all $(U_{\alpha},\varphi_{\alpha})$, $\alpha\in A_s$, then $(U,\varphi)\in\calA_s$. The topological space $(M_s,\calT_s)$ is called locally Euclidean by the equipment of an atlas. Being both locally Euclidean and second-countable implies that the manifold $\calM_s$ as a topological space is locally compact, normal, metrizable, paracompact, path-connected, separable, and admits a countable partition of unity $\{\rho_i\}_{i\in\mathbb{N}}$ subordinate to a predetermined open cover of $M_s$, such that $\forall i\in\mathbb{N}$, $\rho_i\ge 0$, $\forall p\in M_s$, $\sum_{i=1}^{\infty}\rho_i(p)=1$, where the functions $\{\rho_i\}_{i\in\mathbb{N}}$ are $C^k$-smooth, the collection of supports $\{\supp(\rho_i)\}_{i\in\mathbb{N}}$ is locally finite \cite{Warner83,Lang99}; either the predetermined open cover is also countable as $\{U_i\}_{i\in\mathbb{N}}$, and $\supp(\rho_i)\subseteq U_i$ holds $\forall i\in\mathbb{N}$, or else, the $\{\rho_i\}_{i\in\mathbb{N}}$ functions and the open cover sets are not necessarily in one-to-one correspondence, but then $\{\supp(\rho_i)\}_{i\in\mathbb{N}}$ can be assumed all compact \cite{Warner83}. The configuration space $\calC_s=\calM_s\times\calP_s$ is considered to possess a product topology, with $\calP_s$ being interpreted as a compact discrete topological space.

To better facilitate analysis and physics therein, the differentiable manifold $\calM_s$ is further assumed to be equipped with a smooth metric $\{g_p\}_{p\in M_s}$ to make $(\calM_s,g)$ a Riemannian manifold, with the metric $g$ defining an inner product on each tangent space $T_pM_s$, $p\in M_s$, and inducing naturally a length measure $L_g(\gamma)\defeq\int_a^b\langle\gamma'(t)|\gamma'(t)\rangle^{1/2}_gdt$ along any piecewise $C^1$-smooth path $\gamma:\mathbb{R}\supseteq[a,b]\mapsto M_s$ \cite{Lang99,Sakai96}, which in turn defines a distance function $\dist_g(q,r)$, or just $\dist(q,r)$ when there is no risk of ambiguity, as $\dist_g(q,r)\defeq\inf\mathlarger{\mathlarger{\{}}\int_{\gamma}dL_g(\gamma):\gamma\in\Gamma_{[q,r]}(M_s)\mathlarger{\mathlarger{\}}}$, $\forall q,r\in M_s$, with
\begin{equation}
\Gamma_{[q,r]}(M_s)\defeq\left\{\gamma:[a,b]\mapsto M_s,\gamma\;{\rm is\;piecewise}
\;C^1\myhyphen{\rm smooth},\gamma(a)=q,\gamma(b)=r\right\}
\end{equation}
being the set of all piecewise $C^1$-smooth paths on $M_s$ connecting $q$ and $r$, and $\arg\inf\mathlarger{\mathlarger{\{}}\int_{\gamma}dL_g(\gamma):\gamma\in\Gamma_{[q,r]}(M_s)\mathlarger{\mathlarger{\}}}$ being the minimal geodesic on $M_s$ between $q$ and $r$. The distance function $\dist_g(\cdot,\cdot)$ metricizes $M_s$ and the metric topology $(M_s,\dist_g)$ coincides with the atlas topology $\calT_s$. Also, the metric $g$ induces a volume measure $V_g$ on a suitable measurable space $(M_s,\Sigma_s)$ with $\Sigma_s$ being a certain $\sigma$-algebra of subsets of $M_s$ that contains the topology $\calT_s$, where for each $U\in\calT_s$, the volume $V_g(U)\defeq\int_UdV_g\defeq\sum_{i\in\mathbb{N}}\int_{\varphi_{\alpha(i)}(U)}\left[|\det(g)|^{1/2}\rho_i\right]\circ{\raisebox{-.1\height}{$\varphi_{\alpha(i)}^{\mathsmaller{-}1}\,dx_{\alpha(i)}^1\cdots dx_{\alpha(i)}^{d_s}$}}$ \cite{Lang99,Sakai96}, with $\{\rho_i\}_{i\in\mathbb{N}}$ being a countable partition of unity subordinate to the atlas $\{(U_{\alpha},\varphi_{\alpha})\}_{\alpha\in A_s}$, that is, $\forall i\in\mathbb{N}$, there exists an $\alpha(i)\in A_s$ such that $\supp(\rho_i)\subseteq U_{\alpha(i)}$. It is further assumed that the measure space $(M_s,\Sigma_s,V_g)$ is $\sigma$-finite. The smallest $\Sigma_s$ possible is the Borel $\sigma$-algebra of $\calT_s$, denoted by $\Sigma_s^{\mathsmaller{B}}$, on which $V_g$ defines a Borel measure. The Borel measure space $(M_s,\Sigma_s^{\mathsmaller{B}},V_g)$ is separable \cite{Halmos74,Bogachev07} in the sense that, the quotient space of $\Sigma_s^{\mathsmaller{B}}$ modulo sets of measure zero is separable as a topological space, induced by the metric $d([E],[F])\defeq V_g(E\Delta F)\defeq V_g((E\cap F^c)\cup(E^c\cap F))$, called the {\em measure of symmetric difference}, $\forall E,F\in\Sigma_s^{\mathsmaller{B}}$, where $[X]\in\Sigma_s^{\mathsmaller{B}}/V_g$ represents the equivalence class of $X\in\Sigma_s^{\mathsmaller{B}}$ modulo sets of measure zero, that is, $[X]\defeq\{Y\in\Sigma_s^{\mathsmaller{B}}:V_g(X\Delta Y)=0\}$. A drawback of the Borel measure is incompleteness. Fortunately, the Lebesgue measure, being separable and complete as desired, can be constructed by defining an outer measure $V_g^*$ for any subset $X\subseteq M_s$ as $V_g^*(X)\defeq\inf\left\{\sum_{i\in\mathbb{N}}V_g(U_i):\{U_i\}_{i\in\mathbb{N}}\subseteq\calT_s,\bigcup_{i\in\mathbb{N}}U_i\supseteq X\right\}$, calling any subset $Y\subseteq M_s$ Lebesgue measurable if $V_g^*(X)=V_g^*(X\cap Y)+V_g^*(X\cap Y^c)$ holds $\forall X\subseteq M_s$, and identifying $V_g^*$ on the $\sigma$-algebra $\Sigma_s^{\mathsmaller{L}}$ of Lebesgue measurable sets as the Lebesgue measure $V_g$ \cite{Halmos74,Bogachev07}.

In a canonically first-quantized theory with the coordinate representation, the $n_s$ indistinguishable particles of each fermion species $s\in S$ may be artificially labeled, so that their spatial configuration, namely, the collective local coordinates of the $n_s$ identical fermions, can be represented by a point on a product space $\calC_s^{n_s}\defeq(\calM_s\times\calP_s)^{n_s}\defeq\prod_{n=1}^{n_s}(\calM_s\times\calP_s)\cong\calM_s^{n_s}\times\calP_s^{n_s}$, with $\calM_s^{n_s}$ being a Riemannian manifold as an induced manifold structure on the Cartesian set product \cite{Lang99,Tu11}, and $\cong$ signifying a topological/manifold isomorphism. For convenience to make references {\it infra}, the configuration space $\calC_s=\calM_s\times\calP_s$ as a spatial substrate that hosts a number of identical particles of species $s$ is called the {\em substrate space}, and $d_s\defeq\dim(\calC_s)\defeq\dim(\calM_s)+\log_2|\calP_s|$ is called the {\em substrate space dimension}, where $\dim(\calM_s)$ is given by the standard definition of dimension, or called topological dimension, of a connected topological space \cite{Engelking78,Munkres00}, which for a well-behaved topological space, particularly a differentiable manifold $\calM_s$, coincides with the minimum number of independent numerical values, also called {\em scalar components}, that are needed to label and distinguish all of the points in $\calM_s$, while $|\calP_s|$ denotes the cardinality of $\calP_s$, namely, the number of discrete points in the set $\calP_s$, and $\log_2|\calP_s|$ characterizes the minimum number of bits needed in a binary string used to label and distinguish all of the discrete points in $\calP_s$. By contrast, the product manifold $\calC_s^{n_s}$ for describing the spatial configurations of $n_s\in\mathbb{N}$ particles is called the {\em configuration space}, or to emphasize, the {\em many-body configuration space}, which is said to have a {\em configuration space dimension} $\dim(\calC_s^{n_s})\defeq n_sd_s$. Naturally, the interactions among particles and the wave function of the whole system consisting of $S$ fermion species can be treated mathematically under the auspices of a (product) configuration space $\calC\defeq\prod_{s=1}^{\mathsmaller{S}}\calC_s^{n_s}\cong\left(\prod_{s=1}^{\mathsmaller{S}}\calM_s^{n_s}\right)\mathlarger{\mathlarger{\times}}\left(\prod_{s=1}^{\mathsmaller{S}}\calP_s^{n_s}\right)\defeq\calM\times\calP$, which is endowed with the product topology, the product atlas, and the product Riemannian metric still denoted by $g$, together with the product length measure for paths and the product volume measure still being denoted by $L_g$ and $V_g$ respectively. A note is in order on the notation of a {\em configuration point} \cite{Greiner01} or a coordinate variable on a substrate space $\calC_s=\calM_s\times\calP_s$ or a configuration space $\calC=\calM\times\calP$. In many situations, a single variable like $q\in\calC_s$ or $q\in\calC$ would be sufficient and self-explanatory. But when it is necessary or otherwise useful, an ordered pair such as $(q^{\mathsmaller{C}},q^{\mathsmaller{D}})\in\calC_s=\calM_s\times\calP_s$ or $(q^{\mathsmaller{C}},q^{\mathsmaller{D}})\in\calC=\calM\times\calP$ may be used to coordinate a point in a substrate or configuration space, with the continuous and discrete canonical coordinates written down explicitly and separately.

On a many-body configuration space $\calC$, the above-named product length measure $L_g$  induces a product distance function $\dist_g(q,r)\defeq\inf\mathlarger{\mathlarger{\{}}\!\int_{\gamma}dL_g(\gamma):\gamma\in\Gamma_{[q,r]}(\calC)\mathlarger{\mathlarger{\}}}$, $\forall q,r\in\calC$, which respects the product topology on $\calC$. The overall {\em configuration space dimension} is obviously $\dim(\calC)\defeq\sum_{s=1}^{\mathsmaller{S}}n_sd_s$. The redundancy due to the artificial labeling of identical particles can and has to be removed by requiring the wavefunction to satisfy the exchange symmetry under permutations of identical particles. Specifically, for the system of $S$ fermion species, let $\left(q_1,\cdots\!,q_{\mathsmaller{S}}\right)\in\prod_{s=1}^{\mathsmaller{S}}\calC_s^{n_s}$ represent a point in the configuration space $\calC$, with $q_s\defeq\!\left(q_{s1},\cdots\!,q_{sn_s}\right)$, $1\le s\le S$ denoting the configuration of the $n_s$ artificially labeled identical fermions of the $s$-th species, where $\forall s\in[1,S]$, $\forall n\in[1,n_s]$, $q_{sn}\defeq(x_{sn1},\cdots\!,x_{snd_s},v_{sn})$, with $v_{sn}\in\calP_s$, $(x_{sn1}\defeq\varphi_{sn1},\cdots\!,x_{snd_s}\defeq\varphi_{snd_s})$ being a chart (coordinate patch) on an $\calM_s$ manifold, then any $\mathbb{K}$-valued wavefunction
\begin{equation}
\psi\!\left(q_1,\cdots\!,q_{\mathsmaller{S}}\right)\!\defeq\psi\!\left(\left(q_{11},\cdots\!,q_{1n_1}\right)\!,\cdots\!,\left(q_{s1},\cdots\!,q_{sn_s}\right)\!,\cdots\!,\left(q_{\mathsmaller{S}1},\cdots\!,q_{\mathsmaller{S}n_{\mathsmaller{S}}}\right)\right)
\end{equation}
must be in the Hilbert space $L_{\!\mathsmaller{F}}^2(\calC)\defeq\!\left\{\psi\in L^2(\calC)\defeq L^2(\calC;\mathbb{K}):\psi\circ\pi=\sign(\pi)\psi,\forall\pi\in G_{\rm ex}\right\}$, with a scalar field $\mathbb{K}\in\{\mathbb{R},\mathbb{C}\}$, and a group $G_{\rm ex}\defeq\prod_{s=1}^{\mathsmaller{S}}G_s$ called the {\em exchange symmetry group}, where $\forall s\in[1,S]$, $G_s$ is the symmetry group of $n_s$ labels \cite{Sagan01}, {\it i.e.}, the symmetry group of exchanging the $n_s$ identical fermions of the $s$-th species, so that for each $(\pi\in G_{\rm ex})=\prod_{s=1}^{\mathsmaller{S}}(\pi_s\in G_s)$, it holds
\begin{equation}
\psi\circ\pi\!\left(q_1,\cdots\!,q_{\mathsmaller{S}}\right)\!\defeq\psi\!\left(\pi_1q_1,\cdots\!,\pi_{\mathsmaller{S}}q_{\mathsmaller{S}}\right)=\,\sign(\pi)\psi\!\left(q_1,\cdots\!,q_{\mathsmaller{S}}\right) ,
\end{equation}
with $\sign(\pi)\defeq\prod_{s=1}^{\mathsmaller{S}}\sign(\pi_s)$, where
\begin{equation}
\pi_sq_s=\pi_s\!\left(q_{s1},\cdots\!,q_{sn_s}\right)\!\defeq\!\left(q_{s\pi_s(1)},\cdots\!,q_{s\pi_s(n_s)}\right)\!,\;\forall s\in[1,S] \,,
\end{equation}
and $\sign(\pi_s)=\pm 1$ denotes the sign of the permutation $\pi_s\in G_s$ \cite{Sagan01}. Alternatively, the redundancy in the coordinate representation for the system of $S$ fermion species can be removed by using the quotient manifold $\calC/G_{\rm ex}$ of the configuration space $\calC$ modulo the exchange symmetry group $G_{\rm ex}$, whose action on $\calC$ is obviously homeomorphic \cite{Rotman99,Kerber99,Lin94,Bekka00,Jacobson09,Rose09}. More formally, the quotient manifold $\calC/G_{\rm ex}$ forms an orbifold \cite{Thurston78,Adem07}, on which a typical point $[q]\in\calC/G_{\rm ex}$ is actually an orbit $G_{\rm ex}q\defeq\{\pi q:\pi\in G_{\rm ex}\}$ of a certain $q\in\calC$, and a $\mathbb{K}$-valued function can be defined and take value on each $[q]\in\calC/G_{\rm ex}$ independently. The set of square Lebesgue-integral functions on $\calC/G_{\rm ex}$, called wavefunctions on $\calC/G_{\rm ex}$, form a Hilbert space $L^2(\calC/G_{\rm ex})$, which is isomorphic to $L_{\!\mathsmaller{F}}^2(\calC)$, with both being regarded as one and the same. It is worth noting that the Lebesgue spaces $L^2(\calC)$ and $L_{\!\mathsmaller{F}}^2(\calC)$ are necessarily separable Hilbert spaces due to the separability of the Lebesgue measure space $(\calC,\Sigma^{\mathsmaller{L}}(\calC),V_g)$, with $\Sigma^{\mathsmaller{L}}(\calC)=\prod_{s=1}^{\mathsmaller{S}}\!\left(\Sigma_s^{\mathsmaller{L}}\right)^{n_s}$ representing the product $\sigma$-algebra.

Many-body systems considered in this \iftoggle{ForUSPTO} {specification} {presentation} are always assumed to consist of a potentially large number $S\in\mathbb{N}$ of fermion species, where each species $s\in[1,S]$ has at most a fixed number $n_{\max}\in\mathbb{N}$ of identical fermions that move on a substrate space $\calC_s=\calM_s\times\calP_s$ whose substrate space dimension $d_s$ does not exceed a fixed number $d_{\max}\in\mathbb{N}$. As the number $S$ of species and the overall configuration space dimension $\sum_{s=1}^{\mathsmaller{S}}n_sd_s$ increase indefinitely, of great interest is the asymptotic scaling of the computational complexity for simulating such many-body systems.

Since Dirac and von Neumann \cite{Dirac30,vonNeumann32} laid down the foundations, a quantum theory is customarily formulated through a Hilbert space $\calH$ of state vectors and a space $\calL(\calH)$ of linear operators on $\calH$ for physical observables. With a configuration space $\calC$ being chosen, $\calH(\calC)\defeq L^2(\calC)$ and $\calH_{\!\mathsmaller{F}}(\calC)\defeq L_{\!\mathsmaller{F}}^2(\calC)$ are among the most often used Hilbert spaces, although other convenient choices, such as the Sobolev spaces \cite{Hebey99,Adams03,Mazya11} $\calH^1(\calC)\defeq W^{1,2}(\calC)\defeq\{\psi\in L^2(\calC):\partial\psi\in L^2(\calC)\}$ and $\calH_{\!\mathsmaller{F}}^1(\calC)\defeq W^{1,2}(\calC)\defeq\{\psi\in L_{\!\mathsmaller{F}}^2(\calC):\partial\psi\in L^2(\calC)\}$, are employed occasionally, where $\forall\psi\in L^2(\calC)$, $\partial\psi$ denotes the vector of first-order weak partial derivatives of $\psi$ as distributions. Having a Hilbert space $\calH(\calC)$ chosen and understood, let $\calL(\calC)\defeq\calL(\calH(\calC))$ denote the set of linear operators from $\calH(\calC)$ to $\calH(\calC)$ that are possibly unbounded but must be densely defined and closable, $\calL_0(\calC)\subseteq\calL(\calC)$ denote the subset consisting of self-adjoint operators, and $\calB(\calC)\defeq\calB(\calH(\calC))$ denote the Banach space of bounded linear operators \cite{Conway90,Kato80,Blanchard15}. A general Hilbert space $\calH(\calC)$ and the sets of operators $\calL(\calC)$, $\calL_0(\calC)$, $\calB(\calC)$, as well as the quantum theory based on them are said to be supported by the configuration space $\calC$. The quantum physics of such a physical system is governed by one particular self-adjoint operator $H\in\calL_0(\calC)$, called the Hamiltonian, that can be defined through a semibounded symmetric quadratic form $Q_{\!\mathsmaller{H}}(\phi,\psi)$ over a dense subset $\dom(H)$ of the chosen and understood Hilbert space $\calH(\calC)$ \cite{Kato80,Blanchard15}, such that, formally, the Hamiltonian $H$ acts on Dirac bra and ket vectors to make $\langle\phi|H|\psi\rangle\defeq Q_{\!\mathsmaller{H}}(\phi,\psi)$, $\forall\phi,\psi\in\calH'\defeq\dom(H)\subseteq\calH(\calC)$. It is often necessary and/or convenient to consider a restriction of $H$ to a closed Hilbert subspace $\calH\subseteq\calH'$, denoted as $H|_{\mathsmaller{\calH}}:\calH\mapsto\calH$, which is a self-adjoint operator defined on $\calH$ such that $\langle\phi|(H|_{\mathsmaller{\calH}})|\psi\rangle=\langle\phi|H|\psi\rangle$, $\forall\phi,\psi\in\calH$, in other words, $H|_{\mathsmaller{\calH}}=P_{\mathsmaller{\calH}}HP_{\mathsmaller{\calH}}$, now as an operator defined over the entire $\calH'$, where $P_{\mathsmaller{\calH}}$ represents the orthogonal projection onto the closed Hilbert subspace $\calH$. Conversely, the Hamiltonian $H$ is called an extension of $H|_{\mathsmaller{\calH}}$ from the Hilbert space $\calH$ to a Hilbert superspace $\calH'$.

Let $\calD$ be an open subset of the configuration space $\calC$, and $\Diri(\calD,p)\defeq\{\psi\in L^p(\calC):\psi(q)=0,\,\forall q\in\calC\setminus\calD\}$, $p\in\mathbb{R}$, $p\ge 1$ be the closed subspace of  $L^p$-integrable and Dirichlet boundary-conditioned functions. In particular, with $p$ defaulting to the most common value $2$, let $\Diri(\calD)\defeq\Diri(\calD,p=2)$ denote the closed Hilbert subspace of Dirichlet boundary-conditioned functions, called the Dirichlet space on $\calD$. One particularly useful restriction of a Hamiltonian $H\in\calL_0(\calC)$ is $H|_{\mathsmaller{\Diri(\calD)}}=P_{\mathsmaller{\Diri(\calD)}}HP_{\mathsmaller{\Diri(\calD)}}$, or written as $H|_{\mathsmaller{\calD}}=P_{\mathsmaller{\calD}}HP_{\mathsmaller{\calD}}$ in a shorthand manner, where $\cdot|_{\mathsmaller{\calD}}\defeq\cdot|_{\mathsmaller{\Diri(\calD)}}$ and $P_{\mathsmaller{\calD}}\defeq P_{\mathsmaller{\Diri(\calD)}}$ denote respectively a restriction and the orthogonal projection to the Dirichlet space $\Diri(\calD)$. Formally, $H|_{\mathsmaller{\calD}}$ can be represented as $H|_{\mathsmaller{\calD}}=H+V_{\mathsmaller{\calD}}(q)$, $q\in\calC$, interpreted as an extended operator form sum \cite{Kato80,Blanchard15}, where $V_{\mathsmaller{\calD}}(q)=0$, $\forall q\in\calD$, $V_{\mathsmaller{\calD}}(q)=+\infty$, $\forall q\notin\calD$ confines the quantum wave to within the domain $\calD$. A general operator $O\in\calL(\calC)$ is said to satisfy the {\em Hopf lemma} (in $\calC$), also known as the {\em Hopf-Oleinik lemma}, as well as the {\em strong version of the Hopf extremum principle} \cite{Hopf52,Oleinik52,Gilbarg01,Evans10}, if for any given open subset $\calD\subseteq\calC$, and any Dirichlet eigenfunction $\psi\in\Diri(\calD,p)$ of $O$ such that $O\psi = \lambda\psi$ in $\calD$, $\lambda\in\mathbb{K}$, $\mathbb{K}\in\{\mathbb{R},\mathbb{C}\}$, it holds true that $\|\partial\psi(q)\|\neq 0$ for all $q\in\partial\calD$, where $\partial\psi$ represents the vector of first-order partial derivatives of $\psi$, $\|\!\cdot\!\|$ is a certain norm for such vectors, while $\partial\calD$ denotes the boundary of $\calD$; either $p=2$ is implied as conventional, or otherwise, a value for $p\in\mathbb{R}$, $p>1$ will be given explicitly. Conversely, an operator $O\in\calL(\calC)$ is said to satisfy the {\em anti-Hopf lemma} on a specific open subset $\calD\subseteq\calC$, when each Dirichlet eigenfunction $\psi\in\Diri(\calD,p)$, $p\in\mathbb{R}$, $p>1$ of $O$ such that $O\psi = \lambda\psi$ in $\calD$, $\lambda\in\mathbb{K}$, $\mathbb{K}\in\{\mathbb{R},\mathbb{C}\}$ also satisfies the Neumann boundary condition, namely, $\partial\psi(q)=0$ for all $q\in\partial\calD$.

Another useful Hamiltonian restriction has to do with the mathematical description for a system of many identical fermions associated with a configuration space $\calC$, where the governing fermionic Hamiltonian $H_{\!\mathsmaller{F}}$ is a self-adjoint operator on the Hilbert subspace $L_{\!\mathsmaller{F}}^2(\calC)\subseteq L^2(\calC)$, that is always assumed to be derived from $H_{\!\mathsmaller{F}}\defeq H_{\!\mathsmaller{B}}|_{L_{\!\mathsmaller{F}}^2(\calC)}=P_{\!\mathsmaller{F}}H_{\!\mathsmaller{B}}P_{\!\mathsmaller{F}}$, where $H_{\!\mathsmaller{B}}\in\calL_0(\calC)$ is a boltzmannonic Hamiltonian, {\it i.e.}, a Hamiltonian associated with a system of all distinguishable particles, that is exchange-symmetric as $[H_{\!\mathsmaller{B}},P_{\!\mathsmaller{F}}]=0$ on $L^2(\calC)$, while $P_{\!\mathsmaller{F}}\defeq P_{L_{\!\mathsmaller{F}}^2(\calC)}$ is the orthogonal projection from $L^2(\calC)$ onto $L_{\!\mathsmaller{F}}^2(\calC)$, such that $\forall\psi\in L^2(\calC)$, $P_{\!\mathsmaller{F}}\psi\defeq(\prod_{s=1}^{\mathsmaller{S}}n_s!)^{\mathsmaller{-}1/2}\sum_{\pi\in G_{\rm ex}}[\sign(\pi)\times\psi\circ\pi]\in L_{\!\mathsmaller{F}}^2(\calC)$. Such exchange-symmetric $H_{\!\mathsmaller{B}}$ associated with each fermionic Hamiltonian $H_{\!\mathsmaller{F}}$ is called the base boltzmannonic Hamiltonian of $H_{\!\mathsmaller{F}}$, written symbolically as $H_{\!\mathsmaller{B}}=\Base(H_{\!\mathsmaller{F}})$. In the following, some Hamiltonians may have the subscript ``$_{\mathsmaller{B}}$'' or ``$_{\mathsmaller{F}}$'' omitted, when it is unnecessary to make the distinction, or there is no ambiguity in the context as to which is the case. By a slight abuse of notation, let $\Base(H_{\!\mathsmaller{B}})=H_{\!\mathsmaller{B}}$ when the operand $H_{\!\mathsmaller{B}}$ is already a boltzmannonic Hamiltonian, so that the operator $\Base(\cdot)$ may be applied generally to any exchange-symmetric Hamiltonian, either fermionic or boltzmannonic. Although it may not be used frequently in this \iftoggle{ForUSPTO} {specification} {presentation}, the notion of base boltzmannonic Hamiltonian can be straightforwardly generalized to a bosonic Hamiltonian describing a system of multiple identical bosons, as well as a Hamiltonian representing a mixed system consisting of a multitude of particle species, obeying possibly different types of statistics: boltzmannonic, bosonic, or fermionic.

\begin{definition}{(Quantum Physics and Systems)}\label{defiQuantumPhysics}\\
A quantum physics/system, or called a physics/system in short, is an ordered triple $(\calC,\calH,\calB)$, with $\calC$ being a generic continuous-discrete product configuration space, $\calH\defeq\calH(\calC)$ being a Hilbert space of state vectors supported by $\calC$, and $\calB\defeq\calB(\calH)$ being a Banach algebra of bounded linear operators on $\calH$, which contains a strongly continuous semigroup of Gibbs operators $\{\exp({-}\tau H)\}_{\tau}$ parametrized by an imaginary time $\tau\in[0,\infty)\}$ and generated by a lower-bounded and self-adjoint operator $H$, thereby designated as the Hamiltonian measuring the total energy of said quantum system.
\vspace{-1.5ex}
\end{definition}

A lower-bounded and self-adjoint operator will be called a {\em partial Hamiltonian} for brevity. For any partial Hamiltonian $h\in\calL_0(\calH)$, a shorthand notation $h\in\log\calB$ will be used to signify the fact a certain constant $\tau\in\mathbb{R}$, $\tau>0$ exists such that $\exp(-\tau h)\in\calB$.

\begin{definition}{(Dimension and Computational Size of a Quantum Physics/System)}\label{defiSizeQuanPhys}\\
For a given quantum physics/system $(\calC,\calH,\calB)$ associated with a generic continuous-discrete product configuration space $\calC=\calM\times\calP$, with $\calM\cong\prod_{s=1}^{\mathsmaller{S}}\calM_s^{n_s}$ and $\calP\cong\prod_{s=1}^{\mathsmaller{S}}\calP_s^{n_s}$, the dimension of the quantum physics/system is identified with $\dim(\calC)=\sum_{s=1}^{\mathsmaller{S}}n_s[\dim(\calM_s)+\log_2|\calP_s|]$, while the computational size of the quantum physics/system is identified with the computational size of $\calC$ defined as $\size(\calC)\defeq\dim(\calC)+\diam(\calC)$, where $\diam(\calC)\defeq[\diam(\calM)^2+\diam(\calP)^2]^{1/2}$ is called the diameter of $\calC$, with $\diam(\calM)\defeq\sup\{\dist_g(q,r):(q,r)\in\calM^2\}$ and $\diam(\calP)\defeq{\mathlarger{\mathlarger{(}}}\!\sum_{s=1}^{\mathsmaller{S}}|\calP_s|^2{\mathlarger{\mathlarger{)}}}^{1/2}$ being called the diameters of $\calM$ and $\calP$ respectively. The designated Hamiltonian $H\in\log\calB$ measuring the total energy of the quantum physics/system $(\calC,\calH,\calB)$ is also said to have a computational size $\size(H)\defeq\size(\calC)$.
\end{definition}

\begin{definition}{(Homophysical and Isophysical Mapping between Quantum Physics/Systems)}\label{defiHomoIsophysics}\\
Given two quantum physics/systems $(\calC_1,\calH_1,\calB_1)$ and $(\calC_2,\calH_2,\calB_2)$, with each configuration space $\calC_i$, $i\in\{1,2\}$ being associated with an implied measure space $(\calC_i,\calF_i,V_{g_i})$, $\calF_i$ being a $\sigma$-algebra of measurable subsets of $\calC_i$ and containing the supports of all of the wavefunctions in $\calH_i$, $i\in\{1,2\}$, a homophysical mapping $\mathfrak{M}:(\calC_1,\calH_1,\calB_1)\mapsto(\calC_2,\calH_2,\calB_2)$, also called a homophysics from $(\calC_1,\calH_1,\calB_1)$ to $(\calC_2,\calH_2,\calB_2)$, is an injective mapping that sends any measurable set $\calD_1\in\calF_1$ to a measurable set $\calD_2\defeq\mathfrak{M}(\calD_1)\in\calF_2$ which is unique modulo a set of $V_{g_2}$-measure zero, maps any $\psi_1\in\calH_1$ to a unique $\psi_2\defeq\mathfrak{M}(\psi_1)\in\calH_2$, and sends any $O_1\in\calB_1$ to a unique $O_2\defeq\mathfrak{M}(O_1)\in\calB_2$, such that 1) the correspondence $\calF_1\ni\calD\mapsto\mathfrak{M}(\calD)\in\calF_2$ embeds the measure space $(\calC_1,\calF_1,V_{g_1})$ into $(\calC_2,\calF_2,V_{g_2})$ \cite{MeasureSpaceEncycMath,Koushesh14}; 2) $\calH_1\ni\psi\mapsto\mathfrak{M}(\psi)\in\calH_2$ embeds the Hilbert space $\calH_1$ into $\calH_2$; 3) $\calB_1\ni O\mapsto\mathfrak{M}(O)\in\calB_2$ embeds the Banach algebra $\calB_1$ into $\calB_2$; 4) there exists a constant $c\in\mathbb{R}$, $c>0$, called the homophysical constant, with which $\langle\mathfrak{M}(\psi)|\mathfrak{M}(O)|\mathfrak{M}(\phi)\rangle=c\langle\psi|O|\phi\rangle$ holds $\forall\psi,\phi\in\calH_1$, $\forall O\in\calB_1$; 5) in case the physics/systems $(\calC_i,\calH_i,\calB_i)$ are variable and supported by configuration spaces $\calC_i$, $i\in\{1,2\}$ of varying computational sizes, it is further required that $\size(\calC_2)$ should be upper-bounded by a fixed polynomial of $\size(\calC_1)$, and the homophysical constant $c$ satisfy the condition of $c+c^{\mathsmaller{-}1}$ being upper-bounded by another fixed polynomial of $\size(\calC_1)$, where $\size(\calC_1)$ and $\size(\calC_2)$ are the computational sizes of the configuration spaces (as well as of the quantum physics/systems).

A homophysical mapping $\mathfrak{M}:(\calC_1,\calH_1,\calB_1)\mapsto(\calC_2,\calH_2,\calB_2)$ is called isophysical and referred to as an isophysics between $(\calC_1,\calH_1,\calB_1)$ and $(\calC_2,\calH_2,\calB_2)$, when the mapping is also surjective, namely, the three embeddings $\calF_1\ni\calD\mapsto\mathfrak{M}(\calD)\in\calF_2$, $\calH_1\ni\psi\mapsto\mathfrak{M}(\psi)\in\calH_2$, and $\calB_1\ni O\mapsto\mathfrak{M}(O)\in\calB_2$ are all isomorphisms. In general, any homophysics $\mathfrak{M}:(\calC_1,\calH_1,\calB_1)\mapsto(\calC_2,\calH_2,\calB_2)$ induces and can be identified with an isophysics between $(\calC_1,\calH_1,\calB_1)$ and $(\mathfrak{M}(\calC_1),\mathfrak{M}(\calH_1),\mathfrak{M}(\calB_1))$, with the latter called a subphysics of $(\calC_2,\calH_2,\calB_2)$, when $\mathfrak{M}(\calC_1)\defeq\{\mathfrak{M}(q):q\in\calC_1\}\subseteq\calC_2$, $\mathfrak{M}(\calF_1)\defeq\{\mathfrak{M}(D):D\in\calF_1\}\subseteq\calF_2$, $\mathfrak{M}(\calH_1)\defeq\{\mathfrak{M}(\psi):\psi\in\calH_1\}\subseteq\calH_2$, and $\mathfrak{M}(\calB_1)\defeq\{\mathfrak{M}(O):O\in\calB_1\}\subseteq\calB_2$.
\vspace{-1.5ex}
\end{definition}

Under a homophysics $\mathfrak{M}:(\calC_1,\calH_1,\calB_1)\mapsto(\calC_2,\calH_2,\calB_2)$, the image $(\calC_2,\calH_2,\calB_2)$ is said to represent or implement the preimage $(\calC_1,\calH_1,\calB_1)$ through $\mathfrak{M}$. The symbols $\HomoPhysR$ and $\HomoPhysL$ will be used such that, $\psi_1\HomoPhysR\psi_2$ or $\psi_2\HomoPhysL\psi_1$ and $O_1\HomoPhysR O_2$ or $O_2\HomoPhysL O_1$ indicate that a $\psi_1\in\calH_1$ and an $O_1\in\calB_1$ are homophysically mapped to $\psi_2=\mathfrak{M}(\psi_1)\in\calH_2$ and $O_2=\mathfrak{M}(O_1)\in\calB_2$ through the homophysics $\mathfrak{M}$, or in other words, that $\psi_2=\mathfrak{M}(\psi_1)$ and $O_2=\mathfrak{M}(O_1)$ implement $\psi_1\in\calH_1$ and $O_1\in\calB_1$ through the homophysics $\mathfrak{M}$. When said $\mathfrak{M}$ is an isophysics, the two physics/systems are said to be isophysical to each other, and represent or implement each other via the isophysics $\mathfrak{M}$ or $\mathfrak{M}^{\mathsmaller{-}1}$, and the symbol $\IsoPhysLR$ will be used such that, $\psi_1\IsoPhysLR\psi_2$ and $O_1\IsoPhysLR O_2$ state that $\psi_1\in\calH_1$ and $O_1\in\calB_1$ are isophysically equivalent to $\psi_2=\mathfrak{M}(\psi_1)\in\calH_2$ and $O_2=\mathfrak{M}(O_1)\in\calB_2$ via the isophysics $\mathfrak{M}$.

For any partial Hamiltonian $H\in\calB(\calC)$, the operators $\{\exp({-}\tau H)\}$ are well defined $\forall\tau\in[0,\infty)$ and compact $\forall\tau\in(0,\infty)$, which form a strongly continuous semigroup indexed by $\tau\in[0,\infty)$ \cite{Yosida80,Engel00}. Such $\{\exp({-}\tau H)\}_{\tau\ge 0}$ are called the Gibbs operators or thermal density matrices generated by $H$. For any fermionic partial Hamiltonian $H_{\!\mathsmaller{F}}=P_{\!\mathsmaller{F}}H_{\!\mathsmaller{B}}P_{\!\mathsmaller{F}}$, the corresponding Gibbs operator is simply $\exp({-}\tau H_{\!\mathsmaller{F}})=\exp({-}\tau P_{\!\mathsmaller{F}}H_{\!\mathsmaller{B}}P_{\!\mathsmaller{F}})=P_{\!\mathsmaller{F}}\exp({-}\tau H_{\!\mathsmaller{B}})P_{\!\mathsmaller{F}}$. Therefore, it is without loss of generality to focus on Gibbs operators generated by boltzmannonic partial Hamiltonians.

\begin{definition}{(Gibbs Transition Amplitudes, Gibbs Kernels, Gibbs Wavefunctions)}\label{defiGibbsWaveFunc}\\
For any given $(r,q)\in\calC^2$, the point value $\langle r|\exp({-}\tau H)|q\rangle$ of a Gibbs operator $\exp({-}\tau H)$, $\tau>0$ is called the Gibbs transition amplitude between the configuration points $q$ and $r$. The bivariate function $\langle\cdot|\exp({-}\tau H)|\cdot\rangle : (r,q)\in\calC^2 \mapsto \langle r|\exp({-}\tau H)|q\rangle \in \mathbb{K}$ is called a Gibbs kernel, also known as a thermal Green's function, or a thermal density matrix. Acting on a prescribed initial state $\phi\in\calH(\calC)$ as a ket vector, the Gibbs operator $\exp({-}\tau H)$ produces another wavefunction $\psi = \exp({-}\tau H)\phi\in\calH(\calC)$, such that
$$\psi(r)=\langle r|\exp({-}\tau H)|\phi\rangle \,\defeq\, \int_{q\,\in\,\calC}\,\langle r|\exp({-}\tau H)|q\rangle\,\phi(q)\,dq \,, \; \forall r\in\calC ,$$
which is called a forward Gibbs wavefunction, with $H$, $\tau$, $\phi$ being implied and understood. By the self-adjointness of $H$ and $\exp({-}\tau H)$, the Gibbs operator acting on $\phi^*$ as a bra vector produces
$$\psi^*(r)=\langle\phi|\exp({-}\tau H)|r\rangle \,\defeq\, \int_{q\,\in\,\calC}\phi^*(q)\,\langle q|\exp({-}\tau H)|r\rangle\,dq \,, \; \forall r\in\calC ,$$
which is called a backward Gibbs wavefunction, specifically the backward counterpart of the forward Gibbs wavefunction $\langle\cdot|\exp({-}\tau H)|\phi\rangle$, with $H$, $\tau$, $\phi$ being implied and understood.
\vspace{-1.5ex}
\end{definition}

In particular, if $\lambda_0\in\mathbb{R}$ is the smallest eigenvalue of a self-adjoint operator $H\in\calB(\calC)$, then the ground state energy of the shifted partial Hamiltonian $H'\defeq H-\lambda_0$ is zero, and the large-$\tau$ limit of the forward Gibbs wavefunction $\exp({-}\tau H)\phi$, that is $\lim_{\tau\rightarrow+\infty}\exp({-}\tau H)\phi$, coincides with the ground state of $H'$ as well as $H$, so long as the initial state $\phi\in\calH(\calC)$ has a non-vanishing overlap with said ground state. It is for this reason that the ground state of a partial Hamiltonian is considered as a special case, or more precisely, the large-$\tau$ limiting case, of a forward Gibbs wavefunction. Similarly, the restriction of a Gibbs kernel $\langle\cdot|\exp({-}\tau H)|q\rangle : r\in\calC \mapsto \langle r|\exp({-}\tau H)|q\rangle \in \mathbb{K}$ for any fixed $q\in\calC$ is considered to be a forward Gibbs wavefunction, as the limit of a sequence $\{\langle r|\exp({-}\tau H)|\phi\rangle\}$ associated with a sequence $\{\phi\}\subseteq\calH(\calC)$ of wavefunctions, whose supports approach the point set $\{q\}\subseteq\calC$. The self-adjointness of $H$ and $\exp({-}\tau H)$ puts backward and forward Gibbs wavefunctions on the equal footing, in that every property or statement for a forward Gibbs wavefunction holds true for its backward counterpart, moreover, when the scalar field $\mathbb{F}$ contains real numbers only, each backward Gibbs wavefunction $\langle\phi|\exp({-}\tau H)|\cdot\rangle$ as an element in $\calH(\calC)$ becomes the same as its forward counterpart $\langle\cdot|\exp({-}\tau H)|\phi\rangle$. In the following, the adjective ``forward'' or ``backward'' is mostly omitted when there is no risk of ambiguity or no need to differentiate between the two cases, and the generic noun ``Gibbs wavefunction'' will be used to refer to either a ground state wavefunction as a limiting case, or a general Gibbs wavefunction at a finite imaginary time $\tau\in[0,\infty)$, or a restriction of a Gibbs kernel having one of the two configuration coordinates fixed, when there is no risk of ambiguity or it is clear from the context which is the case.

For linear operators acting on a finite-dimensional vector space $L^2(\calP)$, the mathematical formulations and analyses are more or less straightforward. In contrast, significantly more mathematical sophistication is needed to treat quantum motions in a continuous space up to satisfying generality and rigor. In particular, Gibbs operators generated by partial Hamiltonians associated with continuous dynamical variables are best represented using the formulation of path integral \cite{Feynman65}, also known as functional integral \cite{Simon79,Glimm87,Lorinczi11}, based on a Gibbs measure $\mu_{\mathsmaller{G}}$ with respect to a reference measure $\mu_{\mathsmaller{R}}$, on a $\sigma$-algebra $\calF_{\tau}$ of the Wiener space \cite{Ito74,Kuo75,Karatzas91,Bar11} of continuous functions/paths $W_{\tau}\defeq C^0([0,\tau];\calC)\defeq\{\gamma:[0,\tau]\mapsto\calC,\gamma\;\mbox{is\;continuous}\}$, $\tau\ge 0$, where $\calF_{\tau}$ contains the uniform topology induced by the distance function $\dist_{\infty}(\cdot,\cdot):W_{\tau}\times W_{\tau}\mapsto\mathbb{R}$, such that $\forall\gamma_1,\gamma_2\in W_{\tau}$, $\dist_{\infty}(\gamma_1,\gamma_2)\defeq\sup_{0\le\tau'\le\tau}\dist_g(\gamma_1(\tau'),\gamma_2(\tau'))$, and the reference measure $\mu_{\mathsmaller{R}}$ is strictly positive, with respect to which the Gibbs measure $\mu_{\mathsmaller{G}}$ is absolutely continuous, so that the Radon-Nikodym derivative $d\mu_{\mathsmaller{G}}/d\mu_{\mathsmaller{R}}$ exists and is often written as $d\mu_{\mathsmaller{G}}/d\mu_{\mathsmaller{R}}=\exp\{-U_{\!\mathsmaller{G}}[\gamma(\cdot)]\}$, where $U_{\!\mathsmaller{G}}[\gamma(\cdot)]$ is a possibly complex-valued Gibbs energy functional depending on $\gamma(\cdot)\in W_{\tau}$.

When the discrete space $\calP$ contains more than one points, the configuration space $\calC=\calM\times\calP$ comprises mutually disconnected components $\{\calM\times\{v\}:v\in\calP\}$, such that $\calC=\bigcup\{\calM\times\{v\}:v\in\calP\}$, and each component $\calM\times\{v\}$, $v\in\calP$ is a connected continuous submanifold that is isomorphic to $\calM$. Given a boltzmannonic partial Hamiltonian $H_{\!\mathsmaller{B}}\in\calL_0(\calM\times\{v\})$ supported by any component $\calM\times\{v\}$, $v\in\calP$ as a connected submanifold itself, the associated Gibbs measure facilitates a representation of the Gibbs operator $\exp({-}\tau H_{\!\mathsmaller{B}})$, $\tau\ge 0$ via path integral, such that
\begin{align}
\langle\phi|\exp({-}\tau H_{\!\mathsmaller{B}})|\psi\rangle &= \int_{\calM}\int_{\calM}\int_{W^v_{\tau}}\phi^*(r,v)\psi(q,v)\exp\{-U_{\!\mathsmaller{G}}[\gamma(\cdot)]\}d\mu_{\mathsmaller{R}}^{\mathsmaller{[q,r]}}(\gamma)dV_g(q)dV_g(r) \\[0.75ex]
&= \int_{\calM}\int_{\calM}\int_{W^v_{\tau}}\phi^*(r,v)\psi(q,v)d\mu_{\mathsmaller{G}}^{\mathsmaller{[q,r]}}(\gamma)dV_g(q)dV_g(r),\;\forall\phi,\psi\in L^2(\calM\times\{v\}) \,, \nonumber
\end{align}
where $(q,v)$ and $(r,v)$ denote typical points on $\calM\times\{v\}$, with $q,r\in\calM$, $\mu_{\mathsmaller{G}}^{\mathsmaller{[q,r]}}$ is a Gibbs measure, with the reference measure $\mu_{\mathsmaller{R}}^{\mathsmaller{[q,r]}}$ being a conditional Wiener measure, closely related to the so-called Brownian bridge measure \cite{Lorinczi11,Karatzas91,Bar11}, on the subset of constrained continuous paths $\{\gamma(\cdot)\in W^v_{\tau}\defeq C^0([0,t];\calM\times\{v\}):\gamma(0)=(q,v),\,\gamma(\tau)=(r,v)\}$, and the boltzmannonic Hamiltonian $H_{\!\mathsmaller{B}}$ being sufficiently regular to guarantee the existence of the $U_{\!\mathsmaller{G}}[\gamma(\cdot)]$ functional whose real part is lower-bounded. It is clear that $\exp({-}\tau H_{\!\mathsmaller{B}})$, $\tau>0$ is an integral operator, whose integral kernel is precisely the Gibbs kernel, defined as
\begin{equation}
K(\tau H_{\!\mathsmaller{B}};r,q;v)\defeq\langle(r,v)|\exp({-}\tau H_{\!\mathsmaller{B}})|(q,v)\rangle\defeq\int_{W^v_{\tau}}
\exp\{-U_{\!\mathsmaller{G}}[\gamma(\cdot)]\}d\mu_{\mathsmaller{R}}^{\mathsmaller{[q,r]}}(\gamma),\;\forall q,r\in\calM \,,
\label{FeynmanKacKHB}
\end{equation}
such that a typical Gibbs wavefunction reads
\begin{equation}
\left[\exp({-}\tau H_{\!\mathsmaller{B}})\psi\right]\!(r,v)=\int_{\calM}K(\tau H_{\!\mathsmaller{B}};r,q;v)\,\psi(q,v)\,dV_g(q),\;\forall\psi\in L^2(\calM\times\{v\}) \,.
\end{equation}

Given a connected open subset $\calD\subseteq(\calM\times\{v\})\subseteq\calC$, a partial Hamiltonian $H\in\calL_0(\calC)$, being boltzmannonic or otherwise, and an associated Gibbs operator $\exp({-}\tau H)$, $\tau>0$ can be restricted to $\calD$ and subject to the Dirichlet boundary condition on the boundary $\partial\calD$, denoted as $H|_{\Diri(\calD)}$ and $\exp[-\tau H|_{\Diri(\calD)}]$ respectively, by defining and evaluating the Gibbs kernel as
\begin{align}
K(\tau H|_{\Diri(\calD)};r,q) \,\defeq\,& \langle(r,v)|\exp({-}\tau H|_{\mathsmaller{\Diri(\calD)}})|(q,v)\rangle \nonumber \\[0.75ex]
\,\defeq\,& \int_{W_{\tau}}\!\exp\{-U_{\!\mathsmaller{G}}[\gamma(\cdot)]\} \Iver{\gamma\subseteq\calD} d\mu_{\mathsmaller{R}}^{\mathsmaller{[q,r]}}(\gamma) \,,
\label{FeynmanKacKHDiri}
\end{align}
$\forall q,r\in\calM$, where $\Iver{\gamma\subseteq\calD}$ is an Iverson bracket \cite{IversonBracket} which returns a number valued to $1$ if the path $\gamma(\cdot)\subset W_{\tau}$ is fully contained in $\calD$, and $0$ otherwise. Evidently, $K(\tau H|_{\Diri(\calD)};r,q)=0$ whenever $(q,v)\not\in\calD$ or $(r,v)\not\in\calD$, $\forall\tau>0$. More generally, for any partial Hamiltonian $H\in\calL_0(\calC=\calM\times\calP)$ and any topologically open submanifold $\calE\subseteq\calC$, the restriction of $H$ to $\calE$ is denoted and defined as $H|_{\calE}=P_{L^2(\calE)}HP_{L^2(\calE)}$, with $P_{L^2(\calE)}:L^2(\calC)\mapsto L^2(\calE)$ being an orthogonal projection such that
\vspace{-0.5ex}
\begin{equation}
\forall\psi\in L^2(\calC),\;[P_{L^2(\calE)}\psi](q,v)
= \left\{\!\! \begin{array}{rl}
\psi(q,v) \,, &\!\! \mbox{if}\;(q,v)\in\calE , \\
0 \,, &\!\! \mbox{\rm otherwise} \,,
\end{array} \right. \label{defiPsiRestriction}
\end{equation}
\begin{equation}
\lim_{\tau\,\rightarrow\,{+}0} \frac{ \langle(q',v')|\exp({-}\tau H|_{\calE})|(q,v)\rangle  - \langle(q',v')|(q,v)\rangle } { \langle(q',v')|\exp({-}\tau H)|(q,v)\rangle - \langle(q',v')|(q,v)\rangle } =
\left\{\!\! \begin{array}{rl}
1 \,, &\!\! \mbox{if}\;\{(q,v),(q',v')\}\subseteq\calE , \\
0 \,, &\!\! \mbox{\rm otherwise} \,.
\end{array} \right. \label{defiHRestriction}
\end{equation}

\begin{definition}{(Irreducible Partial Hamiltonians and $H$-Connected, $H$-Closed, and $H$-complete Submanifolds)}\label{defiHConnectHClosed}\\
Given a submanifold $\calD\subseteq\calC$ of a configuration space $\calC=\calM\times\calP$, where $\calM$ is a differentiable manifold and $\calP$ is a finite discrete space, a partial Hamiltonian $H\in\calL_0(\calC)$ is called irreducible in $\calD$, when its base boltzmannonic Hamiltonian $\Base(H)$ is irreducible in $\calD$, when in turn the semigroup $\{\exp({-}\tau\Base(H)):\tau\ge 0\}$ is irreducible in $\calD$, in the sense that, for any $V_g$-measurable subsets $\calD_1\subseteq\calD$, $\calD_2\subseteq\calD$ such that $V_g(\calD_1)>0$, $V_g(\calD_2)>0$, there always exist wavefunctions $\psi_1,\psi_2\in L^2(\calD)$ and a constant $\tau>0$ such that
\begin{equation}
\int_{\calD_1}\int_{\calD_2}\langle\psi_2(q_2)|\exp({-}\tau\Base(H))|\psi(q_1)\rangle\,dV_g(q_1\!\in\!\calD_1)\,dV_g(q_2\!\in\!\calD_2)\neq 0 \,,
\end{equation}
in which case, the submanifold $\calD$ is said to be $H$-connected. On the other hand, given a partial Hamiltonian $H\in\calL_0(\calC)$, a submanifold $\calD\subseteq\calC$ is called $H$-closed if, for any $V_g$-measurable subsets $\calD_1\subseteq\calD$, $\calD_2\subseteq(\calC\setminus\calD)$, and for any $\psi_1\in L^2(\calD_1),\psi_2\in L^2(\calD_2)$, $\forall\tau>0$, it holds identically that
\begin{equation}
\int_{\calD_1}\int_{\calD_2}\langle\psi_2(q_2)|\exp({-}\tau\Base(H))|\psi(q_1)\rangle\,dV_g(q_1\!\in\!\calD_1)\,dV_g(q_2\!\in\!\calD_2)=0 \,.
\end{equation}
A submanifold $\calD\subseteq\calC$ is called $H$-complete when it is both $H$-connected and $H$-closed.
\vspace{-1.5ex}
\end{definition}

It is almost a trivial remark that the above-defined notions of $H$-connectedness, $H$-closedness, and $H$-completeness apply equally well when the concerned partial Hamiltonian $H$ is restricted to and involves only a small number of particles or dynamical variables constituting a subsystem of a large physical system in the background. Such situation arises naturally, for example, when $H$ is one of many additive partial Hamiltonians comprising a total Hamiltonian that governs a large physical system associated with a large-sized configuration space $\calC$.

\begin{definition}{(Lie-Trotter-Kato Decomposition of Partial Hamiltonians)}\label{defiLieTrotterKato}\\
A partial Hamiltonian $H$ is called Lie-Trotter-Kato decomposable (or Lie-Trotter-Kato decomposed) \cite{Trotter59,Chernoff68,Chernoff70,Chernoff74,Kato74,Kato78} into $H=\sum_{k=1}^{\mathsmaller{K}}H_k$, $K\in\mathbb{N}$, either as a conventional algebraic operator addition, or more generally, as an extended operator form sum, with each $H_k$, $k\in[1,K]$ being a partial Hamiltonian, such that, as $n\!\rightarrow\!{+}\infty$, $\left(\textstyle{\prod_{k=1}^{\mathsmaller{K}}}\exp(\mathsmaller{-}H_k/n)\right)^n\!$ converges to $\exp(\mathsmaller{-}H)$ in the strong operator topology. The corresponding expression $H=\sum_{k=1}^{\mathsmaller{K}}H_k$, $K\in\mathbb{N}$ is called a Lie-Trotter-Kato decomposition, and $\exp(\mathsmaller{-}H)=s\mbox{-}\lim_{\,n\mathsmaller{\rightarrow}{+}\infty}\left(\textstyle{\prod_{k=1}^{\mathsmaller{K}}}\exp(\mathsmaller{-}H_k/n)\right)^n\!$ is called a Lie-Trotter-Kato product formula, with ``$s\mbox{-}\lim$\!'' indicating limit in the strong operator topology.
\end{definition}

\begin{definition}{(Multiplicatively Coupled Functions or Operators)}\label{defiMulCoup}\\
Given an algebra $\calA$ of functions or operators, and two subalgebras $\calA_1\subseteq\calA$ and $\calA_2\subseteq\calA$, an element $g\in\calA$ is called $\calA_1\calA_2$-multiplicatively coupled when it can be represented as $g=\sum_{i\in\calI}f_{1i}f_{2i}$, with $\{f_{1i}:i\in\calI\}\subseteq\calA_1$, $\{f_{2i}:i\in\calI\}\subseteq\calA_2$, where $\calI$ is an index set, which can be countably infinite, with a suitable topology defined on $\calA$ to make sense of the infinite sum.
\vspace{-1.5ex}
\end{definition}

Let $\calL(\calM)$ and $\calL(\calP)$ be two algebras of linear operators supported by $\calM$ and $\calP$ respectively. Use {\em CD-multiplicatively coupled} as a shorthand abbreviation for an operator $O\in\calL(\calM\times\calP)$ that is $\calL(\calM)\calL(\calP)$-multiplicatively coupled. Armed with Lie-Trotter-Kato product formulas and well known techniques of perturbation theory, it is without loss of generality for us to focus considerations on a particular class of Hamiltonians that are CD-separately moving.

\begin{definition}{(CD-Separately Moving and Irreducible Partial Hamiltonians)}\label{defiCDSepMove}\\
A CD-multiplicatively coupled partial Hamiltonian $H\in\calL_0(\calM\times\calP)$ is called CD-separately moving if it can be written as $H=H^{\mathsmaller{C}}+H^{\mathsmaller{D}}$, with $H^{\mathsmaller{C}}\defeq\sum_{i\in\calI}H_i^{\mathsmaller{C}}V_i^{\mathsmaller{D}}$ and $H^{\mathsmaller{D}}\defeq\sum_{j\in\calJ}V_j^{\mathsmaller{C}}H_j^{\mathsmaller{D}}$ being called the continuous and the discrete parts of $H$ respectively, where $\forall i\in\calI$ and $\forall j\in\calJ$, with the index sets $\calI$ and $\calJ$ being either finite or countably infinite, it holds that $H_i^{\mathsmaller{C}},V_j^{\mathsmaller{C}}\in\calL_0(\calM)$ and $V_i^{\mathsmaller{D}},H_j^{\mathsmaller{D}}\in\calL_0(\calP)$, while $V_i^{\mathsmaller{C}}$ is $\calM$-diagonal and $V_j^{\mathsmaller{D}}$ is $\calP$-diagonal. A partial Hamiltonian $H\in\calL_0(\calM\times\calP)$ is called CD-separately irreducible if it is both CD-separately moving and irreducible in the configuration space $\calM\times\calP$.
\vspace{-1.5ex}
\end{definition}

A discrete part $H^{\mathsmaller{D}}$ of a partial Hamiltonian $H$ occurs naturally, for example, in dealing with a physical system consisting of particles with spins, or in describing quantum tunneling of electrons between nearby nano-structures such as quantum wells, quantum wires, and quantum dots, or in a model of computational physics that involves a discrete dynamical variable, such as a Hubbard model or a lattice field theory. Particularly, in the context of GSQC \cite{Feynman85,Kitaev02,Mizel04,Kempe06}, where a quantum computational algorithm is reduced to solving for the ground state of a designer Hamiltonian $H$ that involves more than one identical particles of a certain species, the bosonic or fermionic exchange symmetry, as a rigid and nonnegotiable requirement imposed on the wavefunction describing the collective state of said identical particles, can be enforced equivalently by superimposing a discrete part $H^{\mathsmaller{D}}=\sum_{\pi\in G_0}[I-\sign(\pi)\times\pi]$ onto the designer Hamiltonian $H$, which induces an energy penalty to prevent any violation of exchange symmetry in the ground state, where $G_0$ denotes a generating set of the exchange symmetry group $G_{\rm ex}$ associated with the concerned species of identical particles, $\sign(\pi)$ is the sign of a particle exchange operation $\pi\in G_{\rm ex}$. In one exemplary embodiment, $G_0$ can be the set of the adjacent transpositions \cite{Sagan01}, which swap pairs of identically particles with artificially assigned indices $i$ and $i{+}1$, $i\in[1,n{-}1]$, $n\in\mathbb{N}$ being the total number of identical particles of the concerned species. Since $G_{\rm ex}\ni\pi\mapsto\sign(\pi)\in\{+1,-1\}$ is a group homomorphism between $G_{\rm ex}$ and the group $(\{+1,-1\},*)$, consistency with the exchange symmetry conditions for permutations in the generating group $G_0$ guarantees that the ground state of $H^{\mathsmaller{D}}$ satisfies the requirement of exchange symmetry for all permutations $\pi\in G_{\rm ex}$. The benefit of having $\pi$ to traverse a generating set $G_0$ instead of the full set $G_{\rm ex}$ is to reduce the number of additive terms in $H^{\mathsmaller{D}}$ from $|G_{\rm ex}|=n!$ to a polynomial of $n$, and specifically $n{-}1$ when $G_0$ is the set of adjacent transpositions. One advantage of enforcing exchange symmetries via discrete parts of Hamiltonians like $H^{\mathsmaller{D}}=\sum_{\pi\in G_0}[I-\sign(\pi)\times\pi]$ is to relax the working domain of state vectors supported by a configuration space $\calC$ from the fermionic Hilbert space $L_{\!\mathsmaller{F}}^2(\calC)$ to the boltzmannonic Hilbert space $L^2(\calC)$, knowing that the ground state of $H^{\mathsmaller{D}}$ must fulfill the symmetry requirements.

Using the so-called perturbation gadgets \cite{Kempe06,Biamonte08,Jordan08,Bravyi08prl,Cao15}, any CD-multiplicatively coupled partial Hamiltonian can be Lie-Trotter-Kato decomposed into a sum of CD-separately moving partial Hamiltonians, since, any monomial $AB$ with $A\in\calL(\calM)$, $B\in\calL(\calP)$ as an additive term of a CD-multiplicatively coupled partial Hamiltonian has a Lie-Trotter-Kato decomposition
\begin{equation}
e^{{-}AB}=s\mbox{\,-}\!\!\!\lim_{n\mathsmaller{\rightarrow}{+}\infty}\left[e^{(A/n)-(A/n)^2}e^{{-}(A+B)/n}e^{(B/n)-(B/n)^2}\right]^{n^2} \!,
\end{equation}
provided that all of the relevant Gibbs operators are well defined, especially when the operator $A\in\calL(\calM)$ is unbounded. CD-multiplicatively coupled operators form a large class, that includes most of the naturally occurring and practically relevant physical interactions in real or model systems. More importantly, as will be clearly seen later, a physical system governed by a CD-separately moving Hamiltonian can already be made homophysically universal, implementing a quantum computer that can be programmed to simulate any given physical system. It is for these reasons that focusing considerations on CD-separately moving partial Hamiltonians entails no loss of generality. In this \iftoggle{ForUSPTO} {specification} {presentation}, any Hamiltonian $H\in\calL(\calC)$ supported by any continuous-discrete product configuration space $\calC=\calM\times\calP$ is always assumed to be CD-separately irreducible.

For all known naturally occurring physics systems, as well as most theoretical model systems that are employed to implement or represent quantum computing \cite{Kitaev02,Kempe06,Cubitt13}, it is usually the case that a many-body Hamiltonian is {\em computationally local} and of an additive form as $H=\sum_{k=1}^{\mathsmaller{K}}h_k$, $K\in\mathbb{N}$, where each partial Hamiltonian $h_k$, $k\in[1,K]$ acts nontrivially only on a small number of artificially labeled particles, that belong to a necessarily small number of fermion species, and are said to be moved by the corresponding $h_k$, while the other particles of the whole system are said to be fixed by $h_k$. More specifically, the following definitions are in order.

\begin{definition}{($h$-Moved and $h$-Fixed Particle Species)}\label{defiHMovedHFixedSpecies}\\
With respect to a quantum physics/system $(\calC,\calH,\calB)$ associated with a generic continuous-discrete product configuration space $\calC=\prod_{s=1}^{\mathsmaller{S}}\calC_s^{n_s}$, $S\in\mathbb{N}$, $n_s\in\mathbb{N}$ for all $s\in[1,S]$, a partial Hamiltonian $h\in\calL_0(\calC)$ is said to move particles of the first species, if a tensor product state $\psi\defeq\psi_1\otimes\psi'_1\in\calH(\calC)$ exists, with $\psi_1\in L^2(\calC_1^{n_1})$, $\psi'_1\in L^2(\calC'_1)$, $\calC'_1\defeq\prod_{s=2}^{\mathsmaller{S}}\calC_s^{n_s}$, for which no wavefunction $\phi'_1\in L^2(\calC'_1)$ exists to satisfy $h\psi=\psi_1\otimes\phi'_1$, $\|\phi'_1\|\neq 0$. More generally, $\forall s\in[1,S]$, a partial Hamiltonian $h\in\calL_0(\calC)$ is said to move particles of the $s$-th species, if by swapping the indices $1$ and $s$, the partial Hamiltonian $h$ is turned into $h'\in\calL_0(\calC')$ which moves particles of the first species, with respect to the new quantum physics/system $(\calC',\calH',\calB')$ that has the particle species relabeled. A partial Hamiltonian $h\in\calL_0(\calC)$ is said to fix particles of the $s$-th species, $s\in[1,S]$, if it does not move particles of the $s$-th species, with respect to said quantum physics/system $(\calC,\calH,\calB)$.
\end{definition}

\begin{definition}{(Few-Body-Moving Interactions)}\label{defiFBMTF}\\
With respect to the same quantum physics/system $(\calC,\calH,\calB)$ as above, a partial Hamiltonian $h\in\log\calB$ is called an $(s_m,n_m,d_m)$-few-body-moving (FBM) interaction, $(s_m,n_m,d_m)\in\mathbb{N}^3$, when it moves a number $s_m$ of particle species and fix all of the other $(S-s_m)$ species, where each moved species labeled by an $s\in[1,S]$ has no more than $n_m$ identical particles that live in a substrate space $\calC_s$ with a substrate space dimension $\dim(\calC_s)\le d_m$. When the quantum physics/system $(\calC,\calH,\calB)$ is variable in size due to a varying $S\in\mathbb{N}$ number of species, a partial Hamiltonian $h\in\log\calB$ is called a logarithmically FBM interaction if it is an $(s_m,n_m,d_m)$-FBM interaction with $d_m=O(1)$ and $n_m=O(1)$ being fixed while $s_m=O(\log(\size(\calC)))$.
\vspace{-1.5ex}
\end{definition}

In this \iftoggle{ForUSPTO} {specification} {presentation}, $O(\cdot)$, $\Omega(\cdot)$, and $\Theta(\cdot)$ are the traditional notations of asymptotics in the Knuth convention, representing an upper bound, a lower bound, and a simultaneous upper and lower bound, respectively \cite{Knuth76}. An $(s_m,n_m,d_m)$-FBM interaction may be simply referred to as an FBM interaction in short, with the triple of bounds $(s_m,n_m,d_m)$ implied and clearly understood. Almost all FBM interactions in this \iftoggle{ForUSPTO} {specification} {presentation} are at most logarithmically FBM, namely, $(s_m,n_m,d_m)$-FBM with $s_m=O(\log(\size(\calC)))$ at the most, $n_m$ and $d_m$ being fixed and small. It is obvious that a logarithmically FBM interaction involves no more than $O(\log(\size(\calC)))$ dynamical variables.

Given any FBM interaction $h\in\log\calB$, let
\begin{equation}
\calI_h\defeq\{s\in[1,S]:h\;\textit{moves the}\;s\textit{-th particle species}\}
\end{equation}
be the set of indices labeling particle species that are moved by $h$, and the complement
\begin{equation}
\calI_h^{\mathsmaller{\perp}}\defeq\{s\in[1,S]:h\;\textit{fixes the}\;s\textit{-th particle species}\}
\end{equation}
be the set of indices labeling particle species that are fixed by $h$. Clearly, $\calI_h$ and $\calI_h^{\mathsmaller{\perp}}$ bi-partition the whole set of indices $[1,S]$, in the sense that $\calI_h\cup\calI_h^{\mathsmaller{\perp}}=[1,S]$ and $\calI_h\cap\calI_h^{\mathsmaller{\perp}}=\emptyset$. Define product manifolds $\calE_h\defeq\prod_{s\in\calI_h}\calC_s^{n_s}$ and $\calE_h^{\mathsmaller{\perp}}\defeq\prod_{s\in\calI_h^{\mathsmaller{\perp}}}\calC_s^{n_s}$. The FBM interaction $h$ is said to induce an {\em orthogonal direct sum decomposition} (ODSD) $\calC=\calE_h\oplus\calE_h^{\mathsmaller{\perp}}$, in the sense that there exist atlases $\{(U_{\alpha},\varphi_{\alpha})\}_{\alpha\in A}$ for $\calE_h$ and $\{(U_{\beta}^{\mathsmaller{\perp}},\varphi_{\beta}^{\mathsmaller{\perp}})\}_{\beta\in B}$ for $\calE_h^{\mathsmaller{\perp}}$, $A$ and $B$ being index sets indexing coordinate patches, such that the Cartesian product $\{(U_{\alpha},\varphi_{\alpha})\times(U_{\beta}^{\mathsmaller{\perp}},\varphi_{\beta}^{\mathsmaller{\perp}})\defeq(U_{\alpha}\times U_{\beta}^{\mathsmaller{\perp}},(\varphi_{\alpha},\varphi_{\beta}^{\mathsmaller{\perp}}))\}_{\alpha\in A,\,\beta\in B}$ serves as an atlas for $\calC$; in particular, any $q\in\calC$ has a unique ODSD $q=u\oplus u^{\mathsmaller{\perp}}$, with $u\in\calE_h$, $u^{\mathsmaller{\perp}}\in\calE_h^{\mathsmaller{\perp}}$, such that for any open neighborhood $U$, $q\in U\subseteq\calC$, there exist open neighborhoods $U_{\alpha}$, $\alpha\in A$, $u\in U_{\alpha}\subseteq\calE_h$ and $U_{\beta}^{\mathsmaller{\perp}}$, $\beta\in B$, $u^{\mathsmaller{\perp}}\in U_{\beta}^{\mathsmaller{\perp}}\subseteq\calE_h^{\mathsmaller{\perp}}$ to make $U\subseteq U_{\alpha}\times U_{\beta}^{\mathsmaller{\perp}}$ true. For any $u^{\mathsmaller{\perp}}\in\calE_h^{\mathsmaller{\perp}}$, the subset $\calE_h\oplus u^{\mathsmaller{\perp}}\defeq\{q\in\calC:\exists u\in\calE_h\;{\rm such\;that}\;q=u\oplus u^{\mathsmaller{\perp}}\}$ is called an $\calE_h$-coset at $u^{\mathsmaller{\perp}}$, which is a submanifold of $\calC$ isomorphic to $\calE_h$ in the sense that the two not only are diffeomorphic as smooth manifolds but also endowed with isometric Riemannian metrics. Similarly, $\forall u\in\calE_h$, the subset $u\oplus\calE_h^{\mathsmaller{\perp}}\defeq\{q\in\calC:\exists u^{\mathsmaller{\perp}}\in\calE_h^{\mathsmaller{\perp}}\;{\rm such\;that}\;q=u\oplus u^{\mathsmaller{\perp}}\}$ is called an $\calE_h^{\mathsmaller{\perp}}$-coset at $u\in\calE_h$, which is a submanifold of $\calC$ isomorphic to $\calE_h^{\mathsmaller{\perp}}$.

\begin{definition}($h$-Moved and $h$-Fixed Factor Spaces and Cosets)\label{defiHMovedHFixedCosets}\\
$\calE_h\defeq\prod_{s\in\calI_h}\calC_s^{n_s}$ and $\calE_h^{\mathsmaller{\perp}}\defeq\prod_{s\in\calI_h^{\mathsmaller{\perp}}}\calC_s^{n_s}$ are called the $h$-moved and $h$-fixed factor (configuration) spaces respectively. An $\calE_h$-coset $\calE_h\oplus u^{\mathsmaller{\perp}}$, $u^{\mathsmaller{\perp}}\in\calE_h^{\mathsmaller{\perp}}$ is called $h$-moved if a $\phi\in L^2(\calE_h\oplus u^{\mathsmaller{\perp}})$ exists such that $\phi$ is not an eigenvector of $h$. Conversely, an $\calE_h$-coset $\calE_h\oplus u^{\mathsmaller{\perp}}$, $u^{\mathsmaller{\perp}}\in\calE_h^{\mathsmaller{\perp}}$ is called $h$-fixed if any $\phi\in L^2(\calE_h\oplus u^{\mathsmaller{\perp}})$ is an eigenvector of $h$.
\vspace{-1.5ex}
\end{definition}

It is clear that for any FBM interaction $h$, on any $h$-fixed $\calE_h$-coset $\calE_h\oplus u^{\mathsmaller{\perp}}$, $u^{\mathsmaller{\perp}}\in\calE_h^{\mathsmaller{\perp}}$, $h$ is $(\calE_h\oplus u^{\mathsmaller{\perp}})$-diagonal whose eigenvectors are trivially computable. Two $\calE_h$-cosets $\calE_h\oplus u^{\mathsmaller{\perp}},\calE_h\oplus v^{\mathsmaller{\perp}}$ with $u^{\mathsmaller{\perp}},v^{\mathsmaller{\perp}}\in\calE_h^{\mathsmaller{\perp}}$ are said to be $h$-isomorphic when the effects of $h$ on the two $\calE_h$-cosets are the same up to a scalar constant, that is, there exist $\tau_1,\tau_2\in(0,\infty)$ such that $\mathlarger{\mathlarger{\langle}}u\oplus u^{\mathsmaller{\perp}}\mathlarger{\mathlarger{|}}\exp(-\tau_1h)\mathlarger{\mathlarger{|}}v\oplus u^{\mathsmaller{\perp}}\mathlarger{\mathlarger{\rangle}}=\mathlarger{\mathlarger{\langle}}u\oplus v^{\mathsmaller{\perp}}\mathlarger{\mathlarger{|}}\exp(-\tau_2h)\mathlarger{\mathlarger{|}}v\oplus v^{\mathsmaller{\perp}}\mathlarger{\mathlarger{\rangle}}$ holds true $\forall u,v\in\calE_h$. Obviously, $h$-isomorphism is an equivalence relation among the $\calE_h$-cosets $\{\calE_h\oplus u^{\mathsmaller{\perp}}:u^{\mathsmaller{\perp}}\in\calE_h^{\mathsmaller{\perp}}\}$, which induces a partition of the collection of all $\calE_h$-cosets into equivalence classes $\mathlarger{\mathlarger{\{}}[\calE_h\oplus u^{\mathsmaller{\perp}}]\defeq\{\calE_h\oplus v^{\mathsmaller{\perp}}:\calE_h\oplus v^{\mathsmaller{\perp}}\;\textrm{is}\;h\textrm{-isomorphic to}\;\calE_h\oplus u^{\mathsmaller{\perp}}\}\mathlarger{\mathlarger{\}}}$.

When an FBM interaction $h_k\in\log\calB$, $k\in\mathbb{N}$ is an additive term in a sum $\sum_{k=1}^{\mathsmaller{K}}h_k$, or a tensor-multiplicative factor in a product $\bigotimes_{k=1}^{\mathsmaller{K}}h_k$, $K\in\mathbb{N}$, the sets of indices $\calI_{h_k}$, $\calI_{h_k}^{\mathsmaller{\perp}}$ and the factor spaces $\calE_{h_k}$, $\calE_{h_k}^{\mathsmaller{\perp}}$ may be denoted alternatively as $\calI_k\defeq\calI_{h_k}$, $\calI_k^{\mathsmaller{\perp}}\defeq\calI_{h_k}^{\mathsmaller{\perp}}$ and $\calE_k\defeq\calE_{h_k}$, $\calE_k^{\mathsmaller{\perp}}\defeq\calE_{h_k}^{\mathsmaller{\perp}}$, $k\in[1,K]$, just for convenience.

\begin{definition}{(FBM Tensor Monomials and Polynomials)}\label{defiTensorMonoPoly}\\
Given a quantum physics/system $(\calC,\calH,\calB)$, a partial Hamiltonian $M\in\log\calB$ is called an FBM tensor monomial, when $M$ is in a form of a tensor product $M=I_0\otimes\left(\bigotimes_{i=1}^nh_i\right)$, $n\in\mathbb{N}$, $n\le\size(\calC)$, where $\forall i\in[1,n]$, $h_i\in\log\calB(\calC_i)$ is an FBM interaction moving a factor space $\calC_i$ with $\dim(\calC_i)=O(\log(\size(\calC)))$, while $\{\calC_i\}_{i=1}^n$ Cartesian-factorizes the whole configuration space as $\calC=\calC_0\times\prod_{i=1}^n\calC_i$, and $I_0$ is the identity operator supported by the factor space $\calC_0\defeq\calC/\prod_{i=1}^n\calC_i$, although $\calC_0$ may reduce to an empty set. The number $\deg(M)\defeq\sum_{i=1}^n\dim(C_i)$ is called the degree of the FBM tensor monomial $M$. A partial Hamiltonian $W\in\log\calB$ is called an FBM tensor polynomial, when $W$ can be written as $W=\sum_{k=1}^{\mathsmaller{K}}M_k$, $K\in\mathbb{N}$, $K=O(\poly(\size(C)))$, with each $M_k$, $k\in[1,K]$ being an FBM tensor monomial. The number $\deg(W)\defeq\max\{\deg(M_k):k\in[1,K]\}$ is called the degree of the FBM tensor polynomial $W$.
\vspace{-1.5ex}
\end{definition}

When useful and necessary, the notions of few-body-moving operators, moved or fixed particle species, and moved or fixed factor spaces and cosets can all be straightforwardly generalized to generic operators in a Banach algebra $\calB$ with respect to a given quantum physics/system $(\calC,\calH,\calB)$. For a general operator $Q \in \calB$ having an adjoint operator (or Hermitian conjugate) $Q^+ \in \calB$, the associated operators $Q_+ \defeq (Q + Q^+)/2 \in \calB$ and $Q_- \defeq (Q - Q^+)/2 \in \calB$ are called {\em the Hermitian part and the skew-Hermitian part of $Q$} respectively. Let ${\bf i} \in \calB$ denote an operator in $\calB$ such that ${\bf i}^2 + I \equiv 0$, with $I \in \calB$ being the identity operator, then ${\bf i}A \in \calB$ is self-adjoint or Hermitian, called a {\em Hermitianization of $A$}, for any given skew-Hermitian operator $A \in \calB$.

\begin{definition}{(General Operators Being Few-Body-Moving)}\label{defiGenOperFBM}\\
Given a quantum physics/system $(\calC,\calH,\calB)$, a general operator $Q \in \calB$ is called few-body-moving, when both its Hermitian part $Q_+ = (Q + Q^+)/2 \in \calB$ and the Hermitianization ${\bf i}Q_-$ of its skew-Hermitian part $Q_- = (Q - Q^+)/2 \in \calB$ are few-body-moving. A particle species, or a factor space, or a coset is called $Q$-moved (conversely, $Q$-fixed), when such a particle species, or a factor space, or a coset is either $Q_+$-moved or ${\bf i}Q_-$-moved (both $Q_+$-fixed and $Q_-$-fixed, respectively).
\vspace{-1.5ex}
\end{definition}

For any given partial Hamiltonian $H$ with respect to a quantum physics/system $(\calC,\calH,\calB)$, let $\lambda_0(H)$ denote the lowest eigenvalue of $H$ (namely, the ground state energy), and $\psi_0(H)$ represent either a unique eigenvector (that is, the ground state wavefunction) when there is no degeneracy, or an orthonormal basis of the ground eigenspace $\psi_0(H)\defeq\{\psi_{0,l}(H):1\le l\le\gmul(H,0)\}$, where $\gmul(H,0)$ denotes the geometric multiplicity of $\lambda_0(H)$. Let $\psi_0(H;q)$ specify the point value(s) of the ground state(s) at $q\in\calC$. When $H$ is well understood, it may be omitted from the notations $\lambda_0(H)$, $\psi_0(H)$, and $\psi_0(H;q)$ without ambiguity. By the same token, let $\{\lambda_n(H)\}_{n\ge 1}$ denote the higher eigenvalues/eigenspaces of $H$ in a non-decreasing order, $\{\psi_n(H)\}_{n\ge 1}$ and $\{\psi_n(H;q)\}_{n\ge 1}$ represent the corresponding eigenstates and their point values, provided that degeneracy of states does not create confusion, or it is properly resolved by using additional indexing and ordering in accordance with additional auxiliary quantum numbers.

\begin{definition}{(Polynomially/Efficiently Computable Partial Hamiltonians)}\label{defiEffCompPartHamil}\\
Given a quantum physics/system $(\calC,\calH,\calB)$, a partial Hamiltonian $H\in\log\calB$ is said to be polynomially computable (or efficiently computable, interchangeably), when the Gibbs operators generated by $\Base(H)$ is polynomially computable (or efficiently computable), in the sense that, there exists an $\epsilon_1=\Omega(1/\poly(\size(\calC)))>0$, such that, for any given $\epsilon_2=\Omega(1/\poly(\size(\calC)))>0$, the Gibbs transition amplitude $K(\tau\Base(H);r,q)=\left\langle r\left|\exp[-\tau\Base(H)]\right|q\right\rangle$ can be computed to within a relative error no more than $\epsilon_2$, $\forall\tau\in(0,\epsilon_1]$, $\forall (r,q)\in\calC^2$ such that $\dist g(r,q)\le\epsilon_1$, at a computational cost that is upper-bounded by $O(\poly(\size(\calC)))$.
\vspace{-1.5ex}
\end{definition}

In the context of GSQC, whenever a computational problem is reduced to determining the ground state energy or sampling from the ground state wavefunction of a Hamiltonian $H$ on a subset of a configuration space $\calC$, or in the context of quantum statistical mechanics, whenever a computational problem is reduced to computing a partition function $\int_{q\,\in\,\calC}\,\langle q|\exp(-\tau H)|q\rangle \, dq$ for a fixed $\tau > 0$ or sampling from a Gibbs kernel $\langle r|\exp(-\tau H)|q\rangle$, $(r,q)\in\calC^2$, the Hamiltonian $H$ has to be polynomially computable in order for any hope of solving the problem efficiently. In particular, whenever a many-body Hamiltonian is given as an FBM tensor polynomial $H=\sum_{k=1}^{\mathsmaller{K}}H_k$, $H_k=\bigotimes_{i=1}^{n_k}h_{ki}$, $n_k\in\mathbb{N}$, $\forall k\in[1,K]$, $K\in\mathbb{N}$, it is always assumed that all of the FBM tensor monomials $H_k$, $k\in[1,K]$ and all of the FBM interactions $h_{ki}$, $i\in[1,n_k]$, $k\in[1,K]$ are polynomially computable. Furthermore, in the context of GSQC, the total Hamiltonian $H$ should have a unique ground state and needs to be polynomially gapped, namely, the ground state should be separated from all excited states by an energy gap whose reciprocal is upper-bounded by a polynomial of the problem input size, in order for the problem to be efficiently solvable.

Given a quantum physics/system $(\calC,\calH,\calB)$ of a variable size, and a typical FBM tensor monomial $M\in\log\calB(\calC)$, a sufficient (though not necessary) condition guaranteeing efficient computability of $M$ is that, either $\deg(M)=O(1)$ is a fixed number independent of the varying $\size(\calC)$, or more generally, $M=P\otimes Q$ is a tensor product with both $P$ and $Q$ being efficiently computable, while $P$ is an  projection operator such that $P^2=P$, although $\deg(P)$ could be as large as $O(\poly(\size(C)))$, in which case, it follows from the identity
\begin{equation}
\exp({-}\tau M) = \exp({-}\tau P\otimes Q) = I + P\otimes[\,\exp({-}\tau Q)-I\,],\;\tau\in(0,\infty) \label{TensorMonoMostProj}
\end{equation}
that $M$ is efficiently computable, because both $\exp(-\tau P)=(I-P) + e^{-\tau}P$ and $\exp(-\tau Q)$ are polynomially computable, $\forall \tau = O(\poly(\size(\calC)))$, $\tau > 0$. Also, if $M\in\log\calB(\calC)$ is a tensor product of the form $M=R\otimes Q$, with both $R$ and $Q$ being efficiently computable, where $R$ is a {\em self-inverse operator} such that $R^2 = I$, then both $M$ and $I-M$ are efficiently computable, because
\begin{equation}
I - R\otimes Q = 2R^+\!\otimes Q^- + 2R^-\!\otimes Q^+, ~\mbox{with}~ R^{\pm} \defeq {\textstyle{\half}}(I\pm R), ~ Q^{\pm} \defeq {\textstyle{\half}}(I\pm Q) \,, \label{TensorMonoSelfInv}
\end{equation}
where $R^{\pm}$ are projection operators such that $R^+R^-=R^-R^+=0$, hence $R^+\!\otimes Q^-$ and $R^-\!\otimes Q^+$ commute. Furthermore, if both $R$ and $Q$ are self-inverse operators, then both $\half(I + R\otimes Q)$ and $\half(I - R\otimes Q)$ are efficiently computable projection operators.

\begin{definition}{(Polynomial Lie-Trotter-Kato Decomposition of Partial Hamiltonians)}\label{defiPolyLieTrotterKato}\\
A Lie-Trotter-Kato decomposed partial Hamiltonian $H=\sum_{k=1}^{\mathsmaller{K}}H_k\defeq\sum_{k=1}^{\mathsmaller{K}}\bigotimes_{i=1}^{n_k}h_{ki}$, with $K\in\mathbb{N}$, $K=O(\poly(\size(H)))$, $n_k\in\mathbb{N}$, $n_k=O(\poly(\size(H)))$, $\forall k\in[1,K]$ is called polynomially Lie-Trotter-Kato decomposable (PLTKD), when $H$ is also an FBM tensor polynomial with each $h_{ki}$, $i\in[1,n_k]$, $n_k\in\mathbb{N}$, $k\in[1,K]$ being an efficiently computable FBM interaction, and for each $k\in[1,K]$, all but an $m_k=O(1)$ number of the FBM interactions $\{h_{ki}:i\in[1,n_k]\}$ are projection operators, wherein $\lim_{n\rightarrow\infty}\left(\textstyle{\prod_{k=1}^{\mathsmaller{K}}}\exp(\mathsmaller{-}H_k/n)\right)^n\!$ converges to $\exp(\mathsmaller{-}H)$ polynomially fast with respect to a suitable operator norm, that is, there exist a fixed integer $b\in\mathbb{N}$ and some positive constants $c_1,c_2,c_3\in\mathbb{R}$, all of which being independent of $K$, such that
\begin{equation}
e^{\lambda_0(H)}\scalebox{1.4}{$\|$}\scalebox{1.2}{$($}\,\textstyle{\prod_{k=1}^{\mathsmaller{K}}}e^{\mathsmaller{-}H_k/mK^b}\scalebox{1.2}{$)$}^{mK^b}\!-e^{\mathsmaller{-}H}\scalebox{1.4}{$\|$}\! \;\le\; c_1m^{\mathsmaller{-}c_2}\min[1,\,\lambda_1(H)-\lambda_0(H)] \label{LieTrotterKatoBounds}
\end{equation}
holds $\forall K\in\mathbb{N}$, $K=O(\poly(\size(H)))$, and $\forall m\in\mathbb{N}$ that satisfies $m\ge c_3$ and $mK^b\ge\tau_k^{{-}1}$ for all $k\in[1,K]$, where $\lambda_0(H)$ and $\lambda_1(H)$ denote the lowest and second lowest eigenvalues of $H$, $\|\cdot\|$ can be either the standard operator norm, or more strongly, the Hilbert-Schmidt norm, $\tau_k\in\mathbb{R}$ is a constant associated with $H_k$ such that the Gibbs operator $\exp(-\tau H_k)$ is efficiently computable, $\forall\tau\in(0,\tau_k)$, $\forall k\in[1,K]$, $\forall K\in\mathbb{N}$. Naturally, the corresponding expression $H=\sum_{k=1}^{\mathsmaller{K}}H_k$ is called a polynomial Lie-Trotter-Kato decomposition of $H$.
\vspace{-1.5ex}
\end{definition}

The class of PLTKD partial Hamiltonians contains a subclass of computationally local Hamiltonians, that are FBM tensor polynomials of the form $H=\sum_{k=1}^{\mathsmaller{K}}h_k$, with each $h_k$, $k\in[1,K]$ being an FBM interaction that moves no more than $O(\log(\size(H)))$ particles or dynamic variables. As will be seen later in this \iftoggle{ForUSPTO} {specification} {presentation}, the subclass of computationally local PLTKD Hamiltonians is already complete and universal for GSQC, and it defines the same quantum computational complexity classes as does the superclass of PLTKD Hamiltonians in general, as far as polynomial reduction is concerned. Still, it is advantageous to consider, analyze, and employ the superclass of PLTKD Hamiltonians, not only because they form a larger set and afford wider generality, but also due to the fact that their employment in certain applications leads to faster algorithms, even though the improvement is only polynomial.

When the partial Hamiltonians $H$ and $\{H_k\}_{k\in[1,\mathsmaller{K}]}$ are all bounded, the polynomial Lie-Trotter-Kato decomposability and the asymptotic bound of (\ref{LieTrotterKatoBounds}) can be easily established using power series expansions of matrix/operator functions and matrix/operator norm analyses, which reduces to the well-known Lie product formula \cite{Lie1888,Poincare1900,vonNeumann1929,Hall03}, as a special case of the Baker-Campbell-Hausdorff formula \cite{Campbell1896,Campbell1897,Baker1905,Hausdorff06}, or the Zassenhaus formula \cite{Zassenhaus39}. However, when some of the operators are unbounded, it is more intricate to formulate a Lie-Trotter-Kato product and prove an asymptotic bound of (\ref{LieTrotterKatoBounds}), due to subtleties in applicable domains as well as additivity of operators. Nevertheless, a significant body of theorems has been proved to establish the Lie-Trotter-Kato product formula for a wide variety of operators \cite{Trotter59,Chernoff68,Chernoff70,Chernoff74,Kato74,Kato78}, mostly in the sense of strong operator topology. Recently, the Lie-Trotter-Kato product formula has been proved to converge in the operator norm, and indeed with the asymptotic error polynomially bounded as in (\ref{LieTrotterKatoBounds}) \cite{Rogava93,Helffer95,Ichinose97,Ichinose98,Neidhardt98,Neidhardt99,Ichinose01,Ichinose02,Ichinose09}, especially for the all-important Schr\"odinger operators $H=-\Delta_g+V(q)$ on a Riemannian manifold $(\calC,g)$, with the potential $V(q\!\in\!\calC)$ subject to certain conditions of regularity. Here we do not limit our Hamiltonians to any explicit and specific form, but just posit that they should have a known polynomial Lie-Trotter-Kato decomposition.

Let $(\Omega,\calF,P)$ be a base probability space. For any random variable $X:(\Omega,\calF,P)\mapsto(\Omega',\calF')$, let $P_{\mathsmaller{X}}$ denote the induced probability on the measurable space $(\Omega',\calF')$ such that $P_{\mathsmaller{X}}(A')\defeq P(X^{\mathsmaller{-}1}(A'))$, $\forall A'\in\calF'$, and call $P_{\mathsmaller{X}}$ the probability law of $X$. For any two sub-$\sigma$-algebras $\calA\subseteq\calF$, $\calB\subseteq\calF$, define $\alpha(\calA,\calB)\defeq\sup_{\mathsmaller{A\in\calA,\,B\in\calB}}\{|P(A\cap B)-P(A)P(B)|\}$ as a measure of dependence between $\calA$ and $\calB$ \cite{Rosenblatt56}. Let $\mathbb{I}$ be a totally ordered monoid as an index set with $\min(\mathbb{I})=0$, {\it e.g.}, $\mathbb{I}=\mathbb{Z}_{\ge 0}$ or $\mathbb{I}=\mathbb{R}_{\ge 0}$. Let $\calX\defeq\{X_i\}_{i\in\mathbb{I}}$ be a sequence of random variables defined on $(\Omega,\calF,P)$. For any $j,k\in\mathbb{I}$, $j<k$, let $\calF_j^k(\calX)=\sigma(\{X_i\}_{j\le i\le k})$ be the $\sigma$-algebra generated by the random variables $\{X_i\}_{j\le i\le k}$ \cite{Bradley86,Bradley05}. For a sequence of random variables $\calX$, call $\alpha(\calX,k)\defeq\sup_{j\in\mathbb{I}}\{\alpha(\calF_{0}^j(\calX),\calF_{j+k}^{\infty}(\calX))\}$, $k\in\mathbb{I}$ its $\alpha$-mixing coefficient \cite{Bradley86,Bradley05,Doukhan94,Lin96}. For two measures $\mu$ and $\nu$ on a common measurable space $(\Omega,\calF)$, define the total variation distance $\|\mu-\nu\|_{\mathsmaller{\rm TV}}\defeq\sup_{\mathsmaller{A\in\calF}}\{|\mu(A)-\nu(A)|\}$ \cite{Levin08}. A sequence of random variables $\calX\defeq\{X_i\}_{i\in\mathbb{I}}$ is called $\alpha$-strongly mixing when $\lim_{k\mathsmaller{\rightarrow}\infty}\alpha(\calX,k)=0$ \cite{Bradley86,Bradley05,Doukhan94,Lin96}.

\begin{definition}{(Strongly Mixing Sequence of Random Variables)}\label{defiStrongMix}\\
A sequence of random variables $\calX\defeq\{X_i\}_{i\in\mathbb{I}}$ is called rapidly mixing toward a random variable $Y$\cite{Levin08}, with a mixing time $k_0>0$, $k_0\in\mathbb{I}$, when both $\alpha(\calX,k)\le\exp(\mathsmaller{-}k/k_0)$ and $\|P_{\mathsmaller{X_k}}-P_{\mathsmaller{Y}}\|_{\mathsmaller{\rm TV}}\le\exp(\mathsmaller{-}k/k_0)$ hold, $\forall k\ge k_0$, $k\in\mathbb{I}$. A sequence of random variables $\calX\defeq\{X_i\}_{i\in\mathbb{I}}$ is said to be well mixed if $\alpha(\calX,k)\le\exp(\mathsmaller{-}k)$, $\forall k\ge 1$, $k\in\mathbb{I}$.
\vspace{-1.5ex}
\end{definition}

Let $(\calC,\calH,\calB)$ be a many-body quantum system with a PLTKD Hamiltonian $H=\sum_{k=1}^{\mathsmaller{K}}H_k\in\log\calB$  that is supported by a continuous-discrete product configuration space $\calC\defeq\prod_{s=1}^{\mathsmaller{S}}(\calM_s^{n_s}\times\calP_s^{n_s})$ equipped with a Riemannian metric $g$, where $S\in\mathbb{N}$ is the number of particle species, each species indexed by $s\in[1,S]$ has $n_s\in\mathbb{N}$ identical particles moving on a substrate space $\calM_s\times\calP_s$.

\begin{definition}{(Computational Problem of Simulating Many-Body Hamiltonians)}\label{defiCompProbSimuManyBody}\\
A computational problem of simulating $H$ is to generate a sequence of a predetermined polynomial number of well mixed random samples of: either 1) configuration points $\{q\in\calC\}$ according to the probability density function $|\psi_0(H;q)|^2dV_g(q)$ in the context of GSQC, when $H$ is guaranteed to be polynomially gapped with a unique ground state $\psi_0(H)$; or 2) configuration points $\{q\in\calC\}$ according to a probability density function $|\psi(H;q;q_0,\tau)|^2dV_g(q)$ in the context of quantum statistical mechanics, where $\psi(H;q;q_0,\tau) \defeq \langle q|\exp\{-\tau[H{-}\lambda_0(H)]\}|q_0\rangle$, $q\in\calC$ is a Gibbs wavefunction with fixed $(q_0,\tau)\in\calC\times(0,\infty)$; or 3) pairs of configuration points $\{(r,q)\in\calC^2\}$ according to the absolute value of a Gibbs kernel $\langle r|\exp(-\tau H)|q\rangle$ in the context of quantum statistical mechanics, with $\tau > 0$ being a predetermined constant; where all said polynomials are in terms of $\size(H)$.
\end{definition}

\begin{definition}{(Spatially Local Partial Hamiltonians)}\label{defiSLH}\\
A boltzmannonic partial Hamiltonian $H_{\!\mathsmaller{B}}\in\calL_0(\calC)$ for a quantum multi-species multi-particle system moving on a many-body configuration space $\calC$ is called spatially local, if for any open subset $\calD\subseteq\calC$ and any $\psi\in\dom(H_{\!\mathsmaller{B}})\subseteq L^2(\calC)$ such that $\psi|_{\mathsmaller{\calC\setminus\cl(D)}}=0$, the overlap integral $\int_{\mathsmaller{\calC\setminus\cl(D)}}\phi^*H_{\!\mathsmaller{B}}\psi dV_g=0$ identically for any $\phi\in L^2(\calC)$. A fermionic partial Hamiltonian $H\in\calL_0(L^2_{\!\mathsmaller{F}}(\calC))$ is called spatially local when its base boltzmannonic Hamiltonian $\Base(H)$ is spatially local.
\vspace{-1.5ex}
\end{definition}

Spatially local boltzmannonic or fermionic partial Hamiltonians are epitomized by, but not limited to, the conventional boltzmannonic or fermionic Schr\"odinger operators of the familiar form $H=-\Delta_g+V(q)$, with $\Delta_g\defeq|\!\det(g)|^{\mathsmaller{-}1/2}\partial_i|\!\det(g)|^{1/2}g^{ij}\partial_j$ being the Laplace-Beltrami operator on a many-body configuration space $\calC$ endowed with a Riemannian metric $g$, $V(q\!\in\!\calC)$ representing a local potential, which as an operator is obviously a $\calC$-diagonal polynomial, and will be referred to as a $\calC$-diagonal potential hereafter. An important property of a spatially local boltzmannonic partial Hamiltonian $H_{\!\mathsmaller{B}}$ is wavefunction support nonexpanding, that is, $\forall\psi\in\dom(H_{\!\mathsmaller{B}})\subseteq L^2(\calC)$, let $\calD\subseteq\calC$ be the support of $\psi$ and $\cl(\calD)$ be the closure of $\calD$, such that $\int_{\mathsmaller{\calC\setminus\cl(\calD)}}\phi^*\psi\,dV_g=0$, $\forall\phi\in L^2(\calC)$, then the support of the wavefunction $H_{\!\mathsmaller{B}}\psi$ modulo null subsets of zero measure is necessarily a subset of $\cl(\calD)$, because it also holds that $\int_{\mathsmaller{\calC\setminus\cl(\calD)}}\phi^*H_{\!\mathsmaller{B}}\psi\,dV_g=0$, $\forall\phi\in L^2(\calC)$. For a general CD-separately moving partial Hamiltonian $H=H^{\mathsmaller{C}}+H^{\mathsmaller{D}}$ supported by a continuous-discrete product configuration space $\calC=\calM\times\calP$, the property of spatial locality is a useful characterization mostly for the continuous component $H^{\mathsmaller{C}}$, while the discrete component $H^{\mathsmaller{D}}$ can be spatially local if and only if it is $\calP$-diagonal, hence $\calC$-diagonal.

\begin{definition}{(Boltzmannonic and Fermionic Schr\"odinger Partial Hamiltonians)}\label{defiBFSH}\\
Supported by a continuous-discrete product manifold $\calC$ as a many-body configuration space, a CD-separately moving boltzmannonic partial Hamiltonian $H_{\!\mathsmaller{B}}=H_{\!\mathsmaller{B}}^{\mathsmaller{C}}+H_{\!\mathsmaller{B}}^{\mathsmaller{D}}\in\calL_0(\calC)$ is called spatially local and stoquastic, or boltzmannonic Schr\"odinger in short, when $H_{\!\mathsmaller{B}}^{\mathsmaller{C}}$ is spatially local and $H_{\!\mathsmaller{B}}$ is stoquastic in the sense that, $\forall\tau>0$, the Gibbs operator $\exp({-}\tau H_{\!\mathsmaller{B}})$ has a real- and non-negative-valued Gibbs kernel $K(\tau H_{\!\mathsmaller{B}};r,q)=\left\langle r\left|\exp({-}\tau H_{\!\mathsmaller{B}})\right|q\right\rangle$, $\forall(r,q)\in\calC\times\calC$. A fermionic partial Hamiltonian $H\in \calL_0(L^2_{\!\mathsmaller{F}}(\calC))$ is called spatially local and fermionic stoquastic, or fermionic Schr\"odinger in short, when its base boltzmannonic partial Hamiltonian $\Base(H)$ is boltzmannonic Schr\"odinger.
\end{definition}

\begin{lemma}{(Tensor Product of Fermionic Schr\"odinger and $\calC$-diagonal Operators)}\label{FStimesCdiagEqFS}\\
If $H_1=H^{\mathsmaller{C}}_1+H^{\mathsmaller{D}}_1\in\calL_0(L^2_{\!\mathsmaller{F}}(\calC_1))$ is a CD-separated fermionic Schr\"odinger partial Hamiltonian that is supported by a continuous-discrete product configuration space $\calC_1\defeq\calM_1\times\calP_1$, and $V\in\calL_0(L^2_{\!\mathsmaller{F}}(\calC))$ is a non-negative $\calC$-diagonal potential supported by a configuration space $\calC$, then the tensor product operator $H_2\defeq H_1\otimes V\in\calL_0(L^2_{\!\mathsmaller{F}}(\calC_2))$ is a fermionic Schr\"odinger partial Hamiltonian supported by the product configuration space $\calC_2\defeq\calC_1\times\calC$.
\end{lemma}
\vspace{-4.0ex}
\begin{proof}[\iftoggle{ForUSPTO} {Demonstration} {Proof}]
$H_2=H_1\otimes V=H^{\mathsmaller{C}}_1\otimes V+H^{\mathsmaller{D}}\otimes V$. It is obvious that $H^{\mathsmaller{C}}_1\otimes V$ is spatially local as long as $H^{\mathsmaller{C}}_1$ is so. Also clear is that $\Base(H_2)=\Base(H_1)\otimes\Base(V)$ is stoquastic, as $\Base(V)$ is $\calC$-diagonal and non-negative. It follows that $H_2$ is fermionic Schr\"odinger.
\vspace{-1.5ex}
\end{proof}

Given any open submanifold $\calD\subseteq\calC$ and any boltzmannonic Schr\"odinger partial Hamiltonian $H_{\!\mathsmaller{B}}$ supported by it, an eigenvalue problem arises naturally for $H_{\!\mathsmaller{B}}$ on $\calD$, subject to a suitable boundary condition of the Dirichlet, Neumann, or Robin type, when the boundary $\partial\calD$ is nonempty. The same is true for a fermionic Schr\"odinger partial Hamiltonian $H$ supported by $\calD$ as well, as long as the submanifold $\calD$ is large enough to contain all of the orbits, so to stay invariant, under the action of the exchange symmetry subgroup $G_{\rm ex}^{\mathsmaller{H}}$, which is defined as the largest subgroup comprising permutations that exchange only among the identical particles being actually moved by $H$. Evidently, all eigenvectors of a boltzmannonic or fermionic Schr\"odinger partial Hamiltonian $H_{\!\mathsmaller{B}}$ or $H$ supported by $\calD$ can be made real-valued, and the ground state of $H_{\!\mathsmaller{B}}$ made nowhere negative. Moreover, with the submanifold $\calD$ being $C^k$-smooth, $k\ge 1$ and $H_{\!\mathsmaller{B}}=\Base(H)$ being sufficiently regular, it can be assumed that the Gibbs operators $\exp({-}\tau H_{\!\mathsmaller{B}})$ and $\exp({-}\tau H)$ for some $\tau>0$ are compact self-adjoint on the separable Hilbert spaces $L^2(\calD)$ and $L_{\!\mathsmaller{F}}^2(\calD)$ respectively, and the Gibbs kernel $K(\tau H_{\!\mathsmaller{B}};r,q)=\langle r|\exp({-}\tau H_{\!\mathsmaller{B}})|q\rangle\in L^2(\calD\times\calD)$ is jointly continuous in $r$ and $q$, so is
\begin{align}
K(\tau H;r,q) \,&=\, \langle r|P_{\!\mathsmaller{F}}\exp({-}\tau H_{\!\mathsmaller{B}})P_{\!\mathsmaller{F}}|q\rangle \,=\, \langle r|P_{\!\mathsmaller{F}}\exp({-}\tau H_{\!\mathsmaller{B}})|q\rangle \nonumber \\[0.75ex]
\,&=\, const \times {\textstyle{ {\scalebox{1.15}{$\sum$}}_{\pi_1,\pi_2\in G_{\rm ex}} }} \, \sign(\pi_1) \, \sign(\pi_2) \, K(\tau H_{\!\mathsmaller{B}};\pi_1r,\pi_2q) \\[0.75ex]
\,&=\, const \times {\textstyle{ {\scalebox{1.15}{$\sum$}}_{\pi\in G_{\rm ex}} }} \, \sign(\pi)\,K(\tau H_{\!\mathsmaller{B}};\pi r,q) \,. \nonumber
\end{align}
It follows from the spectral theorem for compact self-adjoint operators \cite{Conway90,Blanchard15} that the countable set of eigenvectors of $\exp({-}\tau H)$, $\tau>0$ form an orthonormal basis for the Hilbert space $L_{\!\mathsmaller{F}}^2(\calD)$, and the Gibbs operator has the spectral resolution $\exp({-}\tau H)=\sum_{n=0}^{\infty}e^{\mathsmaller{-}\tau\lambda_n(H)}|\psi_n(H)\rangle\langle\psi_n(H)|$, where $\mathlarger{\mathlarger{\{}}e^{\mathsmaller{-}\tau\lambda_n(H)}\mathlarger{\mathlarger{\}}}_{n=0}^{\infty}$ lists the eigenvalues of $\exp({-}\tau H)$ in a decreasing order, and $\mathlarger{\mathlarger{\{}}\psi_n(H)\mathlarger{\mathlarger{\}}}_{n=0}^{\infty}$ lists the corresponding normalized eigenvectors, all of which are necessarily continuous in $\calD$ by the continuity of $K(\tau H;r,q)$ and the $L^2$-integrability of the eigenvectors in the eigen equations
\begin{align}
\left[\exp\left(-\tau H\right)\psi_n(H)\right](r) \,&=\, \int_{\calD}K(\tau H;r,q)\,\psi_n(H;q)\,dV_g(q) \nonumber \\[0.75ex]
\,&=\, e^{\mathsmaller{-}\tau\lambda_n(H)}\psi_n(H;r),\;\forall r\in\calD,\;\forall n\ge 0 \,.
\end{align}

For a conventional boltzmannonic Schr\"odinger operator $H_{\!\mathsmaller{B}}=-\Delta_g+V(q)$ on any connected Riemannian manifold $(\calM,g)$, with $\Delta_g$ being the Laplace-Beltrami operator, $V(q\!\in\!\calM)$ being $\calM$-diagonal and real-valued, the conditional Wiener measure can be taken as the reference measure, with respect to which the Gibbs energy functional reads $U_{\!\mathsmaller{G}}[\gamma(\cdot)]=\int_0^{\tau}V[\gamma(\tau')]d\tau'$ in equation (\ref{FeynmanKacKHB}), which then reduces to the celebrated Feynman-Kac formula. It can be shown that the operator $\exp({-}\tau H_{\!\mathsmaller{B}})$ is trace class, $\forall\tau>0$, the Gibbs kernel $K(\tau H_{\!\mathsmaller{B}};r,q)\in L^2(\calM\times\calM;\mathbb{R})$ is non-negative and continuous in $(\tau,r,q)\in(0,\infty)\times\calM\times\calM$, so long as the potential $V(q)$ is sufficiently regular, {\it e.g.}, Kato-decomposable, that is, $V=V_+-V_-$, with $V_+\in L^1_{\rm loc}(\calM)$ (locally $L^1$-integrable), and $V_-$ being a Kato-class potential having mild singularities \cite{Kato80,Lorinczi11,Simon82}. Moreover, $\forall\tau>0$, the Gibbs kernel $K(\tau H_{\!\mathsmaller{B}};r,q)$ is strictly positive almost everywhere, thus {\em positivity improving} \cite{Lorinczi11} in the sense that, $\forall\psi\in L^2(\calM;\mathbb{R})$ such that $\psi(q)\ge 0$, $\forall q\in\calM$ and $\psi\not\equiv 0$, $\int_{\mathsmaller{\calM}}K(\tau H_{\!\mathsmaller{B}};r,q)\,\psi(q)\,dV_g(q)$ becomes strictly positive for $V_g$-almost every $r\in\calM$, provided that the potential $V(q)$ is {\em Klauder regular}, {\it i.e.}, fulfilling the condition below to avoid the so-called Klauder phenomenon \cite{Simon79,Lorinczi11},
\begin{equation}
\exists\,\tau>0\,\;\mbox{such that }\int_{\calM}\mu_{\mathsmaller{W}}^{\mathsmaller{[q,\cdot]}}\!\left(\left\{
\gamma:\exp[-{\textstyle{\int}}_0^{\tau}V(\gamma(\tau'))d\tau']\!=\!0\right\}\right)\!dV_g(q)=0 \,,
\end{equation}
with $\mu_{\mathsmaller{W}}^{\mathsmaller{[q,\cdot]}}$ being the standard Wiener measure on the subset of paths in $W_{\tau}$ that start from $q$ and end anywhere in $\calM$. Any Kato-decomposable potentials $V(q)$ is Klauder regular, so the associated conventional boltzmannonic Schr\"odinger operator $H_{\!\mathsmaller{B}}=-\Delta_g+V(q)$ is positivity improving. When $H_{\!\mathsmaller{B}}$ is positivity improving, $\forall\tau>0$, the Gibbs operator $\exp({-}\tau H_{\!\mathsmaller{B}})$ has a unique positive eigenvector as the ground state of $H_{\!\mathsmaller{B}}$ \cite{Lorinczi11}, by the celebrated Perron-Frobenius theorem for non-negative matrices \cite{Seneta81,Horn85,Meyer00} and the generalizing Krein-Rutman and de Pagter's theorems for compact operators that are positive with respect to a convex cone \cite{Krein48,dePagter86,Du06}. Besides, when necessary, better analytical properties of $K(\tau H_{\!\mathsmaller{B}};r,q)$ can be established and used when $q$ and $r$ are kept out of a submanifold of singularity $\calM_{\rm sin}$, which has a lower dimension than $\calM$. A case in point is when the potential $V(q)$, $q\in\calM$ is a sum of a bounded smooth function and Coulomb interactions between multiple electrons and multiple nuclei as well as among the multiple electrons and among the multiple nuclei themselves, where $\calM_{\rm sin}$ is the union of {\em collision planes}, each of which is defined by a coordinate coincidence between two particles. The potential $V(q)$ is analytic \cite{Krantz02} at any $q\in\calM\setminus\calM_{\rm sin}$, so the eigenfunctions of $H_{\!\mathsmaller{B}}$ are analytic in $\calM\setminus\calM_{\rm sin}$, by the standard theory of analytic regularity for solutions of partial differential equations \cite{Kato57,Morrey66,Hormander76}. Thus, the Gibbs kernel $K(\tau H_{\!\mathsmaller{B}};r,q)$ is analytic in $(\calM\setminus\calM_{\rm sin})\times(\calM\setminus\calM_{\rm sin})$. Even at points on the submanifold of singularity $\calM_{\rm sin}$, the singular behaviors of the eigenfunctions of $H_{\!\mathsmaller{B}}$ are mild \cite{Fournais05,Fournais09,Ammann12}.

In this \iftoggle{ForUSPTO} {specification} {presentation}, we do not limit boltzmannonic Hamiltonians to the conventional Schr\"odinger operators, instead assume that any boltzmannonic Hamiltonian $H_{\!\mathsmaller{B}}$ should be spatially local and sufficiently regular over a continuous-discrete product Riemannian manifold $(\calC,g)$ excluding a lower-dimensional submanifold $\calC_{\rm sin}\subset\calC$, such that $\forall\tau>0$, the Gibbs operator $\exp(\mathsmaller{-}\tau H_{\!\mathsmaller{B}})$ is trace class, hence Hilbert-Schmidt with a Gibbs kernel $K(\tau H_{\!\mathsmaller{B}};r,q)$ belonging to the class $L^2(\calC\times\calC;\mathbb{R})$ for any fixed $\tau\in(0,\infty)$, as well as being positivity improving and $(\tau,r,q)$-jointly $C^{\infty,2,2}$ over the product topological space $(0,\infty)\times(\calC\setminus\calC_{\rm sin})\times(\calC\setminus\calC_{\rm sin})$, namely, $\forall(n_{\tau},n_i,n_j,n_k,n_l)\in(\mathbb{N}\cup\{0\})^5$ such that $n_i+n_j\le 2$, $n_k+n_l\le 2$, with $(i,j,k,l)\subseteq\mathbb{N}^4$ indexing continuous coordinate components, the partial derivative
$\partial_{\raisebox{0.3\height}{\tiny $\tau$}}^{\raisebox{0.4\height}{\tiny $n_{\tau}$}}
\partial_{\raisebox{0.3\height}{\tiny $r_i$}}^{\raisebox{0.4\height}{\tiny $n_i$}}
\partial_{\raisebox{0.3\height}{\tiny $r_j$}}^{\raisebox{0.4\height}{\tiny $n_j$}}
\partial_{\raisebox{0.3\height}{\tiny $q_k$}}^{\raisebox{0.4\height}{\tiny $n_k$}}
\partial_{\raisebox{0.3\height}{\tiny $q_l$}}^{\raisebox{0.4\height}{\tiny $n_l$}}
K(\tau H_{\!\mathsmaller{B}};r,q)$ exists and is continuous in $(0,\infty)\times(\calC\setminus\calC_{\rm sin})\times(\calC\setminus\calC_{\rm sin})$. Moreover, the Hamiltonian $H_{\!\mathsmaller{B}}$ ought to behave like a uniformly elliptic differential operator and substantiate the so-called Hopf lemma \cite{Hopf52,Oleinik52,Gilbarg01,Evans10}, in that, for any bounded connected open subset $\calD\subseteq\calC'$ of a submanifold $\calC'\subseteq(\calC\setminus\calC_{\rm sin})$, possibly with $\dim(\calC')<\dim(\calC)$ when considering restricted motions on a lower-dimensional submanifold of $\calC$, there exists a constant $E_0(H_{\!\mathsmaller{B}},\calD)\in\mathbb{R}$, called a {\em Hopf bias} for convenience, such that if $\psi\in C^2(\calD)\cap C^1(\cl(\calD))$ satisfies the {\em Dirichlet differential inequality} $[H_{\!\mathsmaller{B}}|_{\calD}+E_0(H_{\!\mathsmaller{B}},\calD)]\psi(q)\ge 0$, $\forall q\in\calD$, then 1) $\psi$ must have a strictly negative outer normal derivative at any boundary point $r\in\partial\calD\cap\cl(B)$ with an open ball $B\subseteq\calD$ such that $\psi(r)<\psi(q)$, $\forall q\in B$; 2) $\psi$ can only be a constant if it attains a non-positive minimum over $\cl(\calD)$ at an interior point $r\in\calD$. In particular, $\forall\psi\in\Diri(\calD)$, $\lambda\in\mathbb{R}$ such that $H_{\!\mathsmaller{B}}|_{\calD}\,\psi(q)=\lambda\psi(q)$, $\psi(q)>0$, $\forall q\in\calD$, it is always true that the gradient vector $\partial_r\psi(r)\defeq(\partial^i\psi)=(g^{ij}\partial_j\psi)\neq 0$ for almost every $r\in\partial\calD$, so long as $\partial\calD$ is piecewise $C^2$-smooth.

Still further, when considering a varying series of fermionic or boltzmannonic Schr\"odinger partial Hamiltonians $\{H(t):t\in\calI\}$ or $\{H_{\!\mathsmaller{B}}(t)=\Base(H(t)):t\in\calI\}$, with the index set $\calI$ being an interval on the real axis, and $\{H(t)\}$ or $\{H_{\!\mathsmaller{B}}(t)\}$ being continuously differentiable or even analytic in $t\in\calI$ \cite{Krantz02,Rellich69,Kriegl97}, it is always assumed that a suitable Wiener measure on the Riemannian manifold $\calC$ exists, and the partial Hamiltonian $H_{\!\mathsmaller{B}}(t)$ behaves sufficiently well on $\calC$ excluding a lower-dimensioned submanifold $\calC_{\rm sin}(t)$, $\forall t\in\calI$, such that the Feynman-Kac integral produces a Gibbs kernel $K(\tau H_{\!\mathsmaller{B}}(t);r,q)$ that is in $L^2(\calC\times\calC;\mathbb{R})$ and positivity improving $\forall(\tau,t)\in(0,\infty)\times\calI$, which is furthermore $(\tau,t,r,q)$-jointly $C^{\infty,1,0,0}$ in $(0,\infty)\times\calI\times\calC\times\calC$ and $(\tau,r,q)$-jointly $C^{\infty,2,2}$ in $(0,\infty)\times(\calC\setminus\calC_{\rm sin}(t))\times(\calC\setminus\calC_{\rm sin}(t))$, $\forall t\in\calI$. Still further, the operator $H_{\!\mathsmaller{B}}(t)$ must substantiate the Hopf lemma and extremum principle, $\forall t\in\calI$. Antisymmetrization yields the fermionic counterparts $H(t)=P_{\!\mathsmaller{F}}H_{\!\mathsmaller{B}}(t)P_{\!\mathsmaller{F}}$, $t\in\calI$, and $K(\tau H(t);r,q)=\left(\prod_{s=1}^{\mathsmaller{S}}n_s!\right)\sum_{\pi,\pi'\in\,G_{\rm ex}}(\mathsmaller{-}1)^{\sign(\pi)\,+\,\sign(\pi')}K(\tau H_{\!\mathsmaller{B}}(t);\pi r,\pi'q)$, which is also $(\tau,t,r,q)$-jointly $C^{\infty,1,0,0}$ in the product space $(0,\infty)\times\calI\times\calC\times\calC$, and $(\tau,r,q)$-jointly $C^{\infty,2,2}$ in $(0,\infty)\times(\calC\setminus\calC_{\rm sin}(t))\times(\calC\setminus\calC_{\rm sin}(t))$, $\forall t\in\calI$. Let $\{\psi_n(t;q)\}_{n=0}^{\infty}$ denote the set of eigenvectors of $\exp(\mathsmaller{-}H(t))$ corresponding to the eigenvalues $\{\lambda_n(H(t))\}_{n=0}^{\infty}$ in an increasing order, $\forall t\in\calI$. It is well known that $K(\tau H(t);r,q)=\sum_{n=0}^{\infty}e^{\mathsmaller{-}\tau\lambda_n(H(t))}\psi_n(t;r)\psi_n^*(t;q)$, $\forall(\tau,t,r,q)\in(0,\infty)\times\calI\times\calC\times\calC$, which converges absolutely and uniformly as $n\rightarrow\infty$, as well as uniformly when $\tau\rightarrow\infty$. Then it follows straightforwardly from
\begin{equation}
\int_{\calC}K(\tau H(t);r,q)\,\psi_n(t;q)\,dV_g(q) \,=\, e^{\mathsmaller{-}\tau\lambda_n(H(t))}\,\psi_n(t;r) \,, \; \forall r\in\calC \,, \; \forall t\in\calI \,, \; \forall\tau>0 \,, \; \forall n\ge 0 \,, \label{PsiIsC12}
\end{equation}
that the $t$-series of eigenstate wavefunctions $\{\psi_n(t;q):(t,q)\in\calI\times\calC\}$ is $(t,q)$-jointly $C^{1,0}$ in $\calI\times\calC$, $\forall n\ge 0$, and $C^2$ in $\calC\setminus\calC_{\rm sin}(t)$, $\forall t\in\calI$, $\forall n\ge 0$.

\begin{definition}{(Densities and Signed Densities)}\label{defiSignDens}\\
Given a measurable space $(\calC,\calF)$, where $\calC$ is a configuration space and $\calF \defeq \calF(\calC)$ is a $\sigma$-algebra of subsets of $\calC$, a scalar density on $\calC$ is a $(\calC,\calF)$-measurable function from $\calC$ to a field $\mathbb{K}$. A density on $\calC$ is a tuple- or vector-valued function having either just one or a plurality of scalar densities as components. In particular, a scalar density on $\calC$ is a density on $\calC$ that has just one component. A density $f$ on $\calC$ is said to be signed, when $f$ has two components $f_1$ and $f_2$ as scalar densities that are not necessarily different, and two points $q_1\in\calC$ and $q_2\in\calC$ exist which need not to differ, such that the value of the quotient $f_1(q_1)/f_2(q_2)$ is different from zero and a positive real number.
\end{definition}

\begin{definition}{(Minimally Entangled Densities)}\label{defiMiniEntaDens}\\
A density on a configuration space $\calC$ of a variable size is said to be minimally entangled, when it can be written as $f(\{g_i\}_{i=1}^n) \defeq f(g_1,\cdots\!,g_n)$, with $n=O(\poly(\size(\calC)))$, $f$ being a Borel measurable function defined on $\mathbb{K}^n$, $\mathbb{K}$ being a field, each $g_i$, $i\in[1,n]$ being a $\mathbb{K}$-valued scalar density on a submanifold $\calC_i$ as a tensor factor of $\calC$ such that $\size(C_i)=O(\log(\size(\calC)))$, wherein all of the functions $f$ and $\{g_i\}_{i=1}^n$ have a representation in a closed mathematical form that is efficiently computable, such that $\forall q\in\calC$, the point-value $f(\{g_i(q|_{\calC_i})\}_{i=1}^n)$, with $q|_{\calC}$ denoting the coordinate restriction of $q$ to the submanifold $\calC_i$, $\forall i\in[1,n]$, can be computed to within any predetermined relative error $\epsilon > 0$, at the cost of an $O(\poly(\size(\calC)+|\log\epsilon|))$ computational complexity.
\end{definition}

\begin{definition}{(Substantially Entangled Densities, Practically Substantially Entangled Densities)}\label{defiSubsEntaDens}\\
A density on a configuration space $\calC$ of a variable size is said to be substantially entangled when it is signed and can not be represented in the form of a minimally entangled density on $\calC$. A density on the same $\calC$ is said to be practically substantially entangled when it is signed and has no known representation in the form of a minimally entangled density on $\calC$.
\vspace{-1.5ex}
\end{definition}

The ground state $\psi_0(H;\cdot)$ or the Gibbs kernel function $K(\tau H;\cdot,\cdot)\defeq\langle\cdot|\exp(-\tau H)|\cdot\rangle$, $\tau>0$ of a many-fermion Hamiltonian $H$ is a signed density on the corresponding configuration space $\calC\defeq\calC^1$ or the product space $\calC\times\calC\defeq\calC^2$. In the past, great computational difficulties have been encountered in sampling from $|\psi_0(H;\cdot)|^2$ or $K(\tau H;\cdot,\cdot)$, $\tau>0$ to derive a numerical estimate for an expectation value of an interested observable, which seem to stem from the fact that the density $\psi_0(H;\cdot)$ or $K(\tau H;\cdot,\cdot)$, $\tau>0$ is substantially entangled, or at least practically substantially entangled.

To tackle a computational problem of simulating a many-fermion Hamiltonian $H$ via Monte Carlo, a mathematical correspondence is needed to connect a density associated with $H$ to a classical probability. Frequently used densities include the ground state $\psi_0(H;q)$, $q\in\calC^1$ and a Gibbs kernel $K(\tau H;q)\defeq K(\tau H;q_1,q_2)\defeq\langle q_1|\exp(-\tau H)|q_2\rangle$, $q\defeq(q_1,q_2)\in\calC^2$, $\tau>0$. Let $\calH\subseteq L^2(\calC^k;\mathbb{K})$, $k\in\mathbb{N}$ be a Hilbert space of functions supported by a configuration or product space $\calC^k\defeq\prod_{i=1}^k\calC$ and over a scalar field $\mathbb{K}\in\{\mathbb{R},\mathbb{C}\}$, and let $T\!\in\!\calB(\calH)$ be a bounded and normal linear operator, defined through an integral kernel $\langle\cdot|T|\cdot\rangle\in L^2(\calC^k\times\calC^k;\mathbb{K})$, such that $[T\phi](r)=\int_{\mathsmaller{\calC^k}}\langle r|T|q\rangle\,\phi(q)\,dV_g(q)$, $\forall\phi\in\calH$, $\forall r\in\calC^k$. Let $\psi$ be an eigenvector such that $T\psi=\lambda\psi$, $\lambda\in\mathbb{K}$. $\psi$ is necessarily continuous in $\calC^k$. Define a new operator $T_{\psi}\defeq[\psi^*]T[\psi^*]^{\mathsmaller{-}1}$\! on the $L^1$ space $\Diri(\calD,1)$, with $\calD\defeq\calC^k\setminus\psi^{\mathsmaller{-}1}(0)$ being the open subset of $\calC^k$ excluding the nodal points of $\psi$, the superscript $\!^*$ denoting complex conjugation when applicable and needed, and $\forall\phi\in L^2(\calD)$, $[\phi]:L^2(\calD)\mapsto L^2(\calD)$ representing a multiplication operator induced by the function $\phi(\cdot)$, which is formally a diagonal operator listing the values $\{\phi(q):q\in\calD\}$ as its diagonal elements, namely, $[\phi]\defeq\diag(\{\phi(q)\!:\!q\in\calD\})\defeq\!\int_{\mathsmaller{\calD}}\phi(q)\,|q\rangle\langle q|\,dV_g(q)$, such that $[\phi]\xi(q)=\phi(q)\xi(q)$ holds true, $\forall\xi\in L^2(\calD)$, $\forall q\in\calD$. The operator $T_{\psi}$ has an integral kernel $\langle r|T_{\psi}|q\rangle\defeq\psi^*(r)\langle r|T|q\rangle\psi^*(q)^{\mathsmaller{-}1}$, $(r,q)\in\calD\times\calD$.

\begin{lemma}{(Quasi-stochastic Operators and Integral Kernels)}\label{QuasiStochasticOper}\\
The operator $T_{\psi}$ and the associated integral kernel $\langle\cdot|T_{\psi}|\cdot\rangle$ are quasi-stochastic, in the sense that the right marginal distribution $\int_{\mathsmaller{\calD}}\langle r|T_{\psi}|q\rangle\,dV_g(r)$ reduces to the constant $\lambda$, $\forall q\in\calD$, moreover, for any eigenvector $\psi'$ of $T$ such that $T|_{\mathsmaller{\calD}}\psi'=\lambda'\psi'$, $\psi'\in\Diri(\calD)$, the function $\psi^*\psi'\in\Diri(\calD,1)$ is an eigenvector of $T_{\psi}$ with the same eigenvalue $\lambda'$. Conversely, when $T$ satisfies the Hopf lemma in $\calC\supseteq\calD$ while $T_{\psi}$ satisfies the anti-Hopf lemma on $\calD$, then $\{\psi^*\psi_k\}_{k\ge 0}$ lists all of the Dirichlet eigenfunctions of $T_{\psi}$ such that $T_{\psi}(\psi^*\psi_k) = \lambda_k\psi^*\psi_k$, $\psi^*\psi_k\in\Diri(\calD,1)$, $\forall k\ge 0$, when $\{\psi_k\}_{k\ge 0}$ lists all of the Dirichlet eigenfunctions of $T$ such that $T\psi_k = \lambda_k\psi_k$, $\psi_k\in\Diri(\calD)$, $\lambda_k\in\mathbb{K}$, $\forall k\ge 0$.
\end{lemma}
\vspace{-4.0ex}
\begin{proof}[\iftoggle{ForUSPTO} {Demonstration} {Proof}]
Let $T^+$ denote the adjoint operator of $T$ on $L^2(\calC^k)$, $k\in\mathbb{N}$. The operator $T-\lambda$ is normal. It follows from $\|(T-\lambda)\psi\|^2=0$ that $\|(T^+-\lambda^*)\psi\|^2=0$, namely, $T^+\psi=\lambda^*\psi$. Therefore, $\int_{\mathsmaller{\calC^k}}\psi^*(r)\langle r|T|q\rangle\,dV_g(r)=\int_{\mathsmaller{\calC^k}}\langle\psi|r\rangle\,\langle r|T|q\rangle\,dV_g(r)=\langle\psi|T|q\rangle=\langle T^+\psi|q\rangle=\lambda\psi^*(q)$, and consequently, $\int_{\mathsmaller{\calD}}\langle r|T_{\psi}|q\rangle\,dV_g(r)=\lambda$, $\forall q\in\calD$. On the other hand, for any $\psi'\in\Diri(\calD)$ such that $T|_{\mathsmaller{\calD}}\psi'=\lambda'\psi'$, $\lambda'\in\mathbb{K}$, it holds true that $T_{\psi}(\psi^*\psi')=[\psi^*]T[\psi^*]^{\mathsmaller{-}1}(\psi^*\psi')=[\psi^*]T\psi'=\lambda'(\psi^*\psi')$. In particular, the operator $T_{\psi}$ fixes the non-negative density $|\psi(q)|^2=\psi^*(q)\psi(q)$, $q\in\calD$, when the operator $T$ is affine-transformed to make $\lambda=1$. Such a non-negative density $|\psi(\cdot)|^2$, when normalized, is called a {\em stationary distribution} or {\em stationary probability vector} of (or associated with) the quasi-stochastic operator $T_{\psi}$. Conversely, when $T$ satisfies the Hopf lemma in $\calC\supseteq\calD$ while $T_{\psi}$ satisfies the anti-Hopf lemma on $\calD$, then for any Dirichlet eigenfunction $\psi''\in\Diri(\calD,1)$ of $T_{\psi}$ such that $T_{\psi}\psi'' = \lambda''\psi''$, $\lambda''\in\mathbb{K}$, it holds that $\psi' \defeq \psi''/\psi^*$ is always well-defined in $\calD$ and continuous in $\cl(\calD)$, where $\cl(\calD) \defeq \calD \cup \partial\calD$ denotes the topological closure of $\calD$, further, $\psi'$  vanishes on $\partial\calD$. It follows from $T_{\psi}\psi'' = [\psi^*]T[\psi^*]^{-1}\psi''= \lambda''\psi''$ that $T(\psi''/\psi^*) = \lambda''(\psi''/\psi^*)$, namely, $\psi' = \psi''/\psi^*$ is necessarily a Dirichlet eigenfunction of $T$ in $\calD$ which is also continuous in $\cl(\calD)$.
\vspace{-1.5ex}
\end{proof}

Clearly, a quasi-stochastic operator $T_{\psi}$ becomes bona fide stochastic when the associated integral kernel $\langle r|T|q\rangle$, $(r,q)\!\in\!\calC^k\times\calC^k$, $k\in\mathbb{N}$ is real-valued and nowhere negative, and the operator is scaled properly to have a unit spectral radius, in which case, the meaning of stationary distribution or stationary probability vector coincides precisely with the standard and well-known definition in association with Markov chains and stochastic matrices.

Specifically, for any conventional boltzmannonic Schr\"odinger operator $H_{\!\mathsmaller{B}}$ that is irreducible on a connected open subset $\calD$ of a many-body configuration space $\calC$, the Gibbs kernel $\langle r|\exp({-}\tau H_{\!\mathsmaller{B}})|q\rangle$ is positive real-valued, $\forall(r,q)\in\calD\times\calD$, $\forall\tau>0$, and the ground state $\psi_0$ of $H_{\!\mathsmaller{B}}$ is necessarily positive everywhere on $\calD$. Thus, the operator $[\psi]\exp({-}\tau H_{\!\mathsmaller{B}})[\psi]^{\mathsmaller{-}1}$, $\tau>0$ and its integral kernel $\psi(r)\langle r|\exp({-}\tau H_{\!\mathsmaller{B}})|q\rangle\psi(q)^{\mathsmaller{-}1}$ are indeed stochastic and can be interpreted as a Markov operator and a Markov kernel. Such boltzmannonic Schr\"odinger Hamiltonians are known as stoquastic and of great importance for Monte Carlo simulations of physical systems \cite{Bravyi08qic,Bravyi10}, where the positivity of the associated Gibbs kernels and the ground state wavefunctions makes it possible to map a quantum density into a classical probability density, which can be sampled by a random walk without encountering the sign problem. Here in this \iftoggle{ForUSPTO} {specification} {presentation}, the notion of stoquasticity will be extended to fermionic systems, which require a resolution of the numerical sign problem.

Despite being based on a stoquastic boltzmannonic Hamiltonian, fermionic Schr\"odinger Hamiltonians have in the past been considered archetypical hosts and victims of the infamous sign problem, due to the ostensible necessity of both positive and negative signs to satisfy the fermion-exchange symmetry of any legitimate wavefunction. Feynman \cite{Feynman82,Feynman87} was one of the pioneers recognizing the need and use of ``negative probability'' as the fundamental and differentiating aspect of quantum mechanics from classical physics and probability, even above the ``spooky'' quantum entanglement. Indeed, Feynman showed explicitly \cite{Feynman87} that when formulated in terms of density matrices, quantum physics paralleled classical probability all but identically, with the only exception of admitting ``probability'' values out of the interval $[0,1]$, which echoed Wigner \cite{Wigner32} and Moyal's \cite{Moyal49} earlier analyses using the density-matrix-equivalent Wigner functions. When a quantum system undergoes a restricted form of dynamics such that a suitable representation of the density matrix or Wigner function can be chosen to avoid ``negative probability'', then one may reasonably expect that the quantum dynamics can be efficiently simulated on a classical probabilistic machine. Such intuition/conjecture has been attested empirically by the vast and successful QMC simulations of boltzmannonic and bosonic systems in the absence of magnetic fields, and recently proved in full mathematical rigor for restricted quantum dynamics involving superoperators represented by positive-definite Wigner functions \cite{Veitch12,Mari12}. On the other hand, the apparently unavoidable ``negative probability'' makes it extremely hard, seemingly impossible, to simulate a general large quantum system on a classical probabilistic computer. Limited to classical physics and probability, unable to actualize and transfer ``negative probability'' as a quantum system can do, a BPP machine simulating a signed measure or density $\nu$ describing a quantum system would have to sample from a classical probability, that is, a non-negative-definite measure $\mu$, and compute expectation values of/involving a signed function $\sigma=d\nu/d\mu$, where catastrophic cancellation could have statistical noise inundating and obscuring the desired averages, when the $\nu$-described system is large and near quantum degeneracy. The seemingly fundamental limitation prompted Feynman to turn the difficulty of simulating quantum mechanics on its head, by proposing to build quantum computers and simulate quantum mechanics on quantum machines, thus ushering in the field of quantum computing, particularly quantum simulations \cite{Lloyd96}. Quantum computers, empowered by the capability to actualize and transfer ``negative probability'', seem to hold a potential to efficiently solve a possibly larger set of computational problems known as BQP. There is a widely held belief that ${\rm BQP}\setminus{\rm BPP}\neq\emptyset$, which shall turn out to be false.

Let $H$ be a general and usually fermionic partial Hamiltonian on a Riemannian manifold $\,(\calC,g)$ as a configuration space, $q_0 \in \calC$ be a fixed configuration point called the {\em reference point} \cite{Ceperley91}, $\tau \in (0,\infty]$ be a fixed imaginary time. It is always assumed that $H$ behaves sufficiently well analytically, such that the Gibbs wavefunction $\psi(H;q;q_0,\tau) \defeq \langle q|\exp\{-\tau[H{-}\lambda_0(H)]\}|q_0\rangle$ is continuous in $q \in \calC$, $\forall q_0 \in \calC$, $\forall \tau \in (0,\infty]$. In particular, the limit $\psi(H;q;q_0,\infty) \defeq \lim_{\tau\rightarrow\infty}\psi(H;q;q_0,\tau)$, $q \in \calC$ becomes substantially independent of the reference point $q_0$ and coincides with the ground state of $H$ up to a global scaling factor, namely, $\psi(H;q;q_0,\infty)/\|\psi(H;q;q_0,\infty)\| = \psi_0(H;q)$, $q \in \calC$, so long as $\psi(H;q_0;q_0,\infty) \neq 0$ and the ground state is non-degenerate. In the following of this \iftoggle{ForUSPTO} {specification} {presentation}, when applicable, {\em we will take advantage of such generality of the Gibbs wavefunction to treat two typical application scenarios on the equal footing,} where in one application scenario a Gibbs wavefunction $\psi(H;\cdot;q_0,\tau\rightarrow\infty)$, $q_0\in\calC$ in the large-$\tau$ limit represents the ground state of $H$, while in the other a Gibbs wavefunction $\psi(H;\cdot;q_0,\tau)$, $q_0\in\calC$ at a finite $\tau\in(0,\infty)$ stands for the Gibbs transition amplitude associated with a Gibbs operator $\exp(-\tau H)$.

\begin{definition}{(Nodal Cell and Nodal Surface)}\label{defiNodalCellSurface}\\
For any given Gibbs wavefunction $\psi(H;\cdot;q_0,\tau)$, $q_0 \in \calC$, $\tau \in (0,\infty]$, the preimage $\psi^{\mathsmaller{-}1}(H;\mathbb{K}\setminus\{0\};q_0,\tau) = \calC\setminus\psi^{\mathsmaller{-}1}(H;\{0\};q_0,\tau) \defeq \left\{r\in\calC:\psi(H;r;q_0,\tau)\neq 0\right\}$ is necessarily an open subset of $\calC$, one connected component of which containing a given point $q \in \calC$ such that $\psi(H;q;q_0,\tau)\neq 0$ is called the nodal cell of $\psi(H;\cdot;q_0,\tau)$ containing $q$, denoted as $\calN(H;q;q_0,\tau)$, whose boundary $\partial\calN(H;q;q_0,\tau)$ is called the nodal surface of $\psi(H;\cdot;q_0,\tau)$ enclosing $q$.
\vspace{-1.5ex}
\end{definition}

Multiple decades worth of research and development in QMC of fermionic Hamiltonians has signified the central importance of nodal cells and surfaces associated with Gibbs wavefunctions, which can be said to condense all of the computational hardness. In particular, for a fermionic Schr\"odinger partial Hamiltonian $H$ on a connected Riemannian manifold $\calM$, if the location of the nodal surfaces were known or otherwise efficiently computable for the ground state $\psi_0(H)$ or a Gibbs kernel function associated with $H$, then solving for the ground state wavefunction everywhere on $\calM$ would be equivalent to solving the ground state of a Dirichlet eigenvalue problem for $\Base(H)$ on a single nodal cell \cite{Lieb62,Ceperley91,Cances06}, or computing the Gibbs kernel function would be equivalent to computing a boltzmannonic path integral over all of the Feynman paths that never cross any nodal surface. That is the essential observation behind the fixed-node diffusion \cite{Anderson75,Anderson76,Ceperley80,Ceperley81,Caffarel88I,Caffarel88II,Ceperley91,Cances06} or constrained path integral \cite{Ceperley91,Ceperley96,Zhang97,Zhang04,Kruger08} methods of QMC.

%

Given a fermionic partial Schr\"odinger Hamiltonian $H=H^{\mathsmaller{C}}+H^{\mathsmaller{D}}$ on a continuous-discrete product configuration space $\calC=\calM\times\calP$, which is CD-separately irreducible and has a non-degenerate ground state $\psi_0(H)$. Let $\calN\subseteq\calM\times\{v\}$, $v\in\calP$ be a nodal cell of $\psi_0(H)$. The continuous component $\Base(H^{\mathsmaller{C}})$ of the base boltzmannonic Hamiltonian $\Base(H)$ is necessarily irreducible in $\calN$, and straightforwardly, the nodal cell $\calN$ is $(\Base(H^{\mathsmaller{C}})|_{\mathsmaller{\calN}})$-complete. However, $\calN$ might not be $\Base(H^{\mathsmaller{D}})$-closed in general. Rather, there could be another nodal cell $\calN'\subseteq\calM\times\{v'\}$ with $v'\in\calP$, $v'\neq v$ that is connected to $\calN$ through $\Base(H^{\mathsmaller{D}})$.

\begin{definition}{((Positively) $H^{\mathsmaller{D}}$-Connected Configuration Points and Nodal Cells)} \label{HDConnectedPointsCells}\\
Let $\calC=\calM\times\calP$, $H=H^{\mathsmaller{C}}+H^{\mathsmaller{D}}$, $\psi_0(H)$, and nodal cells $\calN\subseteq\calM\times\{v\}$, $\calN'\subseteq\calM\times\{v'\}$, $v,v'\in\calP$ of $\psi_0(H)$ be given as in the immediate above. Two configuration points $(q,v)$ and $(q,v')$ with $q\in\calM$ and $v,v'\in\calP$ are said to be $H^{\mathsmaller{D}}$-connected when $\langle(q,v)|I+\Base(H^{\mathsmaller{D}})|(q,v')\rangle\neq 0$. Such a pair of $H^{\mathsmaller{D}}$-connected configuration points $(q,v)$ and $(q,v')$ are called positively $H^{\mathsmaller{D}}$-connected if additionally $\psi_0(H;(q,v))\psi_0(H;(q,v'))>0$. Two nodal cells $\calN\subseteq\calM\times\{v\}$ and $\calN'\subseteq\calM\times\{v'\}$ with $v,v'\in\calP$ are said to be directly (and positively) $H^{\mathsmaller{D}}$-connected when $q\in\calM$ exists such that $(q,v)\in\calN$, $(q,v')\in\calN'$, $(q,v)$ and $(q,v')$ are (positively) $H^{\mathsmaller{D}}$-connected.

Two nodal cells $\calN\subseteq\calM\times\{v\}$ and $\calN'\subseteq\calM\times\{v'\}$ with $v,v'\in\calP$ are said to be (positively) $H^{\mathsmaller{D}}$-connected, when a sequence of nodal cells $\{\calN_i\}_{i=0}^n$, $n\in\mathbb{N}$ exists such that $\calN_0=\calN$, $\calN_n=\calN'$, $\calN_{i{-}1}$ and $\calN_i$ are directly (and positively) $H^{\mathsmaller{D}}$-connected, $\forall i\in[1,n]$, where each such sequence of nodal cells is called a (positive) $H^{\mathsmaller{D}}$-path between $\calN$ and $\calN'$. A non-negative integer-valued (positive) $H^{\mathsmaller{D}}$-distance between such pair of nodal cells $(\calN,\calN')$, denoted by $\dist(\calN,\calN')$, can be defined as $\dist(\calN,\calN')=0$ if $\calN=\calN'$, else $\dist(\calN,\calN')=n$ if $n\in\mathbb{N}$ is the least number such that an $H^{\mathsmaller{D}}$-path $\{\calN_i\}_{i=0}^n$ exists to start and end at $\calN_0=\calN$ and $\calN_n=\calN'$ respectively.
\end{definition}

\begin{definition}{(Nodal Groupoid)}\label{defiNodalGroupoid}\\
Let $H=H^{\mathsmaller{C}}+H^{\mathsmaller{D}}$ be a CD-separately irreducible fermionic Schr\"odinger Hamiltonian supported by a continuous-discrete product configuration space $C=\calM\times\calP$, which has a non-degenerate ground state $\psi_0(H)$. Let $(q,v)\in\calC$ be a configuration point such that $\psi_0(H;(q,v))\neq 0$. The set of all nodal cells of $\psi_0(H)$ that are positively $H^{\mathsmaller{D}}$-connected to $\calN(H;(q,v))$ is called a nodal groupoid of $\psi_0(H)$, which is denoted as $\calN^{\mathsmaller{\clubsuit}}(H;(q,v))$. The same name of nodal groupoid may also refer to the set union of nodal cells $\bigcup\calN^{\mathsmaller{\clubsuit}}(H;(q,v))\defeq\bigcup\{\calN\in\calN^{\mathsmaller{\clubsuit}}(H;(q,v))\}$, when there is no risk of ambiguity in a context, or the mathematical notation makes the specific meaning clear.

For any nodal groupoid $\calN^{\mathsmaller{\clubsuit}}(H;(q,v))$, its continuous boundary denoted by $\partial\calN^{\mathsmaller{\clubsuit}}(H;(q,v))$ is defined to be the union of the usual nodal surfaces, $\partial\calN^{\mathsmaller{\clubsuit}}(H;(q,v))\defeq\bigcup\{\partial\calN:\calN\in\calN^{\mathsmaller{\clubsuit}}(H;(q,v))\}$, while its discrete boundary denoted by $\delta\calN^{\mathsmaller{\clubsuit}}(H;(q,v))$ is defined to be the set of ordered pairs of configuration points $((r,v),(r,v'))$ such that $r\in\calM$, $v,v'\in\calP$, $(r,v)\in\bigcup\calN^{\mathsmaller{\clubsuit}}(H;(q,v))$, $(r,v)$ and $(r,v'))$ are $H^{\mathsmaller{D}}$-connected but not positively $H^{\mathsmaller{D}}$-connected. The boundary of $\calN^{\mathsmaller{\clubsuit}}(H;(q,v))$, called the nodal boundary in short, is denoted and defined as $\nabla\calN^{\mathsmaller{\clubsuit}}(H;(q,v))\defeq\partial\calN^{\mathsmaller{\clubsuit}}(H;(q,v))\cup\delta\calN^{\mathsmaller{\clubsuit}}(H;(q,v))$. Also, the closure of $\calN^{\mathsmaller{\clubsuit}}(H;(q,v))$ is denoted and defined as $\cl[\calN^{\mathsmaller{\clubsuit}}(H;(q,v))]\defeq[\,\bigcup\calN^{\mathsmaller{\clubsuit}}(H;(q,v))]\cup[\nabla\calN^{\mathsmaller{\clubsuit}}(H;(q,v))]$.

A wavefunction $\psi\in L^2(\calC)$ is said to satisfy the Dirichlet boundary conditions on $\nabla\calN^{\mathsmaller{\clubsuit}}(H;(q,v))$, when $\psi(r,v)=0$ for any $(r,v)\in\partial\calN^{\mathsmaller{\clubsuit}}(H;(q,v))$, and $\psi(r,v')/\psi(r,v)=\psi_0(H;(r,v'))/\psi_0(H;(r,v))$ for any $((r,v),(r,v'))\in\delta\calN^{\mathsmaller{\clubsuit}}(H;(q,v))$.
\vspace{-1.5ex}
\end{definition}

The term {\em nodal groupoid} suggests an interpretation of $\calN^{\mathsmaller{\clubsuit}}(H;(q,v))$ forming a small category, where the collection of concerned nodal cells serves as the set of objects, and from each $\calN\in\calN^{\mathsmaller{\clubsuit}}(H;(q,v))$ to each $\calN'\in\calN^{\mathsmaller{\clubsuit}}(H;(q,v))$, the set of morphisms $\hom(\calN,\calN')$ consists of all positive $H^{\mathsmaller{D}}$-paths (considered as directed) between $\calN$ and $\calN'$, so consequently, every morphism is invertible.

The required equation of wavefunction proportionality for the Dirichlet boundary condition on the discrete boundary $\delta\calN^{\mathsmaller{\clubsuit}}(H;(q,v))$ is a generalization of the so-called {\em lever rule} in fixed-node diffusion QMC for lattice fermions \cite{vanBemmel94,Sorella15}, which makes it possible to recover fully and identically the ground state wavefunction $\psi_0(H)$ restricted to the nodal groupoid $\calN^{\mathsmaller{\clubsuit}}(H;(q,v))$, by solving the ground state of a Dirichlet eigenvalue problem $H_{\!\mathsmaller{B}}^{\rm eff}\psi=\lambda\psi$, $\psi\in\Diri(\calN^{\mathsmaller{\clubsuit}}(H;(q,v)),2)$, with an effective Hamiltonian $H_{\!\mathsmaller{B}}^{\rm eff}\defeq\Base(H)|_{\calN^{\mathsmaller{\clubsuit}}(H;(q,v))}+V_{\!\mathsmaller{B}}^{\rm eff}(\cdot,v)$ representing a truncation of the boltzmannonic Hamiltonian $\Base(H)$ to within $\calN^{\mathsmaller{\clubsuit}}(H;(q,v))$, where
\begin{equation}
V_{\!\mathsmaller{B}}^{\rm eff}(r,v)\defeq\sum_{((r,v),(r,v'))\in\delta\calN^{\mathsmaller{\clubsuit}}(H;(q,v))}\frac{\psi_0(H;(r,v'))}{\psi_0(H;(r,v))}\,\langle(r,v)|\Base(H^{\mathsmaller{D}})|(r,v')\rangle \ge 0, \label{defiVBeff}
\end{equation}
and $\Base(H)|_{\calN^{\mathsmaller{\clubsuit}}(H;(q,v))}$ denotes the restriction of $\Base(H)$ to the submanifold $\bigcup\calN^{\mathsmaller{\clubsuit}}(H;(q,v))$ as specified in equation (\ref{defiHRestriction}), which discards all of the $\Base(H)$-induced interactions between any $(r,v)\in\calC$ and $(r,v')\in\calC$ such that $((r,v),(r,v'))\in\delta\calN^{\mathsmaller{\clubsuit}}(H;(q,v))$, henceforth referred to as {\em $\Base(H)$-couplings across the discrete boundary $\delta\calN^{\mathsmaller{\clubsuit}}(H;(q,v))$}, whose collective effect is replaced by the effective $\calC$-diagonal potential $V_{\!\mathsmaller{B}}^{\rm eff}(\cdot,v)$. As mentioned previously, an alternative interpretation of $\Base(H)|_{\calN^{\mathsmaller{\clubsuit}}(H;(q,v))}$ is to view it as a superimposition of $\Base(H)$ and an infinite potential barrier out of $\bigcup\calN^{\mathsmaller{\clubsuit}}(H;(q,v))$, namely, $\Base(H)|_{\calN^{\mathsmaller{\clubsuit}}(H;(q,v))}\cong\Base(H)+V_{\calN^{\mathsmaller{\clubsuit}}(H;(q,v))}(r,v)$, where $V_{\calN^{\mathsmaller{\clubsuit}}(H;(q,v))}(\cdot,v)$ is a $\calC$-diagonal potential such that $V_{\calN^{\mathsmaller{\clubsuit}}(H;(q,v))}(r,v)=0$, $\forall r\in\bigcup\calN^{\mathsmaller{\clubsuit}}(H;(q,v))$ and $V_{\calN^{\mathsmaller{\clubsuit}}(H;(q,v))}(r,v)={+}\infty$, $\forall r\notin\bigcup\calN^{\mathsmaller{\clubsuit}}(H;(q,v))$. That the $\Base(H)$-couplings across $\delta\calN^{\mathsmaller{\clubsuit}}(H;(q,v))$ can be equivalently substituted by the effective $\calC$-diagonal potential $V_{\!\mathsmaller{B}}^{\rm eff}$ has been proved rigorously in \cite{vanBemmel94}, noting that for any function $\psi\in L^2(\calC)$ satisfying the Dirichlet boundary conditions on $\nabla\calN^{\mathsmaller{\clubsuit}}(H;(q,v))$, any occurrence of a wave amplitude $\psi(r,v')$ in any equation can be substituted identically by $[\psi_0(H;(r,v'))/\psi_0(H;(r,v))]\,\psi(r,v)$, so long as $((r,v),(r,v'))\in\delta\calN^{\mathsmaller{\clubsuit}}(H;(q,v))$.

\begin{definition}{(Fundamental Domain or $G$-Tile)}\label{defiFunDomain}\\
Let $(X,\calT)$ be a topological space, on which a group $G$ acts by homeomorphisms \cite{Lin94,Bekka00}. An open set $D\subseteq X$ is called a fundamental domain under $G$, or simply a $G$-tile, if 1) $\forall g_1,g_2\in G$, either $g_1D=g_2D$ or $g_1D\cap g_2D=\emptyset$, and 2) the orbit of $D$, $\orb(D)\defeq\bigcup_{g\in G}gD$, upon topological closure, contains $X$, that is, $\cl[\orb(D)]=X$. An open set $D\subseteq X$ is called a $G$-pretile when only condition 1) is satisfied.
\end{definition}

\begin{lemma}{(Tiling Property of Nodal Groupoids \cite{Ceperley91,Cances06})}\\
Let $\calC=\calM\times\calP$, $\calT$, and $G_{\rm ex}$ denote the configuration space, its defining topology, and the exchange symmetry group of a CD-separately irreducible fermionic Schr\"odinger Hamiltonian $H=H^{\mathsmaller{C}}+H^{\mathsmaller{D}}$, whose ground state is non-degenerate, then each nodal groupoid $\calN^{\mathsmaller{\clubsuit}}(H;q)$, $q\in\calC$, $\psi_0(H;q)\neq 0$ is a $G_{\rm ex}$-tile, with respect to the topological subspace $\mathlarger{\mathlarger{(}}\calC'\defeq\calC\setminus\psi_0^{\mathsmaller{-}1}(H;\{0\},\calT'\defeq\{U\cap\calC':U\in\calT\}\mathlarger{\mathlarger{)}}$.
\label{TilingProperty}
\end{lemma}
\vspace{-4.0ex}
\begin{proof}[\iftoggle{ForUSPTO} {Demonstration} {Proof}]
Firstly, define a binary relation ${\raisebox{-.1\height}{$\simX{H}$}}$ between any $q,r\in\calC'$ such that $q\;{\raisebox{-.1\height}{$\simX{H}$}}\;r$ if $r\in\calN^{\mathsmaller{\clubsuit}}(H;q)$, that is, there exists a sequence of nodal cells $\{\calN_i\}_{i=0}^n$, $n\in\mathbb{N}$ such that $\calN_0=\calN(H;q)$, $\calN_n=\calN(H;r)$, $\calN_{i{-}1}$ and $\calN_i$ are directly and positively $H^{\mathsmaller{D}}$-connected, $\forall i\in[1,n]$. It is easy to verify that ${\raisebox{-.1\height}{$\simX{H}$}}$ is an equivalence relation, which partitions the set $\calC'=\calC\setminus\psi_0^{\mathsmaller{-}1}(H;\{0\})$ into mutually disjoint equivalence classes $\mathlarger{\mathlarger{\{}}[q]_{\mathsmaller{H}}\defeq\{r\in\calC':q\;{\raisebox{-.1\height}{$\simX{H}$}}\;r\}\mathlarger{\mathlarger{\}}}$, that are exactly the nodal groupoids, namely, $[q]_{\mathsmaller{H}}=\calN^{\mathsmaller{\clubsuit}}(H;q)$, $\forall q\in\calC'$. It is also clear that $\forall q,r\in\calC'$, $\forall\pi\in G_{\rm ex}$, $q\;{\raisebox{-.1\height}{$\simX{H}$}}\;r$ if and only if $\pi q\;{\raisebox{-.1\height}{$\simX{H}$}}\;\pi r$, as a consequence of $\psi_0(H;\pi q)=\sign(\pi)\,\psi_0(H;q)$, $\forall q\in\calC$, $\forall \pi\in G_{\rm ex}$, as well as $\pi\in G_{\rm ex}$ being an autohomeomorphism. Therefore, for each nodal groupoid $\calN^{\mathsmaller{\clubsuit}}(H;q)=[q]_{\mathsmaller{H}}$, $q\in\calC'$ and $\forall\pi\in G_{\rm ex}$, the transformed open set $\pi\calN^{\mathsmaller{\clubsuit}}(H;q)=\calN^{\mathsmaller{\clubsuit}}(H;\pi q)=[\pi q]_{\mathsmaller{H}}$ is necessarily a nodal groupoid, which is either identical to or disjoint with $\calN^{\mathsmaller{\clubsuit}}(H;q)=[q]_{\mathsmaller{H}}$. So, every groupoid is a $G$-pretile. The subgroup of $G_{\rm ex}$ under which a nodal groupoid $\calN^{\mathsmaller{\clubsuit}}(H;q)$ is invariant is called the stabilizer of $\calN^{\mathsmaller{\clubsuit}}(H;q)$, denoted and specified as $\stab(\calN^{\mathsmaller{\clubsuit}}(H;q))\defeq\{\pi\in G_{\rm ex}:\pi\calN^{\mathsmaller{\clubsuit}}(H;q)=\calN^{\mathsmaller{\clubsuit}}(H;q)\}$.

Secondly, let $\calD\defeq\calN^{\mathsmaller{\clubsuit}}(H;q)$, $q\in\calC'$ be any groupoid viewed as a $G_{\rm ex}$-pretile, then, as a consequence of the spatial locality of $\Base(H^{\mathsmaller{C}})$, the restriction of $\psi_0(H)$ to $\calD$ constitutes an eigen solution to a Dirichlet problem on $\calD$, $[H_{\!\mathsmaller{B}}^{\rm eff}|_{\mathsmaller{\calD}}][\psi_0(H)|_{\mathsmaller{\calD}}]=\lambda_0(H)[\psi_0(H)|_{\mathsmaller{\calD}}]$, $\psi_0(H)|_{\mathsmaller{\calD}}>0$, subject to the Dirichlet boundary condition on $\partial\calD\defeq\partial\calN^{\mathsmaller{\clubsuit}}(H;q)$, where $H_{\!\mathsmaller{B}}^{\rm eff}=\Base(H)|_{\mathsmaller{\calD}}+V_{\!\mathsmaller{B}}^{\rm eff}$ denotes the effective boltzmannonic Hamiltonian in association with the nodal groupoid $\calD$, with $V_{\!\mathsmaller{B}}^{\rm eff}$ being specified in equations (\ref{defiVBeff}). Conversely, for any Dirichlet eigen solution $\phi_0\in\Diri(\calD)\subseteq L^2(\calC)$ of $H_{\!\mathsmaller{B}}^{\rm eff}$ associated with an eigenvalue $\kappa_0\in\mathbb{R}$, $\kappa_0\le\lambda_0(H)$ such that $(H_{\!\mathsmaller{B}}^{\rm eff}|_{\mathsmaller{\calD}})(\phi_0|_{\mathsmaller{\calD}})=\kappa_0\,\phi_0|_{\mathsmaller{\calD}}$, $\phi_0|_{\mathsmaller{\calD}}>0$, $\phi_0|_{\mathsmaller{\calC\setminus \calD}}=0$, the wavefunction $\phi_0$ can be extended into $\phi_1\in L^2(\cl(\calD)\defeq\calD\cup\nabla\calD)\subseteq L^2(\calC)$ such that $\phi_1|_{\mathsmaller{\calD}}=\phi_0|_{\mathsmaller{\calD}}$, $\phi_1|_{\mathsmaller{\calC\setminus\cl(\calD)}}=0$, and $\phi_1(q')/\phi_1(q)=\psi_0(H;q')/\psi_0(H;q)$ for all $(q,q')\in\calC\times\calC$ satisfying $(q,q')\in\delta\calD$ and $q\in\calD$. Consequently, $\left\{[\Base(H)|_{\mathsmaller{\cl(\calD)}}][\phi_1|_{\mathsmaller{\cl(\calD)}}]\right\}\!(q)=\kappa_0\,\phi_0(q)$ holds for all $q\in\calD\cup\partial\calD$. Moreover, it is without loss of generality to assume that $\phi_0$ is fixed by $\stab(\calN^{\mathsmaller{\clubsuit}}(H;q))$, meaning $\phi_0\circ\pi(q)=\phi_0(q)$, $\forall q\in\bigcup\calN^{\mathsmaller{\clubsuit}}(H;q)$, $\forall\pi\in\stab(\calN^{\mathsmaller{\clubsuit}}(H;q))$, because, otherwise, $\phi_0$ can be replaced by its symmetrized version $const\times\sum_{\pi\,\in\,\stab(\calN^{\mathsmaller{\clubsuit}}(H;q))}\phi_0\circ\pi^{{-}1}$.

Now consider the homeomorphic group action \cite{Rotman99,Kerber99,Lin94,Bekka00,Jacobson09,Rose09} of $G_{\rm ex}$ on $\calC$ and the induced transformation of quantum amplitude distributions, {\it i.e.}, wavefunctions. For any $\pi\in G_{\rm ex}$ acting by homeomorphism, $\pi\calD\defeq\{\pi q: q\in\calD\}$ is still a nodal groupoid, $\pi(\partial\calD)=\partial(\pi\calD)$, $\pi(\delta\calD)=\delta(\pi\calD)$, $\pi(\nabla\calD)=\nabla(\pi\calD)$ are still the boundaries of $\pi\calD$, and the transformed wavefunctions $\phi_0\circ\pi^{{-}1}$, $\phi_1\circ\pi^{{-}1}$ still solve a Dirichlet eigenvalue problem on $\pi\calD$, in particular, $\left\{[\Base(H)|_{\mathsmaller{\cl(\pi\calD)}}][\phi_1\circ\pi^{{-}1}|_{\mathsmaller{\cl(\calD)}}]\right\}\!(r)=\kappa_0\,\phi_0\circ\pi^{{-}1}(r)$ holds $\forall r\in\pi\calD\cup\partial(\pi\calD)$, since $\Base(H)$ is invariant under $\pi$, $\phi_0$ is fixed by $\stab(\calN^{\mathsmaller{\clubsuit}}(H;q))$, while $\psi_0(H)$ is antisymmetric as $\psi_0(H;\pi q)=\sign(\pi)\times\psi_0(H;q)$, $\forall q\in\calC$. Consequently, the wavefunctions $\phi_0$ and $\phi_1$ can be antisymmetrized and extended into $P_{\!\mathsmaller{F}}\phi_i\defeq const\times\sum_{\pi\in G_{\rm ex}}[\sign(\pi)\times(\phi_i\circ\pi^{\mathsmaller{-1}})]\in L_{\!\mathsmaller{F}}^2(\calC)$, $i\in\{0,1\}$ respectively, satisfying
\begin{equation}
\{[H|_{\mathsmaller{\orb(\cl(\calD))}}][P_{\!\mathsmaller{F}}\phi_1]\}(r)=\kappa_0\,[P_{\!\mathsmaller{F}}\phi_0](r),\;\forall r\in\orb(\calD), \label{DiriOnOrbitBeforeLever}
\end{equation}
subject to the boundary condition $[P_{\!\mathsmaller{F}}\phi_1](r')/[P_{\!\mathsmaller{F}}\phi_1](r)=\psi_0(H;r')/\psi_0(H;r)$, $\forall(r,r')\in\calC\times\calC$ such that $r\in\orb(\calD)$, $r'\in\orb(\delta\calD)\setminus\orb(\calD)$, and $\langle r|\Base(H^{\mathsmaller{D}})|r'\rangle\neq 0$, where $\orb(\calD)\defeq\bigcup_{\pi\in G_{\rm ex}}\pi\calD$ and $\orb(\delta\calD)\defeq\bigcup_{\pi\in G_{\rm ex}}\pi(\delta\calD)$ are the orbits of $\calD$ and $\delta\calD$ under the action of $G_{\rm ex}$ respectively. Using the lever rule trick again, a similar change of variables with
\begin{equation}
V_{\!\mathsmaller{B}}^{\rm Eff}(r)\defeq\sum_{r'\,\in\,\orb(\delta\calD)\setminus\orb(\calD)}\frac{\psi_0(H;r')}{\psi_0(H;r)}\,\langle r|\Base(H^{\mathsmaller{D}})|r'\rangle\ge 0,\;\forall r\in\orb(\calD), \label{defiVBEff}
\end{equation}
transforms the homogeneous linear equation (\ref{DiriOnOrbitBeforeLever}) into a Dirichlet eigenvalue equation
\begin{equation}
\{[H|_{\mathsmaller{\orb(\calD)}}+V_{\!\mathsmaller{B}}^{\rm Eff}][P_{\!\mathsmaller{F}}\phi_0]\}(r)=\kappa_0\,[P_{\!\mathsmaller{F}}\phi_0](r),\;\forall r\in\orb(\calD). \label{DiriOnOrbitAfterLever}
\end{equation}

But it must hold that $\kappa_0=\lambda_0(H)$. Otherwise, if $\kappa_0<\lambda_0(H)$, then according to equation (\ref{DiriOnOrbitAfterLever}), the wavefunction $\psi\defeq[P_{\!\mathsmaller{F}}\phi_0]_{\mathsmaller{\orb(\calD)}}\times\mathbbm{1}_{\mathsmaller{\orb(\calD)}}\in L^2_{\!\mathsmaller{F}}(\calC)$, with $\mathbbm{1}_{\mathsmaller{\orb(\calD)}}$ being the indicator function \cite{IndicatorFunction} for the set $\orb(\calD)$, would produce a Rayleigh-Ritz ratio $\langle\psi|H|\psi\rangle/\langle\psi|\psi\rangle=\kappa_0<\lambda_0(H)$, which, by the Courant-Fischer-Weyl min-max principle \cite{Courant89,Reed78}, is impossible unless $P_{\!\mathsmaller{F}}\phi_0$ vanishes almost everywhere on $\orb(\calD)$. But $P_{\!\mathsmaller{F}}\phi_0$ cannot vanish, even just in $\calD$, because, $\forall\pi\in G_{\rm ex}$ such that $\pi\calD\,\cap\,\calD\neq\emptyset$, thus $\pi\calD=\calD=\pi^{\mathsmaller{-}1}\calD$, it follows from $(\phi_0\circ\pi^{\pm 1})(r)=\phi_0(\pi^{\pm 1}r)>0$, $\forall r\in\calD$ that $\sign(\pi)=+1$, hence $\psi(r)=[P_{\!\mathsmaller{F}}\phi_0](r)>0$, $\forall r\in\calD$. Furthermore, $\orb(\calD)$ must be the whole of $\calC'$. Otherwise, if there is a $q'\in\calC'\setminus\orb(\calD)$ such that $\psi_0(H;q')\neq 0$, then the nodal groupoid $\calD'\defeq\bigcup\calN^{\mathsmaller{\clubsuit}}(H;q')$ is nonempty, $\orb(\calD)\cap\orb(\calD')=\emptyset$, and a wavefunction $\phi'_0\defeq\psi_0(H)|_{\mathsmaller{\calD'}}$ constitutes an eigenfunction of a Dirichlet boundary value problem on $\calD'$, associated with an eigenvalue $\kappa'_0\le\lambda_0(H)$. Then the previous steps of wavefunction antisymmetrization and extension can be repeated to instantiate another wavefunction $\psi'\defeq[P_{\!\mathsmaller{F}}\phi'_0]_{\mathsmaller{\orb(\calD')}}\times\mathbbm{1}_{\mathsmaller{\orb(\calD')}}\in L^2_{\!\mathsmaller{F}}(\calC)$ besides $\psi\defeq[P_{\!\mathsmaller{F}}\phi_0]_{\mathsmaller{\orb(\calD)}}\times\mathbbm{1}_{\mathsmaller{\orb(\calD)}}\in L^2_{\!\mathsmaller{F}}(\calC)$, such that $\langle\psi|H|\psi\rangle/\langle\psi|\psi\rangle=\lambda_0(H)$, $\langle\psi'|H|\psi'\rangle/\langle\psi'|\psi'\rangle=\kappa'_0\le\lambda_0(H)$, $\langle\psi|\psi\rangle\neq 0$, $\langle\psi'|\psi'\rangle\neq 0$, $\langle\psi|\psi'\rangle=0$. That contradicts the premise of $H$ having a non-degenerate ground state with the eigenvalue $\lambda_0(H)$. Therefore, $\calC'=\orb(\calD)\subseteq\calC'$, with respect to the topological space $(\calC',\calT')$.
\end{proof}

\begin{corollary}{(Fixed-Node Equivalence)}\\
Simulating a promised unique ground state $\psi_0(H;r\in\calC)$ of a CD-separately irreducible fermionic Schr\"odinger Hamiltonian $H=H^{\mathsmaller{C}}+H^{\mathsmaller{D}}$ supported by a continuous-discrete product configuration space $C=\calM\times\calP$ is equivalent to simulating the lowest Dirichlet eigenvector of $\Base(H)$ on a nodal groupoid $\calD\defeq\calN^{\mathsmaller{\clubsuit}}(H;q)$, $q\in\calC$, $\psi_0(H;q)\neq 0$, provided though that the boundary $\nabla\calD$ and the Dirichlet boundary conditions on $\nabla\calD$ are {\it a priori} known or efficiently computable, in the sense that for any given continuous path $\gamma:[0,1]\mapsto\calC$ of a polynomially bounded length, it can be efficiently decided whether $\gamma$ intersects $\partial\calD$, and $\forall r=(r^{\mathsmaller{C}},r^{\mathsmaller{D}})\in\calD$, it is efficient to compute and enumerate all of the configuration points $\{r'=(r'^{\mathsmaller{C}}=r^{\mathsmaller{C}},r'^{\mathsmaller{D}})\}$ such that $(r,r')\in\delta\calD$, together with the ratios of wavefunction values $\{\psi_0(H;r')/\psi_0(H;r)\}$, $\forall(r,r')\in\delta\calD$.
\label{FixedNodeMethod}
\vspace{-1.5ex}
\end{corollary}

This \iftoggle{ForUSPTO} {Derived Utility} {corollary} signifies the central importance of nodal groupoids and boundaries in QMC, which can be said to condense all of the computational hardness of simulating quantum systems on a classical probabilistic computer. A means to decide quickly the boundary and the Dirichlet boundary conditions of a nodal groupoid translates immediately to an efficient BPP algorithm for solving the ground state and energy of the concerned quantum system. Unfortunately, computing the nodal surface(s) of a general fermionic many-body system has proven to be exceedingly hard, except for the celebrated Lieb-Mattis solution \cite{Lieb62} for many fermions moving on a one-dimensional manifold diffeomorphic to one of the intervals $[0,1]$, $[0,1)$, $(0,1)$ on the real line, where the two-fermion-coincidence hyperplanes in the many-body configuration space are precisely the only nodal surfaces, because, within any interval on the real line, particle exchanges can only happen with two particles passing through each other on the line. With more spatial dimensions, however, two-fermion-coincidence hyperplanes in the many-body configuration space have less dimensions than, therefore cannot fix, the nodal surfaces.

The celebrated Feynman path integral and the corresponding path integral Monte Carlo (PIMC) provide general and powerful methods for computing and simulating Gibbs operators and their associated Gibbs wavefunctions and kernels \cite{Feynman72,Ceperley91,Ceperley95,Zinn-Justin05,Kozlov12,Feynman65,Simon79,Glimm87,Lorinczi11,
Blankenbecler81,Sugiyama86,Caffarel88I,Caffarel88II,Ceperley96,Zhang97,
Baroni98arXiv,Baroni99prl,Bernu02,Zhang04,Kruger08,Carleo10,Binder10}. Generally, it may be desired to compute a Gibbs operator, specifically its associated Gibbs wavefunctions and kernels, corresponding to a piecewise constant sequence of partial Hamiltonians $\{H_m : m \in [1,M]\}$, $M \in \mathbb{N}$, with each $H_m$, $m \in [1,M]$ being a constant operator applied to a many-fermion quantum system during an interval $(\tau_{m-1},\tau_m]$ of (imaginary) time, where $\{\tau_m : m \in [0,M]\} \subseteq [0,\infty)^{\mathsmaller{M}+1}$ is a predetermined sequence of time instants, with $\tau_0 = 0$, $\tau_m > \tau_{m-1}$, $\forall m \in [1,M]$. Within each time interval $(\tau_{m-1},\tau_m]$, the partial Hamiltonian $H_m$, $m \in [1,M]$ effects a Gibbs operator $G_m(-\delta\tau_mH_m) \defeq \exp[-(\tau_m{-}\tau_{m-1})H_m]$, whose integral kernel $K(\delta\tau_m H_m;\cdot,\cdot)$, namely, the Gibbs kernel, is analytically expressed as
\begin{equation}
K(\delta\tau_mH_m;r,q) \,=\, \langle r|P_{\!\mathsmaller{F}}\exp({-}\delta\tau_mH_{m{\mathsmaller{B}}})|q\rangle
\,=\, const \times {\textstyle{ {\scalebox{1.15}{$\sum$}}_{\pi\in G_{\rm ex}} }} \, \sign(\pi) \, K(\delta\tau_mH_{m{\mathsmaller{B}}};\pi r,q) \,, \label{GibbsWaveFuncSymm}
\end{equation}
with $H_{m{\mathsmaller{B}}} \defeq \Base(H_m)$ denoting the base boltzmannonic Hamiltonian of $H_m$. Numerically, in order to compute or approximate a Gibbs transition amplitude $K(\delta\tau_mH_m;r,q)$, $\forall m \in \mathbb{N}$, $\forall (r,q) \in \calC^2$, a PIMC procedure divides the interval $(\tau_{m-1},\tau_m]$ into $N_m \in \mathbb{N}$ smaller steps by inserting $N_m{-}1$ intermediate instants of time $\{\tau_{mn}\}_{n \in [1,\mathsmaller{N}_m-1]}$, and defining $\tau_{m0} \defeq \tau_{m-1}$, $\tau_{m\mathsmaller{N}_m} \defeq \tau_m$, such that each time step $\delta\tau_{mn} \defeq \tau_{mn} - \tau_{m(n-1)}$ is positive and sufficiently small, that the corresponding Gibbs operator $\exp(-\delta\tau_{mn}H_m)$ can be efficiently computed or approximated straightforwardly using a {\em small-time approximation} such as $\exp(-\delta\tau_{mn}H_m) = I - \delta\tau_{mn}H_m + O(\delta\tau_{mn}^2)$, or $\exp(-\delta\tau_{mn}H_m) = \exp(-\delta\tau_{mn}U_m)\exp(-\delta\tau_{mn}V_m) + O(\delta\tau_{mn}^2)$, or a higher-order formula of operator decomposition or splitting \cite{Creutz89,Suzuki90,Yoshida90,Suzuki91,Suzuki92,Yoshida93,Chin02,Suzuki06,Sakkos09,Dornheim15,Yan17}, where $U_m + V_m = H_m$ is an exemplary operator decomposition, such that both $\exp(-\delta\tau_{mn}U_m)$ and $\exp(-\delta\tau_{mn}V_m)$ have a closed-form analytical kernel.

Having the time interval $(0,\tau_{\mathsmaller{M}}]$ partitioned into a number $N_{\rm tot} \defeq \sum_{m=1}^{\mathsmaller{M}}N_m$ of small steps with the ordered sequence of time instants $\{\tau_{10} {\defeq} 0\} \cup \{\tau_{mn}\}_{m \in [1,\mathsmaller{M}], \, n \in [1,\mathsmaller{N}_m]}$, the essence of PIMC in the configuration coordinate representation is to discretize the Wiener space of continuous curves $\{\gamma:[0,\tau_{\mathsmaller{M}}]\mapsto\calC,\, \gamma \in C^0\}$ into a {\em cylinder set} $\prod_{t=0}^{\mathsmaller{N}_{\rm tot}} \calC_t \defeq \{(q_0,\cdots\!,q_{t-1},q_t,\cdots\!,q_{\mathsmaller{N}_{\rm tot}}):q_t\in\calC,\,\forall t\in[0,N_{\rm tot}]\} \cong \calC^{\mathsmaller{N}_{\rm tot}}$ \cite{Ito74,Kuo75,Karatzas91,Bar11}, that is a Cartesian product of $N_{\rm tot}{+}1$ labeled and ordered copies of the configuration space $\{\calC_t \cong \calC : t \in [0,N_{\rm tot}]\}$, with each element $\{q_t\}_{t \in [0,\mathsmaller{N}_{\rm tot}]} \in \prod_{t=0}^{\mathsmaller{N}_{\rm tot}} \calC_t$ representing a zigzag path made of consecutive segments of straight lines, then to approximate a path or functional integral by an $(N_{\rm tot}{+}1)$-fold multiple integral of elementary calculus \cite{Feynman65,Simon79,Glimm87,Lorinczi11,Ito74,Kuo75,Karatzas91,Bar11}. Let $N_{\rm tot}(m) \defeq \sum_{l=1}^mN_l$, $m \in [0,M]$ be an integer-valued function, and $\forall t \in [0,N_{\rm tot}]$, let $m(t) \in \mathbb{Z}$ be the smallest integer $m$ such that $N_{\rm tot}(m) \ge t$.

\begin{definition}{(Feynman Slices, Feynman Slabs, Feynman Stacks, Feynman Substacks, Feynman Paths, Feynman Flights, Feynman Bundles, Feynman Chains, Rectified Feynman Chains)}\label{defiFeynmanStuff}\\
With the same conventions and notations as the immediate above, for a given time instant $t \in [0,N_{\rm tot}]$, a Feynman slice denoted as $(\calC_t,H(t))$ is a couple formed by the configuration space $\calC_t \cong \calC$ and the instantaneous partial Hamiltonian $H(0) \defeq 0$ or $H(t) \defeq H_{m(t)}$, $t \in [1,N_{\rm tot}]$; $\forall m \in [1,M]$, let $t_{m-1} \defeq N_{\rm tot}(m{-}1)$, $t_m \defeq N_{\rm tot}(m)$, the $(N_m+1)$-tuple of Feynman slices $\{(\calC_t,H(t)) : t \in [t_{m-1},t_m]\} \defeq \Slab(\calC_{t_{m-1}},\calC_{t_m},H_m,t_m{-}t_{m-1})$ associated with a constant partial Hamiltonian $H(t) = H_m$ for $t \in (t_{m-1},t_m]$ form a Feynman slab; the $M$-tuple of all of the Feynman slabs $\{\{(\calC_t,H(t)) : t \in [t_{m-1},t_m]\} : m \in [1,M]\} \defeq \{\Slab(\calC_{t_{m-1}},\calC_{t_m},H_m,t_m{-}t_{m-1}) : m \in [1,M]\}$ form a Feynman stack; a single Feynman slab or a sequence of consecutive Feynman slabs form(s) a Feynman substack, which is a Feynman stack in its own right with respect to the corresponding start and end time instants.

With respect to a Feynman stack or substack, a couple formed by a zigzag path $\{q_t\}_{t \in [t',t'']} \in \prod_{t=t'}^{t''} \calC_t$, $0\le t' < t'' \le N_{\rm tot}$ and the sequence of partial Hamiltonians $\{H(t)\}_{t \in (t',t'']}$ is called a Feynman path from $t=t'$ to $t=t''$ that connects $q_{t'} \in \calC_{t'}$ and $q_{t''} \in \calC_{t''}$; a triple $(q_{t-1},q_t,H(t))$ with $t \in [1,N_{\rm tot}]$, $q_{t-1} \in \calC_{t-1}$, $q_t \in \calC_t$ is called a Feynman flight and considered as a portion of a Feynman path; for each Feynman slab $\{(\calC_t,H(t)) : t \in [t_{m-1},t_m]\} = \Slab(\calC_{t_{m-1}},\calC_{t_m},H_m,t_m{-}t_{m-1})$, $m \in [1,M]$, and each pair of configuration points $(r_{m-1},r_m) \in \calC_{t_{m-1}} \times \calC_{t_m}$, the collection of all of the Feynman paths from $t = t_{m-1}$ and $t = t_m$ that connect $r_{m-1}$ and $r_m$ is called a Feynman bundle and denoted as $\Bun(t_{m-1},r_{m-1};t_m,r_m)$.

With respect to a Feynman stack or substack that starts at $t_{m_1}$ and ends at $t_{m_2}$, with $0 \le m_1 < m_2 \le M$ and $t_m \defeq N_{\rm tot}(m)$ for any $m \in [m_1,m_2]$, $\forall (m_2{-}m_1{+}1)$-tuple $\{r_m\}_{m\in[m_1,m_2]}$ with $r_m  \in \calC_{t_m}$ for any $m \in [m_1,m_2]$, the $(m_2{-}m_1)$-tuple of Feynman bundles $\{\Bun(t_{m-1},r_{m-1};t_m,r_m)\}_{m \in (m_1,m_2]}$ is called a Feynman chain from $t_{m_1}$ to $t_{m_2}$, or from the $(m_1{+}1)$-th Feynman slab to the $m_2$-th Feynman slab; a Feynman chain $\{\Bun(t_{m-1},r_{m-1};t_m,r_m)\}_{m \in (m_1,m_2]}$ is called rectified if $\forall m \in (m_1,m_2]$, $r_m$ is always in a positive nodal cell of the Gibbs wavefunction $\langle\cdot|\exp[-(t_m{-}t_{m-1})H_m]|r_{m-1}\rangle$; a rectified Feynman chain $\{\Bun(t_{m-1},r_{m-1};t_m,r_m)\}_{m \in (m_1,m_2]}$ is said to be the rectification of another Feynman chain $\{\Bun(t_{m-1},q_{m-1};t_m,q_m)\}_{m \in (m_1,m_2]}$, $q_{m_1} = r_{m_1}$, $\{q_m\in\calC_{t_m}\}_{m\in(m_1,m_2]}$, or the latter being rectified into the former, when a set of fermion exchange operators $\{\pi_m \in G_{\rm ex}\}_{m\in(m_1,m_2]}$ exist such that $q_m = \pi_m r_m$ for all $m\in(m_1,m_2]$.
\vspace{-1.5ex}
\end{definition}

A specific PIMC method assigns to each Feynman flight a Gibbs transition amplitude through a small-time approximation, and assigns to each Feynman path a product of Gibbs transition amplitudes due to its constituent Feynman flights, then computes any Gibbs wavefunction or kernel by summing up (integrating over) the Gibbs transition amplitudes of all of the permitted Feynman paths. In particular, for each Feynman bundle $\Bun(t_{m-1},r_{m-1};t_m,r_m)$ with $t_{m-1} = N_{\rm tot}(m{-}1)$, $t_m = N_{\rm tot}(m)$, $m \in [1,M]$, $q_{m-1} \in \calC_{t_{m-1}}$, $r_m \in \calC_{t_m}$, the boltzmannonic Gibbs transition amplitude $\langle q_m|\exp[{-}(t_m{-}t_{m-1})H_{m\mathsmaller{B}}]|r_{m-1}\rangle$ associated with $H_{m{\mathsmaller{B}}} \defeq \Base(H_m)$ is obtained by integrating the corresponding Gibbs transition amplitudes of all of the paths in the Feynman bundle, which produces the boltzmannonic Gibbs wavefunction or kernel when $r_m$ or both $r_{m-1}$ and $r_m$ traverse(s) the respective configuration space(s). By equation (\ref{GibbsWaveFuncSymm}), the corresponding fermionic Gibbs wavefunctions or kernels can be obtained by symmetrizing the $r_m$ coordinate, that is, $\langle r_m|\exp[{-}(t_m{-}t_{m-1})H_m]|r_{m-1}\rangle = \langle r_m|P_{\!\mathsmaller{F}}\exp[{-}(t_m{-}t_{m-1})H_{m\mathsmaller{B}}]|r_{m-1}\rangle$, $\forall r_{m-1} \in \calC_{t_{m-1}}$, $\forall r_m \in \calC_{t_m}$. Next, for each Feynman chain $\{\Bun(t_{m-1},r_{m-1};t_m,q_m) : m \in (m_1,m_2]\}$ with respect to a Feynman stack or substack comprising the $(m_1+1)$-th to the $m_2$-th Feynman slabs, $0 \le m_1 < m_2 \le M$, a Gibbs transition amplitude associated with the Feynman chain is the product of the Gibbs transition amplitudes associated with its constituent Feynman bundles. Finally, integrating over all Feynman chains starting at $(t_{m_1},r_{m_1})$ and ending at $(t_{m_2},r_{m_2})$, with the intermediate coordinates traversing their respective configuration spaces, yields the Gibbs transition amplitudes for said Feynman stack or substack, namely,
\begin{align}
& \langle r_{m_2} \,|\, {\textstyle{ \prod_{m=m_1+1}^{m_2} }} \exp[-(t_m{-}t_{m-1})H_{m\mathsmaller{B}}] \,|\, r_{m_1}\rangle \,, \; \forall r_{m_1} \in \calC_{t_{m_1}} , \; \forall r_{m_2} \in \calC_{t_{m_2}} \nonumber \\[0.75ex]
\,=\,\;& \int_{\calC_{t_{m_2-1}}} \!\!\!\!\!\!\!\!\!\! \cdots \,\, \int_{\calC_{t_{m_1+1}}} \! {\textstyle{ \prod_{m=m_1+1}^{m_2} }} \langle r_m \,|\, \exp[-(t_m{-}t_{m-1})H_{m\mathsmaller{B}}] \,|\, r_{m-1}\rangle \; {\textstyle{ \prod_{m=m_1+1}^{m_2-1} }} dr_m \label{BoltzGibbsFeynStack}
\end{align}
for the boltzmannonic Gibbs wavefunctions or kernels, and
\begin{align}
& \langle r_{m_2} \,|\, P_{\!\mathsmaller{F}} \, {\textstyle{ \prod_{m=m_1+1}^{m_2} }} \exp[-(t_m{-}t_{m-1})H_m] \,|\, r_{m_1}\rangle \,, \; \forall r_{m_1} \in \calC_{t_{m_1}} , \; \forall r_{m_2} \in \calC_{t_{m_2}} \nonumber \\[0.75ex]
\,=\,\;& \int_{\calC_{t_{m_2-1}}} \!\!\!\!\!\!\!\!\!\! \cdots \,\, \int_{\calC_{t_{m_1+1}}} \! {\textstyle{ \prod_{m=m_1+1}^{m_2} }} \langle r_m \,|\, P_{\!\mathsmaller{F}} \exp[-(t_m{-}t_{m-1})H_{m\mathsmaller{B}}] \,|\, r_{m-1}\rangle \; {\textstyle{ \prod_{m=m_1+1}^{m_2-1} }} dr_m \label{FermiGibbsFeynStack}
\end{align}
for the fermionic counterparts. It is noted that, in equation (\ref{FermiGibbsFeynStack}), suffice it to keep only one of the fermionic symmetrization projection operators acting on one of the configuration spaces $\{\calC_{t_m}\}_{m\in(m_1,m_2]}$, with all of the rest $P_{\!\mathsmaller{F}}$ operators being redundant whose presence and absence produce the same or equivalent result. Nevertheless, it is sometimes convenient to keep the $P_{\!\mathsmaller{F}}$ operator for each Feynman bundle or Feynman slab, so that the fermionic Gibbs wavefunctions or kernels of the Feynman slab are handy, whose numerical signs and sign changes are of interest.

{\bf Context Switch Notice:} In the following plurality of paragraphs until another reminder of context switch like this one, many-body physical systems are always assumed to be moving on a connected continuous Riemannian manifold, involving no discrete dynamical variables or having all discrete dynamical variables frozen.

\begin{corollary}{(Maximal Property of Nodal Cells for Two-Fermion Systems)}\\
The ground state of any fermionic Schr\"odinger system consisting of two identical fermions moving on a connected Riemannian manifold, when non-degenerate, has exactly one connected positive nodal cell.
\label{OnePositiveNodalCell}
\vspace{-2.5ex}
\end{corollary}
\begin{proof}[\iftoggle{ForUSPTO} {Demonstration} {Proof}]
When limited to systems consisting of two identical fermions moving on a connected Riemannian manifold, the exchange symmetry group has only two elements, one being even, the other being odd. Then it follows from {\myLemma} \ref{TilingProperty} that the ground state $\psi_0$ can only have two nodal cells, one being positive, the other negative, denoted by $\calN^+(\psi_0)$ and $\calN^-(\psi_0)$ respectively.
\vspace{-1.5ex}
\end{proof}

For a general single-species fermionic Schr\"odinger system moving on a connected Riemannian manifold, the maximal property of ground state nodal cells, {\it i.e.}, there is but one positive nodal cell, and consequently only one negative nodal cell, has been conjectured by Ceperley, who was also able to confirm the property numerically for systems consisting of up to $200$ identical fermions moving on a two- or three-dimensional torus \cite{Ceperley91}. Later, Mitas provided a rigorous proof for certain classes of single-species fermionic systems \cite{Mitas06prl,Mitas06arx}, particularly those of spin-polarized non-interacting fermions moving in a quadratic potential well or on a potential-flat torus, whose wavefunctions can be factored as a nodeless function multiplying an alternant \cite{Muir60}, that is, the determinant of an alternant matrix of particle coordinate variables, thus enabling an analytical representation and rigorous determination of the nodal surface. It has remained open whether or when the maximal property of nodal cells holds for a general interacting single-species fermionic Schr\"odinger system, until now. As will be asserted in {\myLemma} \ref{NodalCellTubeOneToOne} below, the maximal property of nodal cells does hold for the ground state of a single-species fermionic Hamiltonian $H(1)$, if there exists a $C^1$ curve $\{H(t)\in\calL_0(\calM):t\in[0,1]\}$ in the vector space $\calL_0(\calM)$, such that the ground state of $H(t)$ is non-degenerate, $\forall t\in[0,1]$, and that of $H(0)$ satisfies the maximal property of nodal cells.

On the other hand, it can be shown that a system of multi-particle multi-species has more than one disconnected positive nodal cells in general. To that end, just consider a system consisting of $S\ge 2$ fermion species, with at least two species indexed by $s=1,2$ having at least two identical fermions, where different species do not interact, so that the Hamiltonian is of a form $H(0)=\sum_{s=1}^{\mathsmaller{S}}H_s$, with each $H_s$, $s\in[1,S]$ representing the total energy of the subsystem of all of the fermions of the $s$-th species. Furthermore, the ground state of each $H_s$, $s\in[1,S]$ is assumed non-degenerate, so that $\psi_0(H_s)$ is a unique ray in the corresponding projective Hilbert space, and the unique ground state of the overall system of $S$-specie is a separated product $\psi_0(H(0))=\prod_{s=1}^{\mathsmaller{S}}\psi_0(H_s)$. Pick one positive nodal cell for each $\psi_0(H_{s'})$, and denote it by $\calN_{s'}^+$, $\forall s'\in[1,S]$. For each $s\in\{1,2\}$, there is at least one nontrivial element, denoted as $\pi_s\neq 1$, in the symmetry group $G_s$ of exchanging $n_s\ge 2$ identical fermions, $n_s\in\mathbb{N}$, so there must be at least one positive and one negative nodal cells, denoted as $\calN_s^+$ and $\calN_s^-$ respectively, such that $\calN_s^+\cap\calN_s^-=\emptyset$ and $\pi_s\calN_s^+=\calN_s^-$. Clearly, the product subspace $\prod_{s=1}^{\mathsmaller{S}}\calN^+_s$ is one positive nodal cell for the overall ground state $\psi_0(H_0)$, and so is $\pi_1\pi_2\prod_{s=1}^{\mathsmaller{S}}\calN^+_i=\calN_1^-\times\calN_2^-\times\prod_{s=3}^{\mathsmaller{S}}\calN^+_i$, such that $\scalebox{1.15}{(}\prod_{s=1}^{\mathsmaller{S}}\calN^+_i\scalebox{1.15}{)}\cap\scalebox{1.15}{(}\calN_1^-\times\calN_2^-\times\prod_{s=3}^{\mathsmaller{S}}\calN^+_i\scalebox{1.15}{)}=\emptyset$.

Let $(\calM,\calF_{\mathsmaller{\calM}},V_g)$ denote a $\sigma$-finite Lebesgue measure space, with $\calM$ being a connected, locally compact Riemannian manifold, and $\calF_{\mathsmaller{\calM}}$ being a $\sigma$-algebra of Lebesgue measurable subsets of $\calM$, and $V_g$ a Lebesgue measure associated with the Riemannian metric $g$ on the manifold $\calM$. Let $[\emptyset]\subset\calF_{\mathsmaller{\calM}}$ be the equivalence class of null sets with $V_g$-measure zero, and $\calF_{\mathsmaller{\calM}}^*$ be the quotient space of $\calF_{\mathsmaller{\calM}}$ modulo $[\emptyset]$. For any pair of measurable sets $E,F\in\calF_{\mathsmaller{\calM}}$, let $[E],[F]\in\calF_{\mathsmaller{\calM}}^*$ denote the corresponding equivalence classes of sets, and $d([E],[F])=V_g((E\cap F^c)\cup(E^c\cap F))$ denote the measure of symmetric difference, then $(\calF_{\mathsmaller{\calM}}^*,d)$ is a metric space, which is a separable topological space \cite{Halmos74,Bogachev07}, and actually second-countable.

\begin{definition}{($\sigma$-Continuity of Set-Valued Functions)}\label{SigmaCont}\\
Let $\calX$ be a topological space. A set-valued function $f:\calX\mapsto\calF_{\mathsmaller{\calM}}^*$ is called $\sigma$-continuous, if there exist a countable number of continuous set-valued functions $f_n:\calX\mapsto\calF_{\mathsmaller{\calM}}^*$, $n\in\mathbb{N}$, such that $f_n(x)$ is a compact subset of $\calM$ for each $(n,x)\in(\mathbb{N}\times\calX)$, $\bigcup_{n\in\mathbb{N}}f_n(x)=\calM$ for each $x\in\calX$, and $f\cap f_n:\calX\mapsto\calF_{\mathsmaller{\calM}}^*$ is continuous for each $n\in\mathbb{N}$, where $f\cap f_n$ denotes a range-restricted function such that $\forall x\in\calX$, $(f\cap f_n)(t)\defeq f(x)\cap f_n(x)$.
\end{definition}

\begin{definition}{($C^{1,2}\!$ Path and Piecewise $C^{1,2}\!$ Path of Hamiltonians)}\label{defiC1Path}\\
Let $\calI$ be a connected subset of $\mathbb{R}$. A $C^{1,2}\!$ path of Hamiltonians is a fermionic Schr\"odinger operator-valued curve $\{H(t)\in\calL_0(\calM):t\in\calI\}$ in the vector space $\calL_0(\calM)$ over a connected Riemannian manifold $\calM$, possibly accompanied by a $\sigma$-continuous set-valued function $\calM_{\rm sin}\!:\calI\mapsto\calF_{\mathsmaller{\calM}}^*$, such that the Gibbs kernel function $\langle r|\exp(\mathsmaller{-}\Base(H(t)))|q\rangle$ is $(t,r,q)$-jointly $C^{1,0,0}$ in $\calI\times\calM^2$, and $\forall t\in\calI$, $\calM_{\rm sin}(t)\subset\calM$ is a submanifold with $\dim(\calM_{\rm sin}(t))<\dim(\calM)$, $\langle r|\exp(\mathsmaller{-}\Base(H(t)))|q\rangle$ is positivity improving and in $L^2(\calM^2)\cap C^{2,2}((\calM\setminus\calM_{\rm sin}(t))^2)$, the ground state of $H(t)$ is non-degenerate and placed at the energy level $0$, while the operator $\Base(H(t))$ substantiates the Hopf lemma, $\forall t\in\calI$. A piecewise $C^{1,2}\!$ path of Hamiltonians is a concatenation (also known as product) of a finite number of $C^{1,2}\!$ paths of Hamiltonians.
\vspace{-1.5ex}
\end{definition}

Let the index subset $\calI=\bigcup_{k=1}^{\mathsmaller{K}}\calI_k$ be a union of consecutive intervals $\{\calI_k\}_{k=1}^{\mathsmaller{K}}$, $K\in\mathbb{N}$, where the closures of adjacent intervals intersect at a single point on the real line, $\cl(\calI_k)\cap\cl(\calI_{k+1})=\{a_k\}$, $a_k\in\mathbb{R}$, $\forall k\in[1,K{-}1]$. For a piecewise $C^{1,2}\!$ path of Hamiltonians $\{H(t)\in\calL_0(\calM):t\in\calI\}$, the Gibbs kernel function $\langle r|\exp(\mathsmaller{-}H(t))|q\rangle$ as a linear combination of exchange-symmetric versions of $\langle r|\exp(\mathsmaller{-}\Base(H(t)))|q\rangle$ is by definition $(t,r,q)$-jointly $C^{1,0,0}$ in each $\calI_k\times\calM^2$, $k\in[1,K]$, and $(r,q)$-jointly $C^{2,2}$ in $(\calM\setminus\calM_{\rm sin}(t))^2$, $\forall t\in\calI$, $\forall k\in[1,K]$. It follows from equation (\ref{PsiIsC12}) that the $t$-dependent, non-degenerate ground state $\psi_0(t;q)=\psi_0(H(t);q)$ is in $C^{1,0}(\calI_k\times\calM)$, $\forall k\in[1,K]$ and in $C^2(\calM\setminus\calM_{\rm sin}(t))$, $\forall t\in\calI$, $\forall k\in[1,K]$.

Therefore, the preimage $\psi_0^{\mathsmaller{-}1}(\mathbb{R}\setminus\{0\})\defeq\{(t,q)\in\calI\times\calM:\psi_0(t;q)\neq 0\}$ is an open subset of $\calI\times\calM$, a connected component of which containing a given point $(t',q')\in\calI\times\calM$, denoted by $\calT_{t'q'}$, is called a {\em nodal tube} containing the point $(t',q')$. Let $\calN_{t'q'}\defeq\calN(H(t');q')$, $(t',q')\in\calI\times\calM$ be a typical nodal cell of the ground state $\psi_0(H(t'))$ of the instantaneous Hamiltonian $H(t')$ at the instant $t=t'$. Clearly, $\calN_{t'q'}$ is contained in the nodal tube $\calT_{t'q'}$, and is actually the intersection of the two submanifolds $\calT_{t'q'}\subseteq(\calI\times\calM)$ and $\{t'\}\times\calM\subseteq(\calI\times\calM)$. It is said that the nodal cell $\calN_{t'q'}$ grows into the nodal tube $\calT_{t'q'}$, or the nodal tube $\calT_{t'q'}$ is grown from the nodal cell $\calN_{t'q'}$.

Similar to \iftoggle{ForUSPTO} {Auxiliary Utility} {Lemma} \ref{TilingProperty}, by considering the group action of $G_{\rm ex}$ on nodal tubes, it can be shown that any nodal tube $\calT_{t'q'}$, $(t',q')\in(\calI\times\calM)$ is a $G_{\rm ex}$-tile of the product manifold $\calI\times\calM$, namely, $\forall\pi_1,\pi_2\in G_{\rm ex}$, either $\pi_1\calT_{t'q'}=\pi_2\calT_{t'q'}$, or $\pi_1\calT_{t'q'}\cap\pi_2\calT_{t'q'}=\emptyset$, what is more, $\bigcup_{\pi\in G_{\rm ex}}\pi\calT_{t'q'}=(\calI\times\calM)$. Having a point $(t',q')\in(\calI\times\calM)$ fixed, consider a set-valued function $\calT_{t'q'}(\cdot):\calI\mapsto\calF_{\mathsmaller{\calM}}^*$, which sends each $t\in\calI$ to $\calT_{t'q'}(t)\defeq\!\left(\calT_{t'q'}\cap(\{t\}\times\calM)\right)\cup[\emptyset]\defeq\!\left\{\left(\calT_{t'q'}\cap(\{t\}\times\calM)\right)\cup Z:Z\in[\emptyset]\right\}\!\in\calF_{\mathsmaller{\calM}}^*$, then the nodal tube $\calT_{t'q'}$ is the graph of the function $\calT_{t'q'}(\cdot)$, modulo $[\emptyset]$.

\begin{lemma}{(One-to-One Correspondence between Nodal Cells and Nodal Tubes)}\label{NodalCellTubeOneToOne}\\
Given a piecewise $C^{1,2}\!$ path of Hamiltonians $\left\{H(t)\in\calL_0(\calM):t\in\calI=\bigcup_{k=1}^{\mathsmaller{K}}I_k\right\}$, and an instant $t'\in\calI$, let $\{\calN_{t'}\}\defeq\{\calN_{t'q'}\!:q'\in\{q^{(n)}\}_{n=1}^{\mathsmaller{N}}\}$, $N\in\mathbb{N}$ denote the set of all nodal cells associated with the instantaneous Hamiltonian $H(t')$, with each $q^{(n)}\in\calM$, $n\in[1,N]$ being contained in and indexing one unique nodal cell, then $\{\calT_{t'}\}\defeq\{\calT_{t'q'}\!:q'\in\{q^{(n)}\}_{n=1}^{\mathsmaller{N}}\}$ lists all of the nodal tubes associated with the piecewise $C^{1,2}\!$ path, and there exists a bijection $\bfg:\{\calN_{t'}\}\mapsto\{\calT_{t'}\}$, which commutes with all of the elements of the exchange symmetry group $G_{\rm ex}$, such that $\bfg(\calN_{t'q'})=\calT_{t'q'}$, $\forall q'\in\{q^{(n)}\}_{n=1}^{\mathsmaller{N}}$, and $\bfg\pi=\pi\bfg$, $\forall\pi\in G_{\rm ex}$. Furthermore, for any fixed point $(t',q')\in(\calI\times\calM)$, the set-valued function $\calT_{t'q'}(\cdot)$ from $\calI$ to the metric space $(\calF_{\mathsmaller{\calM}}^*,d)$ is $\sigma$-continuous, and $\int_{\scalebox{0.65}{$\calT_{t'q'}(t)$}\!}|\psi_0(H(t);q)|^2dV_g(q)$ remains a constant fraction of $\int_{\scalebox{0.65}{$\calM$}\!}|\psi_0(H(t);q)|^2dV_g(q)$, $\forall t\in\calI$.
\end{lemma}
\vspace{-3.0ex}
\begin{proof}[\iftoggle{ForUSPTO} {Demonstration} {Proof}]
Suffice it to show that all of the assertions hold true for a single piece of $C^1$ path $\{H(t)\in\calL_0(\calM):t\in[0,1]\}$ with $H(\cdot)$ being $C^1([0,1])$, such that the Gibbs kernel function $\langle r|\exp(\mathsmaller{-}H(t))|q\rangle$ is $(t,r,q)$-jointly $C^{1,0,0}$ in $[0,1]\times\calM^2$ and $(r,q)$-jointly $C^{2,2}$ in $(\calM\setminus\calM_{\rm sin}(t))^2$, $\forall t\in[0,1]$, thus the never degenerate ground state $\psi_0(H(t);q)$ is $(t,q)$-jointly $C^{1,0}$ in $[0,1]\times\calM$ and $C^2$ in $\calM\setminus\calM_{\rm sin}(t)$, $\forall t\in[0,1]$. Given any $t'\in[0,1]$ and the accordingly defined $\{\calN_{t'}\}$ and $\{\calT_{t'}\}$, it is obvious that each nodal cell $\calN_{t'q'}\in\{\calN_{t'}\}$ with $q'\in\calM$ such that $\psi_0(H(t');q')\neq 0$ corresponds to a nodal tube $T_{t'q'}\in\{\calT_{t'}\}$ through their common reference to the point $(t',q')\in[0,1]\times\calM$, which defines a mapping $\bfg:\{\calN_{t'}\}\mapsto\{\calT_{t'}\}$. By the $G_{\rm ex}$-invariance of $H(t)$ and $|\psi_0(H(t);\cdot)|$, $\forall t\in[0,1]$, it is obvious that $\bfg$ commutes with any $\pi\in G_{\rm ex}$, namely, for each nodal cell $\calN_{t'q'}\in\{\calN_{t'}\}$ that grows into a nodal tube $\bfg(\calN_{t'q'})=T_{t'q'}\in\{\calT_{t'}\}$, the $\pi$-transformed nodal cell $\pi\calN_{t'q'}$ grows into the $\pi$-transformed nodal tube $\pi\calT_{t'q'}$. To prove that $\bfg$ is bijective, it is sufficient to show that for any fixed point $(t',q')\in(\calI\times\calM)$, the set-valued function $\calT_{t'q'}(\cdot):[0,1]\mapsto(\calF_{\mathsmaller{\calM}}^*,d)$ is $\sigma$-continuous and never vanishes over its entire domain, namely, $V_g(\calT_{t'q'}(t))>0$, $\forall t\in[0,1]$.

For any $\delta>0$, for any $t\in[0,1]$, let $\calM_{\raisebox{0.1\height}{\tiny \rm sin}}^{\raisebox{-0.1\height}{\tiny $\delta$}}(t)\defeq\bigcup_{q\,\in\,\calM_{\rm sin}(t)}B_g(q,\delta)$ be a slightly blown-up open set enclosing the submanifold of singularity $\calM_{\rm sin}(t)$, with $B_g(q,\delta)\defeq\{r\in\calM:\dist_g(r,q)<\delta\}$. Pick one point $O\in\calM$ and call it an origin. Given any $\epsilon\in(0,1)$, and for any $t\in[0,1]$, let $\delta(\epsilon)>0$ be chosen such that
\begin{equation}
\int_{\calM_{\rm sin}^{\delta(\epsilon)}(t)}|\psi_0(H(t);q)|^2\,dV_g(q)\,<\,\frac{\epsilon}{2}\int_{\calM}|\psi_0(H(t);q)|^2\,dV_g(q), \nonumber
\end{equation}
let $\bfB_{\epsilon}(t)\defeq\cl(B_g(O,R(\epsilon);t))$ be the smallest closed ball centered at the origin $O$ with a radius $R(\epsilon)>0$, such that, with $\calB_{\epsilon}(t)\defeq\bigcup_{\pi\in G_{\rm ex}}\!\pi\bfB_{\epsilon}(t)$, $\calM_{\epsilon}(t)\defeq(\calM\setminus\calM_{\raisebox{0.1\height}{\tiny \rm sin}}^{\raisebox{-0.1\height}{\tiny $\delta(\epsilon)$}}(t))\cap\calB_{\epsilon}(t)$, it holds that
\begin{equation}
\int_{\calM_{\epsilon}(t)}|\psi_0(H(t);q)|^2\,dV_g(q)\,\ge\,(1-\epsilon)\int_{\calM}|\psi_0(H(t);q)|^2\,dV_g(q). \nonumber
\end{equation}
Then the mapping $\calM_{\epsilon}(\cdot):t\in[0,1]\mapsto\calM_{\epsilon}(t)\in\calF_{\mathsmaller{\calM}}^*$ is a continuous function, and the union of $\calM_{\epsilon}(t)$-restricted nodal surfaces $\partial\calN_{\epsilon}(t)\defeq\!\bigcup_{\calN_{tq}\in\{\calN_t\}}\partial(\calN_{tq}\cap\calM_{\epsilon}(t))$ is a compact submanifold, which is $C^2$-smooth by the implicit function theorem and the $C^2$-continuity of $\psi_0(H(t);\cdot)$. Therefore, the norm of the gradient, $\|\partial_q\psi_0(H(t),q)\|$, does not vanish on $\partial\calN_{\epsilon}(t)$ and attains a minimum $\psi'_{\epsilon}(t)\defeq\min\{\|\partial_q\psi_0(H(t),q)\|:q\in\partial\calN_{\epsilon}(t)\}>0$, $\forall t\in[0,1]$ by the premise of the operator $\Base(H(t))$ substantiating the Hopf lemma. Moreover, the function $\psi'_{\min}(\cdot):[0,1]\mapsto\mathbb{R}$ is $C^1([0,1])$ and strictly positive, so $\psi^*_{\epsilon}=\min\{\psi'_{\epsilon}(t):t\in[0,1]\}$ exists and is positive.

For any fixed point $(t',q')\in(\calI\times\calM)$, and with respect to a fixed $\epsilon\in(0,1)$, the $\calM_{\epsilon}(t)$-restricted nodal tube $\calT_{t'q'}(t)\cap\calM_{\epsilon}(t)$ of $\psi_0(H(t);\cdot)$ is a continuous set-valued function in $t\in[0,1]$, because its boundary $\partial(\calT_{t'q'}(t)\cap\calM_{\epsilon}(t))\subseteq((\partial\calT_{t'q'}(t))\cap\calM_{\epsilon}(t))\cup(\partial\calM_{\epsilon}(t))$ varies continuously in $t\in[0,1]$. That $\partial\calM_{\epsilon}(\cdot)$ is in $C^1([0,1])$ follows straightforwardly from the $C^{1,2}$-continuity and $L^2$-integrability of $\psi_0(H(\cdot),\cdot)$, while the coordinate $q\in\calM$ of points on $(\partial\calT_{t'q'}(t))\cap\calM_{\epsilon}(t)$, $t\in[0,1]$ satisfies the level-set equation \cite{Sethian99} $\psi_0(H(t),q(t))=0$, and the differential version $(\partial\psi_0/\partial t)+(\partial_q\psi_0)\cdot(dq/dt)=0$, which yields an evolution equation for the $t$-dependence of the nodal surface \cite{Sethian99},
\begin{equation}
\frac{dq(t)}{dt}\,=\,-\frac{\partial\psi_0(H(t);q(t))}{\partial t}\;\frac{\partial_q\psi_0(H(t);q(t))}{\|\partial_q\psi_0(H(t);q(t))\|^2}\,,
\end{equation}
and guarantees the $C^1$-continuity of $(\partial\calT_{t'q'}(t))\cap\calM_{\epsilon}(t)$, by the $C^{1,2}$-continuity of $\psi_0(H(t);q)$ and the $\psi^*_{\epsilon}>0$ lower-boundedness of $\|\partial_q\psi_0(H(t);q)\|$. Therefore, the set-valued function $(\calT_{t'q'}\cap\calM_{\epsilon})(\cdot):t\in[0,1]\mapsto\calT_{t'q'}(t)\cap\calM_{\epsilon}(t)\in\calF_{\mathsmaller{\calM}}^*$ is continuous, so are the integrals $\int_{\calT_{t'q'}(t)\cap\,\calM_{\epsilon}(t)}|\psi_0(H(t);q)|^2\,dV_g(q)$ and $\int_{\calM_{\epsilon}(t)}|\psi_0(H(t);q)|^2\,dV_g(q)$ continuous in $t$, $\forall t\in[0,1]$.

By \iftoggle{ForUSPTO} {Auxiliary Utility} {Lemma} \ref{TilingProperty}, for any fixed $t\in[0,1]$, the subset $\calT_{t'q'}(t)\cap\calM_{\epsilon}(t)$ tessellates the whole compact space $\calM_{\epsilon}(t)$, under the homeomorphism action of the exchange symmetry group $G_{\rm ex}$, hence there exists an integer $N(\epsilon;t',q';t)\in[1,|G_{\rm ex}|]$ such that
\begin{equation}
\frac{\int_{\calT_{t'q'}(t)\cap\,\calM_{\epsilon}(t)}|\psi_0(H(t);q)|^2\,dV_g(q)}{\int_{\calM_{\epsilon}(t)}|\psi_0(H(t);q)|^2\,dV_g(q)}\,=\,\frac{1}{N(\epsilon;t',q';t)},\;\;\forall t\in[0,1], \label{OneOverNt}
\end{equation}
since $|\psi_0(H(t);\cdot)|^2$ is invariant under any particle exchange operation in $G_{\rm ex}$. To be simultaneously continuous and discrete-valued, the fraction in equation (\ref{OneOverNt}) can only be a constant $1/N_0$, with $N(\epsilon;t',q';t)=N_0\in\mathbb{N}$, $\forall t\in[0,1]$. In the limit $\epsilon\rightarrow 0$, it is necessary that $\calM_{\epsilon}(t)\rightarrow\calM$, in the sense that $V_g(\calM\setminus\calM_{\epsilon}(t))\rightarrow 0$, $\forall t\in[0,1]$, otherwise, there would exist a $t\in[0,1]$ such that the Gibbs kernel function $\langle r|\exp(\mathsmaller{-}\Base(H(t)))|q\rangle$ is not positivity improving. As $\epsilon$ approaches $0$, the fraction in equation (\ref{OneOverNt}) remains the same constant, and the set-valued function $\calT_{t'q'}(\cdot)\cap\calM_{\epsilon}(\cdot):[0,1]\mapsto(\calF_{\mathsmaller{\calM}}^*,d)$ is always continuous. That establishes the bijectivity of the function $\bfg$, for any predetermined $t'\in[0,1]$.
\end{proof}

\begin{corollary}{(Nodal Maximality for Conventional Fermionic Schr\"odinger Systems)}\\
Ceperley's conjecture on the nodal maximality, positing that any non-degenerate ground state has exactly one positive nodal cell, is true for a single-spices system consisting of $n\in\mathbb{N}$ identical spinless fermions moving in a smooth $(d\ge 2)$-dimensional Riemannian substrate space $(\calM_0,g_0)$, where the physics is governed by a conventional Schr\"odinger operator $H_1={-}\Delta_g+V_1(q\in\calM)$ on the configuration space $(\calM\defeq\calM_0^n,g)$, with the potential $V_1$ being Kato-decomposable \cite{Kato80,Lorinczi11,Simon82} and $C^2$-smooth in $\calM\setminus\calM_{\rm sin}$, where $\calM_{\rm sin}\subset\calM$ is a submanifold with $\dim(\calM_{\rm sin})<\dim(\calM)$.
\label{CeperleysConjectureIsTrue}
\vspace{-2.5ex}
\end{corollary}
\begin{proof}[\iftoggle{ForUSPTO} {Demonstration} {Proof}]
Construct a path of Hamiltonians $\{H(t):t\in[0,1]\}$ that is not only piecewise $C^1$-continuous but actually piecewise analytic \cite{Krantz02,Rellich69,Kriegl97} with
\begin{equation}
H(t) \defeq \begin{cases}
(1-2t)H_0+2tH_1+U_0(t;q\in\calM)+2tU_1(q\in\calM), & t\in[0,1/2],\\
H_1+2(1-t)U_1(q\in\calM), & t\in(1/2,1],
\end{cases}
\end{equation}
where $H_0=-\Delta_g+V_0$, with $V_0$ being a bounded, $C^2(\calM)$-smooth, and quasiconvex potential having a deep global minimum, such that $H_0=-\sum_{i=1}^{nd}g^{ii}(\partial/\partial x^i)^2-\sum_{i=1}^{nd}h^i(\partial/\partial x^i)+\sum_{i=1}^{nd}k_ig_{ii}x^ix^i+\min(V_0)$, in a small neighborhood covered by a local coordinate $(x^1,\cdots\!,x^{nd})$ around the global potential minimum, where $(g_{ii})_{i=1}^{nd}$ and $(g^{ii})_{i=1}^{nd}$ are diagonalized representations of the Riemannian metric $g$ and $g^{{-}1}$ respectively, $h^i$ is a smooth function related to the Riemannian metric, $k_i>0$ is a constant, while $U_0$ is a $t$-analytic, $\calM$-diagonal, bounded, and $C^2(\calM)$-smooth potential that vanishes at $t=0$ and $t=1/2$, and $U_1$ is $t$-independent, $\calM$-diagonal, bounded, and $C^2(\calM)$-smooth. The $V_0$ potential can be designed to have the constants $\{k_i:i\in[1,nd]\}$ in the Hamiltonian $H_0$ so large that all the other coefficients $\{g_{ii},g^{ii},h^i:i\in[1,nd]\}$ can be taken as constants in the neighborhood of $\arg\min(V_0)$, without any substantial change in the relevant low-lying eigenstates of $H_0$, which essentially represents an ideal harmonic oscillator, whose fermionic Schr\"odinger ground state is non-degenerate. More specifically, with a local coordinate $(x^1,\cdots\!,x^{nd})=((x^{\nu\delta})_{\delta=1}^{d})_{\nu=1}^{n}$ for said neighborhood of $\arg\min(V_0)\in\calM$ explicitly manifesting its structure of Cartesian product, in that for each $\nu\in[1,n]$, the $d$-tuple $(x^{\nu\delta})_{\delta=1}^{d}\in\calM_0$ represents the position of the $\nu$-th artificially labeled identical fermion in the substrate space $\calM_0$, one good choice of the parabolic potential is $V_0\sim\min(V_0)+\sum_{\nu=1}^n\sum_{\delta=1}^dk_{\delta}g_{0,\delta\delta}x^{\nu\delta}x^{\nu\delta}$, with $(g_{0,\delta\delta})_{\delta=1}^d$ being a diagonalized representation of the Riemannian metric $g_0$ for the substrate space $\calM_0$, and the positive constants $(k_{\delta})_{\delta=1}^d$ satisfying $n\max(k_1^{\mathsmaller{1/2}}g_{0,11},k_2^{\mathsmaller{1/2}}g_{0,22})<k_{\delta}^{\mathsmaller{1/2}}g_{0,\delta\delta}$, $\forall\delta\in[3,d]$, such that, among the $d$ decoupled one-dimensional modes of harmonic oscillations, each along one orthogonal axis in the neighborhood on the substrate space $\calM_0$, only two can be excited and occupied by the identical fermions, while the rest $(d-2)$ modes are all frozen to their ground states. Furthermore, the constants are chosen to make sure that no integers $n_1\in[1,n]$, $n_2\in[1,n]$ exist to satisfy $n_1k_1^{\mathsmaller{1/2}}g_{0,11}=n_2k_2^{\mathsmaller{1/2}}g_{0,22}$, so the ground state of the $n$-fermion system cannot be degenerate. By definition, the ground state of $H(1)$ is also non-degenerate. If suitable potentials $U_0$ and $U_1$ can be chosen such that the ground state of $H(t)$ is never degenerate, $\forall t\in[0,1]$, then all of the premises in \iftoggle{ForUSPTO} {Auxiliary Utility} {Lemma} \ref{NodalCellTubeOneToOne} are fulfilled to assert that the ground state of $H(1)=H_1$ has exactly the same number of positive nodal cells as that of $H(0)=H_0$, which as a system of non-interacting identical spinless fermions moving in a parabolic potential is known to have Ceperley's conjecture hold true \cite{Mitas06prl}.

It is without loss of generality to assume $\calM$ being compact and $V_1$ being in $C^2(\calM)\cap L^{\infty}(\calM)$, because, otherwise, compact sets $\calB_{\epsilon}(1)\defeq\bigcup_{\pi\in G_{\rm ex}}\!\pi\bfB_{\epsilon}(1)$, $\calM_{\epsilon}(1)=(\calM\setminus\calM_{\raisebox{0.1\height}{\tiny \rm sin}}^{\raisebox{-0.1\height}{\tiny $\delta(\epsilon)$}}(1))\cap\calB_{\epsilon}(1)$ as defined in the \iftoggle{ForUSPTO} {demonstration} {proof} of \iftoggle{ForUSPTO} {Auxiliary Utility} {Lemma} \ref{NodalCellTubeOneToOne} can be identified with $\epsilon>0$, $\delta(\epsilon)>0$, and a bounded $C^2(\calM)$-smooth potential $V'_1$ can be defined, which coincides with $V_1$ in $\calM_{\epsilon}(1)$ and arises quickly to a sufficiently large positive constant outside $\calB_{\epsilon}(1)$, so to provide a potential barrier trapping a sufficient number of bound states. Let $H'_1\defeq-\Delta_g+V'_1$ and $U_1\defeq V'_1-V_1$. As multiplicative operators, $V'_1$ is bounded and $V_1$ is $(-\Delta_g)$-bounded with relative bound zero \cite{Kato80,Hislop96,Lorinczi11}. Therefore, $U_1$ is relatively bounded and infinitesimally small with respect to both $H_0$ and $H_1$, as well as $H'_1$. Since both $H_0$ and $H_1$ have a non-degenerate ground state as a discrete eigenvector, a sufficiently large radius of $\calB_{\epsilon}(1)$ can be chosen such that the original potentials $V_0$ and $V_1$ have exponentially attenuated the ground state wavefunctions \cite{Agmon82,Hislop96} effectively to below a sufficiently small $\epsilon>0$ that restriction of the configuration space to the compact subspace $\calB_{\epsilon}(1)$ does not change the physics as far as the low-lying energy states are concerned, at the same time, the perturbation due to $U_1|_{\mathsmaller{\calM\setminus\calM_{\epsilon}(1)}}$ is so small that the ground state of the analytic path of Hamiltonians $\{H(t):t\in[1/2,1]\}$ stays isolated and non-degenerate throughout. Alternatively and specifically, while always enforcing wavefunction continuity in $\calB_{\epsilon}(1)\setminus\partial\calB_{\epsilon}(1)$, a Dirichlet boundary condition on $\partial\calB_{\epsilon}(1)$ can be suddenly turned on at $t=0$ and abruptly removed at $t=1$ without significantly affecting the low-lying energy states and their nodal structures, as guaranteed by the Courant-Fischer-Weyl min-max principle \cite{Courant89,Reed78}, since the the low-lying eigenstates have their wavefunctions exponentially decayed and decaying outside $\calB_{\epsilon}(1)$.

Therefore, it is only necessary to consider an analytic path of Hamiltonians $H'(t)\defeq(1-2t)H_0+2tH'_1\defeq{-}\Delta_g+V'(t)$, $t\in[0,1/2]$ on a presumably compact configuration space $\calM$, with $V'(t)\defeq(1-2t)V_0+2tV'_1\in C^2(\calM)\cap L^{\infty}(\calM)$ being an $\calM$-diagonal potential, $\forall t\in[0,1/2]$. The eigenvalues and eigenvectors of $H'(t)$ can be parametrized into $t$-analytic curves, by a well-known theorem of Rellich \cite{Rellich42,Kato80,Kriegl03,Kriegl11}. Each $t$-analytic curve of eigenvalues or eigenvectors will be referred to as a $\lambda$-curve or $\psi$-curve respectively. A $t$-continuous curve of ground state energy $\{E_0(H'(t)):t\in[0,1/2]\}$ is obtained by taking the minimum among all of the $\lambda$-curves. Since $V(t)$, $t\in[0,1/2]$ is bounded, the growth of $|\lambda(t)-\lambda(0)|$ is bounded by a constant \cite{Kato80} for each $\lambda$-curve $\{\lambda(t):t\in[0,1/2]\}$, and there can be only a finite number of $\lambda$-curves that ever come close to $\{E_0(H'(t)):t\in[0,1/2]\}$. Consequently, there can be no more than a finite number of occurrences of ground state degeneracy along the $\{E_0(H'(t)):t\in[0,1/2]\}$ curve, because, otherwise, there would be two $\lambda$-curves whose difference, as another $t$-analytic function, has an infinite number of zeros within the bounded domain $[0,1/2]$, then, by the Bolzano-Weierstrass theorem, the sequence of zeros has a limit point, contradicting analyticity of the $\lambda$-curves.

Now it only remains to demonstrate that each of the finite number of isolated points of degeneracy along $E_0(H'(t))$, $t\in[0,1/2]$ can be removed by an above-mentioned $U_0(t,\cdot)$, $t\in[0,1/2]$ potential on the compact Riemannian manifold $\calM$. Around each point of degeneracy at $t=t_1\in(0,1/2)$, within each interval $(t_1-\delta_1,t_1+\delta_1)\subset(0,1/2)$ for a sufficiently small $\delta_1>0$, there can be only a finite number of $\lambda$-curves $\{\lambda_i(t):i\in[0,m]\}$, $m\in\mathbb{N}$ intersecting at $t=t_1$ and nowhere else, while the eigenvalues and associated eigenvectors can be represented as $\lambda_i(t)=\lambda_i(t_1)+O(|t-t_1|)$, $\psi_i(t)=\psi_i(t_1)+O(|t-t_1|)$, $\forall i\in[0,m]$. Since the nodal surface of each eigenvector $\psi_i(t_1)$, $i\in[0,m]$ is compact and smooth, on which the continuous function $\|\partial_q\psi_i(t_1,q)\|$ is strictly positive by the Hopf lemma \cite{Hopf52,Oleinik52,Gilbarg01,Evans10}, the union of all these nodal surfaces has a $V_g$-measure zero, and there must be an open ball $B(q_1,\epsilon_1)\subset\calM$ of radius $\epsilon_1>0$ about a certain $q_1\in\calM$, within which each eigenvector $\psi_i(t_1)$, $i\in[0,m]$ does not vanish, and can be made positive by adjusting its global numerical sign. It can be also assumed that $B(q_1,\epsilon_1)$ be covered by a chart with a local coordinate $(x^1,\cdots\!,x^{nd})$ such that the Laplace-Beltrami operator reads $\Delta_g=\sum_{i=1}^{nd}g^{ii}(\partial/\partial x^i)^2+\sum_{i=1}^{nd}h^i(\partial/\partial x^i)$. Next, sufficiently small $\delta_2\in(0,\delta_1)$ and $\epsilon_2\in(0,\epsilon_1)$ can be found such that $\delta_2+\epsilon_2\ll\min\{\psi_i(t_1;q_1):i\in[0,m]\}$, and $\psi_i(t;q)=\psi_i(t_1;q_1)+O(|t-t_1|+\|q-q_1\|)$ in $\calI_2\times B(q_1,\epsilon_2)$, $\calI_2\defeq(t_1-\delta_2,t_1+\delta_2)$, $\forall i\in[0,m]$. Install a perturbative potential $U_0(t;q=q_1+(x^i))=C\exp[{-}D(t-t_1)^2]\exp({-}Ex_1^2)\sin(Fx_1)$, with positive real-valued constants $C,D,E,F$, both $D$ and $E$ being so large that $U_0(t;q_1+(x^i))$ is essentially localized in $\calI_2\times B(q_1,\epsilon_2)$, while $|F|/E$ being so large that $U_0(t;q)$ induces virtually no direct, first-order mixing among the $\psi$-curves $\{\psi_i(t):i\in[0,m], t\in\calI_2\}$, but couples each of them, in essentially the same manner when $\delta_2$ and $\epsilon_2$ are sufficiently small, to a collection of highly excited Dirichlet eigenstates of $H'(t)$, denoted as $\{\psi_j(t):j\in\calJ\subseteq\mathbb{N}\}\subseteq\Diri(\calM)$, where the coupling matrix element $\langle\psi_j(t)|U_0(t)|\psi_i(t)\rangle\sim\langle\psi_j(t_1)|U_0(t_1)\rangle\psi_i(t_1;q_1)$ is mostly independent of $i\in[0,m]$ and $t\in\calI_2$, while the energy difference $\lambda_j(t)-\lambda_i(t)\sim\lambda_j(t_1)\sim g^{11}F^2$ is largely a constant independent of $i\in[0,m]$, $j\in\calJ$, and $t\in\calI_2$. The $\calJ$-indexed high-energy states mediate interactions among the low-energy $\{\psi_i:i\in[0,m]\}$ states through higher-order perturbations.

A standard series expansion up to the second-order perturbations \cite{Schiff68,Landau77,Sakurai20} indicates that the perturbative coupling effect of all of the $\calJ$-indexed high-energy states combined is equivalent to that of a single effective state $\Psi(q=q_1+(x^i))\sim{-}\sin(Fx_1)$, $q\in\calM$, which may not be a strict eigenvector of $H'(t)$, but behaves much like an energy eigenstate in that $H'(t)\Psi\sim g^{11}F^2\Psi$, $\forall t\in\calI_2$. Note that the effective state $\Psi$ is taken to be independent of $t$, ignoring any small perturbative change, by virtue of the largeness of $g^{11}F^2$ and smallness of $\delta_2$. Restricted to the $(m+2)$-dimensional subspace spanned by $\{\psi_i(t):i\in[0,m]\}$ and $\Psi$, $\forall t\in\calI_2$, the total Hamiltonian $H(t)=H'(t)+U_0(t)$ is effectively $H(t)\sim H''(t)$ as
\begin{equation}
H''(t)=\sum_{i=0}^m\lambda_i(t)|\psi_i(t)\rangle\langle\psi_i(t)|+g^{11}F^2|\Psi\rangle\langle\Psi|-\sum_{i=0}^mA_i(t)\left\{|\Psi\rangle\langle\psi_i(t)|+|\psi_i(t)\rangle\langle\Psi|\right\}, \label{psiPsiCoupling}
\end{equation}
with $A_i(t)\sim\half\psi_i(t_1;q_1)\,Ce^{{-}D(t-t_1)^2}\int_{B(q_1,\epsilon_2)}e^{{-}Ex_1^2(q)}dV_g(q)>0$, $\forall t\in\calI_2$. The Hamiltonian $H''(t)$ is stoquastic, irreducible, and aperiodic, $\forall t\in\calI_2$, whose ground state is guaranteed non-degenerate by the Perron-Frobenius theorem \cite{Seneta81,Horn85,Meyer00}. Alternatively, a Schrieffer-Wolff transformation \cite{Schrieffer66,Bravyi11} can be employed to project out the $\calJ$-indexed high-energy states and obtain an effective Hamiltonian for the low-energy subspace spanned by the $\{\psi_i:i\in[0,m]\}$ states, which should again manifest the non-degeneracy of the ground state of the effective Hamiltonian. In any case, it is \iftoggle{ForUSPTO} {demonstrated} {proved} that the ground state of $H(t)$ has avoided level crossing by the gradual onset of $U_0(t)$ in the concerned interval $\calI_2$. Importantly, toward the end of $\calI_2$, both $\lambda_0(H(t))$ and $\psi_0(H(t))$ return back to $\lambda_0(H'(t))$ and $\psi_0(H'(t))$ respectively, so that the same procedure can be repeated for the other isolated points $\{t_i\in(0,1/2):i\in[2,m']\}$, $m'\in\mathbb{N}$ of degeneracy along $\{E_0(H'(t)):t\in[0,1/2]\}$, to have all of them removed eventually, rest assured that the overlap among the Gaussian-shaped impulses of the form $C\exp[{-}D(t-t_i)^2]$, $i\in[1,m']$ in the $U_0(\cdot,\cdot)$ potential can be made negligibly small for all of the relevant considerations by making $D>0$ sufficiently large.
\end{proof}

\begin{definition}{(Imaginary-Time $C^{\omega}\!$ Path of Hamiltonians)}\label{defiImagTimeAnalyticPath}\\
Let $(\calM,g)$ be a connected, locally compact Riemannian manifold as a many-body configuration space. Use the Sobolev space $\calH^1\defeq L^2_{\!\mathsmaller{F}}(\calM)\cap W^{1,2}(\calM)$ as a Hilbert space of fermionic wavefunctions. An imaginary-time $C^{\omega}\!$ path of Hamiltonians is a fermionic Schr\"odinger operator-valued curve $\{H(\tau)\in\calL_0(\calH^1):\tau\in\calI\defeq([\tau_0,\infty),\,\tau_0>0\}$ in the vector space $\calL_0(\calH^1)$, which generates an imaginary-time-inhomogeneous (ITI) Gibbs kernel function $K(\tau;r,q)$ that is $(\tau,r,q)$-jointly analytic over the domain $\calI\times\calM^2$, denoted as $K(\tau;r,q)\in C^{\omega}(\calI\times\calM^2)$, such that
\begin{align}
K(\tau;r,q) &\,=\, K(\tau;q,r)\in\mathbb{R},\;\forall(r,q)\in\calM^2,\;\forall\tau\in\calI, \\[0.75ex]
{-}\frac{\partial K(\tau;r,q)}{\partial\tau} &\,=\, H(\tau)\,K(\tau;r,q),\;\forall(r,q)\in\calM^2,\;\forall\tau\in\calI, \label{SchrEqnForK} \\[0.75ex]
K(\tau_0;r,q) &\,=\, \langle r|\exp[{-}\tau_0H(\tau_0)]|q\rangle,\;\forall(r,q)\in\calM^2,
\end{align}
where the Hamiltonian $H(\tau)$, $\tau\in\calI$ is applied with respect to the $r\in\calM$ spatial variable in equation (\ref{SchrEqnForK}), and $\forall\tau\in\calI$, the function $K(\tau;\cdot,\cdot)\in L^2(\calM^2)$ is positivity improving, the operator $\Base(H(\tau))$ substantiates the Hopf lemma and the strong extremum principle on any bounded and connected submanifold of $\calM$. An imaginary-time piecewise $C^{\omega}\!$ path of Hamiltonians is a concatenation of a finite number of imaginary-time $C^{\omega}\!$ paths of Hamiltonians.
\vspace{-1.5ex}
\end{definition}

Much of the concepts of nodal cells and the tiling and maximal properties of nodal cells apply to an imaginary-time $C^{\omega}\!$ path of Hamiltonians, particularly one that is associated with a constant Hamiltonian $H(\tau)=H_0\in\calL_0(\calM)$, $\forall\tau\in\calI$ \cite{Ceperley91}, when the ITI Gibbs kernel function is simply the conventional $K(\tau;r,q)=\langle r|\exp({-}\tau H_0)|q\rangle$, $\forall\tau\in\calI$.

\begin{definition}{(Ceperley Reaches and Nodal Cells for Gibbs Kernel Functions)}\label{defiReachAndNodal}\\
With a reference point $(\tau_0,q_0)\in\{\tau_0\}\times\calM$ fixed, the Ceperley reach of $(\tau_0,q_0)$, denoted by $\calR_{\tau_0q_0}$, also called the imaginary-time nodal tube from $(\tau_0,q_0)$, is the set of points $(\tau_1,q_1)\in\calI\times\calM$, $\calI\defeq[\tau_0,\infty)$, $\tau_0>0$ having a reaching path ({\it i.e.}, a continuous curve) $\gamma:[0,1]\mapsto\calI\times\calM$ such that $\gamma(0)=(\tau_0;q_0)$, $\gamma(1)=(\tau_1;q_1)$, and $K(\gamma(s),q_0) \defeq K(\tau_s;q_s,q_0)>0$, for all $s\in[0,1]$ \cite{Ceperley91}. A slice of the Ceperley reach $\calR_{\tau_0q_0}\subseteq\calI\times\calM$ at $\tau\in\calI$ is denoted and defined as $\calR_{\tau_0q_0}(\tau)\defeq\calR_{\tau_0q_0}\cap(\{\tau\}\times\calM)$, while a segment of the Ceperley reach $\calR_{\tau_0q_0}$ between $\tau_1\in\calI$ and $\tau_2\in[\tau_1,\infty)$ is denoted and defined as $\calR_{\tau_0q_0}([\tau_1,\tau_2])\defeq\calR_{\tau_0q_0}\cap([\tau_1,\tau_2]\times\calM)$. For each $\tau\in\calI$, the preimage $K^{{-}1}(\tau;\mathbb{R}\setminus\{0\})\defeq\{r\in\calM:K(\tau;r,q_0)\neq 0\}$ is necessarily an open subset of $\calM$, one connected component of which containing a given point $q_1\in\calM$ is called the nodal cell of $K(\tau;\cdot,q_0)$ around $q_1$, denoted as $\calN_{\tau_0q_0}(K;\tau;q_1)$. A nodal cell $\calN_{\tau_0q_0}(K;\tau;q_1)$ is called positive when $K(\tau;q_1,q_0)>0$.
\vspace{-1.5ex}
\end{definition}

Clearly, each slice $\calR_{\tau_0q_0}(\tau)$ is a union of some positive nodal cells, for each $\tau\in\calI$. For a conventional fermionic Schr\"odinger Hamiltonian $H={-}\Delta_g+V$, Ceperley has proved that the Ceperley reach $\calR_{\tau_0q_0}$ is maximal \cite{Ceperley91}, namely, it bisects the manifold $(0,\infty)\times\calM$, $\forall q_0\in\calM$, using the analytical solvability of $K(\tau;\cdot,q_0)$ in the classical limit of $\tau\rightarrow 0{\mathsmaller{+}}$, when the effect of potential energy becomes immaterial, in conjunction with the differential evolution of $K(\tau;\cdot,q_0)$ in $\tau\in(0,\infty)$. His proof can be generalized straightforwardly to an imaginary-time $C^{\omega}\!$ path of Hamiltonians $H(\tau)={-}\Delta_g+V(\tau)$, $\tau\in(0,\infty)$, in association with a $\tau$-dependent potential $V(\tau)$ that is lower bounded, as such, any slice $\calR_{\tau_0q_0}(\tau)$ must be the union of all positive nodal cells at $\tau$, $\forall\tau\in(0,\infty)$. However, when it comes to the maximality of a nodal cell $\calN_{\tau_0q_0}(K;\tau;q_1)$, $K(\tau;q_1,q_0)\neq 0$, $\tau\in[\tau_0,\infty)$, $\tau_0>0$, a rigorous proof is still missing. The following \myCorollary\,\ref{CeperleysConjectureIsTrueDM} shall fill the gap.

Before that, it is well worth mentioning that because of the analyticity of $K(\cdot;\cdot,q_0)\in C^{\omega}(\calI\times\calM)$, the closure of the Ceperley reach $\calR_{\tau_0q_0}$ and its boundary $\partial\calR_{\tau_0q_0}$, the closure of any slice $\calR_{\tau_0q_0}(\tau)$, $\tau\in\calI$ and its boundary $\partial\calR_{\tau_0q_0}(\tau)$, as well as the closure of any nodal cell $\calN_{\tau_0q_0}(K;\tau;q_1)$, $\tau\in\calI$, $q_1\in\calR_{\tau_0q_0}$ and its boundary $\partial\calN_{\tau_0q_0}(K;\tau;q_1)$, are all examples of a semianalytic set $\calF$, which is amenable to a Whitney stratification \cite{Thom64,Whitney65,Lojasiewicz65,Gabrielov68,Mather70,Hironaka73,Wall75,
Bierstone88,Parusinski94,Shiota97,Dries98,Pflaum01}, meaning that a filtration of closed subsets $\calF=\calF_m\supseteq\calF_{m{-}1}\supseteq\cdots\supseteq\calF_0\supseteq\calF_{{-}1}=\emptyset$ exists, such that $m=\dim(\calF)$, $\calF_i\setminus\calF_{i{-}1}=\bigcup_{\mathsmaller{j\in I(i)}}\calS_{ij}$ is either empty or an $i$-dimensional analytic submanifold for each $i\in[0,m]$, where $I(i)$ is an index set labeling the disjoint and connected components $\{S_{ij}:j\in I(i)\}$ of $\calF_i\setminus\calF_{i{-}1}$, each of which is called a strata of dimension $i$, the decomposition of $\calF$ into strata $\calF=\bigcup_{\mathsmaller{i\in[0,m]}}\bigcup_{\mathsmaller{j\in I(i)}}\calS_{ij}$ is locally finite and satisfies the so-called {\em frontier condition}, namely, any point in $\calF$ has a neighborhood intersecting no more than a finite number of strata, and $\cl(\calS_{ij})\supseteq\calS_{i'j'}$, $\dim(\calS_{ij})\ge\dim(\calS_{i'j'})$ must hold for any pair of strata $\calS_{ij}$, $\calS_{i'j'}$ such that $\cl(\calS_{ij})\cap\calS_{i'j'}\neq\emptyset$, moreover, for any triple $(\calS_{ij},\calS_{i'j}',x_*\in\calF)$ such that $x_*\in\calS_{i'j'}\subset\cl(\calS_{ij})$, $\dim(\calS_{ij})=k\in\mathbb{N}$, with a sequence of points $\{x_{\alpha}:\alpha\in\mathbb{N}\}\subseteq\calS_{ij}$ converging to $x_*$, and the associated sequence of tangent spaces $\{T_{x_{\alpha}}\calS_{ij}\in\Gr(k,T_{x_{\alpha}}\calF)\}$ converging to $\calT\subseteq\mathbb{R}^m$, where $\Gr(k,T_{x_{\alpha}}\calF)\}$ represents the Grassmannian bundle of $k$-dimensional subspaces in $T_{x_{\alpha}}\calF$, and the convergence is in the standard topology of said Grassmannian bundle, one of the two {\em Whitney conditions} must be satisfied: A) $T_{x_*}\calS_{i'j'}\subseteq\calT$; B) If another sequence of points $\{x'_{\alpha}:\alpha\in\mathbb{N}\}$ in $\calS_{i'j'}$ also converges to $x_*$, with $x'_{\alpha}\neq x_{\alpha}$, $\forall\alpha\in\mathbb{N}$, and the secants $\widehat{x_{\alpha}x'_{\alpha}}$ converging in the projective space $\mathbb{R}\mathbb{P}^{m{-}1}$ to a line $\calL\subseteq\mathbb{R}^m$, then $\calL\subseteq\calT$. It is known that the condition B implies A \cite{Mather70}.

In particular, the topological boundaries $\partial\calR_{\tau_0q_0}(\tau)$ and $\partial\calN_{\tau_0q_0}(K;\tau;q_1)$, $\tau\in\calI$, $q_1\in\calR_{\tau_0q_0}$ are all examples of an analytic variety $\calV\subseteq\calM$, which, according to {\em Lojasiewicz's structure theorem for varieties} \cite{Lojasiewicz65,Krantz02}, has a Whitney stratification $\calV=\bigcup_{\mathsmaller{i\in[0,m]}}\bigcup_{\mathsmaller{j\in I(i)}}\calV_{ij}$, $m=\dim(\calV)$ in a neighborhood $Q$ around any point $x_*\in\calV\subseteq\calM$, where, in a suitable chart, $x_*$ can be identified with the origin $(0,\cdots\!,0)\in\mathbb{R}^m$, the neighborhood can be appointed as $Q=\prod_{i=1}^m({-}\delta_i,\delta_i)$, $\delta_i>0$, $\forall i\in[1,m]$, and a finite set of polynomials $\{H_i^k(x_1,\cdots\!,x_i;x_k):i\in[1,m],\,k\in[i{+}1,m]\}$ can be chosen, with $(x_1,\cdots\!,x_i;x_k)\in\mathbb{R}^{i{+}1}$, the discriminant $D_i^k(x_1,\cdots\!,x_i)$ for each $H_i^k(x_1,\cdots\!,x_i;x_k)$ viewed as a univariate polynomial of $x_k$ being non-vanishing in $Q_i\defeq\prod_{j=1}^i({-}\delta_j,\delta_j)$, each root $x_k=\xi$ of the univariate polynomial $H_i^k(x_1,\cdots\!,x_i;\cdot)$ satisfying $|\xi|<\delta_k$, $\forall(x_1,\cdots\!,x_i)\in Q_i$, such that, each stratum $\calV_{ij}$ is an algebraic variety that in an open subset $Q_{ij}\subseteq Q_i\subseteq\mathbb{R}^i$ can be analytically parametrized by a system of $(m{-}i)$ equations as
\begin{equation}
\left. \begin{array}{rcl}
x_{i{+}1} &=& \chi_{ij}^{i{+}1}(x_1,\cdots\!,x_i), \\
x_{i{+}2} &=& \chi_{ij}^{i{+}2}(x_1,\cdots\!,x_i), \\
&\cdots& \\
x_m &=& \chi_{ij}^m(x_1,\cdots\!,x_i), \\
H_i^k(x_1,\cdots\!,x_i;\chi_{ij}^k) &=& 0,\;\forall k\in[i{+}1,m],
\end{array} \right\}
\end{equation}
where each $\chi_{ij}^k$ is an analytic function in $Q_{ij}$, the subsets $\{Q_{ij}\}$ and functions $\{\chi_{ij}^k\}$ are non-redundant, in the sense that, for any triple of indices $(i,j,j')$, either $Q_{ij}\cap Q_{ij'}=\emptyset$ or $Q_{ij}=Q_{ij'}$, and in the latter case, for any $k\in[i{+}1,m]$, either $\chi_{ij}^k\equiv\chi_{ij'}^k$ on $Q_{ij}$ or $\chi_{ij}^k(x_1,\cdots\!,x_i)\neq\chi_{ij'}^k(x_1,\cdots\!,x_i)$ for all $(x_1,\cdots\!,x_i)\in Q_{ij}$.

In any case, each of the submanifolds $\calR_{\tau_0q_0}$, $\calR_{\tau_0q_0}(\tau)$, $\calN_{\tau_0q_0}(K;\tau,q_1)$ or the boundaries $\partial\calR_{\tau_0q_0}$, $\partial\calR_{\tau_0q_0}(\tau)$, $\partial\calN_{\tau_0q_0}(K;\tau,q_1)$, $\tau\in\calI$, $q_1\in\calR_{\tau_0q_0}$ is guaranteed to behave sufficiently regular and much like a smooth manifold with smooth boundaries in a sufficiently small neighborhood of any $q\in\calM$, so that the classical Morse theory \cite{Morse69,Milnor69,Goresky88,Matsumoto02,Nicolaescu11,Borodzik16}, in particular, the stratified Morse theory \cite{Goresky88} and the Morse theory for manifolds with boundaries \cite{Borodzik16}, can be employed to track topological changes of a (sub)manifold by analyzing a differentiable scalar function on the (sub)manifold. Our proof of \myCorollary\,\ref{CeperleysConjectureIsTrueDM} below may not necessarily invoke any deep results of Morse theory directly and explicitly, but the basic idea and methodology being used have much in common.

\begin{lemma}{(Nodal Maximality for Fermionic Density Matrices)}\\
With respect to a fixed reference point $(\tau_0,q_0)\in\calI\times\calM$, $\calI\defeq[\tau_0,\infty)$, $\tau_0>0$, Ceperley's conjecture on the maximality of any nodal cell $\calN_{\tau_0q_0}(K;\tau_1;q_1)$, $K(\tau_1;q_1,q_0)\neq 0$, $q_1\in\calM$ holds true for the ITI Gibbs kernel function $K(\cdot;\cdot,q_0)$ associated with an imaginary-time $C^{\omega}\!$ path of Hamiltonians $\{H(\tau):\tau\in\calI\}$ as defined in Definition \ref{defiImagTimeAnalyticPath}, once it starts at $\tau=\tau_0$ with the nodal cell $\calN_{\tau_0q_0}(K;\tau_0;q_0)$ being maximal.
\label{CeperleysConjectureIsTrueDM}
\vspace{-2.5ex}
\end{lemma}
\begin{proof}[\iftoggle{ForUSPTO} {Demonstration} {Proof}]
Since the Ceperley reach $\calR_{\tau_0q_0}\subseteq\calI\times\calM$ is maximal, each slice $\calR_{\tau_0q_0}(\tau)$ is the union of all positive nodal cells of the ITI Gibbs kernel function $K(\tau;\cdot,q_0)\in L^2(\calM)\cap C^{\omega}(\calM)$, $\forall\tau\in\calI$. The subset $\calJ\defeq\{\tau_*\in\calI:K(\tau;\cdot,q_0)\;\mbox{\rm has only 1 positive nodal cell},\;\forall\tau\in[\tau_0,\tau_*)\}$ is nonempty because $\calN_{\tau_0q_0}(K;\tau_0;q_0)$ is maximal. Assume that $\calJ$ is upper-bounded to use reductio ad absurdum. Then $\tau_1\defeq\sup(\calJ)<\infty$ exists, and there must be a point $(\tau_1,q_1)\in\{\tau_1\}\times\calR_{\tau_0q_0}(\tau_1)$ at which a negative nodal cell starts to nucleate within a surrounding positive nodal cell, namely, on the one hand, $\exists\delta>0$ such that no open set $V$ exists to satisfy $q\in V\subseteq\calM$, $K(\tau_1{-}\delta;q,q_0)<0$ for all $q\in V$, on the other hand, $\forall\delta>0$, there always exists an open set $V\subseteq\calM$, $V\ni q_1$ such that $K(\tau_1{+}\delta;q,q_0)<0$ for all $q\in V$. Consider the analytic function $K(\cdot;\cdot,q_0)\in C^{\omega}(\cl(U))$ in a small open set $U\subseteq\calI\times\calM$, $U\ni(\tau_1,q_1)$, and an analytic variety $\calV=U\cap\partial\calR_{\tau_0q_0}$ as a collection of strata. Clearly, $U$ can be chosen sufficiently small, such that $\calV\cap\{\tau\,{<}\,\tau_1\}=\emptyset$, and $V_g(\calR_{\tau_0q_0}(\tau_1)\cap\calV)=0$, otherwise, the absurdity of $K(\tau_1;\cdot,q_0)\equiv 0$, thus $K(\cdot;\cdot,q_0)\equiv 0$, would follow.

Also obviously, $K(\tau_1;q_1,q_0)=0$, $\partial_{\tau}K(\tau=\tau_1;q_1,q_0)\le 0$. If $\partial_{\tau}K(\tau=\tau_1;q_1,q_0)<0$, then by the continuity of $K(\cdot;\cdot,q_0)$, $U$ can be chosen sufficiently small such that $\partial_{\tau}K(\tau;q,q_0)<0$ for all $(\tau,q)\in U$. Otherwise, if $\partial_{\tau}K(\tau=\tau_1;q_1,q_0)=0$, let $k\ge 1$ be the smallest integer such that $\partial_{\tau}^kK(\tau=\tau_1;q_1,q_0)\neq0$, then according to the Weierstrass-Malgrange preparation theorem \cite{Lojasiewicz65,Malgrange66,Golubitsky73,Krantz02}, the function $K(\cdot;\cdot,q_0)$ can be represented as
\begin{equation}
K(t;x,q_0)=c(t,x)\left[{\textstyle{\sum_{i=0}^{k{-}1}}}\,t^ia_i(x)+t^k\right],
\end{equation}
locally in $U$, with local coordinates $t=\tau-\tau_1\in\mathbb{R}$, $x=(x_1,\cdots\!,x_m)\in\mathbb{R}^m$, $m=\dim(\calM)$, where $c(\cdot,\cdot)$ and $\{a_i(\cdot)\}_{i=0}^{k-1}$ are all analytic functions such that $a_i(0)=0$, $\forall i\in[0,k{-}1]$, and $c(t,x)>0$, $\forall(t,x)\in(t,x)[\cl(U)]$, with $(t,x)[\cl(U)]\defeq\{(t(u),x(u)):u\in\cl(U)\}$. If $a_1(\cdot)\equiv 0$ in $x(U)\defeq\{x(u):u\in U\}$, then $\partial_tK(t=0;x,q_0)=[\partial_tc(t,x)/c(t,x)]K(0;x,q_0)$, $\forall x\in x(U\cap\calM)$. Let $E_1\defeq\max\{|\partial_tc(t,x)|/c(t,x):(t,x)\in(t,x)[\cl(U)]\}$, then $\partial_tK(t=0;x,q_0)-E_1K(0;x,q_0)\le 0$, $\forall x\in x(U\cap\{\tau\,{=}\,\tau_1\})$; Otherwise, it follows from $\calV\cap\{\tau\,{<}\,\tau_1\}=\emptyset$ and $V_g(\calR_{\tau_0q_0}(\tau_1)\cap\calV)=0$ that $a_0(\cdot)$ and $a_1(\cdot)$ must be respectively of the forms
\begin{equation}
a_0(x)=b_0(x)\,{\textstyle{\sum_{|{\sf k}|\,=\,\kappa}}}\,a_{0,{\sf k}}\,x^{2{\sf k}},\;\;a_1(x)=b_1(x)\,{\textstyle{\sum_{|{\sf k}|\,=\,\kappa}}}\,a_{1,{\sf k}}\,x^{2{\sf k}},
\end{equation}
where $\kappa\ge 1$, ${\sf k}\defeq(k_1,\cdots\!,k_m)\in(\mathbb{N}\cup\{0\})^m$ is a multi-index, $|{\sf k}|\defeq\sum_{i=1}^mk_i$, $x^{2{\sf k}}\defeq\prod_{i=1}^mx_i^{2k_i}$ \cite{Krantz02}, $a_{0,{\sf k}}\ge 0$ for all ${\sf k}$ and $\sum_{|{\sf k}|\,=\,\kappa}a_{0,{\sf k}}>0$, both $b_0(\cdot)$ and $b_1(\cdot)$ are analytic in $x(U)$, $b_0(x)>0$ for all $x\in x(U)$, $a_{1,{\sf k}}$ must also vanish whenever $a_{0,{\sf k}}=0$. Therefore, $|a_1(x)|/a_0(x)$ is bounded on $x[\cl(U)]$. Let $E_2\defeq\max\{|a_1(x)|/a_0(x):x\in x[\cl(U)]\}$, and $E_0\defeq E_1+E_2$, then it holds true that $\partial_{\tau}K(\tau=\tau_1;q,q_0)-E_0K(\tau_1;q,q_0)\le 0$, $\forall q\in U\cap\{\tau\,{=}\,\tau_1\}$.

Now, consider the ITI Gibbs kernel function $K'(\tau;q,q_0)\defeq K(\tau;q,q_0)\,e^{{-}E_0\tau}$, which corresponds to a constant energy-shifted $C^{\omega}\!$ path of Hamiltonians $\{H(\tau)+E_0:\tau\in\calI\}$, and shares virtually the same analytical properties as well as exactly the same Ceperley reach and nodal structures with $K(\tau;q,q_0)$, $(\tau,q)\in\calI\times\calM$. In particular, the set
\begin{equation}
\calJ'\defeq\{\tau_*\in\calI:K'(\tau;\cdot,q_0)\;\mbox{\rm has only 1 positive nodal cell},\;\forall\tau\in[\tau_0,\tau_*)\}\equiv\calJ ,
\end{equation}
and $(\tau_1,q_1)$ is exactly the same point at which a negative nodal cell starts to nucleate within a surrounding positive nodal cell of $K'(\cdot;\cdot,q_0)$. It follows from the analyses above that
\begin{equation}
[H(\tau_1)+E_0]\,K'(\tau_1;q,q_0) \,=\, -\partial_{\tau}K'(\tau=\tau_1;q,q_0) \,\ge\, 0, \; \forall q\in U\cap\{\tau\,{=}\,\tau_1\}.
\end{equation}
On the other hand, $K'(\tau_1;\cdot,q_0)$ attains a non-positive minimum over $\cl(\calD)$ with $K'(\tau_1;q_1,q_0)=0$ at $q_1\in U$, the strong version of the Hopf extremum principle dictates that $K'(\tau_1;\cdot,q_0)\equiv 0$ over $\cl(U)$, which contradicts the premise of $V_g(\calR_{\tau_0q_0}(\tau_1)\cap\calV)=0$. It must be concluded that the set $\calJ\equiv\calJ'$ is unbounded, and Ceperley's conjecture is true.
\end{proof}

\begin{corollary}{(Nodal Maximality for Conventional Fermionic Density Matrices)}\\
Let $K(\cdot;\cdot,q_0)\in C^{\omega}(\calI\times\calM)$, $\calI\defeq(0,\infty)$ be the ITI Gibbs kernel function  associated with an imaginary-time $C^{\omega}\!$ path of conventional fermionic Schr\"odinger operators $\{H(\tau)={-}\Delta_g+V(t,q\in\calM):\tau\in\calI\}$, where $\Delta_g$ is the Laplace-Beltrami operator on the configuration space $\calM$, $V(\cdot,\cdot)\in C^{\omega}((\{0\}\cup\calI)\times\calM)$, and $V(\tau,\cdot)\in L^{\infty}(\calM)$ is an $\calM$-diagonal potential, $\forall\tau\in\{0\}\cup\calI$. Then, Ceperley's conjecture holds true for the maximality of any nodal cell $\calN_{q_0}(K;\tau;q_1)$, $K(\tau;q_1,q_0)\neq 0$, $q_1\in\calM$, $\tau\in\calI$, with respect to a fixed reference point $q_0\in\calM$.
\label{CeperleysConjectureIsTrueCDM}
\vspace{-2.5ex}
\end{corollary}
\begin{proof}[\iftoggle{ForUSPTO} {Demonstration} {Proof}]
The proposition follows straightforwardly from \myLemma\;\ref{CeperleysConjectureIsTrueDM}, when a $\tau_0>0$ sufficiently and arbitrarily small is chosen such that, based on the standard Feynman-Kac theory of path integrals \cite{Simon79,Glimm87,Lorinczi11}, the $C^{\omega}\!$ path of conventional fermionic Schr\"odinger operators $\{H(\tau):\tau\in[\tau_0,\infty)\}$ generates a $(\tau,q)$-jointly analytic ITI Gibbs kernel function $K(\cdot;\cdot,q_0)\in C^{\omega}([\tau_0,\infty)\times\calM)$, which fulfills all of the conditions in Definition \ref{defiImagTimeAnalyticPath} as premises for \myLemma\;\ref{CeperleysConjectureIsTrueDM}, also ensuring the nodal maximality for $K(\tau_0;\cdot,q_0)$.
\vspace{-1.5ex}
\end{proof}

It is interesting to note that \myCorollary\;\ref{CeperleysConjectureIsTrueCDM} provides an alternative \iftoggle{ForUSPTO} {demonstration} {proof} for \myCorollary\;\ref{CeperleysConjectureIsTrue}, asserting the nodal maximality for the promised non-degenerate ground state of a conventional fermionic Schr\"odinger Hamiltonian $H_{\mathsmaller{\infty}}={-}\Delta_g+V(q\in\calM)$, with $V(\cdot)$ being $\calM$-diagonal and Kato-decomposable, because such $H_{\mathsmaller{\infty}}$ can be easily made the limit of an imaginary-time $C^{\omega}\!$ path of conventional fermionic Hamiltonians $\{H(\tau)\}$ as detailed in \myCorollary\;\ref{CeperleysConjectureIsTrueCDM}.

The essence of \myLemma\;\ref{CeperleysConjectureIsTrueDM} and \myCorollary\;\ref{CeperleysConjectureIsTrueCDM} is to affirm that, on any boundary surface of a Ceperley reach or imaginary-time nodal tube, there exists no point on which the normal vector of said boundary surface runs parallel to the imaginary time axis. For a conventional fermionic Schr\"odinger Hamiltonian on a connected Riemannian manifold $\calM$, with respect to a fixed reference point $(\tau_0,q_0)\in\{0\}\times\calM$, Ceperley has posed an interesting question on whether two points $(\tau_1,q_1)$ and $(\tau_1,r_1)$ in $\calI\times\calM$ that fall in the same nodal cell of the Gibbs wavefunction $K(\tau_1;\cdot,q_0)$ must also lie in the same Ceperley reach of $(\tau_0,q_0)$ \cite{Ceperley91}. The question is now answered in the affirmative by virtue of \myLemma\;\ref{CeperleysConjectureIsTrueDM} and \myCorollary\;\ref{CeperleysConjectureIsTrueCDM}.

{\bf Context Switch Notice:} From this point onwards, many-body physical systems are assumed to be moving on a general continuous-discrete product configuration space $\calC=\calM\times\calP$ as default, unless explicitly stated otherwise.

It is useful to remind that, in the context of a PLTKD Hamiltonian $H=\sum_{k=1}^{\mathsmaller{K}}\bigotimes_{i=1}^{n_k}h_{ki}$ supported by a continuous-discrete product Riemannian manifold $\calC=\prod_{s=1}^{\mathsmaller{S}}(\calM_s\times\calP_s)^{n_s}$ as a configuration space, where $K\in\mathbb{N}$, $n_k\in\mathbb{N}$ for all $k\in[1,K]$, $S\in\mathbb{N}$, $n_s\in\mathbb{N}$ for all $s\in[1,S]$, there is an $h_{ki}$-moved and an $h_{ki}$-fixed factor spaces, denoted as $\calE_{ki}\defeq\calE_{h_{ki}}$ and $\calE_{ki}^{\mathsmaller{\perp}}\defeq\calE_{h_{ki}}^{\mathsmaller{\perp}}$ respectively, that are associated with each individual FBM interaction $h_{ki}$ for a given $k\in[1,K]$ and a given $i\in[1,n_k]$. Consequently, $\forall k\in[1,K]$, there is an $H_k$-fixed factor space $\calE_k^{\mathsmaller{\perp}}\defeq\calE_{H_k}^{\mathsmaller{\perp}}\defeq\bigotimes_{i=1}^{n_k}\calE_{ki}^{\mathsmaller{\perp}}$ associated with the FBM tensor monomial $H_k\defeq\bigotimes_{i=1}^{n_k}h_{ki}$, in conjunction with a complementary $H_k$-moved factor space $\calE_k\defeq\{q\in\calC:q\not\in\calE_k^{\mathsmaller{\perp}}\}$, such that $\calC=\calE_k\oplus\calE_k^{\mathsmaller{\perp}}$, namely, $\forall q\in\calC$, there is a unique ODSD $q=u_{q,k}\oplus u_{q,k}^{\mathsmaller{\perp}}$ with $u_{q,k}\defeq u_{q,H_k}\in\calE_k$, $u_{q,k}^{\mathsmaller{\perp}}\defeq u_{q,H_k}^{\mathsmaller{\perp}}\in\calE_k^{\mathsmaller{\perp}}$.

\begin{definition}{(Node-Determinate FBM Tensor Polynomials and Monomials)}\label{defiNodeDeterTensorPolyMono}\\
With respect to a Gibbs wavefunction of concern, an FBM tensor polynomial $P$ is called $\epsilon$-almost node-determinate, or said to possess a property of node-determinacy with an error probability upper-bounded by $\epsilon\ge 0$, when either 1) the concerned Gibbs wavefunction is
$$\psi(P;q;q_0,\tau) \defeq \langle q|\exp\{-\tau[P{-}\lambda_0(P)]\}|q_0\rangle \,, \; \forall q \in \calC, \; (q_0,\tau) \in \calC\times(0,\infty) \; \mbox{being fixed} \,,$$
the numerical sign of whose point values can be uniquely and efficiently determined on each of the $\calE_{\mathsmaller{P}}$-cosets $\{\calE_{\mathsmaller{P}}\oplus u^{\mathsmaller{\perp}}:u^{\mathsmaller{\perp}}\in\calE_{\mathsmaller{P}}^{\mathsmaller{\perp}},P\;\textit{moves}\;\calE_{\mathsmaller{P}}\oplus u^{\mathsmaller{\perp}}\}$, except for a subset $\calU(P) \subseteq \calC$ of a measure that is less than $\epsilon$ times the volume of $\calC$; or 2) the concerned Gibbs wavefunction is the ground state $\psi_0(P)$ of $P$, which, on each of the $\calE_{\mathsmaller{P}}$-cosets $\{\calE_{\mathsmaller{P}}\oplus u^{\mathsmaller{\perp}}:u^{\mathsmaller{\perp}}\in\calE_{\mathsmaller{P}}^{\mathsmaller{\perp}},P\;\textit{moves}\;\calE_{\mathsmaller{P}}\oplus u^{\mathsmaller{\perp}}\}$, either 2.1) is one-dimensional (namely, the ground state is non-degenerate); or else 2.2) has an orthonormal basis of ground state eigenvectors $\psi_0(P)\defeq\{\psi_{0,l}(P):l\in[1,\gmul(P,0)]\}$, such that the subset
$$\calU(P)\defeq\{q\in\calE_{\mathsmaller{P}}\oplus\calE_{\mathsmaller{P}}^{\mathsmaller{\perp}}:\exists\,\phi,\phi'\in\psi_0(P)\,\mbox{such that}\;\phi\neq\phi',\,\phi(q)\phi'(q)\neq 0\} \,,$$
called the set of node-uncertain configuration points, has a $|\psi_0(P)|^2$-measure that is upper-bounded by $\epsilon$, namely, $\scalebox{1.15}{$\sum$}_{\,\phi\,\in\,\psi_0(P)}\,\scalebox{1.1}{$\int$}_{\!\calU(P)}\,|\phi(q)|^2\,dV_g(q)\le\epsilon$.

Specifically, an FBM tensor monomial $M=\bigotimes_{i=1}^nh_i$ is called $\epsilon$-almost node-determinate when it is so as an FBM tensor polynomial. Associated with an $\epsilon$-almost node-determinate FBM tensor monomial $M$, the set $\calU^c(M)\defeq(\calE_{\mathsmaller{M}}\oplus\calE_{\mathsmaller{M}}^{\mathsmaller{\perp}})\setminus\calU(M)$ is called the subset of node-certain configuration points. More specifically, a single FBM interaction $h$ is called $\epsilon$-almost node-determinate when it constitutes an $\epsilon$-almost node-determinate FBM tensor monomial $M=h$.
\vspace{-1.5ex}
\end{definition}

Obviously, for any finite $\tau \in (0,\infty)$ and $\forall q_0 \in \calC$, if $P$ is efficiently computable, then the Gibbs wavefunction $\langle q|\exp\{-\tau[P{-}\lambda_0(P)]\}|q_0\rangle$, $q\in\calC$ is efficiently computable, the numerical sign of whose point values can be uniquely and efficiently determined, $\forall q\in\calC$, so $P$ is straightforwardly node-determinate. It is only in the large-$\tau$ limit, when the Gibbs wavefunction of $P$ reduces to a set of degenerate ground states, that the node-determinacy of $P$ becomes non-trivial.

Cast in the language of statistical decision theory or hypothesis testing \cite{Berger85,Liese08}, the set of ground state eigenvectors $\{\psi_{0,l}(P):l\in[1,\gmul(P,0)]\}$ of a tensor polynomial $P$ constitutes a statistical model comprising a number $\gmul(P,0)$ of hypotheses or candidate probability distributions that are parametrized by an index $l\in[1,\gmul(P,0)]$, each of which corresponds to a probability density function $|\psi_{0,l}(P;q)|^2$ of generating a random $q\in\calE_{\mathsmaller{M}}\oplus\calE_{\mathsmaller{M}}^{\mathsmaller{\perp}}$ that is measured or observed. Given an observation or measurement $q\in\calE_{\mathsmaller{M}}\oplus\calE_{\mathsmaller{M}}^{\mathsmaller{\perp}}$, the task of statistical decision making or hypothesis testing is to determine an optimal estimate for the parameter $l\in[1,\gmul(P,0)]$, namely, which of the $\gmul(P,0)$ candidates $\{\psi_{0,l}(P):l\in[1,\gmul(P,0)]\}$ is most likely the underlying probability distribution, where optimality is with respect to minimization of error (also known as the loss function \cite{Berger85,Liese08}). A tensor monomial $P$ being $\epsilon$-almost node-determinate means that a decision rule exists, based on which decisions can be made with an expected error bounded by $\epsilon$. The significance of an FBM tensor monomial $M$ being $\epsilon$-almost node-determinate is that, $\forall q\in\calU^c(M)$, it is polynomially efficient to decide $q$ falls into the support of which of the ground state eigenvectors $\{\psi_{0,l}(M):1\le l\le\gmul(M,0)\}$, and further determine what value the ground state eigenvector assumes at $q$, with a probability of failure upper-bounded by $\epsilon$.

Some illustrative examples may help to gain intuition and appreciation of node-determinate FBM tensor monomials and polynomials. To that end, it might be the simplest, by way of example but no means of limitation, to consider a prototypical system consisting of $n\in\mathbb{N}$ mutually distinguishable particles moving in a discrete product space $\calP=\{0,1\}^n$, where each particle labeled by an $i\in[1,n]$ amounts to a canonical rebit, with two Pauli matrices (or called operators) $X_i\defeq\sigma_i^x$ and $Z_i\defeq\sigma_i^z$ generating a Banach algebra of operators acting on the two-dimensional Hilbert subspace ${\rm span}(|0\rangle_i,|1\rangle_i)$. Let $\calI_1=\{i_1,\cdots\!,i_m\}\subseteq[1,n]$ and $\calI_2\defeq[1,n]\setminus\calI_1$ be a set partition of the indices of the rebits, and $\calP=\calP_1\times\calP_2$ be the corresponding Cartesian factorization of the configuration space. Let $\calB(\calI_1)\defeq\bigotimes_{l=1}^m{\rm span}(I_{i_l},Z_{i_l})$ denote the tensor product of the two-dimensional Banach algebras associated with the rebits indexed by $\calI_1$. Any pure tensor $M(\epsilon_1,\cdots\!,\epsilon_l)\defeq\bigotimes_{l=1}^mZ_{i_l}^{\,\epsilon_l}\in\calB(\calI_1)$, with $\epsilon_l\in\{0,1\}$, $\forall l\in[1,m]$, is clearly a node-determinate FBM tensor monomial, because its ground states are all coordinate eigenstates that surely does not overlap in the $\calP$ space. To be utterly specific, for any given configuration point $q\defeq(s_1,\cdots\!,s_n)\in\calP$, there is at most one unique ground state in the set $\psi_0(M(\epsilon_1,\cdots\!,\epsilon_l))$ that does not vanish at $q$, which can be easily identified by a $(Z^+_{i_1},\cdots\!,Z^+_{i_m})$-joint measurement on the location $q\in\calP$, that is, on the coordinate eigenstate $|q\rangle$, so to obtain a bit array $(\langle q|Z^+_{i_1}|q\rangle,\cdots\!,\langle q|Z^+_{i_m}|q\rangle)=(s_{i_1},\cdots\!,s_{i_m})$ as the result of measurement, where $\forall i\in[1,n]$, $Z^+_i\defeq\half(I+Z_i)$ measures the $\{0,1\}$ coordinate of the $i$-th rebit, such that, for any $i\in[1,n]$, $\langle q|Z^+_i|q\rangle=\langle s_1\cdots s_n|Z^+_i|s_1\cdots s_n\rangle=\langle s_i|Z^+_i|s_i\rangle=s_i\in\{0,1\}$. Furthermore, if $j\in\calI_2$ is a typical index, and $U_j=X_j\sin\theta+Z_j\cos(\theta)$, $\theta\in\mathbb{R}$ is a general single-rebit operator acting on the $j$-th rebit, then a $\calP_1$-diagonally controlled FBM operator $M_{\bm{1}}\otimes U_j$, with $M_{\bm{1}}\defeq M(\epsilon_1,\cdots\!,\epsilon_l)\in\calB(\calI_1)$, $\lambda_0(\pm M_{\bm{1}})=-1$, is also node-determinate, because the operator moves only one rebit, and its ground states in the basis $\{\psi_0(M_{\bm{1}})\otimes\psi_0({-}U_j)\}\cup\{\psi_0({-}M_{\bm{1}})\otimes\psi_0(U_j)\}$ have absolutely non-overlapping supports in the configuration space $\calP=\{0,1\}^n$, since $\psi_0({-}U_j)$ and $\psi_0(U_j)$ are each non-degenerate, and every eigenstate in the potentially degenerate manifold $\psi_0(M_{\bm{1}})\cup\psi_0({-}M_{\bm{1}})$ is a $\calP$-coordinate eigenstate that is distinguished from others in the manifold by its different support in $\calP$.

Also, let $\calI_1=\{1,2\}$, $\calI_2=[1,n]\setminus\calI_1$, $\calP=\calP_1\times\calP_2$ be the corresponding Cartesian factorization of the configuration space, consider the FBM projection operators $X_{12}^{\pm}\defeq\half(I\mp X_{12})$ with $X_{12}\defeq X_1\otimes X_2$, and a controlled FBM operator $H_{xxz}\defeq X_{12}\otimes M_{\bm{2}}$, with $M_{\bm{2}}\in\calB(\calI_2)\defeq\bigotimes_{i=3}^n{\rm span}(I_i,Z_i)$ being a $\calP_2$-diagonal FBM tensor monomial such that $\lambda_0(\pm M_{\bm{2}})=-1$. It turns out that both $X_{12}^+=\half(I+X_{12})$ and $X_{12}^-=\half(I-X_{12})$ are node-determinate, because the operators have the famous Bell states as their ground states, namely,
\begin{align}
\sqrt{2}\,\psi_0(X_{12}^+) \,=\,\;& \{|{+}{-}\rangle+|{-}{+}\rangle,\,|{+}{-}\rangle-|{-}{+}\rangle\} \,=\, \{|00\rangle-|11\rangle,\,|10\rangle-|01\rangle\}, \label{psi0X12P} \\[0.75ex]
\sqrt{2}\,\psi_0(X_{12}^-) \,=\,\;& \{|{+}{+}\rangle+|{-}{-}\rangle,\,|{+}{+}\rangle-|{-}{-}\rangle\} \,=\, \{|00\rangle+|11\rangle,\,|10\rangle+|01\rangle\}, \label{psi0X12N}
\end{align}
such that, for either $X_{12}^+$ or $X_{12}^-$, the two basis ground states in either $\psi_0(X_{12}^+)$ or $\psi_0(X_{12}^-)$ are absolutely non-overlapping in the $\{0,1\}^2=\{00, 01, 10, 11\}$ configuration space, and the two of them can be easily told apart by performing a parity $Z_{12}$ measurement jointly on the two rebits, with the parity operator $Z_{12}\defeq Z^+_1+Z^+_2\pmod*{2}$, such that $\langle s_1s_2|Z_{12}|s_1s_2\rangle=s_1+s_2\pmod*{2}$, $\forall(s_1,s_2)\in\{0,1\}^2$. Moreover, the operator $H_{xxz}=X_{12}\otimes M_{\bm{2}}$ is node-determinate, because the ground states in $\psi_0(H_{xxz})=\{\psi_0(X_{12})\otimes\psi_0({-}M_{\bm{2}})\}\cup\{\psi_0({-}X_{12})\otimes\psi_0(M_{\bm{2}})\}$ have non-overlapping supports in the $\calP$ space and can be told apart unequivocally by a $(Z_{12},\{Z^+_i:i\in\calI_2\})$-joint measurement on any given configuration point $q\in\calP$ that is within the support of $\psi_0(H_{xxz})$, namely, the unique member wavefunction in the set $\psi_0(H_{xxz})$ that does not vanish at a given configuration coordinate $q=(s_1,s_2,\{s_i:i\in\calI_2\})\in\calP$ can be easily identified by performing a $(Z_{12},\{Z^+_i:i\in\calI_2\})$-joint measurement on the point $q$, that is, by measuring the parity $Z_{12}$ on the rebits $1$ and $2$ jointly, and $Z^+_i$-measuring each of the $i$-th rebit, $\forall i\in\calI_2$.

Similarly and straightforwardly, let $\calI_1=\{i_1,\cdots\!,i_m\}\subseteq[1,n]$ and $\calI_2\defeq[1,n]\setminus\calI_1$ be a set partition of the indices of the rebits, and $\calP=\calP_1\times\calP_2$ be the corresponding Cartesian factorization of the configuration space, then within the subsystem supported by the configuration space $\calP_1$, the {\em multi-rebit-simultaneously-flipping} operator $\eta X_{\raisebox{0.1\height}{\tiny $i_1\cdots i_m$}} \defeq \eta \bigotimes_{k=1}^mX_{i_k}$, $\eta=\pm 1$ is node-determinate, with $\scalebox{1.15}{\{} \psi_{0,\,s_1\cdots s_m} (X_{\raisebox{0.1\height}{\tiny $i_1\cdots i_m$}}) \defeq \scalebox{1.1}{(} |s_1\cdots s_m\rangle - \eta|(1-s_1)\cdots(1-s_m)\rangle \scalebox{1.1}{)} / \sqrt{2} : s_1\cdots s_m \in \{0,1\}^m \scalebox{1.15}{\}}$ lists an orthonormal basis of ground state functions, whose supports are mutually non-overlapping. Given any configuration point $q_1=(s_1,\cdots\!,s_m)\in\calP_1$, there is a unique ground state $\psi_{0,\,s_1\cdots s_m} (X_{\raisebox{0.1\height}{\tiny $i_1\cdots i_m$}})$ such that $\langle q_1 | \psi_{0,\,s_1\cdots s_m} (X_{\raisebox{0.1\height}{\tiny $i_1\cdots i_m$}}) \rangle \neq 0$, whose support actually contains only two configuration points, $q_1$ and its Boolean complement $q_1^c \defeq \scalebox{1.1}{(} (1-s_1),\cdots\!,(1-s_m) \scalebox{1.1}{)} \in\calP_1$. Let $M_{\bm{2}}\in\calB(\calI_2)\defeq\bigotimes_{i=3}^n{\rm span}(I_i,Z_i)$ with $\lambda_0(\pm M_{\bm{2}})=-1$ be a general $\calP_2$-diagonal tensor monomial that is not necessarily FBM, then it is easily verified, and will be rigorously proved shortly in the following, that the operator $X_{\raisebox{0.1\height}{\tiny $i_1\cdots i_m$}} \otimes M_{\bm{2}}$ is node-determinate.

More generally, with the same $\calI_1=\{i_1,\cdots\!,i_m\}\subseteq[1,n]$, $\calI_2\defeq[1,n]\setminus\calI_1$, $\calP=\calP_1\times\calP_2$, and a general $\calP_2$-diagonal tensor monomial $M_{\bm{2}}\in\calB(\calI_2)$ with $\lambda_0(\pm M_{\bm{2}})=-1$ as introduced immediately above, then any classical reversible computation \cite{Bennett73,Fredkin82,Bennett82,Feynman96} performed on the computational basis states associated with the subsystem supported by $\calP_1$ amounts to a particular type of self-inverse operator $F_{\bm{1}}$, known as a {\em classical reversible gate array} \cite{Shor97} (CRGA) operator, for which the tensor monomial $F_{\bm{1}}\otimes M_{\bm{2}}$ is also guaranteed to be node-determinate. Specifically, a CGRA operator $F_{\bm{1}}$ is associated with a bijective function $F:\{0,1\}^m\mapsto\{0,1\}^m$, which maps each $(s_1\cdots s_m)\in\{0,1\}^m$ to a unique $F((s_1\cdots s_m)\in\{0,1\}^m$. Therefore, $F$ partitions the set $\{0,1\}^m$ into $2^{m-1}$ $F$-equivalence classes, each $F$-equivalence class containing exactly two elements: an $(s_1\cdots s_m)\in\{0,1\}^m$, together with its $F$-image $F(s_1\cdots s_m)\in\{0,1\}^m$. Let $\{0,1\}^m/F$ denote the quotient space, that is, the set containing the $2^{m-1}$ $F$-equivalence classes. Then the operator $F_{\bm{1}}$ can be represented as
\begin{align}
F_{\bm{1}} \,=\,\;& {\textstyle{ \scalebox{1.35}{$\sum$}_{\{(s_1\cdots s_m),\,F(s_1\cdots s_m)\}\,\in\,\{0,1\}^m/F} }} \, |\psi_+(s_1\cdots s_m)\rangle \langle\psi_+(s_1\cdots s_m)| \nonumber \\[0.75ex]
\,-\,\;& {\textstyle{ \scalebox{1.35}{$\sum$}_{\{(s_1\cdots s_m),\,F(s_1\cdots s_m)\}\,\in\,\{0,1\}^m/F} }} \, |\psi_-(s_1\cdots s_m)\rangle \langle\psi_-(s_1\cdots s_m)|\,, \label{NodeDeterOfCRGA}
\end{align}
with $\psi_{\pm}(s_1\cdots s_m) \defeq \{|(s_1\cdots s_m) \pm |F(s_1\cdots s_m)\} /\sqrt{2}$ being two and the only two eigenstates of the self-inverse operator $F_{\bm{1}}$, corresponding to each $\{(s_1\cdots s_m),\,F(s_1\cdots s_m)\}\in\{0,1\}^m/F$. Clearly, both $F_{\bm{1}}$ and $-F_{\bm{1}}$ have a node-determinate manifold of ground states. Consequently, the tensor monomial $F_{\bm{1}} \otimes M_{\bm{2}}$ is node-determinate, as will be proven below.

\begin{definition}{(Doubly Node-Determinate FBM Tensor Monomials)}\label{defiDoubNodDet}\\
An FBM tensor monomial $M$ is called doubly $\epsilon$-almost node-determinate or doubly node-determinate up to $\epsilon>0$, when both $\lambda_0(M)<0$ and $\lambda_0({-}M)<0$, and orthonormal bases $\psi_0(M)\defeq\{\psi_{0,l}(M):l\in[1,\gmul(M,0)]\}$ and $\psi_0({-}M)\defeq\{\psi_{0,l}({-}M):l\in[1,\gmul({-}M,0)]\}$ exist on each of the $\calE_{\mathsmaller{M}}$-cosets $\{\calE_{\mathsmaller{M}}\oplus u^{\mathsmaller{\perp}}:u^{\mathsmaller{\perp}}\in\calE_{\mathsmaller{M}}^{\mathsmaller{\perp}},M\;\textit{moves}\;\calE_{\mathsmaller{M}}\oplus u^{\mathsmaller{\perp}}\}$, such that the subset
\vspace{-1.25ex}
\begin{equation}
\calU_{\pm}(M)\defeq\{q\in\calE_{\mathsmaller{M}}\oplus\calE_{\mathsmaller{M}}^{\mathsmaller{\perp}}:\exists\,\phi,\phi'\in\psi_0(M)\cup\psi_0({-}M),\,\mbox{such that}\;\phi\neq\phi',\,\phi(q)\phi'(q)\neq 0\} \nonumber
\vspace{-1.25ex}
\end{equation}
has a $|\psi_0(\pm M)|^2$-measure $\scalebox{1.15}{$\sum$}_{\,\phi\,\in\,\psi_0(M)\,\cup\,\psi_0({-}M)}\,\scalebox{1.1}{$\int$}_{\!\calU_{\pm}(M)}\,|\phi(q)|^2\,dV_g(q)\le\epsilon$.
\end{definition}

\begin{definition}{(Universally Node-Determinate FBM Tensor Monomials)}\label{defiUinvNodDet}\\
An FBM tensor monomial $M$ is called universally $\epsilon$-almost node-determinate, or universally node-determinate up to $\epsilon>0$, when on each of the $\calE_{\mathsmaller{M}}$-cosets $\{\calE_{\mathsmaller{M}}\oplus u^{\mathsmaller{\perp}}:u^{\mathsmaller{\perp}}\in\calE_{\mathsmaller{M}}^{\mathsmaller{\perp}},M\;\textit{moves}\,\;\calE_{\mathsmaller{M}}\oplus u^{\mathsmaller{\perp}}\}$, a suitable choice of all of its eigenfunctions always exists and forms an orthonormal basis $\psi_*(M)\defeq\bigcup_{n\,\in\,\{0\}\,\cup\,\mathbb{N}}\{\psi_n(M)\}$, such that the subset
\vspace{-1.25ex}
\begin{equation}
\calU_*(M)\defeq\{q\in\calE_{\mathsmaller{M}}\oplus\calE_{\mathsmaller{M}}^{\mathsmaller{\perp}}:\exists\,\phi,\phi'\in\psi_*(M),\,\mbox{such that}\;\phi\neq\phi',\,\phi(q)\phi'(q)\neq 0\} \nonumber
\vspace{-1.25ex}
\end{equation}
has a $|\psi_*(M)|^2$-measure $\scalebox{1.15}{$\sum$}_{\,\phi\,\in\,\psi_*(M)}\,\scalebox{1.1}{$\int$}_{\!\calU_*(M)}\,|\phi(q)|^2\,dV_g(q)\le\epsilon$.
\vspace{-1.5ex}
\end{definition}

The notion of a universally node-determinate FBM tensor monomial may seem a little peculiar, but it encapsulates and is epitomized by a class of FBM polynomials that are built from configuration coordinate measuring/projection operators, which are straightforwardly computable and universally node-determinate, yet endow tremendous computational power.

\begin{lemma}{(Node-Determinate FBM Tensor Monomials of the First Kind)}\label{NodeDeterTensorMonoI}\\
If an FBM tensor monomial $M=U\otimes V$ is a tensor product of two FBM tensor monomials $U$ and $V$ supported by configuration spaces $\calC_{\mathsmaller{U}}$ and $\calC_{\mathsmaller{V}}$ respectively, where $U$ is doubly $\epsilon$-almost node-determinate, $\epsilon>0$, while $V$ is either 1) doubly non-degenerate in the sense that both $\lambda_0(V)$ and $\lambda_0(-V)$ are negative and non-degenerate eigenvalues, or 2) doubly $\epsilon$-almost node-determinate, then, respectively, $M$ is either 1) $\epsilon$-almost  node-determinate, or 2) $2\epsilon$-almost node-determinate, and called a node-determinate FBM tensor monomial of the first kind in either case.
\end{lemma}
\vspace{-4.0ex}
\begin{proof}
Depending on whether $\lambda_0(U)\lambda_0({-}V)>\lambda_0({-}U)\lambda_0(V)$, or $\lambda_0(U)\lambda_0({-}V)<\lambda_0({-}U)\lambda_0(V)$, or $\lambda_0(U)\lambda_0({-}V)=\lambda_0({-}U)\lambda_0(V)$ holds, the ground eigenspace of $M$ may be $\psi_0(M)=\psi_0(U)\otimes\psi_0({-}V)$, or $\psi_0(M)=\psi_0({-}U)\otimes\psi_0(V)$, or $\psi_0(M)=\{\psi_0(U)\otimes\psi_0({-}V)\}\cup\{\psi_0({-}U)\otimes\psi_0(V)\}$ respectively. No matter which is the case, the ground eigenspace has an orthonormal basis $\psi_0(M)=\{\phi_{\mathsmaller{U}}\otimes\phi_{\mathsmaller{V}}:\phi_{\mathsmaller{U}}\in\psi_0(U)\cup\psi_0({-}U),\,\phi_{\mathsmaller{V}}\in\psi_0(V)\cup\psi_0({-}V)\}$, where each member wavefunction $\phi_{\mathsmaller{U}}\otimes\phi_{\mathsmaller{V}}$ does not overlap with others, in the configuration space $\calC_{\mathsmaller{U}}\times\calC_{\mathsmaller{V}}$ excluding a subset with a $|\psi_0(M)|^2$-measure no more than $\epsilon$ or $2\epsilon$ corresponding to the cases of $V$ being either doubly non-degenerate or doubly $\epsilon$-almost node-determinate.
\end{proof}

\begin{lemma}{(Node-Determinate FBM Tensor Monomials of the Second Kind)}\label{NodeDeterTensorMonoII}\\
An FBM tensor monomial $M=\bigotimes_{i=1}^mV_i$, $m\in\mathbb{N}$, with each $V_i$, $i\in[1,m]$ being an FBM tensor monomial, is $m\epsilon$-almost node-determinate and non-negative such that $M\ge 0$, if all but one of the $m$ FBM tensor monomials $\{V_i:i\in[1,m]\}$ are universally $\epsilon$-almost node-determinate, $\epsilon>0$, with the only one possible exception being at least  $\epsilon$-almost node-determinate, and all of the FBM tensor monomials are non-negative such that $V_i\ge 0$, $\forall i\in[1,m]$. Such an operator $M$ is called a node-determinate FBM tensor monomial of the second kind.
\end{lemma}
\vspace{-4.0ex}
\begin{proof}
Suffice it to just demonstrate the basic case of $m=2$ that $M=V_1\otimes V_2$ is $(\epsilon_1+\epsilon_2)$-almost node-determinate and non-negative such that $M\ge 0$, if $V_1$, $V_2$ are two FBM tensor monomials supported by configuration spaces $\calC_1$ and $\calC_2$ respectively, such that $V_1$ is universally $\epsilon_1$-almost node-determinate, $\epsilon_1>0$, and $V_2$ is $\epsilon_2$-almost node-determinate, $\epsilon_2>0$, furthermore, $V_1\ge 0$, $V_2\ge 0$. The ground eigenspace of $M$ has a basis
\vspace{-1.25ex}
\begin{equation}
\psi_0(M)\subseteq\{\psi_n(V_1)\otimes\psi_0(V_2):n\in\mathbb{Z},\,n>0\}\cup\{\psi_0(V_1)\otimes|q\rangle:q\in\calC_2\}, \nonumber
\vspace{-1.25ex}
\end{equation}
whose member wavefunctions do not overlap with each other over the configuration space $\calC_1\times\calC_2$ excluding a subset with a $|\psi_0(M)|^2$-measure less than $\epsilon_1+\epsilon_2$. Non-negativity of $M$ is obvious. Evidently, the same are also true for the FBM tensor monomial $M'=V_2\otimes V_1$.
\vspace{-1.5ex}
\end{proof}

Although by no means exhausting all of the ways that FBM tensor monomials can become node-determinate, doubly and universally node-determinate FBM tensor monomials as well as node-determinate FBM tensor monomials of the first and the second kinds do encapsulate a large set of node-determinate FBM tensor monomials. By means of direct sums, they also combine into a large collection of node-determinate FBM tensor polynomials. 

\begin{definition}{(Ground state Frustration-Free Hamiltonian)}\label{defiGFFH}\\
A partial Hamiltonian $H$ on a configuration space $\calC$ of size $N \defeq \size(C)$ is called ground state frustration-free (GFF), when it is PLTKD in the form $H = \sum_{i=1}^{\mathsmaller{J}} H_i$, $J \in \mathbb{N}$, $J = O(\poly(N))$ and has a non-degenerate ground state $\psi_0(H)$, that is separated from all of the excited states by an energy gap sized as $\Omega(1/\poly(N))$, where each $H_i$ is $O(1/\poly(N))$-almost node-determinate and has $\psi_0(H)$ as its ground state or one of its ground states, $\forall i \in [1,J]$. Each such additive partial Hamiltonian $H_i$, $i\in[1,J]$ is called GFF-compatible with respect to $H$.
\vspace{-1.5ex}
\end{definition}

The ground state $\psi_0(H)$ of any efficiently computable GFF partial Hamiltonian $H = \sum_{i=1}^{\mathsmaller{J}} H_i$ can be efficiently simulated through Markov chain Monte Carlo (MCMC) using a Markov transition matrix $[\psi_0(H;r)]\,\langle r\exp[-\tau\Base(H)]|q\rangle\,[\psi_0(H;q)]^{-1}$, $(r,q) \in \calC\times\calC$, $r$ and $q$ being restricted to within the same nodal cell $\calN \subseteq \calC$, where $\psi_0(H;r)/\psi_0(H;q) = \psi_0(H_i;r)/\psi_0(H_i;q)$ must hold true for a certain $i \in [1,J]$, and the quotient can be computed efficiently for all $(r,q) \in \calN\times\calN$ such that $\langle r|\exp[-\tau\Base(H_i)]|q\rangle \neq 0$, by the premise of $H$ being GFF.

\begin{definition}{(Directly Frustration-Free Hamiltonian)}\label{defiDFFH}\\
A partial Hamiltonian $H$ on a configuration space $\calC$ of size $N \defeq \size(C)$ is called directly frustration-free (DFF), when it is a direct sum of partial Hamiltonians in the form $H = \sum_{i=1}^{\mathsmaller{J}} H_i$, $J \in \mathbb{N}$, $J = O(\poly(N))$, with the energy gap between the ground state and the excited states $\Omega(1/\poly(N))$ lower-bounded for $H$ and every $H_i$, $i\in[1,J]$, where each $H_i$, $i\in[1,J]$ has $0$ as the smallest eigenvalue and moves an $O(\log(N))$-sized configuration subspace $\calC_i$, which is annihilated by any other $H_j$, $j\in[1,J]$, $j\ne i$, namely, $H_j\phi_i=0$ holds true, $\forall\phi_i \in L^2(\calC_i)$, so long as $j\neq i$. Each such additive partial Hamiltonian $H_i$, $i\in[1,J]$ is called DFF-compatible with respect to $H$.
\vspace{-1.5ex}
\end{definition}

Any Gibbs operator $\exp(-\tau H)$, $\tau\in(0,\infty]$ of any efficiently computable DFF partial Hamiltonian $H = \sum_{i=1}^{\mathsmaller{J}} H_i$ can be efficiently simulated by an identity $\exp(-\tau H) = I + \sum_{i=1}^{\mathsmaller{J}} [\,\exp(-\tau H_i) - I\,]$, due to the premise that each $H_i$, $i\in[1,J]$ only moves an $O(\log(N))$-sized configuration subspace, which is annihilated by any other $H_j$, $j\in[1,J]$, $j\ne i$, consequently, $H_iH_j = H_jH_i = 0$ holds true so long as $j\neq i$. In particular, when $\tau > 0$ is sufficiently large, with respect to any fixed $q_0\in\calC$, the Gibbs wavefunction $\langle\cdot|\exp(-\tau H)|q_0\rangle$ is essentially a ground state of $H$, which can be efficiently simulated using either a constrained path integral method \cite{Ceperley91,Ceperley96,Zhang97,Zhang04,Kruger08} or a fixed-node diffusion method \cite{Anderson75,Anderson76,Ceperley80,Ceperley81,Caffarel88I,Caffarel88II,Ceperley91,Cances06}, whose nodal cells can always be efficiently and uniquely determined, even though the partial Hamiltonian $H$ may have degenerate ground states.

In a constrained path integral method, a Feynman slab $\Slab(\calC_0,\calC_{\mathsmaller{L}},H,\tau)$, $L \in \mathbb{N}$, also regarded as a Feynman substack consisting of the single Feynman slab, is divided into a number $1+L$ of Feynman slices, including the first with $\calC_0 \cong \calC$ and the last with $\calC_{\mathsmaller{L}} \cong \calC$, and a Feynman path $\{q_m\}_{m\in[0,\mathsmaller{M}]}$ undergoes a MCMC random walk in the space $\{q_0\}\times\prod_{m=1}^{\mathsmaller{M}}\calC_m$, with $\calC_m \cong \calC$ for all $m\in[1,M]$, subject to the constraint of no crossing nodal surfaces, and following a Metropolis-Hastings \cite{Metropolis53,Hastings70} or Gibbs \cite{Geman84,Liu01,Li09,Bremaud17} strategy of importance sampling, which uses the product of Gibbs transition amplitudes assigned to said Feynman path as a {\em density proxy} to the probability distribution of Feynman paths, where the density proxy needs only to be proportional to the probability density of Feynman paths, with a proportionality constant fixed but unknown.

In a fixed-node diffusion method, with respect to any fixed $q_0 \in \calC_0 \cong \calC$ and at a sufficiently large $\tau > 0$, the Gibbs transition amplitude $\phi_0(q_{\mathsmaller{L}}) \defeq \langle q_{\mathsmaller{L}}|\exp(-\tau H)|q_0\rangle$ can be efficiently computed for all $q_{\mathsmaller{L}} \in \calC_{\mathsmaller{L}} \cong \calC$, $L \in \mathbb{N}$, such computed numerical values for the wavefunction $\phi_0(q_{\mathsmaller{L}})$, $q_{\mathsmaller{L}}\in\calC_{\mathsmaller{L}}$, which is necessarily a ground state of $H$, can be used to guide an MCMC random walk of $q_{\mathsmaller{L}}$ over $\calC_{\mathsmaller{L}}$ with Metropolis-Hastings or Gibbs importance sampling using $|\phi_0(q_{\mathsmaller{L}})|^2$ as a density proxy, so to generate random samples of $q_{\mathsmaller{L}} \in \calC_{\mathsmaller{L}}$ according to the probability density $|\phi_0(q_{\mathsmaller{L}})|^2/\|\phi_0(q_{\mathsmaller{L}})\|^2$. It is usually desired to keep $q_{\mathsmaller{L}}$ within a fixed nodal cell of $\phi_0(\cdot)$, and there is no difficulty to locate the nodal surfaces when the function $\phi_0(\cdot)$ itself is efficiently computable. It may be further desired to keep $q_{\mathsmaller{L}}$ within the Ceperley reach with respect to the reference point $q_0\in\calC_0$, whose boundary surface is still not hard to locate for the Feynman slab $\Slab(\calC_0,\calC_{\mathsmaller{L}},H,\tau)$.

\begin{definition}{(Separately Frustration-Free Hamiltonian)}\label{defiSFFH}\\
A partial Hamiltonian $H$ that generates a Gibbs operator $\exp\{-t[H{-}\lambda_0(H)]\}$ of interest at a fixed $t \in (0,\infty]$ is called separately frustration-free (SFF), when it is a sum of partial Hamiltonians as $H = \sum_{k=1}^{\mathsmaller{K}} H_k$, $K\in\mathbb{N}$, $K = O(\poly(N))$, $N \defeq \size(H)$, with each shifted partial Hamiltonian $H_k-\lambda_0(H_k)$, $k \in [1,K]$ being either DFF or GFF, and the Gibbs operator $\exp\{-t[H{-}\lambda_0(H)]\}$ can be simulated, to within an error no more than $O(1/\poly(N))$, by iterating a sequence of Gibbs operators $\{\exp\{-\tau[H_k{-}\lambda_0(H_k)]\} : k\in[1,K]\}$, $\tau \in (0,\infty]$ for no more than $O(\poly(N))$ times.
\vspace{-1.5ex}
\end{definition}

In one exemplary embodiment, where $t \in(0,\infty)$ is finite, the Gibbs operator $\exp\{-t[H{-}\lambda_0(H)]\}$ is approximated by iterating a Lie-Trotter-Kato product $\prod_{k=1}^{\mathsmaller{K}}\exp\{-[\tau(H_k-\lambda_0(H_k)]\}$, $\tau = t/m$ for $m\in\mathbb{N}$ times, with $m = \poly(N)$ being sufficiently large for a sufficiently accurate approximation. In another exemplary embodiment, where $t$ is so large that the Gibbs operator $\exp\{-t[H{-}\lambda_0(H)]\}$ is essentially the same as the ground state projector $|\psi_0(H)\rangle\langle\psi_0(H)|$, $\tau$ is also chosen to be large, and each Gibbs operator $\exp\{-\tau[H_k{-}\lambda_0(H_k)]\}$ is applied in turn to effect a ground state projector $|\psi_0(H_k)\rangle\langle\psi_0(H_k)|$, $\forall k \in [1,K]$, and a process of iterating a sequence of ground state projectors $\{|\psi_0(H_k)\rangle\langle\psi_0(H_k)|\}_{k \in [1,\mathsmaller{K}]}$ mixes rapidly and converges to the projector $|\psi_0(H)\rangle\langle\psi_0(H)|$.

\begin{lemma}{(Homophysical Images of DFF, GFF, and SFF Partial Hamiltonians)}\label{HomoPhysDFFGFFSFF}\\
Let $\mathfrak{M}:(\calC_1,\calH_1,\calB_1)\mapsto(\calC_2,\calH_2,\calB_2)$ be a homophysics. If an FBM tensor monomial $h\in\calB_1$ is (doubly/universally) $\epsilon$-almost node-determinate, then $\mathfrak{M}(h)\in\calB_2$ is also (doubly/universally) $\epsilon$-almost node-determinate correspondingly. Therefore, if $H\in\calB_1$ is DFF, GFF, or SFF, then so is $\mathfrak{M}(H)\in\calB_2$ respectively.
\end{lemma}
\vspace{-4.0ex}
\begin{proof}[\iftoggle{ForUSPTO} {Demonstration} {Proof}]
The proposition follows straightforwardly from definitions. 
\end{proof}

\begin{definition}{(DFF-FS, GFF-FS, SFF-FS Partial Hamiltonians and Systems)}\label{defiSFFFS}\\
A partial Hamiltonian is called DFF fermionic Schr\"odinger (DFF-FS), GFF fermionic Schr\"odinger (GFF-FS), or SFF fermionic Schr\"odinger (SFF-FS), when it is both fermionic Schr\"odinger and DFF, GFF, or SFF, respectively. A quantum system is called SFF-FS when its governing Hamiltonian is SFF-FS.
\end{definition}

\begin{lemma}{(Path Rectification for Feynman Stacks of SFF-FS Hamiltonians)}\label{PathRectFeynStack}\\
For an SFF-FS partial Hamiltonian $H=\sum_{k=1}^{\mathsmaller{K}}H_k$ on a configuration space $\calC$ of size $N\defeq\size(\calC)$, with $K\in\mathbb{N}$, $K=O(\poly(N))$, in conjunction with any $q_0\in\calC$ and any $\tau\in(0,\infty]$, the Gibbs wavefunction $\langle q_{\mathsmaller{K}}|\exp\{-\tau[H-\lambda_0(H)]\}|q_0\rangle$, $q_{\mathsmaller{K}} \in \calC_{\mathsmaller{K}} \cong \calC$ associated with a Feynman stack made of $K$ Feynman slabs $\{\Slab(\calC_{k-1},\calC_k,H_k,\tau) \}_{k\in[1,\mathsmaller{K}]}$ can be computed efficiently, to within an error no more than $O(1/\poly(N))$, by integrating over rectified Feynman chains only, and using a constrained path integral method to simulate the Feynman slabs, when all $\{H_k\}_{k\in[1,\mathsmaller{K}]}$ are DFF.

In case some $H_k$, $k\in[1,K]$ are GFF, then the ground state $\psi_0(H)$ can be efficiently simulated, to within an error no more than $O(1/\poly(N))$, by integrating over rectified Feynman chains only, using a fixed-node diffusion method to simulate each Feynman slab associated with a GFF partial Hamiltonian, and using either a fixed-node diffusion method or a constrained path integral method to simulate each Feynman slab associated with a DFF partial Hamiltonian.
\end{lemma}
\vspace{-4.0ex}
\begin{proof}[Proof]
It is without loss of generality and does not change physics to have all of the FBM tensor monomials $H_k$, $k\in[1,K]$ affine transformed so that $\lambda_0(H_k) = 0$, $\forall k \in [1,K]$. Just as in equation (\ref{FermiGibbsFeynStack}), let $H_0 \defeq 0$ and $\{(C_k \cong \calC,H_k) : k\in[0,K]\}$ be the boundary and intermediate Feynman slices, such that $\forall k \in [1,K]$, the $k$-th Feynman slab starts at $(\calC_{k-1},H_{k-1})$ and ends at $(\calC_k,H_k)$. Then, $\forall q_{\mathsmaller{K}} \in \calC_{\mathsmaller{K}}$, the Gibbs wavefunction for the product of Gibbs operators $\prod_{k=1}^{\mathsmaller{K}}\exp(-\tau H_k)$ associated with said Feynman stack reads
\begin{equation}
\langle q_{\mathsmaller{K}\,} |\, P_{\!\mathsmaller{F}} \, {{\textstyle \prod_{k=1}^{\mathsmaller{K}} }} \exp(-\tau H_k) \,|\, q_0\rangle = \int_{\calC_{\mathsmaller{K}-1}} \!\!\!\!\!\!\! \cdots \, \int_{\calC_1} {\textstyle{ \prod_{k=1}^{\mathsmaller{K}} }} \langle q_k \,|\, P_{\!\mathsmaller{F}} \exp(-\tau H_{k\mathsmaller{B}}) \,|\, q_{k-1}\rangle \; {\textstyle{ \prod_{k=1}^{\mathsmaller{K}-1} }} dq_k \,, \label{PathRectGibbsWaveFunc}
\end{equation}
where $H_{k\mathsmaller{B}} \defeq \Base(H_k)$, $\forall k \in [1,K]$. The Gibbs wavefunction is computed by integrating contributions of Feynman chains $\{\Bun(t_{k-1},q_{k-1};t_k,q_k) : k \in[1,K]\}$ over the intermediate configuration coordinates $\{q_k\}_{k\in[1,\mathsmaller{K})}$.

It turns out sufficient to consider and integrate over rectified Feynman chains only, because $\forall k\in[1,K)$, a fermion exchange operator $\pi_k \in G_{\rm ex}$ acting on a $q_k \in \calC_k$ introduces a $\sign(\pi_k)$ factor to both the Gibbs wavefunction associated with the $k$-th Feynman slab and that with the $(k+1)$-th Feynman slab, yielding a $\sign(\pi_k)^2 = +1$ factor without real effect to equation (\ref{PathRectGibbsWaveFunc}), while for $k=K$, a fermion exchange operator $\pi_{\mathsmaller{K}} \in G_{\rm ex}$ acting on a $q_{\mathsmaller{K}} \in \calC_k$ simply reflects $q_{\mathsmaller{K}}$ between the symmetric positive and negative nodal cells, where a restriction of the Gibbs wavefunction to a positive nodal cell is sufficient to determine and derive all of the quantum physics.

If all $\{H_k\}_{k\in[1,\mathsmaller{K}]}$ are DFF, then every fermionic Gibbs transition amplitude $\langle q_k|\exp(-\tau H_k)|q_{k-1}\rangle$, $\tau\in(0,\infty]$ associated with any Feynman bundle $\Bun(t_{k-1},q_{k-1};t_k,q_k)$ can be computed efficiently by a boltzmannonic path integral for the Gibbs operator $\exp(-t H_{k\mathsmaller{B}})$ over all of the Feynman paths that do not cross any nodal surface of the Gibbs wavefunctions $\{\langle q(t)|\exp(-t H_k)|q_{k-1}\rangle : (t,q(t)) \in [0,\tau] \times \calC\}$, which is a procedure well known as the constrained path integral method \cite{Ceperley91,Ceperley96,Zhang97,Zhang04,Kruger08}. The Gibbs wavefunctions and their  nodal surfaces are efficiently computable, because each $H_k$ is DFF, whose Gibbs wavefunctions with respect to any fixed $q_{k-1} \in \calC_{k-1}$ has an $O(\log(N))$-bounded dimensionality.

When some $H_k$, $k\in[1,K]$ are GFF, the Gibbs operator $\exp(-\tau H) = P_{\!\mathsmaller{F}} \exp[-\tau\Base(H)]$ needs to take the large-$\tau$ limit and reduces to the ground state projector. A $\tau = \Theta(\poly(N))$ can be chosen and be sufficiently large, since the ground states are separated by a polynomial-sized gap from the excited states. For each Feynman slab $\Slab(\calC_{k-1},\calC_k,H_k,\tau)$, $k\in[1,K]$, let $q_{k-1}\in\calC_{k-1}$ be a fixed configuration point and $\langle\cdot|\exp(-\tau H_k)|q_{k-1}\rangle$ be a Gibbs wavefunction of interest, which is essentially the same as the ground state $\psi_0(H_k;\cdot)$, that can be simulated efficiently using the fixed-node diffusion method by \myLemma\;\ref{QuasiStochasticOper} and  \iftoggle{ForUSPTO} {Derived Utility} {corollary} \ref{FixedNodeMethod}, with a Markovian transition matrix $[\psi_0(H_k;r_k)]\,\langle r_k|\exp(-\tau H_{k\mathsmaller{B}})|q_k\rangle\,[\psi_0(H_k;q_k)]^{-1}$, where $q_k\in\calC_k$ and $r_k\in\calC_k$ are in the same nodal cell that contains a point $\pi q_{k-1}$, $\pi\in G_{\rm ex}$ is an even permutation. The quotient $\psi_0(H_k;r_k)/\psi_0(H_k;q_k)$ is efficiently computable, $\forall(r_k,q_k)\in\calC_k\times\calC_k$, since $H_k$ is either DFF or GFF. Obviously, if $H_k$, $k\in[1,K]$ is in fact DFF, then the ground state $\langle\cdot|\exp(-\tau H_k)|q_{k-1}\rangle$ can be simulated efficiently using a constrained path integral method as well, since the Gibbs wavefunction $\langle\cdot|\exp(-t H_k)|q_{k-1}\rangle$ is efficiently computable for all $t\in(0,\tau]$, $\tau = \Theta(\poly(N))$.
\vspace{-1.5ex}
\end{proof}

Besides the fermionic Schr\"odinger Hamiltonians/systems, another class of Hamiltonians/systems that is useful for computations and simulations consists of FBM tensor polynomials as a sum of FBM tensor monomials that are essentially bounded.

\begin{definition}{(Essentially Bounded Partial Hamiltonian)}\label{defiEssenBoundHamil}\\
Given a partial Hamiltonian $H$ on a Hilbert space $\calH$ and an $\epsilon>0$, an $\epsilon$-essential computational subspace for $H$ is a Hilbert subspace $\calH'\subseteq\calH$, onto which $H$ is restricted to become $H'\defeq H|_{\calH'}$, such that
\begin{align}
\left|\langle\psi_0(H)|\psi_0(H')\rangle\right|^2 \ge\, (1-\epsilon)\,\langle\psi_0(H)|\psi_0(H)\rangle\,\langle\psi_0(H')|\psi_0(H')\rangle, \\[0.75ex]
(1-\epsilon) \,\le\, [\lambda_1(H')-\lambda_0(H')]\,/\,[\lambda_1(H)-\lambda_0(H)] \,\le\, (1+\epsilon).
\end{align}
A Lie-Trotter-Kato decomposable partial Hamiltonian $H=\sum_{k=1}^{\mathsmaller{K}}H_k\in\calL_0(\calH)$, $K\in\mathbb{N}$ is said to be $\epsilon$-essentially bounded by $\Lambda$, $\epsilon>0$, $\Lambda>0$, if an $\epsilon$-essential computational subspace $\calH'\subseteq\calH$ exists for $H$ and $\{H_k\}_{k\in[1,K]}$, such that the restriction of $H_k$ to $\calH'$ is bounded by $\Lambda$, namely, the operator norm $\|H_k|_{\calH'}\|\le\Lambda$, $\forall k\in[1,K]$.
\end{definition}


\begin{definition}{(DFF-EB, GFF-EB, SFF-EB Partial Hamiltonians and Systems)}\label{defiSFFEB}\\
A partial Hamiltonian is called DFF essentially bounded (DFF-EB), GFF essentially bounded (GFF-EB), or SFF essentially bounded (SFF-EB), when it is both essentially bounded and DFF, GFF, or SFF, respectively. A quantum system is called SFF-EB when its governing Hamiltonian is SFF-EB.
\vspace{-1.5ex}
\end{definition}

\myLemma\;\ref{QuasiStochasticOper} and \myLemma\;\ref{PathRectFeynStack} imply that the efficient classical Monte Carlo simulatability extends to any SFF-EB Hamiltonian $H=\sum_{k=1}^{\mathsmaller{K}}H_k$, because, for any $\tau>0$, the Gibbs operators $G(\tau H_k)\defeq\exp(\mathsmaller{-}\tau H_k)$, $k\in[1,K]$ can be similarity-transformed into quasi-stochastic operators $P(\tau H_k)\defeq[\psi_0^*(H)]G(\tau H_k)[\psi_0^*(H)]^{\mathsmaller{-}1}$, $k\in[1,K]$, by the same diagonal operator associated with the same wavefunction $\psi_0(H)$, where each operator $P(\tau H_k)$ fixes the probability distribution $|\psi_0(H)|^2$, while $K^{\mathsmaller{-}1}\sum_{k=1}^{\mathsmaller{K}}P(\tau H_k)$ attenuates the excited states of $H$, therefore, repeated random applications of the operators $\{P(\tau H_k)\}_{k\in[1,K]}$ as transition probability matrices, with $\tau>0$ taking a value that is sufficiently large, form a Markov chain that projects out $|\psi_0(H)|^2$ as the stationary distribution.

More specifically, with $\tau\rightarrow\mathsmaller{+}\infty$ asymptotically, $\forall k\in[1,K]$, the Gibbs operator $G(\tau H_k)$ approaches a limit $G(\mathsmaller{+}\infty H_k)\defeq\lim_{\tau\rightarrow\mathsmaller{+}\infty}\exp(\mathsmaller{-}\tau H_k)\defeq\Pi_k=\sum_{l\in\gmul(H_k)}|\psi_{0,l}(H_k)\rangle\langle\psi_{0,l}(H_k)|$ exponentially fast, where $\Pi_k$ is the orthogonal projection onto the null space of $H_k$, {\it i.e.}, its ground state Hilbert subspace, which has $\{\psi_{0,l}(H_k)\}_{l\in\gmul(H_k)}$ as an orthonormal basis. By the SFF property of $H$, the subset $\calU(H)$ of node-uncertain configuration points has an asymptotically negligible $|\psi_0(H)|^2$-measure, such that $\forall k\in[1,K]$,
\begin{equation}
\langle r|P(\textstyle{\mathsmaller{+}\infty}H_k)|q\rangle\approx{\textstyle\sum_{l\in\gmul(H_k)}}\left|\psi_{0,l}(H_k;r)\right|^2\times(q\in\supp(\psi_{0,l}(H_k))),\;\forall q,r\notin\calU(H), \label{r_PinftyHk_q}
\end{equation}
where $P(\mathsmaller{+}\infty H_k)\defeq\lim_{\tau\rightarrow\mathsmaller{+}\infty}P(\tau H_k)$, and $(q\in\supp(\psi_{0,l}(H_k)))$ is binary-valued to either $1$ or $0$ depending on whether or not $q$ falls into the support of $\psi_{0,l}(H_k)$, $\forall l\in\gmul(H_k)$. Furthermore, for any excited state of $H$ represented by a wavefunction $\psi_n(\cdot)\defeq\psi_n(H;\cdot)$, $n>0$, the inequality $\langle\psi_n|\exp(\mathsmaller{-}\tau H)|\psi_n\rangle/\langle\psi_n|\psi_n\rangle\le\exp(\mathsmaller{-}\tau\lambda_1)<1-\tau\lambda_1/2$ holds, $\forall\tau\in(0,1/\lambda_1]$, $\lambda_1\defeq\lambda_1(H)$. Since $H=\sum_{k=1}^{\mathsmaller{K}}H_k$ is an essentially bounded Lie-Trotter-Kato decomposition into non-negative operators, $\forall\tau\in(0,1/\lambda_1]$, for any excited state $\psi_n=\psi_n(H)$, $n>0$, there exists an $m\in\mathbb{N}$ such that
$$\textstyle{
\left[\prod_{k=1}^{\mathsmaller{K}}\exp(\mathsmaller{-}\tau H_k/m)\right]^m\ge\left[\prod_{k=1}^{\mathsmaller{K}}(I-\tau H_k/m)\right]^m\ge\left(I-\sum_{k=1}^{\mathsmaller{K}}\tau H_k/m\right)^m\ge I-\sum_{k=1}^{\mathsmaller{K}}\tau H_k,
}$$
$$\textstyle{
\left\langle\psi_n\!\left|\left[\prod_{k=1}^{\mathsmaller{K}}\exp(\mathsmaller{-}\tau H_k/m)\right]^m\right|\!\psi_n\right\rangle<(1+\tau\lambda_1/6)\left\langle\psi_n\!\left|\exp(\mathsmaller{-}\tau H)\right|\!\psi_n\right\rangle<(1-\tau\lambda_1/3)\,\langle\psi_n|\psi_n\rangle,
}$$
consequently, $\langle\psi_n|I-\sum_{k=1}^{\mathsmaller{K}}\tau H_k|\psi_n\rangle<(1-\tau\lambda_1/3)\,\langle\psi_n|\psi_n\rangle$, and further
$$\left\langle\psi_n\!\left|{\textstyle{\sum_{k=1}^{\mathsmaller{K}}}}(I-\Pi_k)\right|\!\psi_n\right\rangle\ge\left\langle\psi_n\!\left|\Lambda^{\mathsmaller{-}1}\,{\textstyle{\sum_{k=1}^{\mathsmaller{K}}}}H_k\right|\!\psi_n\right\rangle>(\lambda_1/3\Lambda)\,\langle\psi_n|\psi_n\rangle,$$
with $\Lambda\defeq\max_{\,k\in[1,K]}\{\|H_k\|\}$. Therefore, for any excited state $\psi_n=\psi_n(H)$, $n>0$, the inequality
\begin{equation}
\left\langle\psi_n\!\left|K^{\mathsmaller{-}1}\,{\textstyle{\sum_{k=1}^{\mathsmaller{K}}}}P(\mathsmaller{+}\infty H_k)\right|\!\psi_n\right\rangle=\left\langle\psi_n\!\left|K^{\mathsmaller{-}1}\,{\textstyle{\sum_{k=1}^{\mathsmaller{K}}}}\Pi_k\right|\!\psi_n\right\rangle<(1-\lambda_1/3K\Lambda)\,\langle\psi_n|\psi_n\rangle
\end{equation}
holds, meaning that the transition probability matrix $K^{\mathsmaller{-}1}\sum_{k=1}^{\mathsmaller{K}}P(\mathsmaller{+}\infty H_k)$, representing random applications of matrices $\{P(\mathsmaller{+}\infty H_k)\}_{k=1}^{\mathsmaller{K}}$, attenuates any excited state of $H$ rapidly.

In practical numerical computations as well as complexity analyses, the real or imaginary time axis is often discretized so that physical and Markov chain dynamics are represented discretely, where the discretization error can be made arbitrarily small efficiently by using the Lie-Trotter-Kato product formulas, especially the higher-order symmetrized versions of them \cite{Creutz89,Suzuki90,Yoshida90,Suzuki91,Suzuki92,Yoshida93,Chin02,Suzuki06,Sakkos09,Dornheim15,Yan17}

In the following I shall present, by ways of example and by no means of limitation, several Monte Carlo algorithms using the Metropolis-Hastings \cite{Metropolis53,Hastings70} or the Gibbs \cite{Geman84,Liu01,Li09,Bremaud17} method of importance sampling to simulate the ground states of SFF Hamiltonians, which take advantage of \myLemma\;\ref{PathRectFeynStack} and overcome the dreaded sign problem. Even though the exemplary algorithms are mostly illustrations of simulating ground states, from and by virtue of \myLemma\;\ref{PathRectFeynStack}, one skilled in the art will see straightforwardly that the same methods and algorithms can be applied or adapted slightly to simulate any Gibbs kernel at a finite imaginary time, also known as a thermal Green's function or a thermal density matrix, associated with an SFF partial Hamiltonian, so long as the SFF partial Hamiltonian is a sum of all DFF partial Hamiltonians.

Some of the algorithms are exhibited in a framework with the configuration space fully discretized, while some others are presented in the mathematical language of an all-continuous configuration space. As being noted previously, no loss of generality is incurred by choosing any such specific configuration space, because an actual or model physical system with either all-discrete or all-continuous dynamical variables can be made universal, in terms of both numerical simulations for or based on it and the computational power of it as a computing machine. Furthermore, methods and steps as specified in the algorithms can be combined and generalized straightforwardly to handle the case of a continuous-discrete product configuration space with ease.

For any multi-species many-body system of concern in this \iftoggle{ForUSPTO} {specification} {presentation}, even if some of the physical dynamical variables are intrinsically continuous and conventionally represented by canonical coordinates on Riemannian manifolds $\calM_s$, $s\in[1,S]$ as substrate spaces, such continuous substrate space can still be approximated by a set of grid/lattice points sampled at a desired resolution, that serve as vertices to form a connected graph $\calG_s$, $s\in[1,S]$, where each vertex is connected with no more than a fixed number of neighbors, by an edge representing the kinetics of a particle hopping between vertices. Taking advantage of its natural Cartesian product structure, the configuration space of such a multi-species many-body system is approximated and represented nicely by a Cartesian product graph $\calG\defeq\prod_{s=1}^{\mathsmaller{S}}\calG_s^{n_s}$ \cite{Imrich08}, with which and for any FBM interaction $h$, the notions of $h$-moved and $h$-fixed factor spaces (now subgraphs) $\calE_h$ and $\calE_h^{\mathsmaller{\perp}}$, $h$-moved and $h$-fixed $\calE_h$-cosets, $h$-induced ODSDs and its node-determinacy, the notions of a DFF or GFF Hamiltonian and its local node-determinacy, all carry over straightforwardly.

When it is clear that a product graph $\calG\defeq\prod_{s=1}^{\mathsmaller{S}}\calG_s^{n_s}$ as a discrete configuration space is the result of spatial discretization from a bounded Riemannian manifold as a continuous-discrete product configuration space $\calC=\prod_{s=1}^{\mathsmaller{S}}\calC_s^{n_s}$, the finite graph $\calG$ can inherit from $\calC$ the measures of configuration space dimension $\dim(\calG)\defeq\dim(\calC)$, diameter $\diam(\calG)=\diam(\calC)$, and computational size $\size(\calG)\defeq\dim(\calG)+\diam(\calG)=\size(\calC)$. Otherwise, without any reference to a continuous-discrete product configuration space being made, a finite product graph $\calG\defeq\prod_{s=1}^{\mathsmaller{S}}\calG_s^{n_s}$ can be assigned with estimated measures of configuration space dimension $\dim(\calG)\defeq\sum_{s=1}^{\mathsmaller{S}}n_s$, diameter $\diam(\calG)=\left(\sum_{s=1}^{\mathsmaller{S}}\diam(\calG_s)^2\right)^{1/2}$, and computational size $\size(\calG)\defeq\dim(\calG)+\diam(\calG)$, where for each $s\in[1,S]$, the diameter $\diam(\calG_s)$ of each $\calG_s$ is defined as $\diam(\calG_s)\defeq\max\{n\in\mathbb{N}:\calG_s\;\textit{having a path formed by}\;n\;\textit{consecutive edges}\}$. Such estimated measures may not be exactly the same as what might be inherited from a continuous-discrete product configuration space, but they are surely within a constant multiplier, and suffice for the purpose of asymptotic analyses, when $\size(\calG)$ grows as the result of an increasing $S\in\mathbb{N}$.

On a continuous submanifold $\calM\subseteq\calC$ of a general configuration space $\calC$, a nodal cell is an open subset of $\calM$, for which topological connectedness is equivalent to path-connectedness. When $\calM$ is discretized, the nodal surfaces are no longer located precisely, rather they fall mostly between pairs of adjacent lattice points at which the ground state wavefunction assumes different signs. As the lattice resolution increases, simulation errors due to the location uncertainty of the nodal surfaces, which does not exceed half of the grid size, can be made arbitrarily small. Furthermore, the so-called ``lever rule'' for fixed-node diffusion QMC of lattice fermions \cite{vanBemmel94,Sorella15} with a DFF or GFF Hamiltonian $H=\sum_{k=1}^{\mathsmaller{K}}H_k$ can pinpoint the nodal surfaces more accurately, and the introduction of a ``nodal boundary potential''  \cite{vanBemmel94,Sorella15} can represent the effects of $\psi_0(H;r\in\calM)$ nodal surface crossings exactly without error.

With a properly discretized configuration space $\calG$, the Hilbert space $L^2(\calG;\mathbb{K})$ of quantum states and linear operators on $L^2(\calG;\mathbb{K})$, $\mathbb{K}\in\{\mathbb{R},\mathbb{C}\}$, all have finite dimensions, so all quantum physics are described and solved by matrix algebra, saving much mathematical technicality of functional analysis. Particularly, in such lattice formulations, all operators as finite matrices are bounded, and their polynomial Lie-Trotter-Kato decompositions follow straightforwardly from Lie's product formula, as a trivial consequence of the Baker-Campbell-Hausdorff formula. Also, a spatially local Hamiltonian becomes a matrix that couples each lattice point with no more than an $O(\dim(\calG))$ number of nearby neighbors. Recall that the graph-theoretic concept of connectedness is also based on path connections. So path-connectedness is generally applicable to continuous, discrete, as well as continuous-discrete product configuration spaces. Furthermore, a quantitative notion of path-connectedness will prove handy in computational complexity considerations.

\begin{definition}
A bounded and connected subset $\calD$ of a configuration space $\calC$, with $\calC$ being either continuous as a Riemannian manifold or discrete as a graph, is called polynomially path-connected, if any two points in $\calD$ is connected by a path in $\calD$ whose length is $O(\poly(\size(\calC)))$, where the length of a path is measured either through the Riemannian metric or in terms of the number of graph edges, respectively.
\vspace{-1.5ex}
\end{definition}

Polynomial path-connectedness may necessarily hold for any nodal cell of a quantum system on a suitably regular configuration space, especially when the Hamiltonian has a polynomially gapped ground state. But let's save a rigorous treatment for a dedicated mathematical discourse elsewhere, and simply posit the property when necessary in this \iftoggle{ForUSPTO} {specification} {presentation}. Let $\Gamma(M)$ denote the directed graph associated with a finite matrix $M$ \cite{Horn85}. The irreducibility of $M$ is equivalent to the fact that $\Gamma(M)$ is strongly connected. When $M$ is primitive, {\it i.e.}, aperiodic irreducible non-negative \cite{Horn85,Meyer00}, the polynomial boundedness of its index of primitivity $\gamma(M)$ is tantamount to the polynomial path-connectedness of $\Gamma(M)$, namely, that any ordered pair of vertices of $\Gamma(M)$ can be connected by a directed path whose length does not exceed the polynomially bounded integer $\gamma(M)$. With a self-adjoint matrix, directedness is immaterial for the associated graph, because any directed arc is always accompanied by another oppositely directed arc between the same pair of vertices.

It may be useful to note that, the term {\em sign graph} has been proposed \cite{Davies01} to denote a maximally connected subgraph of vertices on which a discrete eigen wavefunction $\psi\in L^2(\calG;\mathbb{K})$ assumes the same numerical sign, as opposed to the term nodal cell (or domain) in the case of a continuous configuration space. And, as far as a conventional boltzmannonic Schr\"odinger operator is concerned, possibly as the base boltzmannonic Hamiltonian of a fermionic Schr\"odinger operator, there have been theorems proved regarding the number of sign graphs in relation to the index of an energy eigenstate \cite{Davies01}, in close analogy to the famous theorem of Courant \cite{Courant89}  for a Schr\"odinger operator on a continuous configuration space.

On a bounded discrete configuration space $\calG$, any DFF or GFF Hamiltonian $H=\sum_{k=1}^{\mathsmaller{K}}H_k$ is discretized into a necessarily finite and {\em tensor-sparse} matrix, with tensor-sparsity referring to the property of $H$ being the sum of a $K=O(\poly(\size(\calG)))$ number of FBM tensor monomials $H_k$, $k\in[1,K]$, each of which being a tensor product of an $O(\poly(\size(\calG)))$ number of FBM interactions, with each FBM interaction coupling any lattice point in $\calG$ to at most an $O(\log(\size(\calG)))$ number of neighbor lattice points. By contrast, the dimensions of $H$, being a $|\calG|\times|\calG|$ matrix, can be much larger, as $|\calG|$ can be an exponential function of $\size(\calG)$. Consider discrete Gibbs operators $G_1(\beta M)\defeq I-\beta\beta_0\,\Base(M)$, with $M\in\{H,H_k,k\in[1,K]\}$, $\beta_0\defeq\left(\max\{\|\Base(H)\|,\|\Base(H_k)\|:k\in[1,K]\}\right)^{\mathsmaller{-}1}$, $\beta\in[0,1)$, and $G_*(\tau M)\defeq\exp(\mathsmaller{-}\tau M)$, $\tau\ge 0$, $M\in\{H,H_k,k\in[1,K]\}$. Let $\Gamma(H)\defeq\Gamma(G_1(H))$ be the graph associated with the Hamiltonian $H$, and $\Gamma^+(H)$ be the subgraph of $\Gamma(H)$ with all ``ground state sign changing'' edges removed, namely, $\Gamma^+(H)$ inherits all vertices of $\Gamma(H)$ but keeps only those edges of $\Gamma(H)$ whose two end vertices see the ground state $\psi_0(H)$ assuming the same sign. For each lattice point $q\in\calG\setminus\{r\in\calG:\psi_0(H;r)=0\}$, let $\Gamma(H;q)$ be the set of points that are connected from $q$ by an edge in $\Gamma(H)$, that is, $\Gamma(H;q)\defeq\{r\in\calG:\langle r|G_1(H)|q\rangle\neq 0\}$, or equivalently, $\Gamma(H;q)=\{r\in\calG:\exists k\in[1,K],\langle r|G_1(H_k)|q\rangle\neq 0\}$. Further, let $\Gamma^+(H;q)$ be the subset of $\Gamma(H;q)$ containing only the points on which the ground state of $H$ assumes the same sign as on $q$, namely, $\Gamma^+(H;q)\defeq\{r\in\calG:\langle r|G_1(H)|q\rangle\neq 0,\psi_0(H;r)\psi_0(H;q)>0\}$. Finally, let $\Gamma^*(H;q)$ be the subset of vertices containing $q$ and all that can be reached from $q$ via a connected path in $\Gamma^+(H)$. Evidently, $\Gamma^*(H;q)$ is exactly the discrete equivalent of the nodal cell $\calN(H;q)$.

Recall that $\calU(H)$ represents the set of node-uncertain configuration points for any DFF or GFF Hamiltonian $H$. Only it is noted that when $H$ is a finite matrix supported by a finite graph $\calG$, $\calU(H)$ is a subset of the vertices of $\calG$. On any subgraph $\calU^c(H)\cap\Gamma(H;q_0)$, with $H$ being DFF or GFF, $\calU^c(H)\defeq\calG\setminus\calU(H)$, $q_0\in\calU^c(H)$, $\psi_0(H;q_0)\neq 0$, a classical random walk can be defined in association with $H$, as represented by a transition probability matrix $Q_{\!\mathsmaller{H}}\defeq\{\Pr(r|q):q,r\in\calU^c(H)\cap\Gamma(H;q_0)\}=[\psi_0^*(H)]G_{\mathsmaller{\Box}}(H)[\psi_0^*(H)]^{\mathsmaller{-}1}$, with $[\psi_0^*(H)]=\diag(\{\psi_0^*(H;q)\!:\!q\in\calU^c(H)\cap\Gamma(H;q_0)\})$ being a diagonal matrix, $G_{\mathsmaller{\Box}}(H)\defeq G_1(H)$ when $H$ is DFF-FS or GFF-FS, while $G_{\mathsmaller{\Box}}(H)\defeq K^{\mathbb{-}1}\sum_{k=1}^{\mathsmaller{K}}G_*(\mathsmaller{+}\infty H_k)$ when $H$ is SFF-EB, where $G_*(\mathsmaller{+}\infty H_k)\defeq\lim_{\tau\mathsmaller{\rightarrow+\infty}}G_*(\tau H_k)$, $\forall k\in[1,K]$, such that $\forall q,r\in\calU^c(H)$, the transition probability
\begin{equation}
\Pr(r|q)=\left\{\begin{array}{l}\vspace{1.0ex}
\psi_0(H;r)\left\langle r|G_1(H)|q\right\rangle\psi_0(H;q)^{\mathsmaller{-}1} \times \Iver{q\in\supp(\psi_0(H))},\; \mbox{if}\;H\;\mbox{is FS}, \\
K^{\mathsmaller{-}1}\sum_{k=1}^{\mathsmaller{K}}\sum_{l=1}^{\gmul(H_k)}\left|\psi_{0,l}(H_k;r)\right|^2 \times \Iver{q\in\supp(\psi_{0,l}(H_k))},\; \mbox{if}\;H\;\mbox{is EB},
\end{array}\right. \nonumber
\end{equation}
in which the Iverson brackets $\Iver{q\in\supp(\psi_0(H))}$ and $\Iver{q\in\supp(\psi_{0,l}(H_k))}$ are $0$ or $1$ binary-values to exclude starting points outside of a ground state. In case $H$ is both DFF-FS and DFF-EB, or both GFF-FS and GFF-EB, either $G_{\mathsmaller{\Box}}(H)=G_1(H)$ or $G_{\mathsmaller{\Box}}(H)=K^{\mathbb{-}1}\sum_{k=1}^{\mathsmaller{K}}G_*(\mathsmaller{+}\infty H_k)$ can be used in conjunction with the corresponding transition probability matrix. Starting from $q_0\in\calU^c(H)$, $\psi_0(H;q_0)\neq 0$ and following the transition probability matrix $Q_{\!\mathsmaller{H}}$, a walker traverses a subset of configuration points
$$\Range(Q_{\!\mathsmaller{H}};q_0)\defeq\!\left\{r\in\calU^c(H)\cap\Gamma(H;q_0):\exists\,n\in\mathbb{N}\;\mbox{such that}\;\langle r|Q_{\!\mathsmaller{H}}^n|q_0\rangle\neq 0\right\},$$
that is called the range of $Q_{\!\mathsmaller{H}}$ from $q_0$. It follows from \iftoggle{ForUSPTO} {Auxiliary Utility} {Lemma} \ref{QuasiStochasticOper} and equation (\ref{r_PinftyHk_q}) that the matrix $Q_{\!\mathsmaller{H}}$ is bona fide stochastic and defines an irreducible Markov chain on $\Range(Q_{\!\mathsmaller{H}};q_0)\subseteq\calG$, which shares the same spectrum with the Gibbs operator $G_{\mathsmaller{\Box}}(H)$, with the stationary distribution (stationary probability vector) of $Q_{\!\mathsmaller{H}}$ being the stationary distribution of $G_{\mathsmaller{\Box}}(H)$ point-wise squared, much similar to that in the well-known quantum-stochastic mapping for boltzmannonic stoquastic Hamiltonians \cite{Bravyi10,Aharonov07a,Aharonov07b}. The Markov chain defined by the primitive ({\it i.e.}, aperiodic irreducible non-negative) matrix $Q_{\!\mathsmaller{H}}$ on $\Range(Q_{\!\mathsmaller{H}};q_0)$ is reversible and bound to converge to the unique stationary distribution $|\psi_0(H;q\in\Range(Q_{\!\mathsmaller{H}};q_0))|^2$.

Here goes one Monte Carlo algorithm for sampling the ground state of a discretized DFF-FS or GFF-FS Hamiltonian $H=\sum_{k=1}^{\mathsmaller{K}}H_k$, $K\in\mathbb{N}$ defined on a discrete, bounded configuration space $\calG$, with each $H_k$, $k\in[1,K]$ being either DFF or GFF, which solves the computational problem of simulating a many-body Hamiltonian as specified in Definition \ref{defiCompProbSimuManyBody} using a fixed-node diffusion type of approach, when the system is known to be DFF-FS or GFF-FS.

\renewcommand{\labelenumi}{\thealgorithm.\arabic{enumi}}
\begin{algorithm}{(Monte Carlo  Simulation of a DFF-FS or GFF-FS Hamiltonian)} \label{MCSimuDisc}\\
\vspace{-5.0ex}
\begin{enumerate}[start=0]
\item Choose a predetermined bound $\delta>0$ for numerical accuracy; Choose a warm start $q=q_0$, that is a random point $q_0\in\calU^c(H)$ with the promise of $|\log\psi_0(H;q_0)|$ being polynomially bounded; Set an iteration counter $N\in\mathbb{Z}$ to $0$; Choose a predetermined upper bound $N_{\max}\in\mathbb{N}$ for the iteration counter; Choose a predetermined number of samples $S_{\max}\in\mathbb{N}$; Initialize an index $k\in[1,K]$ of the FBM tensor monomials to 1; \label{AlgDiscStepZero}
\item Let $q=u_{q,k}\oplus u_{q,k}^{\mathsmaller{\perp}}$, $u_{q,k}\in\calE_k$, $u_{q,k}^{\mathsmaller{\perp}}\in\calE_k^{\mathsmaller{\perp}}$ be the unique $H_k$-induced ODSD of $q$, compute the ground states $\psi_0(H_k;\calE_k\oplus u_{q,k}^{\mathsmaller{\perp}})$, with the $L^2$-normed error guaranteed not to exceed $\delta$; \label{AlgDiscSolvePsi0Hk}
\item Compute the set $\Gamma^+(H_k;q)$ by enumerating all non-vanishing entries in the $u_{q,k}$-th column of $H_k$ and checking against $\psi_0(H_k;\calE_k\oplus u_{q,k}^{\mathsmaller{\perp}})$ for wavefunction sign changes; \label{AlgDiscGetGammaQPlus}
\item Compute the Gibbs transition amplitude $\langle r|G_1(H_k)|q\rangle$, $\forall r\in\Gamma^+(H_k;q)$, with the $L^2$-normed error guaranteed not to exceed $\delta$; \label{AlgDiscComputeG1Hk}
\item Advance the walker for a single step from $q$ to a random point $r\in\calU^c(H)\cap\Gamma^+(H_k;q)$ using importance sampling according to the transition probability
$$\Pr(r|q)={\textstyle{\scalebox{1.25}{$\sum$}_{l=1}^{\gmul(H_k)}}}\psi_{0,l}(H_k;r)\,\langle r|G_1(H_k)|q\rangle\,\psi_{0,l}(H_k;q)^{\mathsmaller{-}1} \times \Iver{q\in\supp(\psi_{0,l}(H_k))};$$ \label{AlgDiscQtoRJump}
\vspace{-5.0ex}
\item Assign the value of $r$ to $q$; Increase the index $k$ by $1$; \label{AlgDiscUpdateQ}
\item Repeat steps \ref{MCSimuDisc}.\ref{AlgDiscSolvePsi0Hk} through \ref{MCSimuDisc}.\ref{AlgDiscUpdateQ} until $k>K$, upon which time reset $k$ to $1$ and increase $N$ by $1$; \label{AlgDiscHkLoop}
\item Repeat the loop \ref{MCSimuDisc}.\ref{AlgDiscHkLoop} until $N=N_{\max}$, transfer the well-mixed walker location $q$ to a result container, and reset the value of $N$ to $0$; \label{AlgDiscMixing}
\item Repeat the loop \ref{MCSimuDisc}.\ref{AlgDiscMixing} $S_{\max}$ times to get $S_{\max}$ samples of walker locations. \label{AlgDiscRepeatPolyTimes}
\end{enumerate}
\end{algorithm}

\vspace{-1.0ex}
Another class of discrete-configuration-space-based algorithms randomize the applications of Gibbs operators $G_{\mathsmaller{\Box}}(H_k)\defeq G_1(H_k)$ or $G_{\mathsmaller{\Box}}(H_k)\defeq G_*(\mathsmaller{+}\infty H_k)$, $k\in[1,K]$ in association with the FBM tensor monomials of a Hamiltonian $H=\sum_{k=1}^{\mathsmaller{K}}H_k$ according to $H$ being DFF or GFF and FS or EB, while still adopting a fixed-node diffusion strategy.

\renewcommand{\labelenumi}{\thealgorithm.\arabic{enumi}}
\begin{algorithm}{(Monte Carlo  Simulation of a DFF or GFF and FS or EB Hamiltonian)} \label{MCSimuDiscV1}\\
\vspace{-5.0ex}
\begin{enumerate}[start=0]
\item Choose a predetermined bound $\delta>0$ for numerical accuracy; Choose a warm start $q=q_0$, that is a random point $q_0\in\calU^c(H)$ with the promise of $|\log\psi_0(H;q_0)|$ being polynomially bounded; Set an iteration counter $N\in\mathbb{Z}$ to $0$; Choose a predetermined upper bound $N_{\max}\in\mathbb{N}$ for the iteration counter; Choose a predetermined number of samples $S_{\max}\in\mathbb{N}$; \label{AlgDiscV1StepZero}
\item Pick a $k\in[1,K]$ randomly, let $q=u_{q,k}\oplus u_{q,k}^{\mathsmaller{\perp}}$, $u_{q,k}\in\calE_k$, $u_{q,k}^{\mathsmaller{\perp}}\in\calE_k^{\mathsmaller{\perp}}$ be the unique $H_k$-induced ODSD of $q$; If $\calE_k\neq\emptyset$, go on to the next step; Else, repeat this step; \label{AlgDiscV1PickK}
\item Compute the wavefunction $\psi_0(H_k;\calE_k\oplus u_{q,k}^{\mathsmaller{\perp}})$; If $H$ is DFF-EB or GFF-EB, go on to the next step; Else if $H$ is DFF-FS or GFF-FS, determine the nodal cell $\Gamma^+(H_k;q)$; \label{AlgDiscV1GetGammaQPlus}
\item If $H$ is DFF-EB or GFF-EB, go on to the next step; Else if $H$ is DFF-FS or GFF-FS, compute the Gibbs transition amplitude $\langle r|G_1(H_k)|q\rangle$, $r\in\Gamma^+(H_k;q)$, with the $L^2$-normed error not to exceed $\delta$; \label{AlgDiscV1ComputeG1Hk}
\item Advance the walker for a single step from $q$ to a random point $r\in\calU^c(H)$ using importance sampling, according to the transition probability
\begin{align}
\Pr(r|q) =\;& {\textstyle{\scalebox{1.25}{$\sum$}_{l=1}^{\gmul(H_k)}}}\psi_{0,l}(H_k;r)\,\langle r|G_1(H_k)|q\rangle\,\psi_{0,l}(H_k;q)^{\mathsmaller{-}1} \nonumber \\[0.75ex]
\times\;& \Iver{r\in\Gamma^+(H_k;q)} \,\times\, \Iver{q\in\supp(\psi_{0,l}(H_k))} \nonumber
\end{align}
for the case of $H$ being DFF-FS or GFF-FS, or
$$\Pr(r|q)={\textstyle{\scalebox{1.25}{$\sum$}_{l=1}^{\gmul(H_k)}}}\left|\psi_{0,l}(H_k;r)\right|^2 \times \Iver{q\in\supp(\psi_{0,l}(H_k))}$$
when $H$ is DFF-EB or GFF-EB; \label{AlgDiscV1QtoRJump}
\item Assign the value of $r$ to $q$; Increase the iteration counter $N$ by $1$; \label{AlgDiscV1UpdateQ}
\item Repeat steps \ref{MCSimuDiscV1}.\ref{AlgDiscV1PickK} through \ref{MCSimuDiscV1}.\ref{AlgDiscV1UpdateQ} until $N=N_{\max}$, transfer the well-mixed walker location $q$ to a result container, and reset the value of $N$ to $0$; \label{AlgDiscV1Mixing}
\item Repeat the loop \ref{MCSimuDiscV1}.\ref{AlgDiscV1Mixing} $S_{\max}$ times to get $S_{\max}$ samples of walker locations. \label{AlgDiscV1RepeatPolyTimes}
\end{enumerate}
\end{algorithm}

\vspace{-1.0ex}
A variation of Algorithm \ref{MCSimuDiscV1} is applicable when the DFF or GFF and FS or EB Hamiltonian $H=\sum_{k=1}^{\mathsmaller{K}}H_k$ is such that the FBM tensor monomials $H_k$, $k\in[1,K]$, up to operator isomorphisms, all belong to a prescribed collection of a polynomially bounded number of FBM Hamiltonians, called {\em template Hamiltonians}, $\forall q=u_{q,k}\oplus u_{q,k}^{\mathsmaller{\perp}}\in\calG$, $u_{q,k}\in\calE_k$, $u_{q,k}^{\mathsmaller{\perp}}\in\calE_k^{\mathsmaller{\perp}}$. The variation also illustrates, by way of example but no means of limitation, Monte Carlo samplings without a fixed-node diffusion strategy.

\renewcommand{\labelenumi}{\thealgorithm.\arabic{enumi}}
\begin{algorithm}{(Monte Carlo  Simulation of a DFF or GFF and FS or EB Hamiltonian)} \label{MCSimuDiscV2}\\
\vspace{-5.0ex}
\begin{enumerate}[start=0]
\item Choose a predetermined bound $\delta>0$ for numerical accuracy; Choose a warm start $q=q_0$, that is a random point $q_0\in\calU^c(H)$ with the promise of $|\log\psi_0(H;q_0)|$ being polynomially bounded; Label $q=q_0$ by $\lab(q_0)=+1$; Set an iteration counter $N\in\mathbb{Z}$ to $0$; Choose a predetermined upper bound $N_{\max}\in\mathbb{N}$ for the iteration counter; Choose a predetermined number of samples $S_{\max}\in\mathbb{N}$;

For each template Hamiltonian, precompute and store the ground state wavefunctions, the nodal cells, additionally the Gibbs transition amplitudes for the case of H being DFF-FS or GFF-FS, all with $L^2$-normed errors guaranteed not to exceed $\delta$;
\label{AlgDiscV2StepZero}
\item Pick a $k\in[1,K]$ randomly, let $q=u_{q,k}\oplus u_{q,k}^{\mathsmaller{\perp}}$, $u_{q,k}\in\calE_k$, $u_{q,k}^{\mathsmaller{\perp}}\in\calE_k^{\mathsmaller{\perp}}$ be the unique $H_k$-induced ODSD of $q$; If $\calE_k\neq\emptyset$, go on to the next step; Else, repeat this step; \label{AlgDiscV2PickK}
\item Look up the wavefunction $\psi_0(H_k;\calE_k\oplus u_{q,k}^{\mathsmaller{\perp}})$ from the precomputed and stored solutions; \label{AlgDiscV2GetGammaQPlus}
\item If $H$ is DFF-EB or GFF-EB, go on to the next step; Else if $H$ is DFF-FS or GFF-FS, look up the Gibbs transition amplitude $\langle r|G_1(H_k)|q\rangle$, $r\in\Gamma(H_k;q)$ from the precomputed and stored solutions; \label{AlgDiscV2ComputeG1Hk}
\item Advance the walker for a single step from $q$ to a random point $r\in\calU^c(H)$ using importance sampling, according to the transition probability
\begin{align}
\Pr(r|q) =\;& {\textstyle{\scalebox{1.25}{$\sum$}_{l=1}^{\gmul(H_k)}}}\psi_{0,l}(H_k;r)\,\langle r|G_1(H_k)|q\rangle\,\psi_{0,l}(H_k;q)^{\mathsmaller{-}1} \nonumber \\[0.75ex]
\times\;& \Iver{r\in\Gamma(H_k;q)} \,\times\, \Iver{q\in\supp(\psi_{0,l}(H_k))} \nonumber
\end{align}
for the case of $H$ being DFF-FS or GFF-FS, or
$$\Pr(r|q)={\textstyle{\scalebox{1.25}{$\sum$}_{l=1}^{\gmul(H_k)}}}\left|\psi_{0,l}(H_k;r)\right|^2 \times \Iver{q\in\supp(\psi_{0,l}(H_k))}$$
when $H$ is DFF-EB or GFF-EB; If $\exists\,l\in[1,\gmul(H_k)]$ such that $\psi_{0,l}(H_k;r)/\psi_{0,l}(H_k;q) < 0$, then label $r$ by $\lab(r) = -\lab(q)$; Otherwise, label $r$ by $\lab(r) = \lab(q)$; \label{AlgDiscV2QtoRJump}
\item Assign the value and label of $r$ to $q$; Increase the iteration counter $N$ by $1$; \label{AlgDiscV2UpdateQ}
\item Repeat steps \ref{MCSimuDiscV2}.\ref{AlgDiscV2PickK} through \ref{MCSimuDiscV2}.\ref{AlgDiscV2UpdateQ} until $N=N_{\max}$, record the well-mixed walker location $q$, and reset the value of $N$ to $0$; \label{AlgDiscV2Mixing}
\item Repeat the loop \ref{MCSimuDiscV2}.\ref{AlgDiscV2Mixing} $S_{\max}$ times to get $S_{\max}$ samples of walker locations. \label{AlgDiscV2RepeatPolyTimes}
\end{enumerate}
\end{algorithm}

\vspace{-1.0ex}
On a continuous configuration space $\calM$, a general state space Markov chain \cite{Orey71,Roberts04,Petritis15} in association with a DFF or GFF Hamiltonian $H=\sum_{k=1}^{\mathsmaller{K}}H_k$ can be defined for a discrete time random walk with a quasi-stochastic operator $Q_{\!\mathsmaller{H}}\defeq\{\Pr(r|q)\}_{r,q\,\in\,\calU^c(H)}=[\psi_0^*(H)]G_{\mathsmaller{\Box}}(H)[\psi_0^*(H)]^{\mathsmaller{-}1}$ constructed according to \iftoggle{ForUSPTO} {Auxiliary Utility} {Lemma} \ref{QuasiStochasticOper}, where $\calU^c(H)\defeq\calM\setminus\calU(H)$, $[\psi_0^*(H)]=\diag(\{\psi_0^*(H;q)\!:\!q\in\calU^c(H)\})$ is $\calM$-diagonal, $G_{\mathsmaller{\Box}}(H)\defeq G_*(\epsilon H)$ with a scaling factor $\epsilon=\Omega(1/\poly(\size(H)))$, $\epsilon>0$ or $G_{\mathsmaller{\Box}}(H)\defeq K^{\mathsmaller{-}1}\sum_{k=1}^{\mathsmaller{K}}G_*(\mathsmaller{+}\infty H_k)$ respectively depending upon $H$ being DFF or GFF and FS or EB, with $G_*(\mathsmaller{+}\infty H_k)\defeq\lim_{\tau\mathsmaller{\rightarrow+}\infty}G_*(\tau H_k)$, $k\in[1,K]$, such that $\forall q,r\in\calU^c(H)$, the transition probability
\begin{equation}
\Pr(r|q)=\left\{\begin{array}{l}\vspace{1.0ex}
\psi_0(H;r)\left\langle r|G_*(\epsilon H)|q\right\rangle\psi_0(H;q)^{\mathsmaller{-}1} \times \Iver{q\in\supp(\psi_0(H))},\; \mbox{if}\;H\;\mbox{is FS}, \\
K^{\mathsmaller{-}1}\sum_{k=1}^{\mathsmaller{K}}\sum_{l=1}^{\gmul(H_k)}\left|\psi_{0,l}(H_k;r)\right|^2 \times \Iver{q\in\supp(\psi_{0,l}(H_k))},\; \mbox{if}\;H\;\mbox{is EB}.
\end{array}\right. \nonumber
\end{equation}
Either $G_{\mathsmaller{\Box}}(H)=G_*(\epsilon H)$ or $G_{\mathsmaller{\Box}}(H)=K^{\mathbb{-}1}\sum_{k=1}^{\mathsmaller{K}}G_*(\mathsmaller{+}\infty H_k)$ may be used, in conjunction with the corresponding quasi-stochastic operator, when $H$ is both FS and EB. Starting from $q_0\in\calU^c(H)$ such that $\psi_0(H;q_0)\neq 0$, and following the quasi-stochastic operator $Q_{\!\mathsmaller{H}}$, a walker traverses a subset of configuration points
$$\Range(Q_{\!\mathsmaller{H}};q_0)\defeq\!\left\{r\in\calU^c(H):\exists\,n\in\mathbb{N}\;\mbox{such that}\;\langle r|Q_{\!\mathsmaller{H}}^n|q_0\rangle\neq 0\right\},$$
that is called the range of $Q_{\!\mathsmaller{H}}$ from $q_0$. It is clear that the operator $Q_{\!\mathsmaller{H}}$ is primitive, {\it i.e.}, aperiodic irreducible non-negative, and shares the same eigenvalues with $G_{\mathsmaller{\Box}}(H)$, further has the corresponding eigenvectors $\{\psi_0(H;q)\psi_n(H;q)\}_{n\ge 0}$ simply related to those eigenvectors $\{\psi_n(H;q)\}_{n\ge 0}$ of $G_{\mathsmaller{\Box}}(H)$ \cite{Bravyi10,Aharonov07a,Aharonov07b}. In particular, the irreducible Markov chain is also reversible with respect to the ground state probability distribution, thus has $|\psi_0(H;\cdot)|^2$ as its unique stationary distribution.

The following is one Monte Carlo algorithm for a DFF-FS or GFF-FS Hamiltonian $H=\sum_{k=1}^{\mathsmaller{K}}H_k$ on a continuous configuration space $\calM$, using a fixed-node diffusion type of approach.

\renewcommand{\labelenumi}{\thealgorithm.\arabic{enumi}}
\begin{algorithm}{(Monte Carlo  Simulation of a DFF-FS or GFF-FS Hamiltonian)} \label{MCSimuCont}\\
\vspace{-5.0ex}
\begin{enumerate}[start=0]
\item Choose a predetermined bound $\delta>0$ for numerical accuracy; Choose a warm start $q=q_0$, that is a random point $q_0\in\calU^c(H)$ with the promise of $|\log\psi_0(H;q_0)|$ being polynomially bounded; Set an iteration counter $N\in\mathbb{Z}$ to $0$; Choose a predetermined upper bound $N_{\max}\in\mathbb{N}$ for the iteration counter; Choose a predetermined number of samples $S_{\max}\in\mathbb{N}$; Choose an $m\in\mathbb{N}$ such that the approximation error $c_1m^{\mathsmaller{-}c_2}$ in (\ref{LieTrotterKatoBounds}) does not exceed $\delta$; Initialize an iteration counter $m'$ to $0$; Initialize an index $k\in[1,K]$ of the FBM tensor monomials to 1; \label{AlgContStepZero}
\item Let $q=u_{q,k}\oplus u_{q,k}^{\mathsmaller{\perp}}$, $u_{q,k}\in\calE_k$, $u_{q,k}^{\mathsmaller{\perp}}\in\calE_k^{\mathsmaller{\perp}}$ be the unique $H_k$-induced ODSD of $q$, compute the wavefunction $\psi_0(H_k;\calE_k\oplus u_{q,k}^{\mathsmaller{\perp}})$, determine the nodal cell $\calN(H_k;q)$; \label{AlgContSolveHkOverN}
\item Compute the Gibbs transition amplitude $\langle r|G_*(H_k/mK^b)|q\rangle$, restricted to the connected open set $\{r\in\calN(H_k;q)\}$ and subject to the Dirichlet boundary condition on $\{r\in\partial\calN(H_k;q)\}$, with the $L^2$-normed error guaranteed not to exceed $\delta/mK^b$, $b\in\mathbb{N}$ as specified in (\ref{LieTrotterKatoBounds}); \label{AlgContSolveGstarHkOverN}
\item Advance the walker for a single step from $q$ to a random point $r\in\calU^c(H)\cap\calN(H_k;q)$ via importance sampling according to the transition probability
$$\Pr(r|q)={\textstyle{\scalebox{1.25}{$\sum$}_{l=1}^{\gmul(H_k)}}}\psi_{0,l}(H_k;r)\,\langle r|G_*(H_k/mK^b)|q\rangle\,\psi_{0,l}(H_k;q)^{\mathsmaller{-}1}\times \Iver{q\in\supp(\psi_{0,l}(H_k))};$$ \label{AlgContQtoRJump} \vspace{-4.5ex}
\item Assign the value of $r$ to $q$; Increase the index $k$ by $1$; \label{AlgContUpdateQ}
\item Repeat steps \ref{MCSimuCont}.\ref{AlgContSolveHkOverN} through \ref{MCSimuCont}.\ref{AlgContUpdateQ} until $k>K$, upon which time reset $k$ to $1$ and increase $m'$ by $1$; \label{AlgContHkLoop}
\item Repeat the loop \ref{MCSimuCont}.\ref{AlgContHkLoop} until $m'=mK^b$, upon which time reset $m'$ to $0$ and increase $N$ by $1$; \label{AlgContLieTrotterKatoLoop}
\item Repeat the loop \ref{MCSimuCont}.\ref{AlgContLieTrotterKatoLoop} until $N=N_{max}$, record the well-mixed walker location $q$, and reset the value of $N$ to $0$; \label{AlgContMixing}
\item Repeat the loop \ref{MCSimuCont}.\ref{AlgContMixing} $S_{\max}$ times to get $S_{\max}$ samples of walker locations. \label{AlgContRepeatPolyTimes}
\end{enumerate}
\end{algorithm}

\vspace{-1.0ex}
Another class of continuous-configuration-space-based algorithms randomize the applications of Gibbs operators $G_{\mathsmaller{\Box}}(H_k)\defeq G_*(\epsilon H_k)$ with a polynomially small scaling factor $\epsilon>0$ or respectively $G_{\mathsmaller{\Box}}(H_k)\defeq G_*(\mathsmaller{+}\infty H_k)$, $k\in[1,K]$ in association with the FBM tensor monomials of a DFF or GFF and FS or EB Hamiltonian $H=\sum_{k=1}^{\mathsmaller{K}}H_k$, and apply the so-called Gibbs sampling method to form a random walk that follows a series of conditional distributions \cite{Geman84,Liu01,Li09,Bremaud17} corresponding to the individual FBM tensor monomials of $H$, while still adopting a fixed-node diffusion strategy.

\renewcommand{\labelenumi}{\thealgorithm.\arabic{enumi}}
\begin{algorithm}{(Monte Carlo  Simulation of a DFF or GFF and FS or EB Hamiltonian)} \label{MCSimuContV1}\\
\vspace{-5.0ex}
\begin{enumerate}[start=0]
\item Choose a predetermined bound $\delta>0$ for numerical accuracy; Choose a warm start $q=q_0$, that is a random point $q_0\in\calU^c(H)$ with the promise of $|\log\psi_0(H;q_0)|$ being polynomially bounded; Set an iteration counter $N\in\mathbb{Z}$ to $0$; Choose a predetermined upper bound $N_{\max}\in\mathbb{N}$ for the iteration counter; Choose a predetermined number of samples $S_{\max}\in\mathbb{N}$; If $H$ is DFF-FS or GFF-FS, choose a predetermined scaling factor $\epsilon>0$; \label{AlgContV1StepZero}
\item Pick a $k\in[1,K]$ randomly, let $q=u_{q,k}\oplus u_{q,k}^{\mathsmaller{\perp}}$, $u_{q,k}\in\calE_k$, $u_{q,k}^{\mathsmaller{\perp}}\in\calE_k^{\mathsmaller{\perp}}$ be the unique $H_k$-induced ODSD of $q$; If $\calE_k\neq\emptyset$, go on to the next step; Else, repeat this step; \label{AlgContV1PickK}
\item Compute the wavefunction $\psi_0(H_k;\calE_k\oplus u_{q,k}^{\mathsmaller{\perp}})$; If $H$ is DFF-EB or GFF-EB, go on to the next step; Else if $H$ is DFF-FS or GFF-FS, determine the nodal cell $\calN(H_k;q)$; \label{AlgContV1EvalPsi0Hk}
\item If $H$ is DFF-EB or GFF-EB, go on to the next step; Else if $H$ is DFF-FS or GFF-FS, compute the Gibbs transition amplitude $\langle r|G_*(\epsilon H_k)|q\rangle$, restricted to the connected open set $\{r\in\calN(H_k;q)\}$ and subject to the Dirichlet boundary condition on $\{r\in\partial\calN(H_k;q)\}$, with the $L^2$-normed error guaranteed not to exceed $\delta$; \label{AlgContV1EvalGtarEpsHk}
\item Advance the walker for a single step from $q$ to a random point $r\in\calU^c(H)$ using importance sampling, according to the transition probability
\begin{align}
\Pr(r|q) =\;& {\textstyle{\scalebox{1.25}{$\sum$}_{l=1}^{\gmul(H_k)}}}\psi_{0,l}(H_k;r)\,\langle r|G_*(\epsilon H_k)|q\rangle\,\psi_{0,l}(H_k;q)^{\mathsmaller{-}1} \nonumber \\[0.75ex]
\times\;& \Iver{r\in\calN(H_k;q)} \,\times\, \Iver{q\in\supp(\psi_{0,l}(H_k))} \nonumber
\end{align}
for the case of $H$ being DFF-FS or GFF-FS, or
$$\Pr(r|q)={\textstyle{\scalebox{1.25}{$\sum$}_{l=1}^{\gmul(H_k)}}}\left|\psi_{0,l}(H_k;r)\right|^2\times \Iver{q\in\supp(\psi_{0,l}(H_k))}$$
when $H$ is DFF-EB or GFF-EB; \label{AlgContV1QtoRJump}
\item Assign the value of $r$ to $q$; Increase the iteration counter $N$ by $1$; \label{AlgContV1UpdateQ}
\item Repeat steps \ref{MCSimuContV1}.\ref{AlgContV1PickK} through \ref{MCSimuContV1}.\ref{AlgContV1UpdateQ} until $N=N_{\max}$, record the well-mixed walker location $q$, and reset the value of $N$ to $0$; \label{AlgContV1Mixing}
\item Repeat the loop \ref{MCSimuContV1}.\ref{AlgContV1Mixing} $S_{\max}$ times to get $S_{\max}$ samples of walker locations. \label{AlgContV1RepeatPolyTimes}
\end{enumerate}
\end{algorithm}

\vspace{-1.0ex}
A variation of Algorithm \ref{MCSimuContV1} is useful when the concerned DFF or GFF and FS or EB Hamiltonian is a sum of FBM tensor monomials $\{H_k\}_{k=1}^{\mathsmaller{K}}$, all of which belong to a prescribed collection of a polynomially bounded number of template Hamiltonians. The variation also illustrates, by way of example but no means of limitation, Monte Carlo samplings without a fixed-node diffusion strategy.

\renewcommand{\labelenumi}{\thealgorithm.\arabic{enumi}}
\begin{algorithm}{(Monte Carlo  Simulation of a DFF or GFF and FS or EB Hamiltonian)} \label{MCSimuContV2}\\
\vspace{-5.0ex}
\begin{enumerate}[start=0]
\item Choose a predetermined bound $\delta>0$ for numerical accuracy; Choose a warm start $q=q_0$, that is a random point $q_0\in\calU^c(H)$ with the promise of $|\log\psi_0(H;q_0)|$ being polynomially bounded; Label $q=q_0$ by $\lab(q_0)=+1$; Set an iteration counter $N\in\mathbb{Z}$ to $0$; Choose a predetermined upper bound $N_{\max}\in\mathbb{N}$ for the iteration counter; Choose a predetermined number of samples $S_{\max}\in\mathbb{N}$; If $H$ is DFF-FS or GFF-FS, choose a predetermined scaling factor $\epsilon>0$;

For each template Hamiltonian, precompute and store the ground state wavefunctions, the nodal cells, additionally the Gibbs transition amplitudes for the case of H being DFF-FS or GFF-FS, all with $L^2$-normed errors guaranteed not to exceed $\delta$; \label{AlgContV2StepZero}
\item Pick a $k\in[1,K]$ randomly, let $q=u_{q,k}\oplus u_{q,k}^{\mathsmaller{\perp}}$, $u_{q,k}\in\calE_k$, $u_{q,k}^{\mathsmaller{\perp}}\in\calE_k^{\mathsmaller{\perp}}$ be the unique $H_k$-induced ODSD of $q$; If $\calE_k\neq\emptyset$, go on to the next step; Else, repeat this step; \label{AlgContV2PickK}
\item Look up the wavefunction $\psi_0(H_k;\calE_k\oplus u_{q,k}^{\mathsmaller{\perp}})$ from the precomputed and stored solutions; \label{AlgContV2EvalPsi0Hk}
\item If $H$ is DFF-EB or GFF-EB, go on to the next step; Else if $H$ is DFF-FS or GFF-FS, look up the Gibbs transition amplitude $\langle r|G_*(\epsilon H_k)|q\rangle$ from the precomputed and stored solutions; \label{AlgContV2EvalGtarEpsHk}
\item Advance the walker for a single step from $q$ to a random point $r\in\calU^c(H)$ using importance sampling, according to the transition probability
\begin{align}
\Pr(r|q) =\;& {\textstyle{\scalebox{1.25}{$\sum$}_{l=1}^{\gmul(H_k)}}}\psi_{0,l}(H_k;r)\,\langle r|G_*(\epsilon H_k)|q\rangle\,\psi_{0,l}(H_k;q)^{\mathsmaller{-}1} \nonumber \\[0.75ex]
\times\;& \Iver{r\in\calU^c(H_k)} \,\times\, \Iver{q\in\supp(\psi_{0,l}(H_k))} \nonumber
\end{align}
for the case of $H$ being DFF-FS or GFF-FS, or
$$\Pr(r|q)={\textstyle{\scalebox{1.25}{$\sum$}_{l=1}^{\gmul(H_k)}}}\left|\psi_{0,l}(H_k;r)\right|^2 \times \Iver{q\in\supp(\psi_{0,l}(H_k))}$$
when $H$ is DFF-EB or GFF-EB; If $\exists\,l\in[1,\gmul(H_k)]$ such that $\psi_{0,l}(H_k;r)/\psi_{0,l}(H_k;q) < 0$, then label $r$ by $\lab(r) = -\lab(q)$; Otherwise, label $r$ by $\lab(r) = \lab(q)$; \label{AlgContV2QtoRJump}
\item Assign the value and label of $r$ to $q$; Increase the iteration counter $N$ by $1$; \label{AlgContV2UpdateQ}
\item Repeat steps \ref{MCSimuContV2}.\ref{AlgContV2PickK} through \ref{MCSimuContV2}.\ref{AlgContV2UpdateQ} until $N=N_{\max}$, record the well-mixed walker location $q$, and reset the value of $N$ to $0$; \label{AlgContV2Mixing}
\item Repeat the loop \ref{MCSimuContV2}.\ref{AlgContV2Mixing} $S_{\max}$ times to get $S_{\max}$ samples of walker locations. \label{AlgContV2RepeatPolyTimes}
\end{enumerate}
\end{algorithm}

\vspace{-1.0ex}
In Algorithms \ref{MCSimuCont} through \ref{MCSimuContV2}, when dealing with a DFF-FS or GFF-FS Hamiltonian, Gibbs operators of the form $W_*(M) \defeq G_*(M)|_{\Diri(\calD)}$ or $W_*(M) \defeq G_*(M)$ need to be computed, where $M=\bigotimes_{i=1}^nh_i$ is an FBM monomial which moves a submanifold $\calD=\prod_{i=1}^n\calE_i$ as a Cartesian product of low-dimensional factor spaces, which, according to equations (\ref{TensorMonoMostProj}) and (\ref{TensorMonoSelfInv}), requires the computation of $W_*(h)$ for a fermionic Schr\"odinger FBM interaction $h\defeq h_i$ that moves a low-dimensional factor space $\calE\defeq\calE_i$, $i\in[1,m]$, $m=O(1)\in[1,n]$. Such a Gibbs operator $W_*(h)$ can be computed directly either by solving the eigenvalue problem of $h$ on the domain $\calE$, or via a subroutine evaluating the Feynman-Kac path integral \cite{Bernu02,Binder10} that corresponds to $h$ on $\calE$, subject to the proper boundary condition. Alternatively, each Gibbs operator $W_*(h)$ of any FBM interaction $h$ can itself be estimated through an MCMC subroutine, where $W_*(h)$ is broken up into many small-step Gibbs operators as $W_*(h)|=\prod_{j=1}^{\mathsmaller{N}}W_*(h/N)$, with $N\in\mathbb{N}$ sufficiently large though still polynomially bounded, such that each small-step Gibbs operator $W_*(h/N)$ is highly localized in space, meaning that there exist positive constants $c_1,c_2,c_3$ such that $\langle r|W_*(h/N)|q\rangle<\exp\left(-c_1\|r-q\|^{c_2}N^{c_3}\right)$, $\forall q,r\in\calE$, in which case, the small-step Gibbs operator $W_*(h/N)$ in a small neighborhood around any starting point $q\in\calE$ can be approximated very well by an analytical or semi-analytical solution for the case of an infinite or semi-infinite configuration space permeated globally with a simply flat, or linearly sloped, or quadratically varying potential landscape that locally matches the potential energy surface in the small neighborhood around $q\in\calE$. Particularly, a semi-infinite configuration space separated by a hyperplane from a region of infinite potential can be employed to model a small neighborhood in $\calE$ near the boundary $\partial\calE$.

Yet another alternative, applicable to an FBM interaction with a conventional Schr\"odinger operator $h|_{\Diri(\calE)}=-\Delta_g+V_{\mathsmaller{\calE}}(q)$ restricted to an open domain $\calE$, where the potential $V_{\mathsmaller{\calE}}(q)$ rises up steeply near $\partial\calE$, from being finite-valued within $\calE$ to being infinity-valued outside, in order to enforce the Dirichlet boundary condition, implements a single step of random walk described by a transition probability $\Pr(r|q)=\psi_0(r)\,\langle r|G_*(h)|_{\Diri(\calE)}|q\rangle\,\psi_0(q)^{-1}$ by mapping the Schr\"odinger-Dirichlet eigenvalue problem $\exp({-}h|_{\Diri(\calE)})\psi(q)=e^{-\lambda}\psi(q)$, $\forall q\in\calE$, $\lambda\in\mathbb{R}$, while $\psi(q)=0$ for all $q\notin\calE$, to a Fokker-Planck-Kolmogorov eigenvalue problem $\exp({-}L|_{\Diri(\calE,1)})\psi_0(q)\psi(q)=e^{-(\lambda-\lambda_0)}\psi_0(q)\psi(q)$, $\forall q\in\calE$, $\lambda\in\mathbb{R}$, while $\psi_0(q)\psi(q)=0$ for all $q\notin\calE$, with $\lambda_0\defeq\lambda_0(h|_{\Diri(\calE)})$ and $\psi_0\defeq\psi_0(h|_{\Diri(\calE)})$ characterizing the ground state of $h|_{\Diri(\calE)}$, through a {\em quantum-stochastic operator similarity transformation} \cite{Caffarel88I,Caffarel88II,Nelson66,Favella67,Albeverio74,Jona-Lasinio81,Risken96,Bogachev15} in the same spirit of \iftoggle{ForUSPTO} {Auxiliary Utility} {Lemma} \ref{QuasiStochasticOper}, where
\begin{equation}
L|_{\Diri(\calE,1)}\defeq\left[\psi_0\right]\times\left(h|_{\Diri(\calE)}-\lambda_0\right)\times\left[\psi_0\right]^{\mathsmaller{-}1}=\left[\psi_0\right]\times\left[-\Delta_g+V_{\mathsmaller{\calE}}(q)-\lambda_0\right]\times\left[\psi_0\right]^{\mathsmaller{-}1}
\label{SchrodingerFokkerPlanckKolmogorovOperator}
\end{equation}
is a Fokker-Planck-Kolmogorov operator generating a pure drift-diffusion process without branching and killing \cite{Nemec10}, and can be simulated by standard and stable MCMC \cite{Caffarel88I,Caffarel88II,Risken96}. It is obvious that the semigroup $\exp({-}L|_{\Diri(\calE,1)})=[\psi_0]\times\exp[-(h|_{\Diri(\calE)}-\lambda_0)]\times[\psi_0]^{\mathsmaller{-}1}$ of the Fokker-Planck-Kolmogorov drift-diffusion process does precisely realize the desired quasi-stochastic operator $[\psi_0^*(h)]G_*(h)|_{\Diri(\calE)}[\psi_0^*(h)]^{\mathsmaller{-}1}$ associated with $G_*(h)|_{\Diri(\calE)}$.

Constrained path integral methods \cite{Ceperley91,Ceperley96,Zhang97,Zhang04,Kruger08} provide another class of Monte Carlo algorithms for efficiently simulating DFF or GFF Hamiltonians, and the closely related reptation quantum Monte Carlo \cite{Baroni98arXiv,Baroni99prl,Carleo10} is rightly suitable for simulating an adiabatically varying process in association with a sequence of Hamiltonians that change gradually in time. Given a DFF or GFF Hamiltonian $H=H^{\mathsmaller{C}}+H^{\mathsmaller{D}}:L^2_{\mathsmaller{F}}(\calC)\mapsto L^2_{\mathsmaller{F}}(\calC)$, that is CD-separately irreducible and supported by a continuous-discrete product configuration space $C=\calM\times\calP$, an adiabatic sequence of DFF or GFF Hamiltonians $\{H'_l:L^2_{\mathsmaller{F}}(\calC)\mapsto L^2_{\mathsmaller{F}}(\calC),\,l\in[0,L],\,L\in\mathbb{N}\}$ may be constructed, in which $\psi_0(H'_0)$ is easily solved, $H'_{\mathsmaller{L}}=H$, and the Hamiltonians vary slowly from $H'_0$ to $H'_{\mathsmaller{L}}$ to satisfy conditions of the adiabatic theorem, such that the Gibbs operators $\{\exp(\mathsmaller{-}\tau H'_l):l\in[0,L]\}$, with a predetermined $\tau>0$, progressively and recursively project out a series of ground states $\{\psi_0(H'_l):l\in[0,L]\}$ from the starting point $\psi_0(H'_0)$, in the sense that $\langle\psi_0(H'_l)|\exp(\mathsmaller{-}\tau H'_l)|\psi_0(H'_{l-1})\rangle=\|\psi_0(H'_{l-1})\|\times[1+O(\tau^2/L^2)]$, much similar to the quantum search algorithm that uses repeated measurements to drive an adiabatic evolution of ground states \cite{Childs02}. For each $l\in[0,L]$, let $\LTK(H'_l)$ denote a set of FBM tensor monomials that Lie-Trotter-Kato decompose $H'_l$, such that
\begin{equation}
\left\|\left[\textstyle{\prod_{\mathsmaller{H\in\LTK(H'_l)}}}\,e^{-\tau H/m}\right]^m-e^{-\tau H'_l}\right\| \,\le\, c_1m^{\mathsmaller{-}c_2}\min[1,\lambda_1(\tau H'_l)], \label{LTKforHprime}
\end{equation}
for an $m=O(\poly(\max\{\size(H'_l):l\in[0,L]\}))$, $m\in\mathbb{N}$, and fixed constants $c_1>0$, $c_2>0$. Then let $\{H_t:t\in[0,T],\,T\in\mathbb{N}\}\defeq\bigcup_{l=0}^{\mathsmaller{L}}\{\LTK(H'_l)\}^m\!$ denote the ordered sequence of FBM tensor monomials collecting all of the $m$-fold repeated Lie-Trotter-Kato decompositions of $\{H'_l:l\in[0,L]\}$, such that the continued products of the Gibbs operators $\{\exp(\mathsmaller{-}\tau H'_l):l\in[0,L]\}$ are well approximated by the continued products of the Gibbs operators $\{\exp(\mathsmaller{-}\tau H_t/m):t\in[0,T]\}$, which progressively and recursively project out a series of ground states ${\mathlarger{\mathlarger{\{}}}\phi_t\defeq\psi_0({H'}_{\!\!\raisebox{0.4\height}{\tiny $\LTK^{-1}(t)$}}):t\in[0,T]{\mathlarger{\mathlarger{\}}}}$, where $\forall t\in[0,T]$, the function $\LTK^{-1}(t)$ returns the unique index $l\in[0,L]$ such that $H_t$ is there as an FBM tensor monomial for the Lie-Trotter-Kato decomposition of $H'_l$. Given such a projected ground state $\psi_0(H)=\prod_{t=0}^{\mathsmaller{T}}\exp(\mathsmaller{-}\tau H_t/m)|\phi_0\rangle$, the expected value of any operator $O:L^2_{\mathsmaller{F}}(\calC)\mapsto L^2_{\mathsmaller{F}}(\calC)$ can be computed as
\begin{equation}
\langle O\rangle=\frac{\langle\psi_0(H)|O|\psi_0(H)\rangle}{\langle\psi_0(H)|\psi_0(H)\rangle}=\frac
{\mathlarger{\mathlarger{\langle}}\phi_0\mathlarger{\mathlarger{|}}\!\left[\prod_{t=0}^{\mathsmaller{T}}\exp(\mathsmaller{-}\tau H_t/m)\right]^{\!+\!}O\!\left[\prod_{t=0}^{\mathsmaller{T}}\exp(\mathsmaller{-}\tau H_t/m)\right]\!\mathlarger{\mathlarger{|}}\phi_0\mathlarger{\mathlarger{\rangle}}}
{\mathlarger{\mathlarger{\langle}}\phi_0\mathlarger{\mathlarger{|}}\!\left[\prod_{t=0}^{\mathsmaller{T}}\exp(\mathsmaller{-}\tau H_t/m)\right]^{\!+\!}\!\left[\prod_{t=0}^{\mathsmaller{T}}\exp(\mathsmaller{-}\tau H_t/m)\right]\!\mathlarger{\mathlarger{|}}\phi_0\mathlarger{\mathlarger{\rangle}}}.
\end{equation}
The basic idea of PIMC in the configuration coordinate representation is to work with a cylinder set $\calC^{2\mathsmaller{T}+1}=\{(q_{2\mathsmaller{T}},\cdots\!,q_{\mathsmaller{T}},q_{\mathsmaller{T}-1},\cdots\!,q_0):q_t\in\calC,\,\forall t\in[0,2T]\}$, insert a resolution of the identity $I=\int_{\calC}dq_t\,|q_t\rangle\langle q_t|$ to the immediate left of either the operator $\exp(\mathsmaller{-}\tau H_t/m)$ for each $t\in[0,T]$ or the operator $\exp(\mathsmaller{-}\tau H_{2\mathsmaller{T}-t+1}/m)^+$ for each $t\in[T{+}1,2T]$, such that
\begin{equation}
\langle O\rangle=\frac
{\scalebox{1.35}{$\int$}\phi_0(q_{2\mathsmaller{T}})\mathlarger{\mathlarger{[}}\prod_{t=\mathsmaller{T}+1}^{2\mathsmaller{T}}\langle q_t|G_t|q_{t-1}\rangle\mathlarger{\mathlarger{]}}\langle q_{\mathsmaller{T}}|Oe^{\mathsmaller{-}\tau H_{\mathsmaller{T}}/m}|q_{\mathsmaller{T}-1}\rangle\mathlarger{\mathlarger{[}}\prod_{t=1}^{\mathsmaller{T}-1}\langle q_t|G_t|q_{t-1}\rangle\mathlarger{\mathlarger{]}}\phi_0(q_0)\prod_{t=0}^{2\mathsmaller{T}}dq_t}
{\scalebox{1.35}{$\int$}\phi_0(q_{2\mathsmaller{T}})\mathlarger{\mathlarger{[}}\prod_{t=\mathsmaller{T}+1}^{2\mathsmaller{T}}\langle q_t|G_t|q_{t-1}\rangle\mathlarger{\mathlarger{]}}\langle q_{\mathsmaller{T}}|G_{\mathsmaller{T}}|q_{\mathsmaller{T}-1}\rangle\mathlarger{\mathlarger{[}}\prod_{t=1}^{\mathsmaller{T}-1}\langle q_t|G_t|q_{t-1}\rangle\mathlarger{\mathlarger{]}}\phi_0(q_0)\prod_{t=0}^{2\mathsmaller{T}}dq_t},\!
\label{CoordPIMC}
\end{equation}
where shorthand notations are used with
\begin{equation}
G_t \,\defeq\,
\left\{\begin{array}{rl}
\exp(\mathsmaller{-}\tau H_t/m), & \mbox{when}\;t\in[0,T], \\
\exp(\mathsmaller{-}\tau H_{2\mathsmaller{T}-t+1}/m), & \mbox{when}\;t\in[T{+}1,2T],
\end{array}\right.
\end{equation}
while $\forall q\in\calC$, $|q\rangle=const\times\sum_{\pi\in G_{\rm ex}}\!\sign(\pi)\,|\pi q)$ is a Dirac delta distribution as a limit of wavefunctions in $L^2_{\mathsmaller{F}}(\calC)$ approximating an antisymmetric configuration coordinate eigenvector, and for each $\pi q=r\in\calC$, $|r)$ denotes a Dirac delta distribution as a limit of functions in $L^2(\calC)$ approximating a boltzmannonic state that has artificially labeled particles, of the same or different species, being localized and positioned around the configuration point $r\in\calC$.

With each $H_t\defeq\bigotimes_{i=1}^{n_t}h_{ti}$, $n_t\in\mathbb{N}$, $t\in[0,T]$ being an FBM tensor monomial, and most of the FBM interactions $\{h_{ti}:i\in[1,n_t]\}$ being orthogonal projection operators, the fermionic Gibbs transition amplitude $G_t(r,q)\defeq\langle r|\exp(\mathsmaller{-}\tau H_t)|q\rangle$, $\forall(r,q)\in\calC^2$, $\forall\tau>0$, $\tau=O(1)$ can be computed as in equation (\ref{TensorMonoMostProj}) from the transition amplitudes associated with the individual FBM interactions. For each $i\in[1,n_t]$, it is perfectly feasible to have the FBM interaction $h_{ti}$ exactly diagonalized in the $h_{ti}$-moved factor space $\calC_i=\calM_i\times\calP_i$, from which, either the transition amplitude $\langle r_i|h_i|q_i\rangle$, when $h_i$ is an orthogonal projection operator, or else the Gibbs transition amplitude $\langle r_i|\exp(\mathsmaller{-}\tau h_i)|q_i\rangle$, can be easily computed for any $(r_i,q_i)\in\calC_i^2$. Note that, when $h_i$ is a projection operator, computing the Gibbs transition amplitude $\langle r_i|\exp(\mathsmaller{-}\tau h_i)|q_i\rangle$ for any $\tau>0$ such that $\tau+\tau^{{-}1}=O(\poly(\size(H)))$ is completely equivalent to computing $\langle r_i|h_i|q_i\rangle$, by virtue of the operator identity $\exp(\mathsmaller{-}\tau P)=I+(e^{{-}\tau}-1)P$ for all $P$ such that $P^2=P$. Alternatively, each fermionic Gibbs transition amplitude $\langle r_i|\exp(\mathsmaller{-}\tau h_i)|q_i\rangle$, with $h_i$ being an FBM interaction, can be evaluated by straightforwardly summing up an $O(\poly(\size(H)))$ number of boltzmannonic transition amplitudes for all of the exchange permutations of the configuration $(r_i,q_i)$, where, for each such permuted configuration $(r'_i,q'_i)\in\calC_i^2$, a boltzmannonic transition amplitude can be computed approximately by taking a direct ``Rayleigh flight'', either along the geodesic line segment from $q'_i$ to $r'_i$, when $q'_i,r'_i\in\calM\times\{v\}$ for some $v\in\calP$, and $\langle r'_i|\exp[{-}\tau\Base(H^{\mathsmaller{C}})]|q'_i\rangle\neq 0$, or through a single $H^{\mathsmaller{D}}$-hop, when $q'_i,r'_i\in\{x\}\times\calP$ for some $x\in\calM$, and $\langle r'_i|\Base(H^{\mathsmaller{D}})|q'_i\rangle\neq 0$, provided that $\tau/m>0$ is sufficiently small. Still alternatively, using the Hubbard-Stratonovich transformation \cite{Stratonovich58,Hubbard59}, every Gibbs operator $\exp(\mathsmaller{-}\tau h_{ti}/m)$ with a small $\tau/m>0$ can be converted into a product of single-particle Gibbs operators, each of which is free of inter-particle interactions, but may involve an individual particle subject to an external auxiliary random field, such that the product of single-particle Gibbs operators under a fixed auxiliary random field configuration is analytically solvable, and the Gibbs operator $\exp(\mathsmaller{-}\tau h_{ti}/m)$ can be computed as an expectation value of said product against the probability distribution of the auxiliary random fields \cite{Blankenbecler81,Sugiyama86}. In any case, equation (\ref{CoordPIMC}) is numerically implemented by approximating both the denominator and the numerator on the right hand side as a sum of quantum amplitudes associated with zigzags henceforth called {\em Feynman paths}, each of which corresponds to a unique point in the cylinder set and is of the form $\bfq\defeq(q_{2\mathsmaller{T}},\cdots\!,q_{\mathsmaller{T}},q_{\mathsmaller{T}-1},\cdots\!,q_0)\in\calC^{2\mathsmaller{T}+1}$, where each $q_t\in\calC$ is called the configuration coordinate of the $t$-th time slice, $\forall t\in[0,2T]$. With respect to each specific Feynman path $\bfq$, a partial zigzag formed by connecting a consecutive subset of slice coordinates, $\bfq(t_2,t_1)\defeq(q_{t_2},q_{t_2{-}1},\cdots\!,q_{t_1{+}1},q_{t_1})$, with $0\le t_1<t_2\le 2T$, is called a {\em Feynman segment} of $\bfq$, and the quantum amplitude
\begin{equation}
\bfG[\bfq(t_2,t_1)] \,\defeq\, \phi_{\lfloor t_2\rceil}(q_{t_2})\,{\mathlarger{\mathlarger{[}}}\,{\textstyle{\prod_{t=t_1{+}1}^{t_2}}}\langle q_t|G_t|q_{t-1}\rangle{\mathlarger{\mathlarger{]}}}\,\phi_{\lfloor t_1\rceil}(q_{t_1})
\end{equation}
is called a {\em partial path integral} associated with the Feynman segment $\bfq(t_2,t_1)$, where
\begin{equation}
\lfloor t\rceil\defeq\left\{\begin{array}{rl}
t,\! & \!\mbox{when}\;t\in[0,T], \\
2T-t,\! & \!\mbox{when}\;t\in[T{+}1,2T].
\end{array}\right.
\end{equation}
It is well known and recognized that Feynman's method and theory of Euclidean path integrals can be regarded as an isophysical mapping between a quantum system on a configuration space $\calC$ and a classical system on a configuration space $\calC\otimes[0,2T] \cong \calC^{2\mathsmaller{T}+1}$, with which a quantum Hamiltonian supported by $\calC$ is mapped to a classical potential energy function on $\calC^{2\mathsmaller{T}+1}$, and a quantum Euclidean theory becomes identical to a classical theory of statistical mechanics, although for a fermionic quantum system, the requirement of exchange antisymmetry, persisting into the classically mapped problem, remains a strong quantum reminder and signature.

A PIMC procedure does importance sampling from such zigzag Feynman paths weighted by their associated quantum amplitudes and produces an approximate fraction to estimate $\langle O\rangle$. To overcome the sign problem, a constrained PIMC \cite{Ceperley91,Ceperley96,Zhang97,Zhang04,Kruger08} starts with and accepts only {\em positive-definite Feynman paths} $\bfq=(q_{2\mathsmaller{T}},\cdots\!,q_{\mathsmaller{T}},q_{\mathsmaller{T}-1},\cdots\!,q_0)\in\calC^{2\mathsmaller{T}+1}$ such that $\phi_0(q_0)>0$ and $\bfG[\bfq(t,0)]>0$, $\forall t\in[1,2T]$, or equivalently,
\begin{equation}
\phi_0(q_0)>0\,\;\mbox{\rm and}\;\,\bfG[\bfq(t,t{-}1)]>0,\;\forall t\in[1,2T]. \label{PhiGtPhiPositive}
\end{equation}
During a constrained PIMC procedure, when a positive-definite Feynman path $\bfq\in\calC^{2\mathsmaller{T}+1}$ is being wiggled into $\bfq'\in\calC^{2\mathsmaller{T}+1}$, by changing and updating the configuration coordinates for a tuple of time slices $\{t_k,\cdots\!,t_i,\cdots\!,t_0\}$ from an old tuple of configuration coordinates $\{q_{t_k},\cdots\!,q_{t_i}\cdots\!,q_{t_0}\}$ to a new tuple $\{q'_{t_k},\cdots\!,q'_{t_i}\cdots\!,q'_{t_0}\}$, with $0\le t_0\le\cdots\le t_i\le\cdots\le t_k\le 2T$, $0\le i\le k\le 2T$, but no two time instants should coincide when $t_k\neq t_0$, the positive definiteness of Feynman paths must be maintained by requiring that the two fractions
\begin{align}
\frac{\bfG[\bfq'(t_i,t_i{-}1)]}{\bfG[\bfq(t_i,t_i{-}1)]} \,=\,\;& \frac{\psi_0(H_{\lfloor t_i\rceil};q'_{t_i})\,G_{t_i}(q'_{t_i},q'_{t_i-1})\,\psi_0(H_{\lfloor t_i-1\rceil};q'_{t_i-1})}{\psi_0(H_{\lfloor t_i\rceil};q_{t_i})\,G_{t_i}(q_{t_i},q_{t_i-1})\,\psi_0(H_{\lfloor t_i-1\rceil};q_{t_i-1})} \label{PositivePathCondRight} \\[0.75ex]
\frac{\bfG[\bfq'(t_i{+}1,t_i)]}{\bfG[\bfq(t_i{+}1,t_i)]} \,=\,\;& \frac{\psi_0(H_{\lfloor t_i+1\rceil};q'_{t_i+1})\,G_{t_i+1}(q'_{t_i+1},q'_{t_i})\,\psi_0(H_{\lfloor t_i\rceil};q'_{t_i})}{\psi_0(H_{\lfloor t_i+1\rceil};q_{t_i+1})\,G_{t_i+1}(q_{t_i+1},q_{t_i})\,\psi_0(H_{\lfloor t_i\rceil};q_{t_i})} \label{PositivePathCondLeft}
\end{align}
both remain positive, $\forall i\in[0,k]$, where the local node-determinacy of the DFF or GFF Hamiltonians $\{H'_l:l\in[0,L]\}$ has been used in equations (\ref{PositivePathCondRight}) and (\ref{PositivePathCondLeft}), which enables efficient local determination of the ground state nodal structures. It is noted that some harmless redundancy may be present when the tuple of time slices $\{t_k,\cdots\!,t_i\cdots\!,t_0\}$ consists of consecutive time instants, namely, $\exists\,i\in[0,k]$ such that $t_{i+1}=t_i$ or $t_{i-1}=t_i$. Such redundancy is easily identified and avoided. Among the legitimate new Feynman paths $\bfq'=(q'_{2\mathsmaller{T}},\cdots\!,q'_{\mathsmaller{T}},q'_{\mathsmaller{T}-1},\cdots\!,q'_0)\in\calC^{2\mathsmaller{T}+1}$, a Metropolis-Hastings importance sampling strategy can be used to select a transition from $\bfq$ to $\bfq'$ randomly according to a probability
\begin{equation}
\Pr(\bfq'\leftarrow\bfq) \,\sim\, \frac{\bfG[\bfq'(2T,0)]}{\bfG[\bfq(2T,0)]} \,=\, \frac{\phi_0(q'_{2\mathsmaller{T}})\,{\mathlarger{\mathlarger{[}}}\,{\textstyle{\prod_{t=1}^{2\mathsmaller{T}}}}\langle q'_t|G_t|q'_{t-1}\rangle{\mathlarger{\mathlarger{]}}}\,\phi_0(q'_0)}{\phi_0(q_{2\mathsmaller{T}})\,{\mathlarger{\mathlarger{[}}}\,{\textstyle{\prod_{t=1}^{2\mathsmaller{T}}}}\langle q_t|G_t|q_{t-1}\rangle{\mathlarger{\mathlarger{]}}}\,\phi_0(q_0)} \,>\, 0, \label{FeynPathWiggleTrans}
\end{equation}
where it is obvious that many of the multiplying terms in the numerator and denominator of the rightmost fraction are common factors and cancel out, except for those involving a configuration coordinate that is being changed.

The following algorithm is an example of constrained PIMC simulating the ground state of a CD-separately irreducible DFF or GFF Hamiltonian $H=H^{\mathsmaller{C}}+H^{\mathsmaller{D}}=H'_{\mathsmaller{L}}$ supported by a continuous-discrete product configuration space $C=\calM\times\calP$, with $H'_{\mathsmaller{L}}$ terminating an adiabatic sequence of DFF or GFF Hamiltonians $\{H'_l:L^2_{\mathsmaller{F}}(\calC)\mapsto L^2_{\mathsmaller{F}}(\calC),\,l\in[0,L]\}$, $L\in\mathbb{N}$, whose initial Hamiltonian $H'_0$ has a known ground state $\psi_0(H'_0)$. Let $m=O(\poly(\max\{\size(H'_l):l\in[0,L]\}))$, $m\in\mathbb{N}$ be chosen sufficiently large so that each Gibbs operator $\exp({-}\tau H'_l)$, $l\in[0,L]$ is approximated as in (\ref{LTKforHprime}) to a predetermined accuracy, and let $\{H_t:t\in[0,T],\,T\in\mathbb{N}\}\defeq\bigcup_{l=0}^{\mathsmaller{L}}\{\LTK(H'_l)\}^{\otimes m}\!$ denote the ordered sequence of FBM tensor monomials collecting all of the $m$-fold repeated Lie-Trotter-Kato decompositions of $\{H'_l:l\in[0,L]\}$, with grounds states ${\mathlarger{\mathlarger{\{}}}\phi_t\defeq\psi_0({H'}_{\!\!\raisebox{0.4\height}{\tiny $\LTK^{-1}(t)$}}):t\in[0,T]{\mathlarger{\mathlarger{\}}}}$.

\renewcommand{\labelenumi}{\thealgorithm.\arabic{enumi}}
\begin{algorithm}{(Constrained PIMC  Simulation of a DFF or GFF Hamiltonian)} \label{MCSimuConsPathInt}\\
\vspace{-5.0ex}
\begin{enumerate}[start=0]
\item Choose a predetermined bound $\delta>0$ for numerical accuracy; Choose a warm start $q=r_0$, that is a random point $r_0\in\calU^c(H)$ with the promise of $|\log\psi_0(H'_0;r_0)|$ being polynomially bounded; Set an iteration counters $N\in\mathbb{Z}$ to $0$; Choose a predetermined upper bound $N_{\max}\in\mathbb{N}$ for the iteration counter; Set a predetermined number of samples $S_{\max}\in\mathbb{N}$; Choose a predetermined number $2T{+}1$ of time slices, with $T\in\mathbb{N}$;

Initialize a $2T$-tuple of Gibbs operators $\{G_t:t\in[1,2T]\}$ to $G_t=\exp({-}\tau H_0/m)$, $\forall t\in[1,2T]$; Initialize a Feynman path, that is a $(2T{+}1)$-tuple $\bfq\defeq(q_{2\mathsmaller{T}},\cdots\!,q_{\mathsmaller{T}},q_{\mathsmaller{T}-1},\cdots\!,q_0)\in\calC^{2\mathsmaller{T}+1}$ to $q_t=r_0$, $\forall t\in[0,2T]$; Initialize a time slice index $t\in\mathbb{N}$ to $0$; \label{AlgConsPathIntStepZero}
\item Increase the time slice index $t$ by $1$; Find a $q'_t\in\calU^c(H_t)$ such that
$$\frac{\bfG[\bfq'(t,t{-}1)]}{\bfG[\bfq(t,t{-}1)]}\,=\,\frac{\psi_0(H_t;q'_t)\,G_t(q'_t,q_{t-1})}{\psi_0(H_t;q_t)\,G_t(q_t,q_{t-1})} \,>\, 0;$$ \label{AlgConsPathIntGrowPath1} \vspace{-3.5ex}
\item $\forall t'\in[t,2T{-}t{+}1]\cap\mathbb{N}$, assign the newly found value of $q'_t$ to each $q_{t'}$, and update each Gibbs operator $G_{t'}$ to $G_{t'}=\exp({-}\tau H_t/m)$; \label{AlgConsPathIntGrowPath2}
\item Repeat steps \ref{MCSimuConsPathInt}.\ref{AlgConsPathIntGrowPath1} and \ref{MCSimuConsPathInt}.\ref{AlgConsPathIntGrowPath2} until $t=T$, upon which time break the loop and go to step \ref{MCSimuConsPathInt}.\ref{AlgConsPathIntChooseTimeTuple}; \label{AlgConsPathIntGrowPath3}
\item Increase the iteration counters $N$ by $1$; Choose a random tuple of time slices $\{t_k,\cdots\!,t_i\cdots\!,t_0\}$ with $0\le t_0\le\cdots\!\le t_i\le\cdots\!\le t_k\le 2T$, $0\le i\le k\le 2T$, but no two time instants should coincide when $t_k\neq t_0$; \label{AlgConsPathIntChooseTimeTuple}
\item Choose a random tuple of configuration coordinates $\{q'_{t_k},\cdots\!,q'_{t_i}\cdots\!,q'_{t_0}\}\in\calC^{k+1}$ that satisfies the positive-definite constraints of equations (\ref{PositivePathCondRight}) and (\ref{PositivePathCondLeft}); \label{AlgConsPathIntFindNewCoords}
\item Update the Feynman path from $\bfq$ to $\bfq'$ by substituting the tuple of configuration coordinates $\{q_{t_k},\cdots\!,q_{t_i}\cdots\!,q_{t_0}\}\subseteq\bfq$ with $\{q'_{t_k},\cdots\!,q'_{t_i}\cdots\!,q'_{t_0}\}\subseteq\bfq'$, via Metropolis-Hastings importance sampling according to the transition probability $\Pr(\bfq'\leftarrow\bfq)\sim\bfG[\bfq'(2T,0)]/\bfG[\bfq(2T,0)]\!$ as specified in equation (\ref{FeynPathWiggleTrans}); \label{AlgConsPathIntTransToNewCoords}
\item Repeat steps \ref{MCSimuConsPathInt}.\ref{AlgConsPathIntChooseTimeTuple} through \ref{MCSimuConsPathInt}.\ref{AlgConsPathIntTransToNewCoords} until $N=N_{\max}$, record the well-mixed Feynman path, and reset the value of $N$ to $0$; \label{AlgConsPathIntMixing}
\item Repeat the loop \ref{MCSimuConsPathInt}.\ref{AlgConsPathIntMixing} $S_{\max}$ times to get $S_{\max}$ samples of Feynman paths. \label{AlgConsPathIntRepeatPolyTimes}
\end{enumerate}
\end{algorithm}

The following lists an alternative method of PIMC for simulating a DFF-FS or GFF-FS system, which does not forbid nodal crossings between two adjacent configuration coordinates of two consecutive time slices, but label each Feynman path $\bfq=(q_{2\mathsmaller{T}},\cdots\!,q_{\mathsmaller{T}},q_{\mathsmaller{T}-1},\cdots\!,q_0)\in\calC^{2\mathsmaller{T}+1}$ by a $\lab(\bfq) = \pm 1$, that counts the total number of nodal crossings incurred by the Feynman path, and is efficiently computed from a $\lab(q_0) \defeq \sign(\psi_0(H_0;q_0))$, a $\lab(q_{\mathsmaller{T}}) \defeq \sign(\psi_0(H_0;q_{\mathsmaller{T}}))$, as well as the local nodal structure of the intermediate Hamiltonians $\{H_t:t\in[0,T],\,T\in\mathbb{N}\}$. It is noted, as been observed by many previous authors \cite{Ceperley91,Ceperley96,Zhang97,Zhang04,Kruger08}, that, for fermionic Schr\"odinger systems, the expectation value of an observable with respect to the ground state can be numerically estimated by discarding the {\em odd Feynman paths} labeled by $-1$ and using only the {\em even Feynman paths} labeled by $+1$, because, for any odd Feynman path, there exists a compensating even Feynman path that is a mirror image of the odd Feynman path, obtained by firstly identifying a time slice with a configuration point $q_*$ sitting on a nodal surface, then mirror-reflecting either the Feynman segment from $q_0$ to $q_*$ or the Feynman segment from $q_*$ to $q_{2\mathsmaller{T}}$, said mirror reflection being done by applying an odd permutation that exchange identical fermions.

\renewcommand{\labelenumi}{\thealgorithm.\arabic{enumi}}
\begin{algorithm}{(Constrained PIMC  Simulation of a DFF or GFF Hamiltonian)} \label{MCSimuNoCoPathInt}\\
\vspace{-5.0ex}
\begin{enumerate}[start=0]
\item Choose a predetermined bound $\delta>0$ for numerical accuracy; Choose a warm start $q=r_0$, that is a random point $r_0\in\calU^c(H)$ with the promise of $|\log\psi_0(H'_0;r_0)|$ being polynomially bounded; Set an iteration counters $N\in\mathbb{Z}$ to $0$; Choose a predetermined upper bound $N_{\max}\in\mathbb{N}$ for the iteration counter; Set a predetermined number of samples $S_{\max}\in\mathbb{N}$; Choose a predetermined number $2T{+}1$ of time slices, with $T\in\mathbb{N}$;

Initialize a $2T$-tuple of Gibbs operators $\{G_t:t\in[1,2T]\}$ to $G_t=\exp({-}\tau H_0/m)$, $\forall t\in[1,2T]$; Initialize a Feynman path, that is a $(2T{+}1)$-tuple $\bfq\defeq(q_{2\mathsmaller{T}},\cdots\!,q_{\mathsmaller{T}},q_{\mathsmaller{T}-1},\cdots\!,q_0)\in\calC^{2\mathsmaller{T}+1}$ to $q_t=r_0$, $\forall t\in[0,2T]$; Label the initial Feynman path by $\lab(\bfq) = +1$; Initialize a time slice index $t\in\mathbb{N}$ to $0$; \label{AlgNoCoPathIntStepZero}
\item Increase the time slice index $t$ by $1$; Find a $q'_t\in\calU^c(H_t)$ such that
$$\frac{\bfG[\bfq'(t,t{-}1)]}{\bfG[\bfq(t,t{-}1)]}\,=\,\frac{\psi_0(H_t;q'_t)\,G_t(q'_t,q_{t-1})}{\psi_0(H_t;q_t)\,G_t(q_t,q_{t-1})} \,\neq\, 0;$$
\label{AlgNoCoPathIntGrowPath1} \vspace{-3.5ex}
\item $\forall t'\in[t,2T{-}t{+}1]\cap\mathbb{N}$, assign the newly found value of $q'_t$ to each $q_{t'}$, and update each Gibbs operator $G_{t'}$ to $G_{t'}=\exp({-}\tau H_t/m)$; \label{AlgNoCoPathIntGrowPath2}
\item Repeat steps \ref{MCSimuNoCoPathInt}.\ref{AlgNoCoPathIntGrowPath1} and \ref{MCSimuNoCoPathInt}.\ref{AlgNoCoPathIntGrowPath2} until $t=T$, upon which time break the loop and go to step \ref{MCSimuNoCoPathInt}.\ref{AlgNoCoPathIntChooseTimeTuple}; \label{AlgNoCoPathIntGrowPath3}
\item Increase the iteration counters $N$ by $1$; Choose a random tuple of time slices $\{t_k,\cdots\!,t_i\cdots\!,t_0\}$ with $0\le t_0\le\cdots\!\le t_i\le\cdots\!\le t_k\le 2T$, $0\le i\le k\le 2T$, but no two time instants should coincide when $t_k\neq t_0$; \label{AlgNoCoPathIntChooseTimeTuple}
\item Choose a random tuple of configuration coordinates $\{q'_{t_k},\cdots\!,q'_{t_i}\cdots\!,q'_{t_0}\}\in\calC^{k+1}$ such that none of the two fractions in equations (\ref{PositivePathCondRight}) and (\ref{PositivePathCondLeft}) vanishes; \label{AlgNoCoPathIntFindNewCoords}
\item Update the Feynman path from $\bfq$ to $\bfq'$ by substituting the tuple of configuration coordinates $\{q_{t_k},\cdots\!,q_{t_i}\cdots\!,q_{t_0}\}\subseteq\bfq$ with $\{q'_{t_k},\cdots\!,q'_{t_i}\cdots\!,q'_{t_0}\}\subseteq\bfq'$, via Metropolis-Hastings importance sampling according to the transition probability $\Pr(\bfq'\leftarrow\bfq) \sim \scalebox{1.15}{$|$}\bfG[\bfq'(2T,0)]/\bfG[\bfq(2T,0)]\scalebox{1.15}{$|$}\!$ as specified in equation (\ref{FeynPathWiggleTrans}); If $\bfG[\bfq'(2T,0)]/\bfG[\bfq(2T,0)] < 0$, then label $\bfq'$ by $\lab(\bfq') \defeq -\lab(\bfq)$; Otherwise, label $\bfq'$ by $\lab(\bfq') \defeq \lab(\bfq)$; \label{AlgNoCoPathIntTransToNewCoords}
\item Repeat steps \ref{MCSimuNoCoPathInt}.\ref{AlgNoCoPathIntChooseTimeTuple} through \ref{MCSimuNoCoPathInt}.\ref{AlgNoCoPathIntTransToNewCoords} until $N=N_{\max}$, record the well-mixed Feynman path, and reset the value of $N$ to $0$; \label{AlgNoCoPathIntMixing}
\item Repeat the loop \ref{MCSimuNoCoPathInt}.\ref{AlgNoCoPathIntMixing} $S_{\max}$ times to get $S_{\max}$ samples of Feynman paths. \label{AlgNoCoPathIntRepeatPolyTimes}
\end{enumerate}
\end{algorithm}

An SFF-FS or SFF-EB partial Hamiltonian $H=\sum_{i=1}^{\mathsmaller{J}}H_i$, $J \in \mathbb{N}$, with each partial Hamiltonian $H_i$, $i \in [1,J]$ being DFF or GFF and FS or EB, can be efficiently simulated via Monte Carlo by iteratively simulating DFF or GFF and FS or EB partial Hamiltonians using one of the Algorithms \ref{MCSimuDisc} through \ref{MCSimuNoCoPathInt}. It is often the ground state $\psi_0(H)$ or the associated probability density $|\psi_0(H)|^2$ that is of interest, which can be obtained by running a time-inhomogeneous Markov chain \cite{Seneta81,Winkler03,Stroock05} that iterates the ground state projectors $\{[\psi_0(H_i)] \exp(-\tau H_i) [\psi_0(H_i)]^{-1}\}_{i \in [1,J]}$ in sequence, as implemented in the following algorithm.

\renewcommand{\labelenumi}{\thealgorithm.\arabic{enumi}}
\begin{algorithm}{(Monte Carlo Simulation of a SFF-FS or SFF-EB Hamiltonian)} \label{MCSimuSFF}\\
\vspace{-5.0ex}
\begin{enumerate}[start=0]
\item Receive descriptions of an SFF-FS or SFF-EB partial Hamiltonian $H=\sum_{i=1}^{\mathsmaller{J}}H_i$, $J \in \mathbb{N}$, with each partial Hamiltonian $H_i$, $i \in [1,J]$ being DFF or GFF and FS or EB; Initialize a random walker; \label{MCSimuSFFInit}
\item Perform Monte Carlo simulations driven by $\{[\psi_0(H_i)] \exp(-\tau H_i) [\psi_0(H_i)]^{-1}\}_{i \in [1,J]}$ as Markov operators in sequence, with the $i$-th Markov operator driving the random walker toward a steady state distribution $|psi_0(H_i)|^2$, $\forall i \in [1,J]$; \label{MCSimuSFFOneIter}
\item Repeat step \ref{MCSimuSFF}.\ref{MCSimuSFFOneIter} for a $\Theta(\poly(\size(H)))$ number of times.
\end{enumerate}
\end{algorithm}

\begin{theorem}{(A First {\myTheorem} of Monte Carlo Quantum Computing)}\\
With a warm start promised, the computational problem (as specified in Definition \ref{defiCompProbSimuManyBody}) of simulating an SFF-FS or SFF-EB Hamiltonian $H=\sum_{k=1}^{\mathsmaller{K}}H_k$ is in the class BPP.
\label{FirstTheorem}
\end{theorem}
\vspace{-4.0ex}
\begin{proof}[\iftoggle{ForUSPTO} {Demonstration} {Proof}]
According to the analyses above and Algorithm \ref{MCSimuSFF}, it is only necessary to perform a $\Theta(\poly(\size(H)))$ number of Monte Carlo simulations that sample from ground states of DFF or GFF and FS or EB partial Hamiltonians. Suffice it to prove that any such Monte Carlo simulation for a DFF or GFF and FS or EB partial Hamiltonian can be done with a polynomial-bounded computational complexity.

Firstly, it is without loss of generality to assume that the configuration space $\calC$ is either all-continuous or all-discrete, because any computational problem as specified in Definition \ref{defiCompProbSimuManyBody} with a continuous-discrete product configuration space can be transformed, through a polynomial reduction, into another problem of the same class but based on an all-continuous or all-discrete configuration space. By \myLemma\;\ref{QuasiStochasticOper} and \myCorollary\;\ref{FixedNodeMethod}, suffice it to show that any of the Algorithms \ref{MCSimuDisc} through \ref{MCSimuContV2} solves in probabilistic polynomial time the equivalent problem of either a Dirichlet boundary-conditioned boltzmannonic system confined in one nodal cell when using a fixed-node diffusion type of approach, or a quantum-stochastic operator similarity-transformed random walk based on equation (\ref{r_PinftyHk_q}) over the entire configuration space otherwise.

In accordance with the configuration space $\calC$ being discrete or continuous and $H_i$, $i \in [1,J]$ being FS or EB, a Gibbs operator $G_{\mathsmaller{\Box}}(H_i)$ is suitably defined, and the unique and polynomially gapped ground state $\psi_0(G_{\mathsmaller{\Box}}(H_i))$ is sought after. Irrespectively, it always holds true that the associated quasi-stochastic operator $Q_{\!\mathsmaller{H}_i}=[\psi_0(G_{\mathsmaller{\Box}}(H_i))]G_{\mathsmaller{\Box}}(H_i)[\psi_0(G_{\mathsmaller{\Box}}(H_i))]^{-1}\!$ shares exactly the same eigenvalues with $G_{\mathsmaller{\Box}}(H_i)$. The quotient $\psi_0(G_{\mathsmaller{\Box}}(H_i);r)/\psi_0(G_{\mathsmaller{\Box}}(H_i);q)$ between any pair of configuration points $q,r\in\calU^c(H_i)\subseteq\calC$ that are relevant for simulating $Q_{\!\mathsmaller{H_i}}$-induced transitions is always efficiently computable through a solution of one of its FBM tensor monomials, which is computed efficiently by solving FBM interactions over low-dimensional factor spaces, thereby rendering the transition probability matrix $Q_{\!\mathsmaller{H}_i}$ itself efficiently computable and always positive, thus bona fide stochastic, defining a Markov chain random walk.

As will be justified immediately below, it is safe to assume that $G_{\mathsmaller{\Box}}(H_i)$ is aperiodic, consequently, $Q_{\!\mathsmaller{H}_i}$ is primitive, namely, aperiodic irreducible non-negative. Then the Markov chain $Q_{\!\mathsmaller{H}_i}$ is bound to converge to its unique stationary distribution $|\psi_0(H_i;q)|^2$ from any warm start $q_0$, and the mixing time is polynomial as determined by the spectral gap \cite{Seneta81,Levin08}. $\forall A\in\{\ref{MCSimuDisc},\ref{MCSimuDiscV1},\ref{MCSimuDiscV2},\ref{MCSimuCont},\ref{MCSimuContV1},\ref{MCSimuContV2}\}$, for each iteration of steps $A.1$ through $A.5$ of Algorithm $A$, the computational complexity is polynomial because the number of $(s_m,n_m,d_m)$-few-body-moving tensor monomials is polynomially bounded, while $s_m=O(\log(\size(H_i)))$, $n_m=O(1)$, $d_m=O(1)$. Finally, the predetermined upper bound $N_{\max}$ in all algorithms needs only to be polynomially sized to obtain a polynomial number of rapidly mixing samples of $q\in\calC$ according to the distribution $|\psi_0(H_i;q)|^2$. To conclude, the computational problem of simulating an DFF or GFF and FS or EB Hamiltonian $H_i$ as specified in Definition \ref{defiCompProbSimuManyBody} can be solved via MCMC in polynomial time.
\vspace{-1.5ex}
\end{proof}

In view of the polynomial gappedness of $H$, it is indeed without loss of generality to assume that $\exp(-H)$ be aperiodic and irreducible, therefore, $Q_{\!\mathsmaller{H}} \defeq [\psi_0(H)] \exp(-H) [\psi_0(H)]^{-1\!}$ be primitive, regardless of the configuration space being discrete or continuous, because a Laplace-Beltrami omnidirectional diffusion operator $-\epsilon{\sf L}$, with $\epsilon=\Omega(\poly(\size(H)))$, $\epsilon>0$, could always be added to the Hamiltonian $H$ without substantially altering the ground state wavefunction and the energy gap, so to render the Hamiltonian definitely irreducible. More specifically, with a discrete configuration space $\calG=\prod_{s=1}^{\mathsmaller{S}}\calG_s^{n_s}$ as a Cartesian product graph, the Laplace-Beltrami operator could be ${\sf L}=\sum{\raisebox{.2\height}{$_{(i,j)\in E(\calG)}$}}\left(|i\rangle\langle j|+|j\rangle\langle i|\right)$, where $i$ and $j$ index graph vertices, $(i,j)$ represents a graph edge, and $E(\calG)$ is the set of all graph edges. With a continuous configuration space, ${\sf L}$ could be a Laplace-Beltrami operator $\Delta_g=|\!\det(g)|^{\mathsmaller{-}1/2}\partial_i|\!\det(g)|^{1/2}g^{ij}\partial_j$ associated with a smooth Riemannian metric $g$ on a Riemannian manifold $(\calM,g)$. The aperiodicity of $\exp(-H)$, thus the primitivity of $Q_{\!\mathsmaller{H}}$, follows immediately from the aperiodicity of the Hamiltonian $H$. Or even simpler, any Markovian transition matrix $Q_{\!\mathsmaller{H}}$ can be straightforwardly modified into a lazy version $\half(I+Q_{\!\mathsmaller{H}})$ \cite{Seneta81,Levin08} to ensure aperiodicity.

A warm start is guaranteed for a polynomially gapped, CD-separately irreducible Hamiltonian $H=H^{\mathsmaller{C}}+H^{\mathsmaller{D}}$ supported by a continuous-discrete product configuration space $C=\calM\times\calP$, when two mild conditions are fulfilled: 1) the diagonal elements of $H$ are all polynomially bounded, that is, $|\langle q|H|q\rangle|\le\poly(\size(H))$, $\forall q\in\calC$; 2) each nodal cell is polynomially path-connected, and the positive $H^{\mathsmaller{D}}$-distance between any two nodal cells in any nodal groupoid is polynomially bounded. In which case, any location $q_0$ that is at least a distance $\epsilon>0$ away from any nodal point of $\psi_0(H)$, with $\epsilon^{\mathsmaller{-}1}=O(\poly(\size(H)))$, is a warm start, because, without loss of generality, again a Laplace-Beltrami omnidirectional diffusion operator $-\epsilon{\sf L}$ can be added to the Hamiltonian $H$ without substantially altering the ground state wavefunction and the energy gap, so to establish $\exp[-O(|\log\epsilon|+\poly(\size(H)))L_{\sup}]$ as a lower bound for $|\psi_0(H;q)|$, $\forall q\in\calC$ that is a distance of $\epsilon$ away from any nodal point of $\psi_0(H)$.

Alternatively, it is mathematically less sophisticated and physically more intuitive that, Algorithms \ref{MCSimuDisc} through \ref{MCSimuSFF} can be used to simulate a time series of SFF-FS or SFF-EB Hamiltonians $\{H(t)\}_{0\le t\le 1}$, evolving adiabatically from a pre-solved $H(0)$ with a known $\psi_0(H(0))$ to a target Hamiltonian $H(1)$, such that a warm start for $H(0)$ is readily available, and running any Algorithm \ref{MCSimuDisc} through \ref{MCSimuContV2} for $H(t)$, $0\le t<1$ till convergence supplies warm starts for simulating $H(t+\delta t)$, $\delta t>0$, so long as $\|H(t+\delta t)-H(t)\|$ is sufficiently small and $H(t)$, $0\le t<1$ is always sufficiently gapped to satisfy conditions of the adiabatic theorem \cite{Born28,Kato50,Berry84,Farhi01,Jansen07,Amin09}. Similar adiabatic techniques have been used in the context of boltzmannonic stoquastic Hamiltonians \cite{Bravyi10,Aharonov07a,Aharonov07b}.

Still further, it often happens, and always does in the context of GSQC \cite{Feynman85,Kitaev02,Mizel04,Kempe06}, that the sought-after ground state encodes the history of a quantum evolution and has the quantum probability distributed more or less uniformly among a polynomial number of temporal snapshots of the evolution history, in which one or a few temporal snapshot(s) representing the initial condition of the quantum evolution is/are associated with a known distribution in space carrying a polynomially sized portion of the quantum probability, from which a warm start can be picked easily.

\iftoggle{ForUSPTO} {One skilled in the art shall have no difficulty to recognize that} {It can be easily recognized that} the presented methods of solving a sign problem in numerical simulations are not limited to solutions of many-body quantum systems. Rather, the essential ideas generalize straightforwardly to applications and problems in other areas of science, technology, and engineering, including but not limited to statistics and optimizations, where a high-dimensional density is involved and/or to be simulated, based on which a linear or quadratic functional or another mathematical form is defined as an expectation value, and a computational task is to derive a good numerical estimate to the expectation value. Familiar examples include the theory and applications of Markov random fields, Gibbs distributions, and Gibbs samplings in probability theory, Bayesian statistics, image processing, and stochastic optimizations \cite{Geman84,Liu01,Li09,Bremaud17}, which have demonstrated and enjoyed tremendous successes in providing incredibly efficient numerical methods and algorithms solving many problems of great practical importance, especially when the underlying density is everywhere non-negative and no sign problem is there to spoil numerical integration.

However, in many other applications, a high-dimensional signed density is involved and a sign problem is there to impose great difficulty. That is where the present methods generalize and help. It is often the case, although not always nor does have to be so, that a Gibbs operator, a transition matrix, or another linear operator exists and is explicitly identified, broadly referred to as a {\em transition density matrix}, which has the high-dimensional (signed) density as a {\em stationary density vector} (SDV), in the sense that an operation of said transition density matrix on said stationary density vector produces a (signed) density that is substantially the same as the stationary density vector. There may or may not be a Hamiltonian, or an energy operator, or an energy functional explicitly identified as a generator of the transition density matrix. Even the transition density matrix may be implicit. Regardless, the essential ideas of the present methods are to identify and use an {\em SFF property} in association with a collection of low-dimensional transition density matrices (LD-TDMs) related to a high-dimensional signed density of interest, where each LD-TDM in turn is associated with a plurality of efficiently computable low-dimensional signed densities as low-dimensional stationary density vectors (LD-SDVs) in the sense that an operation of said LD-TDM on each one of the LD-SDVs associated with it produces a low-dimensional signed density that is substantially the same as the one being operated upon, furthermore, said high-dimensional signed density is one stationary density vector associated with all of the LD-TDMs that are related to it. Once such an SFF property is identified and established, it follows from substantially the same mathematical theory, derivations, and proofs as presented in the above, that the high-dimensional signed density of interest can be efficiently simulated on a classical computer, by way of example but no means of limitation, using an MCMC procedure of Gibbs samplings following Markovian state transition rules related to the LD-TDMs.

\iftoggle{ForUSPTO} {
\subsection*{\bf Bi-Fermion Rebits} \label{BiFermionRebits}
} {
\section{Bi-Fermion Rebits} \label{BiFermionRebits}
}
With their BPP solvability established, it is natural to ask about the generality of SFF-FS or SFF-EB systems, whether SFF-FS or SFF-EB Hamiltonians possess sufficient computational power to simulate other quantum systems in general. The answer is yes, as may be greatly surprising to many. To proceed, let's first construct a so-called bi-fermion system as a rebit and devise a
universal set of FBM interactions and tensor polynomials, which can be combined to effect a {\em designer Hamiltonian} that is either SFF-FS or SFF-EB, or even both SFF-FS and SFF-EB, thus dubbed SFF doubly universal (SFF-DU), which has a unique and polynomially gapped ground state encoding an entire history of quantum state evolution as a result of executing a quantum algorithm. A designer Hamiltonian is so named because it is specially programmed and tailor-made to have a ground state encoding a quantum computation, in particular, the solution to a prescribed computational problem that an accordingly designed quantum algorithm is able to solve.

\iftoggle{ForUSPTO} {
} {
\vspace{-2.0ex}
\begin{figure}[ht]
\centering
\includegraphics[width=0.6\textwidth]{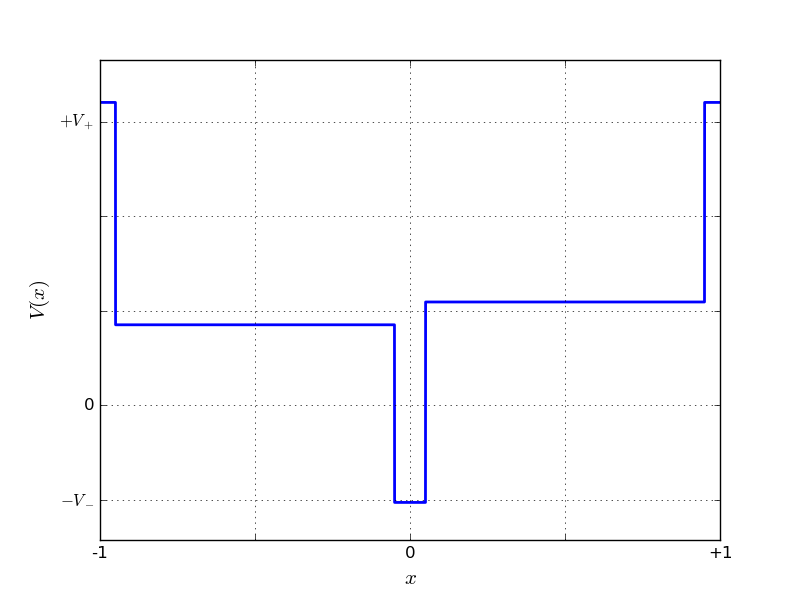}
\caption{A three-well potential on a circle.}
\label{ThreeWellVx}
\end{figure}
\vspace{2.0ex}
}

Let $\mathbb{T}\defeq\mathbb{R}/2\mathbb{Z}$ represent the one-dimensional compact and connected Riemannian Lie group, known as the circle group, which is algebraically a quotient of the additive group of real numbers $(\mathbb{R},+)$ modulo the lattice of even integers $(2\mathbb{Z},+)$, with the $\!\!\pmod*{2}$ operation producing a real-valued remainder within $[-1,1)$, and geometrically a Riemannian manifold endowed with a Riemannian metric defining a distance $d:\mathbb{T}\times\mathbb{T}\mapsto\mathbb{R}_{\ge 0}$ such that, each point on $\mathbb{T}$ is indexed by, namely, bijectively mapped to, a number in the interval $[-1,1)\subset\mathbb{R}$, as will be done throughout this \iftoggle{ForUSPTO} {specification} {presentation}, and $\forall x,y\in\mathbb{T}$, $d(x,y)=\min\{|x-y+2n|:n\in\mathbb{Z}\}$, with $|r|$ representing the absolute value of $r\in\mathbb{R}$. It is obvious that $d(-1,1)=0$, with the consequence of $x=\pm 1$ being identified to represent the same point on $\mathbb{T}$. This coordinate scheme is very convenient, although it deviates from the standard method of using multiple overlapping charts on open sets. Consider two non-interacting identical spinless fermions moving on $\mathbb{T}$ under a three-well potential $V(x)$ as depicted in Fig.\;\ref{ThreeWellVx} and specified analytically as
\begin{equation}
V(x)=\left\{\begin{array}{rl}
\vspace{0.75ex}
-V_-=-C_-V_0,&d(x,0)\le a<\textstyle{\half},\\
\vspace{0.75ex}
V_+=+C_+V_0,&d(x,0)\ge 1-a,\\
\vspace{0.75ex}
v_0,&d\!\left(x,\textstyle{+\half}\right)<\textstyle{\half}-a,\\
-v_0,&d\!\left(x,\textstyle{-\half}\right)<\textstyle{\half}-a,\\
\end{array}\right.
\label{BiFermionV}
\end{equation}
where $C_{\pm}>0$ are fixed constants, $V_0\gg aV_0\gg 1\gg a>0$, $|v_0|\ll 1$, although in the contexts of computations and simulations, $V_0$, $aV_0$, $a^{\mathsmaller{-}1}$, and $|v_0|^{\mathsmaller{-}1}$ all have to be polynomially bounded in terms of the problem size.

For reasons that may become more convincing below, the potential well in the region $\{x:d(x,0)\le a\}$ is referred to as the {\em Pauli well}, while the two potential wells at $\{x:d(x,-1/2)<1/2-a\}$ and $\{x:d(x,1/2)<1/2-a\}$ are called the {\em left and right logic wells} respectively. The Pauli well is said to form a {\em $\pi$ junction}, while the potential barrier in the region $\{x:d(x,0)\ge 1-a\}$ is said to form a {\em normal junction}, between the two regions of logic wells. In a natural unit system with $\hbar=m=1$, the single-particle Schr\"odinger equation reads
\begin{equation}
-\textstyle{\half}\psi''(x)+V(x)\psi(x)=E\psi(x), \label{SingleFermionSchr}
\end{equation}
which is already semi-analytically solvable by matching sinusoidal or hyper-sinusoidal solutions at edges of potential jumps \cite{Schiff68,Flugge94,Messiah99}, but becomes especially easy to analyze in the asymptotic limit of $a\rightarrow 0$, $V_0\rightarrow\infty$, while $\gamma_0\defeq aV_0$ remains constant, namely, when the potential barrier or well approaches a Dirac delta function $\gamma_0\delta(x-x_0)$, $\gamma_0\in\mathbb{R}$, $x_0\in\mathbb{T}$, and the effect of such a Dirac $\delta$ potential reduces to a peculiar boundary condition, that is, while the wavefunction remains continuous at $x=x_0$, its derivative undergoes a definitive discontinuity as $\lim_{\epsilon\mathsmaller{\rightarrow}0,\epsilon>0}[\psi'(x_0+\epsilon)-\psi'(x_0-\epsilon)]=2\gamma_0\psi(x_0)$ \cite{Flugge94,Goldstein94,Belloni14}.

\iftoggle{ForUSPTO} {
} {
\vspace{-2.0ex}
\begin{figure}[ht]
\centering
\includegraphics[width=0.6\textwidth]{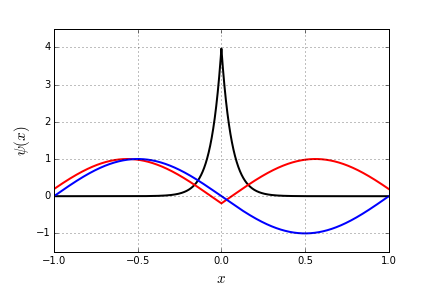}
\caption{The eigenstates $\psi_{\mathsmaller{P}}$, $\psi_+$, and $\psi_-$ in black, red, and blue respectively.}
\label{ThreeWellPsiEvenOdd}
\end{figure}
\vspace{2.0ex}

\vspace{-2.0ex}
\begin{figure}[ht]
\centering
\includegraphics[width=0.6\textwidth]{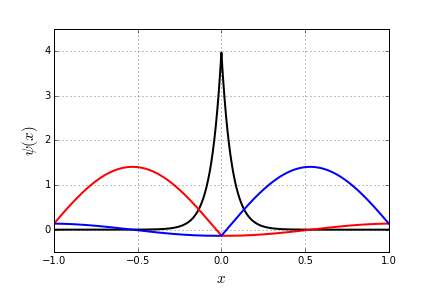}
\caption{The eigenstates $\psi_{\mathsmaller{P}}$, $\psi_{\mathsmaller{L}}$, and $\psi_{\mathsmaller{R}}$ in black, red, and blue respectively.}
\label{ThreeWellPsiPsiLPsiR}
\end{figure}
\vspace{2.0ex}
}

In the simple case with $v_0=0$, $\gamma_+\defeq\lim_{a\mathsmaller{\rightarrow}0,V_+\mathsmaller{\rightarrow}\infty}aV_+=C_+\gamma_0=\Theta(\gamma_0)$, $\gamma_-\defeq\lim_{a\mathsmaller{\rightarrow}0,V_-\mathsmaller{\rightarrow}\infty}aV_-=C_-\gamma_0=\Theta(\gamma_0)$, $\gamma_0\gg 1$, the $\delta$ potential Pauli well hosts a deeply bound state $\psi_{\mathsmaller{P}}(x)=\gamma_-^{\mathsmaller{1/2}}e^{\mathsmaller{-}\gamma_-d(x,0)}$ with energy $E_0=-\gamma_-^2/2$, which is tightly localized around the point $x=0$, while the two logic wells mostly accommodate the first two excited states $\psi_+(x)=\sin k_+[d(x,0)-\alpha_+]$ of even parity, and $\psi_-(x)={-}\sin\pi x$ of odd parity, corresponding to eigen energies $E_+=k_+^2/2$ and $E_-=\pi^2/2$ respectively. The $\psi_-$ state vanishes at $x=0$ as well as $x=\pm 1$, and assumes an eigen energy that is independent of the $\delta$ potentials. For the $\psi_+$ state, the peculiar boundary conditions due to the $\delta$ potentials determine the parameters as $\alpha_+=\gamma_-^{\mathsmaller{-}1}+O\!\left(\gamma_0^{\mathsmaller{-}3}\right)$, $k_+=\pi(1+\gamma_-^{\mathsmaller{-}1}-\gamma_+^{\mathsmaller{-}1})+O\!\left(\gamma_0^{\mathsmaller{-}2}\right)$, such that $E_+=\pi^2/2+\pi^2(\gamma_-^{\mathsmaller{-}1}-\gamma_+^{\mathsmaller{-}1})+O\!\left(\gamma_0^{\mathsmaller{-}2}\right)$. When $\gamma_+=\gamma_-=\gamma_0$, the balance between the positive and negative $\delta$ potentials leads to energy degeneracy at $E_+=E_-=\pi^2/2$. When the $\delta$ potentials are off balance, the even and odd states are split by an energy gap $E_+-E_-=\pi^2(\gamma_-^{\mathsmaller{-}1}-\gamma_+^{\mathsmaller{-}1})+O\!\left(\gamma_0^{\mathsmaller{-}2}\right)$. In a unitarily transformed basis consisting of states
\begin{align}
\psi_{\mathsmaller{L}}(x) \,\defeq\,& \textstyle{\frac{1}{\sqrt{2}}}\!\left[\psi_+(x)+\psi_-(x)\right]\,=\,\sqrt{2}\sin[\pi d(x,0)]\Iver{d(x,-1/2)<1/2} + O\!\left(\gamma_0^{\mathsmaller{-}1}\right)\!, \\[0.75ex]
\psi_{\mathsmaller{R}}(x) \,\defeq\,& \textstyle{\frac{1}{\sqrt{2}}}\!\left[\psi_+(x)-\psi_-(x)\right]\,=\,\sqrt{2}\sin[\pi d(x,0)]
\Iver{d(x,+1/2)<1/2} + O\!\left(\gamma_0^{\mathsmaller{-}1}\right)\!,
\end{align}
with $\Iver{\,\cdot\,}$ being the Iverson bracket, the wave amplitude is mostly localized in the left and right logic wells respectively. The combined effect of the $\delta$ potentials may be interpreted as a coupling via tunneling between the $\psi_{\mathsmaller{L}}$ and $\psi_{\mathsmaller{R}}$ states in conjunction with an overall energy shift, altogether represented by a Hamiltonian $\Gamma+\Gamma|\psi_{\mathsmaller{L}}\rangle\langle\psi_{\mathsmaller{R}}|+\Gamma|\psi_{\mathsmaller{R}}\rangle\langle\psi_{\mathsmaller{L}}|$, with the coupling strength $\Gamma\defeq\pi^2(\gamma_-^{\mathsmaller{-}1}-\gamma_+^{\mathsmaller{-}1})/2+O\!\left(\gamma_0^{\mathsmaller{-}2}\right)$, whereas a relative potential shift $2v_0\neq 0$ between the left and right logic wells can be considered simply as to offset the eigen energies of the $\psi_{\mathsmaller{L}}$ and $\psi_{\mathsmaller{R}}$ states. By the same token, in the $\{\psi_+,\psi_-\}$ basis, the  combined effect of the $\delta$ potentials may be understood as to offset the eigen energies of the $\psi_+$ and $\psi_-$ states, whereas a relative potential shift $2v_0\neq 0$ between the left and right logic wells can be regarded as a coupling between the $\psi_+$ and $\psi_-$ states. It is fairly clear that, when $0\neq|\gamma_+-\gamma_-|\ll\gamma_0$ and $v_0=0$, the second and third eigen states remain largely between $\psi_+$ and $\psi_-$, with a small deviation in the wavefunctions that is $O\!\left(|\gamma_+-\gamma_-|/\gamma_0^2\right)$, although the state degeneracy will be lifted with an energy gap that is $\Theta\!\left(|\gamma_+-\gamma_-|/\gamma_0\right)$. Similarly, it can be easily verified that, when $\gamma_+=\gamma_-=\gamma_0$ and $0\neq|v_0|\ll 1\ll\gamma_0$, the second and third eigen states remain largely between $\psi_{\mathsmaller{L}}$ and $\psi_{\mathsmaller{R}}$, with a small deviation in the wavefunctions that is $O\!\left(|v_0|/\gamma_0\right)$, although the state degeneracy will be lifted with an energy gap that is $\Theta\!\left(|v_0|\right)$. Fig.\;\ref{ThreeWellPsiEvenOdd} shows the first three single-particle eigenstates $\psi_{\mathsmaller{P}}$, $\psi_+$, and $\psi_-$ on the circle group with a balanced three-well potential, namely, when $v_0=0$, $\gamma_+=\gamma_-=\gamma_0$, where $\psi_{\mathsmaller{P}}(x)$, $\psi_+(x)$, and $\psi_-(x)$ are plotted in black, red, and blue colors respectively. In contrast, Fig.\;\ref{ThreeWellPsiPsiLPsiR} shows the eigenstates $\psi_{\mathsmaller{P}}$, $\psi_{\mathsmaller{L}}$, and $\psi_{\mathsmaller{R}}$ as a different representation. A relatively small value of $\gamma_0=16$ is chosen only to reduce strains in graphic displaying and viewing, whereas in real applications, $\gamma_0$ is usually much larger.

\iftoggle{ForUSPTO} {
} {
\vspace{-2.0ex}
\begin{figure}[ht]
\centering
\begin{tabular}{ccc}
\subfloat[Subfigure 1 list of figures text][$\Phi_+(x_1,x_2)$]{\includegraphics[width=0.54\textwidth]
{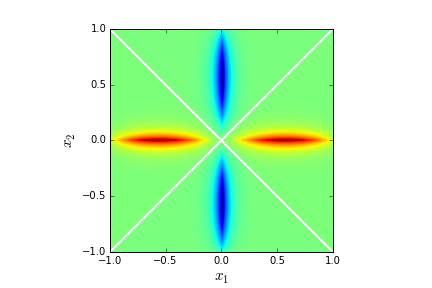}\label{fig:FourPhis1}} & \hspace{-7em} &
\subfloat[Subfigure 2 list of figures text][$\Phi_-(x_1,x_2)$]{\includegraphics[width=0.54\textwidth]
{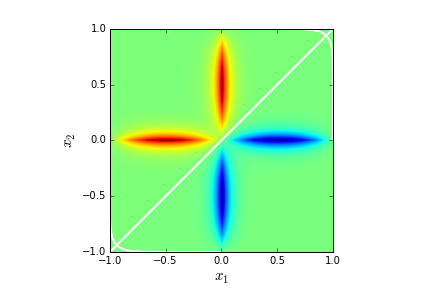}\label{fig:FourPhis2}} \\
\subfloat[Subfigure 3 list of figures text][$\Phi_{\mathsmaller{L}}(x_1,x_2)$]{\includegraphics[width=0.54\textwidth]
{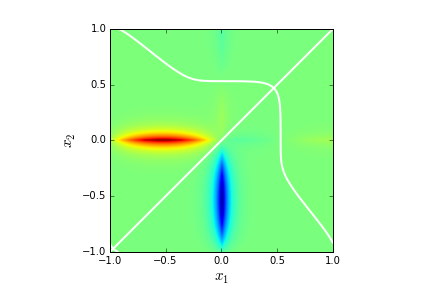}\label{fig:FourPhis3}} & \hspace{-7em} &
\subfloat[Subfigure 4 list of figures text][$\Phi_{\mathsmaller{R}}(x_1,x_2)$]{\includegraphics[width=0.54\textwidth]
{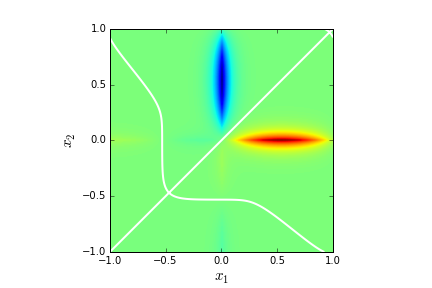}\label{fig:FourPhis4}}
\end{tabular}
\caption{The bi-fermion wavefunctions $\Phi_+(x_1,x_2)$, $\Phi_-(x_1,x_2)$, $\Phi_{\mathsmaller{L}}(x_1,x_2)$, $\Phi_{\mathsmaller{R}}(x_1,x_2)$.}
\label{FourPhisAndNodalCurves}
\end{figure}
\vspace{2.0ex}
}

With $\alpha\defeq\alpha(\gamma_0)\defeq 2\gamma_0^{{-}1}\log\gamma_0$, define a closed circular arc $\mathbb{T}^-_{\alpha}\defeq\{x\in\mathbb{T}:d(x,0)\le\alpha\}$ and its open complement $\mathbb{T}^+_{\alpha}\defeq\{x\in\mathbb{T}:d(x,0)>\alpha\}$, corresponding to the regions that are nearby and far away from the Pauli well respectively. Consider a probability of collision $P(\psi_{\pm},\psi_{\mathsmaller{P}};\mathbb{T}^-_{\alpha})$ between the $\psi_{\pm}$ and the $\psi_{\mathsmaller{P}}$ orbits with respect to $\mathbb{T}^-_{\alpha}$, which is defined and bounded as
\begin{align}
P(\psi_{\pm},\psi_{\mathsmaller{P}};\mathbb{T}^-_{\alpha}) \,\defeq\,& {\textstyle{\scalebox{1.2}{$\int$}_{\!\mathbb{T}^-_{\alpha}}}}|\psi_{\pm}(x)|^2dx \, {\textstyle{\scalebox{1.2}{$\int$}_{\!\mathbb{T}^-_{\alpha}}}}|\psi_{\mathsmaller{P}}(x)|^2dx \,+\, {\textstyle{\scalebox{1.2}{$\int$}_{\!\mathbb{T}^+_{\alpha}}}}|\psi_{\pm}(x)|^2dx \, {\textstyle{\scalebox{1.2}{$\int$}_{\!\mathbb{T}^+_{\alpha}}}}|\psi_{\mathsmaller{P}}(x)|^2dx \nonumber \\[0.75ex]
\le\;\,& {\textstyle{\scalebox{1.2}{$\int$}_{\!\mathbb{T}^-_{\alpha}}}}|\psi_{\pm}(x)|^2dx \,+\, {\textstyle{\scalebox{1.2}{$\int$}_{\!\mathbb{T}^+_{\alpha}}}}|\psi_{\mathsmaller{P}}(x)|^2dx \label{OrbitCollisionProb} \\[0.75ex]
=\;\,& O\!\left(\gamma_0^{{-}3}\log^3\gamma_0\right) \,+\, O\!\left(\gamma_0^{{-}4}\right) \,=\, O\!\left(\gamma_0^{{-}3}\log^3\gamma_0\right)\!. \nonumber
\end{align}
Also, define a closed semicircle $\mathbb{T}_-\defeq\{x\in\mathbb{T}:d(x,{-}1/2)\le 1/2\}$ and and its open complement $\mathbb{T}_+\defeq\{x\in\mathbb{T}:d(x,1/2)<1/2\}$, that correspond to the left and the right logic wells respectively. Consider the overlap integral between $\psi_+$ and $\psi_-$ within each of the two semicircles,
\begin{equation}
2\,{\textstyle{\scalebox{1.2}{$\int$}_{\!\mathbb{T}_-}}}\psi_+(x)\psi_-(x)dx \,=\, {-}2\,{\textstyle{\scalebox{1.2}{$\int$}_{\!\mathbb{T}_+}}}\psi_+(x)\psi_-(x)dx \,=\, \cos\pi\alpha_+ \,=\, 1-O\!\left(\gamma_0^{{-}2}\right)\!. \label{PsiPPsiMOverlap}
\end{equation}
By choosing a sufficiently large value for $\gamma_0$, the probability of collision $P(\psi_{\pm},\psi_{\mathsmaller{P}};\mathbb{T}^-_{\alpha})$ can be made $O\!\left(\gamma_0^{{-}3}\log^3\gamma_0\right)$ infinitesimal, while the overlap integral between $\psi_+$ and $\psi_-$ within each of the two logic wells can be made $O\!\left(\gamma_0^{{-}2}\right)$ close to unity or perfection. Finally, define a closed semicircle $\mathbb{T}_0\defeq\{x\in\mathbb{T}:d(x,0)\le 1/2\}$ and its open complement $\mathbb{T}_1\defeq\{x\in\mathbb{T}:d(x,0)>1/2\}$, that enclose the Pauli well and the potential barrier respectively. It is interesting and potentially useful to note that, the two eigenstates $\psi_+$ and $\psi_-$ have an appreciable difference in their cumulative probability distributions between the two semicircles $\mathbb{T}_0$ and $\mathbb{T}_1$. While $\psi_-$ is exactly equidistributed between the two semicircles, that is, $\int_{\mathbb{T}_0}|\psi_-(x)|^2dx=\int_{\mathbb{T}_1}|\psi_-(x)|^2dx=1/2$, the $\psi_+$ state is more concentrated in $\mathbb{T}_1$ than in $\mathbb{T}_0$, and specifically,
\begin{equation}
{\textstyle{\scalebox{1.2}{$\int$}_{\!\mathbb{T}_1}}}|\psi_+(x)|^2dx - {\textstyle{\scalebox{1.2}{$\half$}}} \,=\, {\textstyle{\scalebox{1.2}{$\half$}}} - {\textstyle{\scalebox{1.2}{$\int$}_{\!\mathbb{T}_0}}}|\psi_+(x)|^2dx \,=\, \pi^{{-}1}\sin 2\pi\alpha_+ \,=\, 2\gamma_0^{{-}1}+O\!\left(\gamma_0^{{-}2}\right)\!. \label{PsiPlusUnevenDistribution}
\end{equation}
It is most worth noting that the disparity between the cumulative probabilities $\int_{\mathbb{T}_0}|\psi_+(x)|^2dx$ and $\int_{\mathbb{T}_1}|\psi_+(x)|^2dx$ scales asymptotically as $\Theta\!\left(\gamma_0^{{-}1}\right)$ when $\gamma_0\gg 1$, as opposed to the $O\!\left(\gamma_0^{{-}3}\log^3\gamma_0\right)$ probability of collision between the $\psi_{\pm}$ and $\psi_{\mathsmaller{P}}$ orbits, and the $O\!\left(\gamma_0^{{-}2}\right)$ deviation from perfect overlapping between $\psi_+$ and $\psi_+$ within each of the logic wells.

To employ the space spanned by $\{\psi_+,\psi_-\}$ or $\{\psi_{\mathsmaller{L}},\psi_{\mathsmaller{R}}\}$ for GSQC, two identical fermions are employed to invoke the Pauli exclusion principle, such that one fermion fills and blocks the deeply bound $\psi_{\mathsmaller{P}}$ state, while the other particle lives in the two-dimensional space ${\rm span}\{\psi_+,\psi_-\}$ to implement a rebit. The configuration is an example of Pauli blockade. Such construct of two identical fermions moving in a three-well potential on a circle is called a bi-fermion, whose two lowest-lying energy states form one rebit's worth of computational space, in which universal GSQC can be realized by designer Hamiltonians with $\calC$-diagonal potential perturbations, where $\calC=\mathbb{T}^2$ for a single bi-fermion. More specifically, a bi-fermion has the nominal Hamiltonian
\begin{equation}
H_{\mathsmaller{\rm BF},0} \,\defeq\, \frac{\gamma_0^2-\pi^2}{2}+\sum_{k=1}^2\left[-{\textstyle{\half}}\partial_k^2-\gamma_0\delta(x_k)+\gamma_0\delta(x_k+1)\right]\!, \label{HBF0}
\end{equation}
$\partial_k\defeq\partial/\partial x_k$, $\partial_k^2\defeq\partial^2\!/\partial x_k^2$, $\forall k\in\{1,2\}$, with two exchange antisymmetric states
\begin{align}
\Phi_+(x_1,x_2) \,\defeq\,& \textstyle{\frac{1}{\sqrt{2}}}\left[\psi_+(x_1)\psi_{\mathsmaller{P}}(x_2)-\psi_+(x_2)\psi_{\mathsmaller{P}}(x_1)\right]\!, \label{PhiPdefi} \\[0.75ex]
\Phi_-(x_1,x_2) \,\defeq\,& \textstyle{\frac{1}{\sqrt{2}}}\left[\psi_-(x_1)\psi_{\mathsmaller{P}}(x_2)-\psi_-(x_2)\psi_{\mathsmaller{P}}(x_1)\right]\!, \label{PhiMdefi}
\end{align}
that are degenerate at the lowest (zero-valued) energy and span a rebit manifold. All excited states lie above with an energy gap at least $\Omega(1)$. The unitarily transformed states
\begin{equation}
\Phi_{\mathsmaller{L}}\defeq\textstyle{\frac{1}{\sqrt{2}}}\left(\Phi_++\Phi_-\right),\;\;\Phi_{\mathsmaller{R}}\defeq\textstyle{\frac{1}{\sqrt{2}}}\left(\Phi_+-\Phi_-\right) \label{PhiLRdefi}
\end{equation}
have one fermion trapped by the Pauli well and tightly localized around $x=0$, while the other fermion mostly localized in either the left or the right logic well, although the spatial localization is not quite exact. \iftoggle{ForUSPTO} {Images of the bi-fermion wavefunctions $\Phi_+(x_1,x_2)$, $\Phi_-(x_1,x_2)$, $\Phi_{\mathsmaller{L}}(x_1,x_2)$, and $\Phi_{\mathsmaller{R}}(x_1,x_2)$ are shown in Fig.\;\ref{FourPhisAndNodalCurves} via contour plots, where dashed lines labeled by the numbers $0.1$, $0.3$, $0.7$ indicate where a wavefunction is positive-valued at levels equal to $0.1$, $0.3$, $0.7$ respectively, while dark solid lines labeled by the numbers $-0.1$, $-0.3$, $-0.7$ signify where a wavefunction is negative-valued at levels equal to $-0.1$, $-0.3$, $-0.7$ respectively. In each image, there are gray solid lines labeled by the number $0.0$ to indicate nodal surfaces, which are actually curves in the present context.} {Images of the bi-fermion wavefunctions $\Phi_+(x_1,x_2)$, $\Phi_-(x_1,x_2)$, $\Phi_{\mathsmaller{L}}(x_1,x_2)$, and $\Phi_{\mathsmaller{R}}(x_1,x_2)$ are shown in Fig.\;\ref{FourPhisAndNodalCurves} via scaled colors, where the color of image pixels extends from dark blue to dark red in accordance with the wave amplitude varying from maximally negative to maximally positive values. In each image, there are white lines indicating nodal surfaces, which are actually curves in the present context.} It is noted that for all of the bi-fermion wavefunctions $\{\Phi_i(x_1,x_2)\}$, $i\in\{+,-,L,R\}$, the two nodal curves are always orthogonal to each other when they intersect, as it follows from the fact that $(\partial_1^2+\partial_2^2)\Phi_i(x_1,x_2)=0$ on the nodal curves, $\forall i\in\{+,-,L,R\}$ \cite{Ceperley91}. It is also noted that the nodal curves always have the bi-fermion configuration space divided into precisely one positive nodal cell $\calN^+(\Phi_i(x_1,x_2))$ and one negative nodal cell $\calN^-(\Phi_i(x_1,x_2))$, $i\in\{+,-,L,R\}$, as dictated by {\myCorollary} \ref{OnePositiveNodalCell} of \iftoggle{ForUSPTO}{Auxiliary Utility}{Lemma} \ref{TilingProperty}.

For the most convenient implementations of logic operations among multiple rebits, it may be preferable to work with basis states that are strictly localized in the logic wells. Again set $\alpha=2\gamma_0^{{-}1}\log\gamma_0$, define three well-localized single-particle states
\begin{align}
|P\rangle \,\defeq\,& \mathlarger{\mathlarger{[}}\gamma_0+O\!\left(\gamma_0^{{-}4}\right)\mathlarger{\mathlarger{]}}^{1/2}e^{\mathsmaller{-}\gamma_0d(x,0)} \Iver{d(x,0)\le\alpha}, \\[0.75ex]
|L\rangle \,\defeq\,& \mathlarger{\mathlarger{[}}\sqrt{2}+O\!\left(\alpha^3\right)\mathlarger{\mathlarger{]}}\sin\pi[d(x,0)-(2\gamma_0)^{{-}1}] \Iver{d(x,-(1+\alpha)/2)<(1-\alpha)/2}, \\[0.75ex]
|R\rangle \,\defeq\,& \mathlarger{\mathlarger{[}}\sqrt{2}+O\!\left(\alpha^3\right)\mathlarger{\mathlarger{]}}\sin\pi[d(x,0)-(2\gamma_0)^{{-}1}] \Iver{d(x,+(1+\alpha)/2)<(1-\alpha)/2},
\end{align}
with $\Iver{\,\cdot\,}$ being the Iverson bracket, such that
\begin{align}
\langle\psi_{\mathsmaller{P}}|P\rangle \,=\, 1+O\!\left(\gamma_0^{{-}4}\right)\!, \\[0.75ex]
\textstyle{\frac{1}{\sqrt{2}}}\langle\psi_{\pm}|\left(|L\rangle\pm|R\rangle\right) \,=\, 1+O\!\left(\gamma_0^{{-}2}\right)\!, \\[0.75ex]
\langle\psi_+|\left(|L\rangle+|R\rangle\right) - \langle\psi_-|\left(|L\rangle-|R\rangle\right) \,=\, O\!\left(\gamma_0^{{-}3}\right)\!, \label{psipmLRpmAgreeWell}
\end{align}
where the $\pi(2\gamma_0)^{{-}1}$ phase shift in the wavefunction $\sin\pi[d(x,0)-(2\gamma_0)^{{-}1}]$ is introduced deliberately to make equation (\ref{psipmLRpmAgreeWell}) true. Then designate the bi-fermion states
\begin{align}
& |\!\downarrow\rangle \,\defeq\, \textstyle{\frac{1}{\sqrt{2}}}\left(|L_1\rangle|P_2\rangle-|L_2\rangle|P_1\rangle\right)
\,\defeq\, \textstyle{\frac{1}{\sqrt{2}}}\left(|L_1P_2\rangle-|L_2P_1\rangle\right) \nonumber \\[0.75ex]
& \!\defeq \sin\pi[d(x_1,0)-(2\gamma_0)^{{-}1}] \Iver{d(x_1,-(1+\alpha)/2)<(1-\alpha)/2}
\gamma_0^{1/2}e^{\mathsmaller{-}\gamma_0d(x_2,0)} \Iver{d(x_2,0)\le\alpha} \nonumber \\[0.75ex]
& \!-\; \sin\pi[d(x_2,0)-(2\gamma_0)^{{-}1}] \Iver{d(x_2,-(1+\alpha)/2)<(1-\alpha)/2}
\gamma_0^{1/2}e^{\mathsmaller{-}\gamma_0d(x_1,0)} \Iver{d(x_1,0)\le\alpha} \nonumber \\[0.75ex]
& \!+\; O\!\left(\alpha^3\right), \label{BiFermionDownDef}
\end{align}
and
\begin{align}
& |\!\uparrow\rangle \,\defeq\, \textstyle{\frac{1}{\sqrt{2}}}\left(|R_1\rangle|P_2\rangle-|R_2\rangle|P_1\rangle\right)
\,\defeq\, \textstyle{\frac{1}{\sqrt{2}}}\left(|R_1P_2\rangle-|R_2P_1\rangle\right) \nonumber \\[0.75ex]
& \!\defeq \sin\pi[d(x_1,0)-(2\gamma_0)^{{-}1}] \Iver{d(x_1,+(1+\alpha)/2)<(1-\alpha)/2}
\gamma_0^{1/2}e^{\mathsmaller{-}\gamma_0d(x_2,0)} \Iver{d(x_2,0)\le\alpha} \nonumber \\[0.75ex]
& \!-\; \sin\pi[d(x_2,0)-(2\gamma_0)^{{-}1}] \Iver{d(x_2,+(1+\alpha)/2)<(1-\alpha)/2}
\gamma_0^{1/2}e^{\mathsmaller{-}\gamma_0d(x_1,0)} \Iver{d(x_1,0)\le\alpha} \nonumber \\[0.75ex]
& \!+\; O\!\left(\alpha^3\right), \label{BiFermionUpDef}
\end{align}
which can be equivalently written in the second-quantized representation as
\begin{align}
|\!\downarrow\rangle \,=\,& a_{\mathsmaller{L}}^+a_{\mathsmaller{P}}^+|{\sf vac}\rangle, \\[0.75ex]
|\!\uparrow\rangle \,=\,& a_{\mathsmaller{R}}^+a_{\mathsmaller{P}}^+|{\sf vac}\rangle,
\end{align}
where $|{\sf vac}\rangle$ denotes the vacuum state, $a_{\mathsmaller{P}}^+$, $a_{\mathsmaller{L}}^+$, and $a_{\mathsmaller{R}}^+$ are fermion creation operators for the $|P\rangle$, $|L\rangle$, and $|R\rangle$ singe-particle orbits respectively. At times, it is also convenient to use the unitarily transformed basis with states
\begin{align}
|+\rangle \,\defeq\,& \textstyle{\frac{1}{\sqrt{2}}}(|\!\downarrow\rangle+|\!\uparrow\rangle) \,=\, \textstyle{\frac{1}{\sqrt{2}}}(a_{\mathsmaller{L}}^++a_{\mathsmaller{R}}^+)a_{\mathsmaller{P}}^+|{\sf vac}\rangle, \label{BiFermionPlusDef} \\[0.75ex]
|-\rangle \,\defeq\,& \textstyle{\frac{1}{\sqrt{2}}}(|\!\downarrow\rangle-|\!\uparrow\rangle) \,=\, \textstyle{\frac{1}{\sqrt{2}}}(a_{\mathsmaller{L}}^+-a_{\mathsmaller{R}}^+)a_{\mathsmaller{P}}^+|{\sf vac}\rangle. \label{BiFermionMinusDef}
\end{align}
To make definitive and unambiguous references in the following, the states $|\!\downarrow\rangle$ and $|\!\uparrow\rangle$ or $|+\rangle$ and $|-\rangle$ are called the {\em effective computational basis states}, and ${\rm span}\{|\!\downarrow\rangle,|\!\uparrow\rangle\}={\rm span}\{|+\rangle,|-\rangle\}$ is called the {\em computational bi-fermion space}. By contrast, the states $\Phi_{\mathsmaller{L}}$ and $\Phi_{\mathsmaller{R}}$ or $\Phi_+$ and $\Phi_-$ are called the {\em physical bi-fermion basis states}, and ${\rm span}\{\Phi_{\mathsmaller{L}},\Phi_{\mathsmaller{R}}\}={\rm span}\{\Phi_+,\Phi_-\}$ is called the {\em physical bi-fermion space}. The two spaces and their corresponding basis states are nearly the same, as the overlaps $\langle\downarrow\!|\Phi_{\mathsmaller{L}}\rangle$, $\langle\uparrow\!|\Phi_{\mathsmaller{R}}\rangle$, $\langle+|\Phi_+\rangle$, and $\langle-|\Phi_-\rangle$ are all valued within $O\!\left(\gamma_0^{{-}2}\right)$ to $1$.

Interestingly, just for a little digression, it is noted that in the special case of a balanced three-well potential with $v_0=0$ and $\gamma_+=\gamma_-=\gamma_0\gg 1$ on the circle group, apart from the deeply bound state $\psi_{\mathsmaller{P}}(x)=\gamma_0^{1/2}e^{\mathsmaller{-}\gamma_0d(x,0)}$ with energy $E_0=-\gamma_0^2/2$, all the other low-lying single-particle energy levels $E_n=n^2\pi^2/2$, $n\in\mathbb{N}$, $n\ll\gamma_0$ are two-fold degenerate, where the two degenerate states per each energy level $E_n$, $n\ge 1$ preferably chosen to have a definitive parity, either even or odd, namely, $\psi_n^+=\sin n\pi[d(x,0)-\gamma_0^{\mathsmaller{-}1}+O\!\left(\gamma_0^{\mathsmaller{-}3}\right)]$, or $\psi_n^-={-}\sin n\pi x$. For all $n\in\mathbb{N}$, $n\ll\gamma_0$, $x=0$ and $x=\pm 1$ are the nodal points for the odd-parity states $\{\psi_n^-\}_{n}$, while $x=\pm\gamma_0^{\mathsmaller{-}1}$ are approximately the nodal points for the even-parity states $\{\psi_n^+\}_{n}$, with an extremely small error $O\!\left(\gamma_0^{\mathsmaller{-}3}\right)$ when $\gamma_0\gg 1$. A bi-fermion, with two non-interacting identical spinless fermions moving in the balanced three-well potential on the circle group, will have one fermion filling the $\psi_{\mathsmaller{P}}$ state and another fermion occupying one of the $\{\psi_n^{\pm}\}_n$ orbits, so that the low-lying energy levels $E_0+E_n$, $n\in\mathbb{N}$, $n\ll\gamma_0$ are all two-fold degenerate, with even-parity and odd-parity eigenstates characterized respectively by the antisymmetric two-fermion wavefunctions $\Psi_n^{\pm}(x_1,x_2)\defeq\psi_n^{\pm}(x_1)\psi_{\mathsmaller{P}}(x_2)-\psi_n^{\pm}(x_2)\psi_{\mathsmaller{P}}(x_1)$, $n\in\mathbb{N}$, where all of the even-parity states $\{\Psi_n^+\}_n$ have nodal lines defined by $x_1-x_2=0$ and $x_1+x_2=0$, $x_1,x_2\in\mathbb{T}$, as illustrated in Fig.\;\ref{fig:FourPhis1}, while all of the odd-parity states $\{\Psi_n^-\}_n$ have one nodal line defined by $x_1-x_2=0$, $x_1,x_2\in\mathbb{T}$, and another nodal curve that, when $\gamma_0\gg 1$, is well approximated by $\{x_1=\pm 1\}\cup\{x_2=\pm1\}$, as illustrated in Fig.\;\ref{fig:FourPhis2}. The open square box $\{(x_1,x_2):x_1,x_2\in(-1,1)\}\subset\mathbb{R}^2$ is divided by the two straight lines $x_1-x_2=0$ and $x_1+x_2=0$, $x_1,x_2\in(-1,1)$ into four isosceles right triangles, the interiors of which are called the E (east), N (north), W (west), and S (south) quarters respectively. Then the even-parity state $\Psi_1^+$ has the E and W quarters as well as the line boundary $\{x_1=\pm 1\}$ between them fused into a connected open set $\calN_{\mathsmaller{\rm EW}}$ as one nodal cell on $\mathbb{T}^2$, and the S and N quarters as well as the line boundary $\{x_2=\pm 1\}$ between them fused into a connected open set $\calN_{\mathsmaller{\rm SN}}$ as the other nodal cell on $\mathbb{T}^2$, with the wavefunction vanishing almost everywhere on the line $x_1+x_2=0$ and absolutely everywhere on the line $x_1-x_2=0$, $x_1,x_2\in\mathbb{T}$, while the odd-parity state $\Psi_1^-$ has the N and W quarters as well as the line boundary $\{x_1+x_2=0,x_1<0\}$ between them fused into a connected open set $\calN_{\mathsmaller{\rm NW}}$ as one nodal cell on $\mathbb{T}^2$, and the S and E quarters as well as the line boundary $\{x_1+x_2=0,x_1>0\}$ between them fused into a connected open set $\calN_{\mathsmaller{\rm SE}}$ as the other nodal cell on $\mathbb{T}^2$, with the wavefunction vanishing almost everywhere on the lines $\{x_1=\pm 1\}\cup\{x_2=\pm1\}$ and absolutely everywhere on the line $x_1-x_2=0$, $x_1,x_2\in\mathbb{T}$. Furthermore, all of the even-parity states $\{\Psi_n^+\}_{n\ll\gamma_0}$ can be recovered as eigenstates of the Dirichlet boundary-conditioned Hamiltonians $H_{\mathsmaller{\rm BF},0}|\calN_{\mathsmaller{\rm EW}}$ and $H_{\mathsmaller{\rm BF},0}|\calN_{\mathsmaller{\rm SN}}$, and by exactly the same token, all of the odd-parity states $\{\Psi_n^-\}_{n\ll\gamma_0}$ as eigenstates of $H_{\mathsmaller{\rm BF},0}|\calN_{\mathsmaller{\rm NW}}$ and $H_{\mathsmaller{\rm BF},0}|\calN_{\mathsmaller{\rm SE}}$. In other words, under a balanced three-well potential, the nodal cells $\calN_{\mathsmaller{\rm EW}}$ and $\calN_{\mathsmaller{\rm NW}}$ (or $\calN_{\mathsmaller{\rm SN}}$ and $\calN_{\mathsmaller{\rm SE}}$) are isospectral manifolds, thus, one cannot distinguish between the two differently shaped bi-fermion ``drums'' $H_{\mathsmaller{\rm BF},0}|\calN_{\mathsmaller{\rm EW}}$ and $H_{\mathsmaller{\rm BF},0}|\calN_{\mathsmaller{\rm NW}}$ (or $H_{\mathsmaller{\rm BF},0}|\calN_{\mathsmaller{\rm SN}}$ and $H_{\mathsmaller{\rm BF},0}|\calN_{\mathsmaller{\rm SE}}$), by just ``hearing'' the spectra of their vibrations \cite{Kac66,Sunada85,Gordon92,Buser94}.

Now come back to our main course of discussions regarding the use of bi-fermions for GSQC. The nominal Hamiltonian $H_{\mathsmaller{\rm BF},0}$ ensures that a bi-fermion at its lowest energy stays in the physical bi-fermion space, which is substantially the same as the computational bi-fermion space, when both parameters $\gamma_0$ and $\alpha^{{-}1}=\gamma_0/(2\log\gamma_0)$ are large yet polynomially bounded. Although containment of the ground state(s) in ${\rm span}\{|\!\downarrow\rangle,|\!\uparrow\rangle\}$ is not absolute, the rate of leakage error as measured by the trace of a so-called {\em leakage error operator} $\sfE_{\rm leak}(\gamma_0)\defeq I-|\!\downarrow\rangle\langle\downarrow\!|-|\!\uparrow\rangle\langle\uparrow\!|\in\calB(\{\Phi_{\mathsmaller{L}},\Phi_{\mathsmaller{R}}\})$, can be made arbitrarily small as $\Tr(\sfE_{\rm leak}(\gamma_0))=O\!\left(\alpha^3\right)=O\!\left(\gamma_0^{{-}3}\log^3\gamma_0\right)$, at the price of proportionally increased peak-to-valley swing of the three-well potential $V(x)$, which may entail higher costs to simulate, either quantum computationally or by means of classical Monte Carlo. Furthermore, there are quantum error suppression techniques known as subsystem, operator, or Hamiltonian encoding \cite{Kribs05a,Kribs05b,Bacon06,Jordan06}, that map intended Hamiltonians into self-adjoint operators acting on an encoded qubit space, and introduce energy penalty terms to suppress one-local, or $k$-local with a fixed small $k$, error transitions affecting the raw qubits before the error suppression encoding. Generally, quantum gate operations on the bi-fermion rebits are not always perfect in terms of the effective computational basis states, but subject to small errors, even in theory. The error rate should and can be made so sufficiently low that there is an appreciable probability, polynomially bounded from below, that a required long sequence of ground state quantum gates on the bi-fermions could be carried through error-free. Also, the well-developed theory and techniques of quantum error correction \cite{Shor95,Steane96,Steane97} can be employed to counter the detrimental effects of errors, and the celebrated quantum threshold theorem promises that, with concatenated fault-tolerant encoding for quantum error correction, an arbitrarily sized quantum circuit can be realized using error-prone quantum gates with at most polylogarithmic overheads, so long as the rate of error for each error-prone quantum gate is below a certain threshold \cite{Aharonov96,Kitaev97a,Kitaev97b,Knill98,Aliferis06}.

It may be noted that the nominal Hamiltonian $H_{\mathsmaller{\rm BF},0}$ of bi-fermions can be rendered doubly universal, and incorporated into an SFF-DU Hamiltonian for GSQC, which is amenable to probabilistic simulations using either of the two corresponding families of Monte Carlo algorithms, because $H_{\mathsmaller{\rm BF},0}$ becomes essentially bounded in a practical implementation, since the Dirac $\delta$ potentials in equation (\ref{HBF0}) are really idealizations that are meant to represent the effects of a bounded potential such as $V(\cdot)$ of equation (\ref{BiFermionV}), where the parameters $V_0>0$ and $a^{{-}1}>0$ are constants, whose values are large but still polynomially bounded. Besides, although it is necessary to keep the Pauli well narrow and deep, approximating a fairly strong Dirac delta trap so to minimize the leakage error of a bi-fermion rebit, there is no practical necessity other than convenience of mathematical analysis, to install literally a steep potential barrier around $x=\pm 1$ on the circle group of a bi-fermion. Rather, it is perfectly fine to place a relatively wide and low potential barrier, as long as it has the width and height chosen properly to be commensurate with the Delta-like potential well around $x=0$, such that the nominal Hamiltonian of the bi-fermion system defines a degenerate two-state Hilbert space approximating an effective computational basis. Along another line of reasoning that follows literally the definition of essential boundedness, the quantum physics of any bi-fermion can be described to within any desired and predetermined accuracy by selecting a finite number of continuous basis functions to span a Hilbert subspace, and projecting or restricting all relevant operators to the Hilbert subspace, such that all projected or restricted operators, including Dirac $\delta$ potentials, become bounded. {\em In any case, it is concluded that even Hamiltonians containing Dirac $\delta$ potentials can be made essentially bounded, a factoid that may be invoked repeatedly in the following without reiterating the justifications.}

On top of the nominal Hamiltonian $H_{\mathsmaller{\rm BF},0}$, the following $\calC$-diagonal potential,
\begin{equation}
V_{\mathsmaller{X},\eta}(x_1,x_2)\defeq\sum_{i=1}^2\left\{\frac{\eta\gamma_0}{1-\eta}\delta(x_i+1)-\frac{\eta\pi^2}{2\gamma_0} \Iver{d(x_i,0)\le\alpha} \right\}\!,\;-1<\eta<1, \label{VxPotentialDef}
\end{equation}
with $\Iver{\,\cdot\,}$ being the Iverson bracket, can be added to realize the Pauli $\pm\sigma^x$ operator, in that
\begin{align}
\lim_{\tau\mathsmaller{\rightarrow}+\infty}\frac{e^{\mathsmaller{-}\tau(H_{\mathsmaller{\rm BF},0}+V_{X,\eta})}}
{\Tr\left[e^{\mathsmaller{-}\tau(H_{\mathsmaller{\rm BF},0}+V_{X,\eta})}\right]}
\,&\HomoPhysL
\lim_{\tau\mathsmaller{\rightarrow}+\infty}\frac{e^{\mathsmaller{-}\tau\sigma^x\sign(\eta)}}{\Tr\left[e^{\mathsmaller{-}\tau\sigma^x\sign(\eta)}\right]} \,+\, O\!\left(\gamma_0^{\mathsmaller{-}1}\right) \\[0.75ex]
&\;=\;
\left\{\begin{array}{rl}\vspace{0.75ex}|-\rangle\langle-| \,+\, O\!\left(\gamma_0^{\mathsmaller{-}1}\right)\!,\!\!&{\rm if}\;{\eta>0},\\
I \,+\, O\!\left(\gamma_0^{\mathsmaller{-}1}\right)\!,\!\!&{\rm if}\;{\eta=0},\\
| \,+\, \rangle\langle+|+O\!\left(\gamma_0^{\mathsmaller{-}1}\right)\!,\!\!&{\rm if}\;{\eta<0},\end{array}\right. \nonumber
\end{align}
where $\eta$ is a constant in $(-1,1)$. It is quite obvious that a nonzero value of $v_0$ in equation (\ref{BiFermionV}) induces $\pm\sigma^z$ interactions. In general, an arbitrary rebit operator $\sigma^x\sin\theta+\sigma^z\cos\theta$, $\theta\in[-\pi,\pi)$ can be realized by adding a $\calC$-diagonal potential $V_{\mathsmaller{X},\eta\sin\theta}(x_1,x_2)+V_{\mathsmaller{Z},\eta\cos\theta}(x_1,x_2)$ to the nominal Hamiltonian $H_{\mathsmaller{\rm BF},0}$, with
\begin{align}
V_{\mathsmaller{Z},\eta\cos\theta}(x_1,x_2) \,\defeq\,& \frac{\eta\pi^2\cos\theta}{2\gamma_0}\sum_{i=1}^2 \Iver{d(x_i,-(1+\alpha)/2)<(1-\alpha)/2} \nonumber \\[0.75ex]
\,-\,\;& \frac{\eta\pi^2\cos\theta}{2\gamma_0}\sum_{i=1}^2 \Iver{d(x_i,+(1+\alpha)/2)<(1-\alpha)/2}, \label{VzPotentialDef}
\end{align}
in the sense that
\begin{align}
& \lim_{\tau\mathsmaller{\rightarrow}+\infty}\frac{e^{\mathsmaller{-}\tau(\sigma^x\sin\theta+\sigma^z\cos\theta)}}{\Tr\left[e^{\mathsmaller{-}\tau(\sigma^x\sin\theta+\sigma^z\cos\theta)}\right]} \,\HomoPhysL
\lim_{\tau\mathsmaller{\rightarrow}+\infty}\frac{e^{\mathsmaller{-}\tau(\sigma^x\sin\theta+\sigma^z\cos\theta)}}{\Tr\left[e^{\mathsmaller{-}\tau(\sigma^x\sin\theta+\sigma^z\cos\theta)}\right]} \,+\, O\!\left(\gamma_0^{\mathsmaller{-}1}\right) \\[0.75ex]
=\;& \left[\,\!\cos\left(\theta/2\right)|\!\uparrow\rangle-\sin\left(\theta/2\right)|\!\downarrow\rangle\,\right]\left[\,\!\cos\left(\theta/2\right)\langle\uparrow\!|-\sin\left(\theta/2\right)\langle\downarrow\!|\,\right] \,+\, O\!\left(\gamma_0^{\mathsmaller{-}1}\right)\!. \nonumber
\end{align}
That is so, because, projected onto the two-dimensional manifold spanned by the nearly degenerate two states of concern, the single-particle Hamiltonian
\begin{align}
H_{\mathsmaller{\rm SP}} \,\defeq& -\half\frac{d^2}{dx^2}-\gamma_0\delta(x)+\gamma_0\delta(x+1)-\frac{\pi^2}{2} \nonumber \\[0.75ex]
\,+\,\;& \frac{\eta\gamma_0\sin\theta}{1-\eta\sin\theta}\delta(x+1)-\frac{\eta\pi^2\sin\theta}{2\gamma_0} \Iver{d(x,0)\le\alpha} \nonumber \\[0.75ex]
\,+\,\;& \frac{\eta\pi^2\cos\theta}{2\gamma_0} \Iver{d(x,-(1+\alpha)/2)<(1-\alpha)/2} \nonumber \\[0.75ex]
\,-\,\;& \frac{\eta\pi^2\cos\theta}{2\gamma_0} \Iver{d(x,+(1+\alpha)/2)<(1-\alpha)/2}
\end{align}
is approximated as
\begin{equation}
H_{\mathsmaller{{\rm SP},LR}} \,\defeq\,
\!\left[\!\begin{array}{cc}
\langle\psi_{\mathsmaller{L}}|H_{\mathsmaller{\rm SP}}|\psi_{\mathsmaller{L}}\rangle & \langle\psi_{\mathsmaller{L}}|H_{\mathsmaller{\rm SP}}|\psi_{\mathsmaller{R}}\rangle \\
\langle\psi_{\mathsmaller{L}}|H_{\mathsmaller{\rm SP}}|\psi_{\mathsmaller{R}}\rangle & \langle\psi_{\mathsmaller{R}}|H_{\mathsmaller{\rm SP}}|\psi_{\mathsmaller{R}}\rangle
\end{array}\!\right]\!
\,=\,\frac{\eta\pi^2}{2\gamma_0}\left[\!\begin{array}{cc}
\cos\theta & \sin\theta \\
\sin\theta & -\cos\theta
\end{array}\!\right]+O\!\left(\gamma_0^{\mathsmaller{-}2}\right)
\end{equation}
in the basis of single-particle states $\{\psi_{\mathsmaller{L}},\psi_{\mathsmaller{R}}\}$, or equivalently
\begin{equation}
H_{\mathsmaller{\rm SP},\pm} \,\defeq\,
\!\left[\!\begin{array}{cc}
\langle\psi_+|H_{\mathsmaller{\rm SP}}|\psi_+\rangle & \langle\psi_+|H_{\mathsmaller{\rm SP}}|\psi_-\rangle \\
\langle\psi_-|H_{\mathsmaller{\rm SP}}|\psi_+\rangle & \langle\psi_-|H_{\mathsmaller{\rm SP}}|\psi_-\rangle
\end{array}\!\right]\!
\,=\,\frac{\eta\pi^2}{2\gamma_0}\left[\!\begin{array}{cc}
\sin\theta & \cos\theta \\
\cos\theta & -\sin\theta
\end{array}\!\right]+O\!\left(\gamma_0^{\mathsmaller{-}2}\right)
\end{equation}
in the basis of single-particle states $\{\psi_+,\psi_-\}$. Therefore, universal single-rebit ground state logic can be realized by a bi-fermion restricted to the two-state manifold of lowest energy, with the fermionic Schr\"odinger Hamiltonian $H_{\mathsmaller{\rm BF},0}+V_{\mathsmaller{X},\eta\sin\theta}+V_{\mathsmaller{Z},\eta\cos\theta}$, $\theta\in[-\pi,\pi)$, which is effectively
\begin{align}
& H_{\mathsmaller{\rm BF},0}+V_{\mathsmaller{X},\eta\sin\theta}+V_{\mathsmaller{Z},\eta\cos\theta} \nonumber \\[0.75ex]
\IsoPhysLR\,& \frac{\eta\pi^2}{2\gamma_0}\left[(a_{\mathsmaller{L}}^+a_{\mathsmaller{R}}+
a_{\mathsmaller{R}}^+a_{\mathsmaller{L}})\sin\theta+(a_{\mathsmaller{L}}^+a_{\mathsmaller{L}}-a_{\mathsmaller{R}}^+a_{\mathsmaller{R}})\cos\theta\right]a_{\mathsmaller{P}}^+a_{\mathsmaller{P}}
+O\!\left(\gamma_0^{\mathsmaller{-}2}\right) \label{HamiltonianForRTheta} \\[0.75ex]
\HomoPhysL\;& \frac{\eta\pi^2}{2\gamma_0}(\sigma^x\sin\theta+\sigma^z\cos\theta)+O\!\left(\gamma_0^{\mathsmaller{-}2}\right)\!, \nonumber
\end{align}
in the first-quantized, second-quantized, and rebit representations respectively.

\iftoggle{ForUSPTO} {
} {
\vspace{-2.0ex}
\begin{figure}[ht]
\vspace{-2.0ex}
\centering
\begin{tabular}{ccc}
\subfloat[Subfigure 1 list of figures text][$\theta=-3\pi/4$]{\includegraphics[width=0.54\textwidth]
{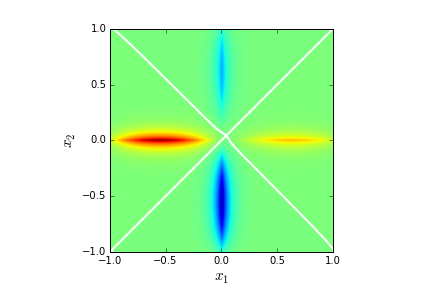}\label{fig:FourXsinZcos1}} & \hspace{-7em} &
\subfloat[Subfigure 2 list of figures text][$\theta=-\pi/4$]{\includegraphics[width=0.54\textwidth]
{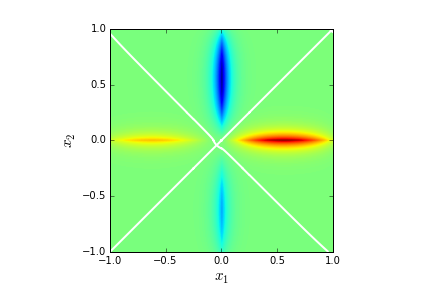}\label{fig:FourXsinZcos2}} \\
\subfloat[Subfigure 3 list of figures text][$\theta=\pi/4$]{\includegraphics[width=0.54\textwidth]
{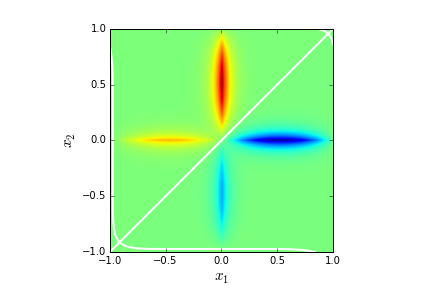}\label{fig:FourXsinZcos3}} & \hspace{-7em} &
\subfloat[Subfigure 4 list of figures text][$\theta=3\pi/4$]{\includegraphics[width=0.54\textwidth]
{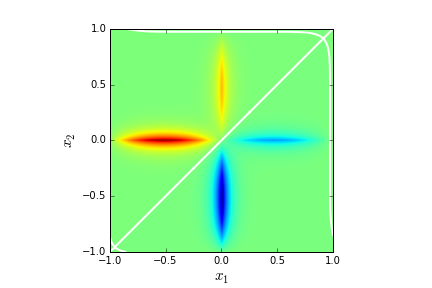}\label{fig:FourXsinZcos4}}
\end{tabular}
\caption{The ground states of $H_{\mathsmaller{\rm BF},0}+V_{\mathsmaller{X},\eta\sin\theta}+V_{\mathsmaller{Z},\eta\cos\theta}$ for $\theta=-3\pi/4,-\pi/4,\pi/4,3\pi/4$.}
\label{FourXsinZcosNodalCurves}
\end{figure}
\vspace{2.0ex}
}

\iftoggle{ForUSPTO} {In contour plots, with dashed lines labeled by the numbers $0.1$, $0.3$, $0.7$ indicating where a wavefunction is positive-valued at levels equal to $0.1$, $0.3$, $0.7$ respectively, while dark solid lines labeled by the numbers $-0.1$, $-0.3$, $-0.7$ signifying where a wavefunction is negative-valued at levels equal to $-0.1$, $-0.3$, $-0.7$ respectively, and gray solid lines labeled by the number $0.0$ indicating nodal surfaces} {In scaled colors}, Fig.\;\ref{FourXsinZcosNodalCurves} shows the ground states of the bi-fermion Hamiltonian $H_{\mathsmaller{\rm BF},0}+V_{\mathsmaller{X},\eta\sin\theta}+V_{\mathsmaller{Z},\eta\cos\theta}$ for $\theta=-3\pi/4,-\pi/4,\pi/4,3\pi/4$ respectively, with $\gamma_0=16$, $\alpha^{{-}1}=4$ (which is not exactly $\gamma_0/(2\log\gamma_0)$), and $\eta=1/2$, where it can be clearly seen that the white-colored nodal curves divide the configuration space into exactly one positive and one negative nodal cells. Again, the relatively small values of $\gamma_0=16$ and $\alpha^{\mathsmaller{-}1}=4$ are chosen only to reduce strains in graphic displaying and viewing, while in real applications, the values of $\gamma_0$ and $\alpha^{\mathsmaller{-}1}$ could and should be much larger. In general, for most of the $\theta$ angles such that $\sin\theta$ is valued not too close to $0$, {\it i.e.}, $|\!\sin\theta|\gg\gamma_0^{\mathsmaller{-}1}$, the ground state of $H_{\mathsmaller{\rm BF},0}+V_{\mathsmaller{X},\eta\sin\theta}+V_{\mathsmaller{Z},\eta\cos\theta}$ is either {\em peri-even} when $\theta\in(-\pi,0)$ or {\em peri-odd} when $\theta\in(0,\pi)$, in the sense that a {\em peri-even} ground state has, approximately, the closures of the E and W quarters of the open square $\{(x_1,x_2):x_1,x_2\in(-1,1)\}$ fused into the closure of one nodal cell on $\mathbb{T}^2$, the closures of the S and N quarters fused into the closure of another nodal cell on $\mathbb{T}^2$, where the two nodal cells are divided by the nodal line $x_1-x_2=0$ and another nodal curve that is closely aligned with the line $x_1+x_2=0$, $x_1,x_2,\in\mathbb{T}$, while a {\em peri-odd} ground state has, approximately, the closures of the N and W quarters of the open square $\{(x_1,x_2):x_1,x_2\in(-1,1)\}$ fused into the closure of one nodal cell on $\mathbb{T}^2$, the closures of the S and E quarters fused into the closure of another nodal cell on $\mathbb{T}^2$, where the two nodal cells are divided by the nodal line $x_1-x_2=0$ and another nodal curve that is closely aligned with $\{x_1=\pm 1\}\cup\{x_2=\pm1\}$, $x_1,x_2,\in\mathbb{T}$. In practical implementations and simulations, when a large $\gamma_0$ is chosen, the nodal curves of the {\em peri-even} and {\em peri-odd} states can be well approximated by the lines $x_1+x_2=0$ and $\{x_1=\pm 1\}\cup\{x_2=\pm1\}$ respectively, in addition to the always exact nodal line $x_1-x_2=0$, $x_1,x_2,\in\mathbb{T}$.

A moment of reflection on  Fig.\;\ref{FourXsinZcosNodalCurves} helps to raise an intuition that, when one fermion is artificially labeled as the first and assigned to the single-particle state $|P\rangle\sim|\psi_{\mathsmaller{P}}\rangle$, the motion of the other fermion in the three-well potential is isophysical to a quantum computational rebit, whose wavefunction is restricted to be real-valued, but can be signed freely, either positive or negative. Still though, the system of the two fermions as a whole, that is the bi-fermion, can be simulated probabilistically on a classical computer, using MCMC in conjunction with either the fixed-node diffusion method or the quantum-stochastic operator similarity-transformation based on equation (\ref{r_PinftyHk_q}), where a random walk over a classical probability space, necessarily with a non-negative valued measure, gleans a sufficient statistic of the bi-fermion system in its quantum-mechanical entirety, involving both positive and negative valued quantum amplitudes. In a sense, if the ability of a quantum machine to represent and manipulate ``negative probability'' is considered a certain computational resource, then a classical probabilistic machine apparently lacks such a computational resource. Nevertheless, the deficiency can be made up for by using an exchange of two identical fermions to represent the sign flip of probability amplitude, which turns the so-called fermion sign problem on its head, in that, instead of being dreaded as a curse that brings trouble to quantum Monte Carlo simulations, the exchange antisymmetry of fermions is embraced as a blessing that endows a classical probabilistic simulator of bi-fermions the ability to represent and manipulate ``negative probability''.

More specifically, let $X\defeq\sigma^x$, $Z\defeq\sigma^z$, and $R(\theta)\defeq X\sin\theta+Z\cos\theta$, $\theta\in[-\pi,\pi)$ represent an ``$R$'' gate on a bi-fermion rebit. Let $R^{\pm}(\theta)\defeq\half[I\pm R(\theta)]$ denote the affine-transformed operators, which preserve the two eigen states of $R(\theta)$, and send one of them to the ground at zero energy. As will be discussed shortly below, applications of such ``$R$'' gates on a specific bi-fermion rebit are often controlled by the state of other rebits, however the quantum physics and Monte Carlo dynamics of the concerned bi-fermion rebit remain essentially the same for the same $R$-gate, although exactly which $R$-gate is applied shall depend upon the state of the controlling rebits. With $R(\theta)$ being an involution, or a self-inverse operator, $R^{\pm}(\theta)$ are orthogonal projections, or self-adjoint idempotent operators, and the associated Gibbs operators have a rather simple formula,
\begin{equation}
G[\tau R^{\pm}(\theta)] \,\defeq\, \exp[-\tau R^{\pm}(\theta)]\,=\,R^{\mp}(\theta)+e^{\mathsmaller{-}\tau}R^{\pm}(\theta),\;\forall\tau\ge 0,
\end{equation}
thus $\lim_{\tau\mathsmaller{\rightarrow}+\infty}G[\tau R^{\pm}(\theta)]=R^{\mp}(\theta)$. The two eigenstates of $R(\theta)$, corresponding to the eigenvalues $+1$ and $-1$ respectively, and being orthogonal necessarily, read
\begin{align}
\Phi^+(\theta) \,=\,& \cos(\theta/2)\,|\!\downarrow\rangle+\sin(\theta/2)\,|\!\uparrow\rangle, \label{Phi+Theta01} \\[0.75ex]
\Phi^-(\theta) \,=\,& \cos(\theta/2)\,|\!\uparrow\rangle-\sin(\theta/2)\,|\!\downarrow\rangle, \label{Phi-Theta01}
\end{align}
in the computational bi-fermion representation, which are equivalently
\begin{align}
\Phi^+(\theta) \,=\,& \left[\cos(\theta/2)|L_1\rangle+\sin(\theta/2)|R_1\rangle\right]|P_2\rangle
-\left[\cos(\theta/2)|L_2\rangle+\sin(\theta/2)|R_2\rangle\right]|P_1\rangle, \label{Phi+ThetaLR} \\[0.75ex]
\Phi^-(\theta) \,=\,& \left[\cos(\theta/2)|R_1\rangle-\sin(\theta/2)|L_1\rangle\right]|P_2\rangle
-\left[\cos(\theta/2)|R_2\rangle-\sin(\theta/2)|L_2\rangle\right]|P_1\rangle, \label{Phi-ThetaLR}
\end{align}
in terms of the effective computational basis states in an explicitly exchange-antisymmetric form. It holds true in any representation that
\begin{equation}
\lim_{\tau\mathsmaller{\rightarrow}+\infty}G[\tau R^{\pm}(\theta)]\,=\,R^{\mp}(\theta)\,=\,|\Phi^{\mp}(\theta)\rangle\langle\Phi^{\mp}(\theta)|\,=\,I-|\Phi^{\pm}(\theta)\rangle\langle\Phi^{\pm}(\theta)|.
\end{equation}
It does not need more than a note in passing that any Hamiltonian $R^{\pm}(\theta)$, $\theta\in[-\pi,\pi)$ is trivially $\epsilon$-almost node-determinate, with $\epsilon=O\!\left(\gamma_0^{{-}2}\right)>0$. When Monte Carlo-simulated on a classical computer, the bi-fermion ground state $\Phi^{\pm}(\theta;x_1,x_2)$ of $R^{\mp}(\theta)$ corresponds to a probability distribution $\left|\Phi^{\pm}(\theta;x_1,x_2)\right|^2$, and the quantum statistical mechanics of the bi-fermion under the Gibbs operator $G[\tau R^{\pm}(\theta)]$, $\forall\tau>0$ maps to a classical random walk with a stochastic operator
\begin{align}
M_{\tau R^{\pm}(\theta)} \,\defeq\;& [\Phi_{\mp}(\theta)]\,G(\tau R^{\pm}(\theta))\,[\Phi_{\mp}(\theta)]^{\mathsmaller{-}1} \label{PtauRpmTheta} \\[0.75ex]
\,=\;\;& [\Phi_{\mp}(\theta)]\,\{R^{\mp}(\theta)+e^{\mathsmaller{-}\tau}R^{\pm}(\theta)\}\,[\Phi_{\mp}(\theta)]^{\mathsmaller{-}1}, \nonumber
\end{align}
on any fixed nodal cell $\calN(R^{\pm}(\theta); x''_1,x''_2)$ enclosing an $(x''_1,x''_2)\in\mathbb{T}^2$ such that $R^{\pm}(\theta; x''_1,x''_2)\neq 0$, where the operators $R^{\mp}(\theta)=|\Phi^{\mp}(\theta)\rangle\langle\Phi^{\mp}(\theta)|$ and $G[\tau R^{\pm}(\theta)]$, as well as $M_{\tau R^{\pm}(\theta)}$ are all guaranteed to be non-negative in the $\mathbb{T}^2$-coordinate representation. Furthermore, in the limit of $\tau\rightarrow{+}\infty$, the operator $M_{\tau R^{\pm}(\theta)}$ is actually positive over the entire configuration space $\mathbb{T}^2\setminus\psi_0^{-1}(R^{\mp}(\theta))$.

More formally, a {\em quantum bi-fermion rebit/system} lives in the two-dimensional {\em quantum state space}
\begin{equation}
\calQ \,\defeq \left\{\textstyle{\frac{1}{\sqrt{2}}}\left(c_0|L_1\rangle|P_2\rangle + c_1|R_1\rangle|P_2\rangle + {\rm E.S.T.}\right):c_0\in\mathbb{R}_{\ge 0},\,c_1\in\mathbb{R}\right\}\!, \label{DefHilbertQ}
\end{equation}
where ${\rm E.S.T.}$ stands for the {\em exchange-symmetric terms} including $(-c_0)|L_2\rangle|P_1\rangle$ and $(-c_1)|R_2\rangle|P_1\rangle$, which are obtained by applying an operator $\pi_{12}$ to the wavefunctions $c_0|L_2\rangle|P_1\rangle$ and $c_1|R_2\rangle|P_1\rangle$ respectively, with $\pi_{12}$ denoting the well-known {\em fermionic exchange operator} that swaps labels for the two identical fermions and induces a negative sign. Any normalized wavefunction $\phi\in\calQ$ can be made into the ground state of a Hamiltonian $H(\phi)\defeq I-|\phi\rangle\langle\phi|=|\phi^{\smallperp}\rangle\langle\phi^{\smallperp}|\in\calB(\calQ)$, with $\phi^{\smallperp}\in\calQ$ being normalized and $\langle\phi|\phi^{\smallperp}\rangle=0$, $\calB(\calQ)$ denoting the Banach algebra of operators acting on the Hilbert space $\calQ$, over the field of real numbers. The quantum bi-fermion rebit can be mapped into a {\em rectified bi-fermion rebit/system} on an ostensibly four-dimensional but actually two-dimensional {\em rectified state space}
\begin{equation}
\calR \,\defeq \left\{c_0|L_1\rangle|P_2\rangle+(-\pi_{12})^{(c_1<0)}c_1|R_1\rangle|P_2\rangle:c_0\in\mathbb{R}_{\ge 0},\,c_1\in\mathbb{R}\right\}\!, \label{DefHilbertR}
\end{equation}
through a one-to-one correspondence $\mathfrak{M}:\left(\mathbb{T}^2,\calQ,\calB(\calQ)\right)\mapsto\left(\mathbb{T}_*^2,\calR,\calB(\calR)\right)$, with
\begin{equation}
\mathbb{T}_*^2 \,\defeq\, \supp(|L_1\rangle|P_2\rangle)\cup\supp(|R_1\rangle|P_2\rangle)\cup\supp(|R_2\rangle|P_1\rangle)
\end{equation}
being called the {\em rectified bi-fermion configuration space}, as opposed to the {\em full bi-fermion configuration space} $\mathbb{T}^2$, such that $\mathfrak{M}(\mathbb{T}^2)=\mathbb{T}_*^2\subseteq\mathbb{T}^2$, and
\begin{align}
\forall\psi(r,\theta) \,\defeq\,& r\,\textstyle{\frac{1}{\sqrt{2}}}\left[\cos(\theta/2)|L_1\rangle|P_2\rangle+\sin(\theta/2)|R_1\rangle|P_2\rangle\right]+{\rm E.S.T.}\in\calQ, \nonumber \\[0.75ex]
\mathfrak{M}(\psi(r,\theta)) \,=\,\;& \Pi_{\mathsmaller{\rm R}}\left\{r\!\left[\cos(\theta/2)|L_1\rangle|P_2\rangle+\sin(\theta/2)|R_1\rangle|P_2\rangle\right]\right\}\in\calR \label{defFrakMPsi} \\[0.75ex]
\,\defeq\,& r\cos(\theta/2)|L_1\rangle|P_2\rangle+(-\pi_{12})^{(r\sin(\theta/2)<0)}r\sin(\theta/2)|R_1\rangle|P_2\rangle, \nonumber
\end{align}
for the wavefunctions, $\forall r\in\mathbb{R}$, $\forall\theta\in[-\pi,\pi]$, and
\begin{align}
\forall H(\theta) \,\defeq\,& I-|\phi(\theta)\rangle\langle\phi(\theta)| \,=\, |\phi^{\smallperp}(\theta)\rangle\langle\phi^{\smallperp}(\theta)|\in\calB(\calQ), \nonumber \\[0.75ex]
\mathfrak{M}(H(\theta)) \,=\,\;& \Pi_{\mathsmaller{\rm R}}H(\theta)\Pi_{\mathsmaller{\rm R}} \,=\, |\mathfrak{M}(\phi^{\smallperp}(\theta))\rangle\langle\mathfrak{M}(\phi^{\smallperp}(\theta))|\in\calB(\calR), \label{defFrakMH}
\end{align}
for the orthogonal projection-valued Hamiltonians, $\theta\in[-\pi,\pi]$, where
\begin{align}
\phi(\theta) \,\defeq\,& \textstyle{\frac{1}{\sqrt{2}}}\left[\cos(\theta/2)|L_1\rangle|P_2\rangle+\sin(\theta/2)|R_1\rangle|P_2\rangle\right]+{\rm E.S.T.}, \label{defiPhiSelfTheta} \\[0.75ex]
\phi^{\smallperp}(\theta) \,\defeq\,& \textstyle{\frac{1}{\sqrt{2}}}\left[\cos(\theta/2)|R_1\rangle|P_2\rangle-\sin(\theta/2)|L_1\rangle|P_2\rangle\right]+{\rm E.S.T.}, \label{defiPhiPerpTheta}
\end{align}
and $\Pi_{\mathsmaller{\rm R}}$ as defined in equation (\ref{defFrakMPsi}) is a one-to-one and self-inverse mapping between $\calQ$ and $\calR$, which is actually an isomorphism between two Hilbert spaces, as will be shown rigorously below. It is noted that a useful convention, henceforth referred to as the {\em logic-$0$-positive convention or representation}, is adopted in equations (\ref{DefHilbertQ}) and (\ref{DefHilbertR}), which imposes the condition that the coefficient $c_0$ should be non-negative, without inducing any loss of generality, since a global phase factor (a +/- sign, to be exact) in or out of the wavefunction of a rebit does not change any physics.

Essentially, the bijective mapping $\mathfrak{M}$ embodies a method of {\em fermion sign rectification} using isomorphic group representations $\mathfrak{s}$ and $\mathfrak{m}$, where
\begin{equation}
\mathfrak{s}:(\{+1,-1\},*)\mapsto(\{I,\pi_{12}\},*),\;{\rm such\;that}\;\mathfrak{s}(+1)=I,\;\mathfrak{s}(-1)=\pi_{12},
\end{equation}
is an isomorphism between the multiplicative group $(\{+1,-1\},*)$ of numerical signs and the multiplicative group $(\{I,\pi_{12}\},*)$ of linear operators, while
\begin{align}
& \mathfrak{m}:(\mathbb{R}_{\neq 0},*)=(\mathbb{R}_{>0},*)\times(\{+1,-1\},*)\mapsto(\mathbb{R}_{>0},*)\times(\{I,\pi_{12}\},*)\defeq(\mathbb{R}'_{\neq 0},*),
\nonumber \\[0.75ex]
& \mathfrak{m}(r)=(|r|,\pi_{12}{\raisebox{-.05\height}{$^{(r<0)}$}})\,\defeq\,|r|\pi_{12}{\raisebox{-.05\height}{$^{(r<0)}$}}=\textstyle{\half}\left(|r|+r\right)I+\textstyle{\half}\left(|r|-r\right)\pi_{12},\;\forall r\in\mathbb{R}_{\neq 0}, \label{defiRprimeMul}
\end{align}
is an isomorphism between the multiplicative group of real numbers and the direct product group $(\mathbb{R}_{>0},*)\times(\{I,\pi_{12}\},*)$, where $\mathbb{R}_{\neq 0}\defeq\{r\in\mathbb{R}:r\neq 0\}$, $\mathbb{R}_{>0}\defeq\{r\in\mathbb{R}:r>0\}$, and $\forall(|r|,\pi_{12}{\raisebox{-.05\height}{$^{(r<0)}$}})\in(\mathbb{R}'_{\neq 0},*)$, the equation $(|r|,\pi_{12}{\raisebox{-.05\height}{$^{(r<0)}$}})=\textstyle{\half}\left(|r|+r\right)I+\textstyle{\half}\left(|r|-r\right)\pi_{12}$ makes its sense with the group elements $\{(|r|,\pi_{12}{\raisebox{-.05\height}{$^{(r<0)}$}})\}_{r\in\mathbb{R}_{\neq 0}}$ interpreted as linear operators acting on the rectified state space $\calR$. $(\mathbb{R}_{\neq 0},*)$ is of course the multiplicative structure in the field of real numbers $(\mathbb{R},+,*)$. With zero elements added, and the domain of the mapping $\mathfrak{m}$ extended to $\mathbb{R}$ such that $\mathfrak{m}(0)=(0,I)=(0,\pi_{12})\defeq 0'\in\mathbb{R}_{\ge 0}\times\{I,\pi_{12}\}$, with $\mathbb{R}_{\ge 0}\defeq\{r\in\mathbb{R}:r\ge 0\}$, then the product set $\mathbb{R}_{\ge 0}\times\{I,\pi_{12}\}$ can be made into an abelian additive group with a ``$+$'' operation defined as
\begin{equation}
r'_1+r'_2\,\defeq\,\mathfrak{m}[\mathfrak{m}^{\mathsmaller{-}1}(r'_1)+\mathfrak{m}^{\mathsmaller{-}1}(r'_2)],\;\forall r'_1,r'_2\in\mathbb{R}_{\ge 0}\times\{I,\pi_{12}\} \,\defeq\, \mathbb{R}', \label{defiRprimeAdd}
\end{equation}
so that the structure $(\mathbb{R}',+,*)$ becomes an ordered field, henceforth referred to as the {\em field of rectified numbers}, that is isomorphic to the conventional {\em field of real numbers} $(\mathbb{R},+,*)$, with an element $r'\in(\mathbb{R}',+,*)$ being positive, denoted as $r'>0'$ or $r'\in\mathbb{R}'_{>0'}$, if and only if $\mathfrak{m}^{\mathsmaller{-}1}(r')\in\mathbb{R}_{>0}$. An alternative construction of the field $(\mathbb{R}',+,*)$ starts with the direct product group $(\mathbb{R}''_{\neq 0},*)\defeq(\mathbb{R}_{\neq 0},*)\times(\{I,\pi_{12}\},*)=\left\{(r,\pi_{12}^s):r\in\mathbb{R}_{\neq 0},s\in\{0,1\}\right\}$, which is overcomplete to serve as a group of multiplicative coefficients for wavefunctions of bi-fermions, but the unwanted redundancy can be removed by imposing an equivalence relation ${\raisebox{-.1\height}{$\simX{$\pi$}$}}$ such that $(r,\pi_{12})\;{\raisebox{-.1\height}{$\simX{$\pi$}$}}\;(-r,I)$, $\forall r\in\mathbb{R}_{\neq 0}$, which then effects a group endomorphism $f_{\pi}:(\mathbb{R}''_{\neq 0},*)\mapsto(\mathbb{R}''_{\neq 0},*)$ such that $f_{\pi}(r,\pi_{12}^s)=(|r|,\pi_{12}{\raisebox{-.05\height}{$^{r(-1)^s<0}$}})$, $\forall r\in\mathbb{R}$, $\forall s\in\{0,1\}$. The resulted quotient group $(\mathbb{R}''_{\neq 0},*)/{\raisebox{-.1\height}{$\simX{$\pi$}$}}$ is isomorphic to $(\mathbb{R}'_{\neq 0},*)$. Then, a zero element $0'$ and an additive group structure can be introduced to the direct product set $\mathbb{R}''\defeq\left\{(r,\pi_{12}^s):r\in\mathbb{R},s\in\{0,1\}\right\}$ in a manner similar to equation (\ref{defiRprimeAdd}), with the equivalence relation ${\raisebox{-.1\height}{$\simX{$\pi$}$}}$ extended naturally such that $(r,I)+(-r,I)\;{\raisebox{-.1\height}{$\simX{$\pi$}$}}\;(r,I)+(r,\pi_{12})\;{\raisebox{-.1\height}{$\simX{$\pi$}$}}\;0'$, $\forall r\in\mathbb{R}$. It becomes clear now that the quotient set $\mathbb{R}''/{\raisebox{-.1\height}{$\simX{$\pi$}$}}$ forms an ordered field $(\mathbb{R}''/{\raisebox{-.1\height}{$\simX{$\pi$}$}},+,*)$, which is isomorphic to the field $(\mathbb{R}',+,*)$. Viewing the field $(\mathbb{R}',+,*)$ as $(\mathbb{R}''/{\raisebox{-.1\height}{$\simX{$\pi$}$}},+,*)$ affords an advantageous perspective and the algebraic convenience for transforming mathematical formulas involving bi-fermion wavefunctions and operators, by freely inserting an identity factor $(-1)\pi_{12}$ anywhere in a multiplicative expression, or a zero term $1+\pi_{12}$ anywhere in an additive expression. In particular, a stoquastic operator on a Hilbert space over the field of rectified numbers $(\mathbb{R}',+,*)$ can be deliberately written as $D-A$, instead of $D+\pi_{12}A$, to signify the stoquastic property and the consequence that $\exp[-\tau(D-A)]$ is non-negative $\forall\tau>0$, when the operator $D$ is represented by a diagonal matrix with elements in $\mathbb{R}\cong\mathbb{R}'$, while the operator $A$ is represented by a non-negative matrix with all elements in $\mathbb{R}_{\ge 0}$. It is also noted that the definition of the rectified state space $\calR$ in (\ref{DefHilbertR}) has removed a two-fold redundancy in the single bi-fermion configuration space and wave distribution, by omitting the negative nodal cell and the negative-valued wave distribution in it, which are always isomorphic to the retained positive nodal cell and positive wave distribution therein.

The quantum state space $\calQ$ is clearly a Hilbert space over the field $(\mathbb{R},+,*)$, with the conventionally defined inner product $\langle\psi_1|\psi_2\rangle\defeq a_1a_2+b_1b_2$, $\forall\psi_i=a_i|L_1\rangle|P_2\rangle+b_i|R_1\rangle|P_2\rangle\in\calQ$, $a_i,b_i\in\mathbb{R}$, $i\in\{1,2\}$. The quantum physics of a rebit is also concerned with the vector space $\calB(\calQ)$ of bounded linear operators on the Hilbert space $\calQ$, which forms a Banach algebra under the operator norm. With a chosen basis of $\calQ$, such as $\{|L_1\rangle|P_2\rangle{-}|L_2\rangle|P_1\rangle,|R_1\rangle|P_2\rangle{-}|R_2\rangle|P_1\rangle\}$, each element in $\calB(\calQ)$ is represented by a $2\times 2$ matrix with $\mathbb{R}$-valued entries. Thanks to the one-to-one correspondence $\mathfrak{M}$ and the bijection $\mathfrak{m}$, the rectified state space $\calR$ can also be made into a vector space, over the field $(\mathbb{R}',+,*)$, with the required vector addition and scalar multiplication operations defined as
\begin{align}
& \psi'_1+\psi'_2 \,\defeq\, \mathfrak{M}\left[\mathfrak{M}^{\mathsmaller{-}1}(\psi'_1)+\mathfrak{M}^{\mathsmaller{-}1}(\psi'_2)\right],\;\forall\psi'_1,\psi'_2\in\calR, \\[0.75ex]
& r'\psi' \,\defeq\, \mathfrak{M}\left[\mathfrak{m}^{\mathsmaller{-}1}(r')\mathfrak{M}^{\mathsmaller{-}1}(\psi')\right],\;\forall r'\in(\mathbb{R}',+,*),\;\forall\psi'\in\calR,
\end{align}
then straightforwardly made into a Hilbert space with an inner product defined as
\begin{equation}
\langle\psi'_1|\psi'_2\rangle\defeq\left(a_1a_2,I\right)+\left(b_1b_2,\pi_{12}^{s_1+s_2}\right)\in(\mathbb{R}',+,*),
\end{equation}
$\forall\psi'_i=(a_i,I)|L_1\rangle|P_2\rangle+(b_i,\pi_{12}^{s_i})|R_1\rangle|P_2\rangle\in\calR$, $(a_i,I)\in\mathbb{R}'$, $(b_i,\pi_{12}^{s_i})\in\mathbb{R}'$, $s_i\in\{0,1\}$, $i\in\{1,2\}$. In the Hilbert space $\calR$, the value of any wavefunction at any location in the bi-fermion configuration space involves only non-negative numbers in $\mathbb{R}$. Still, two non-zero vectors $\psi'_1,\psi'_2\in\calR$ can be orthogonal and produce a $0'$-valued inner product $\langle\psi'_1|\psi'_2\rangle=(0,I)=(0,\pi_{12})\in(\mathbb{R}',+,*)$. The vector space $\calB(\calR)$ of bounded linear operators over the Hilbert space $\calR$ forms a Banach algebra under the operator norm, which is isomorphic to $\calB(\calQ)$ and describes the equivalent physics. With a chosen basis of $\calR$, such as $\{|L_1\rangle|P_2\rangle,|R_1\rangle|P_2\rangle\}$, each element in $\calB(\calR)$ is represented by a $2\times 2$ matrix with $\mathbb{R}'$-valued entries that involve only non-negative numbers in $\mathbb{R}$ and possibly the operator $\pi_{12}$ exchanging identical fermions. It is interesting to note that both $\calQ$ and $\calR$ can be considered as subsets of a larger Hilbert space $\calX$ over the field $\mathbb{R}$, where
\begin{equation}
\calX \,\defeq\, \left\{c_0|L_1\rangle|P_2\rangle+c_1|R_1\rangle|P_2\rangle+c_2|L_2\rangle|P_1\rangle+c_3|R_2\rangle|P_1\rangle:c_0\in\mathbb{R}_{\ge 0},\,c_1,c_2,c_3\in\mathbb{R}\right\}\!, \label{DefHilbertS}
\end{equation}
being equipped with the conventionally defined inner product, while both $\calB(\calQ)$ and $\calB(\calR)$ can be regarded as subsets of $\calB(\calX)$. $\calQ$ and $\calB(\calQ)$ are just the results of $\calX$ and $\calB(\calX)$ modulo the antisymmetric condition upon exchange of identical fermions. For the rectified state space, any wavefunction
\begin{align}
\psi' =\;& (a,I)|L_1\rangle|P_2\rangle+(b,\pi_{12}^s)|R_1\rangle|P_2\rangle\in\calR \nonumber \\[0.75ex]
=\;& a\,(|L_1\rangle|P_2\rangle)+b\,(\pi_{12}^t|R_1\rangle|P_2\rangle)\in\calX,\;a,b\in\mathbb{R}_{\ge 0},\;s\in\{0,1\}
\end{align}
is in the conical hull of $\calX_{\star}\defeq\{|L_1\rangle|P_2\rangle,|R_1\rangle|P_2\rangle,|R_2\rangle|P_1\rangle\}$, denoted by $\coni(\calX_{\star})$, where all of the basis wavefunctions in $\calX_{\star}$ are non-negative-valued in the bi-fermion configuration space, the union of whose supports is the rectified bi-fermion configuration space $\mathbb{T}_*^2$. Therefore, $\calR$ is inside a positive convex cone consisting of non-negative definite wavefunctions in $\calX$. By the same token, any dyadic operator in $\calB(\calR)$ that is of the form $\sum_{ij}|\psi'_i\rangle\langle\phi'_j|$, $\psi'_i,\phi'_j\in\calR$, $i,j\in\mathbb{N}$ can be interpreted as an element in the positive convex cone $\coni(\calX_{\star}\calX_{\star}\!{}\!^+)\subset\calB(\calX)$, with $\calX_{\star}\calX_{\star}\!{}\!^+\defeq\{|u'\rangle\langle v'|:u'\in\calX_{\star},\,v'\in\calX_{\star}\}$. Such embeddings of $\calR$ into $\calX$ and $\calB(\calR)$ into $\calB(\calX)$ coincide with the embeddings of $\calQ$ into $\calX$ and $\calB(\calQ)$ into $\calB(\calX)$ as far as the positive parts of the wavefunctions are concerned, thus enabling fermion sign rectification, which is to do quantum computing and Monte Carlo simulations without ever encountering a negative number in $\mathbb{R}$, through the use of rectified bi-fermions in conjunction with the fixed-node diffusion method or the quantum-stochastic operator similarity transformation as in \iftoggle{ForUSPTO} {Auxiliary Utility} {Lemma} \ref{QuasiStochasticOper}.

Clearly, the mapping $\mathfrak{M}:\left(\mathbb{T}^2,\calQ,\calB(\calQ)\right)\mapsto\left(\mathbb{T}_*^2,\calR,\calB(\calR)\right)$ amounts to an isophysics, through which the rectified bi-fermion physics $\left(\mathbb{T}_*^2,\calR,\calB(\calR)\right)$ constitutes a representation of the quantum bi-fermion physics $\left(\mathbb{T}^2,\calQ,\calB(\calQ)\right)$, where any observable $H\in\calB(\calQ)$ other than the identity operator, and the corresponding $H'=\mathfrak{M}(H)\in\calB(\calR)$ may be designated as the Hamiltonians for the respective physics. Since the states in $\calR$ are all non-negative distributions, any ground state physics of $\left(\mathbb{T}_*^2,\calR,\calB(\calR)\right)$ can be isophysically mapped into a classical random walk on the set union of the supports of all wavefunctions in the state space $\calR$, governed by a stochastic operator of the form $[\psi_0(H')]G'(\tau H')[\psi_0(H')]^{\mathsmaller{-}1}$, with $\tau>0$, $H'=\mathfrak{M}(H)$, $G'(\tau H')\defeq\mathfrak{M}[G(\tau H)]\defeq\mathfrak{M}[\exp({-}\tau H)]$, as having been discussed extensively in the previous section leading to {\myTheorem} \ref{FirstTheorem}, as well about equation (\ref{PtauRpmTheta}) a few paragraphs above. It is reminded that for any Hamiltonian $H'$, $\psi_0(H')$ denotes the ground space of $H'$, which consists of one unique ground state wavefunction for the present case of a single bi-fermion system. For any $H'\in\calB(\calR)$ with $\psi_0(H')\in\calR$, it is obvious that one connected component of the support of $\psi_0(H')$ is exactly the positive nodal cell of the ground state $\mathfrak{M}^{\mathsmaller{-}1}(\psi_0(H'))$ of the quantum bi-fermion Hamiltonian $\mathfrak{M}^{\mathsmaller{-}1}(H')$. Importantly, for a set $\calB(\calQ)$ that contains orthogonal projection-valued operators of the form $H=R^{\pm}(\theta)$, $\theta\in[-\pi,\pi)$, which are complete for ground state quantum computations, each of the associated ground states always has one unique positive nodal cell that is topologically connected, such that the Markovian random walks following $\{[\psi_0(H')]G'(\tau H')[\psi_0(H')]^{\mathsmaller{-}1}\!:H'\in\mathfrak{M}[\calB(\calQ)],\,\tau>0\}$ are always ergodic and rapidly mixing. Therefore, the isophysics $\mathfrak{M}:\left(\mathbb{T}^2,\calQ,\calB(\calQ)\right)\mapsto\left(\mathbb{T}_*^2,\calR,\calB(\calR)\right)$ and {\myCorollary} \ref{FixedNodeMethod} of \iftoggle{ForUSPTO} {Auxiliary Utility} {Lemma} \ref{TilingProperty} agree, as they must do, on asserting rigorously that the fermionic ground state $\psi_0(H)$ of a quantum bi-fermion system is uniquely determined by the boltzmannonic ground state $\mathfrak{M}[\psi_0(H)]$ of the corresponding rectified bi-fermion system on the positive nodal cell of $\psi_0(H)$, where the two ground states are both supported and become identical.

It is important to note that a continuous substrate space is not indispensable for devising useful bi-fermions. Rather, a discrete version of the bi-fermion system can be quite simply realized by placing two identical non-interacting fermions on a cyclic grid of $2n$ points $\mathbb{L}_{2n}\defeq\{x\in\mathbb{Z}:-n\le x\le n\}$, $n\ge 2$, $n\in\mathbb{N}$, where $x=\pm n$ are identified as labeling the same lattice point, with single-particle tunneling/coupling between neighbor sites that is translation-invariant in strength, and an on-site single-particle potential $V_x$ that consists of a deep potential well at $x=0$ and a steep potential barrier at $x=\pm n$, so that the nominal bi-fermion Hamiltonian reads
\begin{equation}
H^{\rm lattice}_{\mathsmaller{\rm BF},0} = C_0 \,+\, \sum_{k=1}^2\sum_{x_k=-n}^{n{-}1}\left(\!{}-|x_k\rangle\langle x_k{+}1|-|x_k{+}1\rangle\langle x_k|-\gamma_0\delta_{x_k,0}+\gamma_0\delta_{x_k,\pm n}\right),
\end{equation}
where $\gamma_0>0$, $\gamma_0\gg 1$, and $C_0\in\mathbb{R}$ is a constant suitably chosen such that $\lambda_0\scalebox{1.15}{(}H^{\rm lattice}_{\mathsmaller{\rm BF},0}\scalebox{1.15}{)}=0$, and $\delta_{x,y}\defeq\Iver{x=y}$ denotes the Kronecker delta function for $(x,y)\in\mathbb{Z}^2$. The nominal bi-fermion Hamiltonian $H_{\mathsmaller{\rm BF},0}^{{\rm lattice}}$ has two degenerate ground states $\Phi_{\mathsmaller{L}}\in L^2(\mathbb{L}_{2n}^2)$ and $\Phi_{\mathsmaller{R}}\in L^2(\mathbb{L}_{2n}^2)$ as the discrete counterparts of the wavefunctions as defined in equations (\ref{PhiPdefi}-\ref{PhiLRdefi}), which span a two-dimensional Hilbert space as well. Such a bi-fermion made of two identical fermions moving in a discrete substrate space like $\mathbb{L}_{2n}$ is referred to as a {\em discrete bi-fermion}, as opposed to the {\em continuous bi-fermion} being discussed extensively above, which is supported by a continuous substrate space like $\mathbb{T}$. In perfect parallelism to the case of a continuous bi-fermion, the restricted states
\begin{align}
|\!\downarrow\rangle \,\defeq\,& \Phi_{\mathsmaller{L}}(x_1,x_2) \Iver{\exists\, k\in\{1,2\}\;\mbox{\rm such that}\;({-}x_k)\in[1,n-1]\;\mbox{\rm and}\;x_{3-k}=0}, \\[0.75ex]
|\!\uparrow\rangle \,\defeq\,& \Phi_{\mathsmaller{R}}(x_1,x_2) \Iver{\exists\, k\in\{1,2\}\;\mbox{\rm such that}\;({+}x_k)\in[1,n-1]\;\mbox{\rm and}\;x_{3-k}=0},
\end{align}
with $\Iver{\,\cdot\,}$ being the Iverson bracket, are called the {\em effective computational basis states} spanning a {\em computational bi-fermion space} ${\rm span}(\{|\!\downarrow\rangle,|\!\uparrow\rangle\})$, which are more convenient and useful for GSQC, although such a computational bi-fermion space is not absolutely complete, manifested in that the leakage error operator $\sfE_{\rm leak}(\gamma_0)\defeq I-|\!\downarrow\rangle\langle\downarrow\!|-|\!\uparrow\rangle\langle\uparrow\!|\in\calB(\{\Phi_{\mathsmaller{L}},\Phi_{\mathsmaller{R}}\})$ does not vanish identically. It can be easily verified that the leakage error can be made as small as $\Tr(\sfE_{\rm leak}(\gamma_0))=O\!\left(\gamma_0^{{-}3}\log^3\gamma_0\right)$ in the limit of $n\rightarrow\infty$, or $\Tr(\sfE_{\rm leak}(\gamma_0))=O\!\left(\gamma_0^{{-}2}\right)$ in the limit of $n\rightarrow 2$. Next, perturbative on-site single-particle potentials can be introduced, so that the fermionic Schr\"odinger operator
\begin{align}
H^{\rm lattice}_{\mathsmaller{\rm BF},0} \;+\;& \left\{ V^{\rm lattice}_{\mathsmaller{X},\eta\sin\theta} \,\defeq\, \frac{\eta\gamma_0\sin\theta}{1-\eta\sin\theta}\delta_{x,\pm n} \,-\, \frac{\eta\sin\theta}{\gamma_0}\delta_{x,0} \right\}
\label{HlatticeForRTheta} \\[0.75ex]
\;+\;& \left\{ V^{\rm lattice}_{\mathsmaller{Z},\eta\cos\theta} \,\defeq\, \frac{\eta\cos\theta}{\gamma_0} \Iver{(-x)\in[1,n-1]} \,-\, \frac{\eta\cos\theta}{\gamma_0} \Iver{x\in[1,n-1]} \right\} \nonumber
\end{align}
effects a computational rebit operator $(\eta/\gamma_0)\left(\sigma^x\sin\theta+\sigma^z\cos\theta\right)$, $\eta\in(-1,1)$, $\theta\in[-\pi,\pi)$. Said operator is simultaneously fermionic Schr\"odinger and essentially bounded, thus doubly universal.

For yet another example of a discrete bi-fermion, two identical and non-interacting fermions move in a three-point lattice $\mathbb{L}_3\defeq\{x:x\in\{0,1,2\}\}$, with on-site energies $\gamma_0+\cos\theta,\,\gamma_0-\cos\theta,\,{-}\gamma_0$ at $x=0,\,1,\,2$ respectively, $\gamma_0>0$, $\theta\in[-\pi,\pi)$, where a particle can possibly hop between the two sites $0$ and $1$, but never transfer between $2$ and the other sites, such that a bi-fermion is formed with a Hamiltonian
\begin{align}
H^{\rm tripod}_{\mathsmaller{\rm BF}} =\,\;& \scalebox{1.2}{(}H^{\rm tripod}_{\mathsmaller{\rm BF},0} \defeq \gamma_0a_0^+a_0+\gamma_0a_1^+a_1-\gamma_0a_2^+a_2\scalebox{1.2}{)} \nonumber \\[0.75ex]
+\,\;& (a_0^+a_1+a_1^+a_0)\sin\theta \,+\, (a_0^+a_0-a_1^+a_1)\cos\theta \label{HtripodForRTheta}
\end{align}
in the second-quantized representation, where $a_k$ and $a_k^+$ are the fermionic annihilation and creation operators respectively, $\forall k\in\{0,1,2\}$, such that $a_ka_l+a_la_k=0$, $a^+_ka^+_l+a^+_la^+_k=0$, $a_ka^+_l+a^+_la_k=\delta_{kl}$, $\forall k,l\in\{0,1,2\}$. The Hamiltonian $H^{\rm tripod}_{\mathsmaller{\rm BF}}$ effects a computational rebit operator $\sigma^x\sin\theta+\sigma^z\cos\theta$ over a {\em computational bi-fermion space} spanned by the states $|\!\downarrow\rangle=a_0^+a_2^+|{\sf vac}\rangle$ and $|\!\uparrow\rangle=a_1^+a_2^+|{\sf vac}\rangle$, which is identical to the {\em physical bi-fermion space} spanned by the lowest two rigorous eigenstates of $H^{\rm tripod}_{\mathsmaller{\rm BF}}$. In the first-quantized picture, with the system resting at its ground state(s), it is obvious that the low-energy site $2$, which is named the first orbit, must always be filled by one fermion, that may be arbitrarily labeled as the first, while the second fermion can only occupy possibly a superposition between sites $0$ and $1$, that is named a second orbit. When the two fermions have their positions exchanged, the wavefunction of the whole bi-fermion system acquires a negative sign. The sole purpose of the first orbit at site $2$ hosting one of the identical fermions is just sign-tracking via the exchange antisymmetry. Although not always fermionic Schr\"odinger, the Hamiltonian $H^{\rm tripod}_{\mathsmaller{\rm BF}}$ is surely essentially bounded, suitable for GSQC through an SFF-EB Hamiltonian. On the other hand, with such a discrete bi-fermion supported by a three-point lattice $\mathbb{L}_3$, henceforth referred to as the $\mathbb{L}_3$-supported {\em leakage-free bi-fermion}, there is absolutely no rebit leakage error to worry about, as the operator $I-|\!\downarrow\rangle\langle\downarrow\!|-|\!\downarrow\rangle\langle\downarrow\!|$ strictly vanishes over the physical bi-fermion space, thus, $\Tr(\sfE_{\rm leak}(\gamma_0))=0$ hold true, $\forall\gamma_0\neq 0$. Consequently, there is no need to set $\gamma_0$ to a substantially large value. A choice of $\gamma_0=O(1)$ suffices. These may be considered as compensations for sacrificing the fermionic Schr\"odinger property of $H^{\rm tripod}_{\mathsmaller{\rm BF}}$.

From this point on, unless explicitly stated otherwise, a bi-fermion without a modifier shall refer to a general bi-fermion construct that can be wither continuous or discrete, which is made of two identical fermions moving on a general substrate space $\mathbb{X}$ that can be either continuous as $\mathbb{X}=\mathbb{T}$, or discrete as $\mathbb{X}=\mathbb{L}_3$ or $\mathbb{X}=\mathbb{L}_{2n}$, $n\ge 2$, or another variant of the substrate space. The two states with the lowest energy of such general bi-fermion, denoted as $|\!\downarrow\rangle\in L^2_{\mathsmaller{F}}(\mathbb{X}^2)$ and $|\!\uparrow\rangle\in L^2_{\mathsmaller{F}}(\mathbb{X}^2)$, that may be suitably restricted in the configuration space $\mathbb{X}^2$ if useful, necessary, and desired, constitute a Hilbert space $\calQ\defeq\calQ(\mathbb{X})\defeq{\rm span}(\{|\!\downarrow\rangle,|\!\uparrow\rangle\})\subseteq L^2_{\mathsmaller{F}}(\mathbb{X}^2)$, called the {\em computational bi-fermion space}. The triple $(\mathbb{X}^2,\calQ,\calB(\calQ))$, with $\calB(\calQ)$ being a Banach algebra of bounded linear operators over the Hilbert space $\calQ$, forms a quantum physics/system associated with a single bi-fermion, which implements a computational rebit in the sense of homophysics. It is also straightforward to generalize the discussions between equations (\ref{DefHilbertR}) and (\ref{defiPhiPerpTheta}) to construct a {\em rectified bi-fermion rebit} with a quantum physics/system $(\mathbb{X}_*^2,\calR,\calB(\calR))$, with $\mathbb{X}_*^2=\calN^+(|\!\downarrow\rangle)\cup\calN^+(|\!\uparrow\rangle)\cup\calN^+(\pi_{12}|\!\uparrow\rangle)$ being called the {\em rectified configuration space}, $\calN^+(\phi)$ denoting the positive nodal cell of a wavefunction $\phi\in L^2_{\mathsmaller{F}}(\mathbb{X}^2)$, $\pi_{12}$ being the familiar two-fermion exchange operator, $\calR$ being the Hilbert space spanned by the states $|\!\downarrow\rangle|_{\mathbb{X}_*^2}$, $|\!\uparrow\rangle|_{\mathbb{X}_*^2}$, and $(\pi_{12}|\!\uparrow\rangle)|_{\mathbb{X}_*^2}$ over the field of rectified numbers $(\mathbb{R}',+,*)$, and $\calB(\calR)$ being a Banach algebra of bounded linear operators over the Hilbert space $\calR$, with respect to the field of rectified numbers $(\mathbb{R}',+,*)$ of course.

Even though a general substrate space $\mathbb{X}\in\{\mathbb{L}_3,\mathbb{L}_{2n},\mathbb{T}\}$ and the quantum physics $(\mathbb{X}^2,\calQ,\calB(\calQ))$ or $(\mathbb{X}_*^2,\calR,\calB(\calR))$ have already covered a rather broad range of preferred embodiments of bi-fermions, they should be understood by way of examples and by no way of limitations. Indeed, there should be no difficulty for one skilled in the art to devise a construct similar to the bi-fermions that are exemplified in the above, by choosing another variant of the substrate space, or a different number of a preferred species of identical fermions, or subjecting the identical fermions to a differently arranged potential or different ways of interactions. In particular, it can be stressed that the dimension of the substrate space to host identical fermions that constitute a bi-fermion or any similar construct worthy of a quantum computational qubit or rebit, does not have to be $1$, and the number of identical fermions comprising each bi-fermion or similar construct can exceed $2$, so that there can be $2$ or more fermions filling $2$ or more lower energy orbits and another fermion occupying and roaming between two nearly degenerate states to encode a qubit/rebit. Still further, a bi-fermion or similar construct can span a $d$-dimensional Hilbert space, $d\ge 3$ to constitute a qudit or {\em redit} (that is a qudit with a $d$-dimensional Hilbert space over the field of real numbers).

Also, there can be inter-particle interactions among the identical fermions constituting a single qubit/rebit, and no difficulty in homophysical implementations and Monte Carlo simulations is entailed by such inter-particle interactions, so long as they are described by diagonal self-adjoint operators in the many-particle configuration space representation. On the other hand, it is important to employ a substrate space that is either at least two-dimensional, or in the one-dimensional case, necessarily containing a subspace that is is not simply connected and homeomorphic to the circle group or a cycle graph, regardless of the substrate space being continuous as a manifold or discrete as a topological graph. And, it is essential to have at least two identical fermions moving on such a multi-connected substrate space, if the Hamiltonian has to be of the conventional Schr\"odinger type. Because, as discussed in the previous section, and it is without loss of generality to take the case of a continuous configuration space for example, the ground state of a single particle under a fermionic Schr\"odinger, thus equivalent to boltzmannonic stoquastic, Hamiltonian $H=-\Delta_g+V(q\!\in\!\calM)$ is always non-degenerate and positive, when the arbitrarily dimensioned Riemannian manifold $(\calM,g)$ is compact and smooth, while the potential $V(\cdot)$ is Kato-decomposable and Klauder regular. In another situation, according to a beautiful theorem of Lieb and Mattis \cite{Lieb62}, a simply connected one-dimensional manifold can never host a degenerate ground state of any system of a number of particles with a base boltzmannonic Hamiltonian of the form $\Base(H_{\!\mathsmaller{B}})=-\sum_ig_{ii}{\raisebox{-.1\height}{$^{\mathsmaller{-}1/2}$}}\partial_i\,g_{ii}{\raisebox{-.1\height}{$^{3/2}$}}\partial_i+V(\{x_i\}_i)$, where $\forall i\in\mathbb{N}$, $x_i$ denotes the position of the $i$-th labeled particle on the one-dimensional manifold, $V(\cdot)$ is Klauder regular, $\partial_i\defeq\partial/\partial x_i$, $g_{ii}(x_i)\in\mathbb{R}_{>0}$, $\forall x_i$, regardless of such particles being all identical, all distinguishable, or a mixture of multiple species, and whether or not there are inter-particle interactions. By contrast, it is certainly possible for an at least two-dimensional substrate space in conjunction with a not necessarily Klauder regular potential to accommodate a degenerate ground state, as the previously discussed bi-fermion physics can be equally realized by two fermions confined to within an annular potential well on a two-dimensional plane, where two infinitely high potential walls are narrowly separated to define an annulus with such a small width that the particle motion in the radial direction is effectively frozen to the fundamental mode of lowest energy.

It is also useful to note that an experimental construction or numerical simulation of a bi-fermion-like computational qubit or rebit, being exploited for fermion sign rectification, does not have to involve an actual species of physical fermions as elementary particles or their compositions found in nature. What is truly essential and of only importance is the use of at least two identical and indistinguishable quantum atomic entities that can be artificially labeled by natural numbers in $\mathbb{N}$, where each said quantum atomic entity has at least two different states and is associated with a pair of annihilation and creation operators $a_k$ and $a^+_k$, $k\in\mathbb{N}$ that support a Lie algebra as specified by $a_ka_l+a_la_k=0$, $a^+_ka^+_l+a^+_la^+_k=0$, $a_ka^+_l+a^+_la_k=\delta_{kl}$, $\forall k,l\in\mathbb{N}$, such that, the exchange of two identical said quantum atomic entities amounts to an equivalent algebraic representation of the numerical negative sign.

No matter what variants of bi-fermions or similar constructs may be employed to serve as basic units for carrying and processing quantum information, it is straightforward by following virtually the same exact recipes to construct a system of such bi-fermions or similar constructs that either is able to simulate homophysically any quantum physics of interest or has a GSQC Hamiltonian whose ground state can be designed to encode the solution of any prescribed BQP computational problem. In the latter case, the GSQC Hamiltonian can be so judiciously designed by translating the notions of SFF Hamiltonians and their property of local node-determinacy to said system of bi-fermions or similar constructs, such that the GSQC Hamiltonian can be as well simulated efficiently using MCMC. The only thing to note here, once and for all, is that a discrete bi-fermion may require the use a substrate space $\mathbb{X}\neq\mathbb{L}_3$, if it is preferred to construct an SFF-FS or SFF-DU Hamiltonian for GSQC. On the other hand, it should also be borne in mind that the fermionic Schr\"odinger property is not indispensable for universal and efficient GSQC as well as MCQC.

Specifically, in a many bi-fermion system implementing a GSQC machine, the bi-fermions are individually distinguishable and interacting through couplings between their constituent fermions of different species, while the two fermions within each bi-fermion are identical and usually non-interacting. Each abstract quantum computational rebit associated with a binary configuration space $\{0,1\}$ and a two-dimensional Hilbert space $L^2(\{0,1\})$ is homophysically implemented as a bi-fermion rebit associated with a configuration space $\mathbb{X}^2$ and a two-dimensional Hilbert subspace $\calQ\in L^2_{\mathsmaller{F}}(\mathbb{X}^2)$ as specified in equation (\ref{DefHilbertQ}). An abstract quantum computational $n$-rebit system, $n\in\mathbb{N}$ associated with a $2^n$-valued configuration space $\{0,1\}^n$ and a $2^n$-dimensional Hilbert space $L^2(\{0,1\}^n)$ is homophysically mapped to a {\em quantum $n$-bi-fermion system} that consists of $n$ distinguishable quantum bi-fermion rebits associated with a configuration space $(\mathbb{X}^2)^n=\mathbb{X}^{2n}$ and a tensor product Hilbert space $\calQ^{\otimes n}\subseteq L^2_{\mathsmaller{F}}(\mathbb{X}^{2n})$, henceforth called the {\em quantum state space of $n$ bi-fermions}, which further can be isophysically mapped to a {\em rectified $n$-bi-fermion system} that consists of $n$ distinguishable rectified bi-fermion rebits associated with a configuration space $\mathbb{X}_*^{2n}\defeq(\mathbb{X}_*^2)^n$, $\forall n\ge 1$, henceforth referred to as the {\em simultaneously individually rectified} (SIR) configuration space of $n$ bi-fermions, and a tensor product Hilbert space $\calR^{\otimes n}\subseteq L^2(\mathbb{X}_*^{2n})$, henceforth referred to as the {\em rectified state space of $n$ bi-fermions}, following the standard procedure of constructing/defining a tensor product vector space \cite{Lang02} to represent the quantum state of a many-body system.

An abstract $n$-rebit ($n\in\mathbb{N}$) quantum computer can be realized as either a quantum or a rectified $n$-bi-fermion system via the following homophysics $\mathfrak{M}$. An abstract $n$-rebit quantum state
\begin{equation}
|\psi\rangle=\sum_{s_1\cdots s_n}c_{s_1\cdots s_n}|s_1\cdots s_n\rangle\in L^2(\{0,1\}^n),\;c_{0\cdots 0}\in\mathbb{R}_{\ge 0},\;c_{s_1\cdots s_n}\in\mathbb{R},\;\forall(s_1\cdots s_n)\in\{0,1\}^n, \label{GeneralMultiRebitPsi}
\end{equation}
is embodied by a system of $n$ bi-fermions at the state
\begin{equation}
\mathfrak{M}(|\psi\rangle)\,\defeq\,|\psi\rangle_{\mathsmaller{\rm BF}}=\sum_{s_1\cdots s_n}\mathfrak{M}(c_{s_1\cdots s_n})\bigotimes_{i=1}^n\left\{\mathlarger{\mathlarger{[}}(1-s_i)a_{\mathsmaller{L}i}^++s_ia_{\mathsmaller{R}i}^+\mathlarger{\mathlarger{]}}a_{\mathsmaller{P}i}^+\right\}|{\sf vac}\rangle\,\in\,\calY^{\otimes n} \label{AbstractRebitsToBiFermionRebitOne}
\end{equation}
in the second-quantization representation, where $\calY=\calQ$ or $\calY=\calR$, $a_{\mathsmaller{W}i}^+$ with $W=L,P,R$ denotes respectively the fermion creation operator for the $|L\rangle$, $|P\rangle$, $|R\rangle$ singe-particle orbit of an $i$-th fermion species that moves in its substrate space, or
\begin{align}
\mathfrak{M}(|\psi\rangle) \,\defeq\,& |\psi\rangle_{\mathsmaller{\rm BF}} \nonumber \\[0.75ex]
\;=\;& \sum_{s_1\cdots s_n}\mathfrak{M}(c_{s_1\cdots s_n})\bigotimes_{i=1}^n\textstyle{\frac{1}{\sqrt{2}}}\left[(1-s_i)\left(|L_1\rangle_i|P_2\rangle_i-|L_2\rangle_i|P_1\rangle_i\right)+s_i\left(|R_1\rangle_i|P_2\rangle_i-|R_2\rangle_i|P_1\rangle_i\right)\right] \nonumber \\[0.75ex]
\;=\;& \sum_{s_1\cdots s_n}\mathfrak{M}(c_{s_1\cdots s_n})\bigotimes_{i=1}^n\textstyle{\frac{1}{\sqrt{2}}}\left\{\left[(1-s_i)|L_1\rangle_i+s_i|R_1\rangle_i\right]|P_2\rangle_i-\left[(1-s_i)|L_2\rangle_i+s_i|R_2\rangle_i\right]|P_1\rangle_i\right\} \nonumber \\[0.75ex]
\;=\;& 2^{\mathsmaller{-}n/2}\sum_{\chi_1=0}^1\cdots\sum_{\chi_n=0}^1(-1)^{\sum_{i=1}^n\chi_i}|\psi_{\chi_1\cdots\chi_n}\rangle_{\mathsmaller{\rm BF}}\,\in\,\calY^{\otimes n} \label{AbstractRebitsToBiFermionRebitTwo}
\end{align}
in the first-quantization representation, with
\begin{equation}
|\psi_{\chi_1\cdots\chi_n}\rangle_{\mathsmaller{\rm BF}} \,\defeq\, \sum_{s_1\cdots s_n}\mathfrak{M}(c_{s_1\cdots s_n})\bigotimes_{i=1}^n\left[(1-s_i)|L_{1+\chi_i}\rangle_i+s_i|R_{1+\chi_i}\rangle_i\right]|P_{2-\chi_i}\rangle_i, \label{AbstractRebitsToBiFermionRebitThree}
\end{equation}
$\forall(\chi_1\cdots\chi_n)\in\{0,1\}^n$, where $\forall i\in[1,n]$, $\forall j\in\{1,2\}$, $|W_j\rangle_i$ with $W_j=L,P,R$ denotes the $j$-th artificially labeled identical fermion of the $i$-th species occupying respectively the $|L\rangle,|P\rangle,|R\rangle$ single-particle orbit. In equations (\ref{AbstractRebitsToBiFermionRebitOne}-\ref{AbstractRebitsToBiFermionRebitThree}), the coefficients $\mathfrak{M}(c_{s_1\cdots s_n})$, $(s_1\cdots s_n)\in\{0,1\}^n$ are either $\mathbb{R}$-valued or $\mathbb{R}'$-valued, depending upon the codomain of the $\mathfrak{M}$ homophysics is associated with a quantum or a rectified $n$-bi-fermion system.

It should be reemphasized that, adopting the so-called {\em logic-$0$-positive convention or representation}, namely, imposing the condition that the coefficient $c_0$ in equations (\ref{DefHilbertQ}) and (\ref{DefHilbertR}) or the coefficient $c_{0\cdots 0}$ in equation (\ref{GeneralMultiRebitPsi}) should be non-negative, is for conventional and notational convenience, which does not induce any loss of generality, since a global phase factor (a +/- sign, to be exact) in or out of a global wavefunction does not change any physics. Besides, were it really necessary to track the overall phase of a wavefunction representing a collection of bi-fermions, such collection of bi-fermions could always be augmented by the addition of a single bi-fermion, called the {\em sign keeper}, to which only two possible wavefunctions $|R_1\rangle|P_2\rangle-|R_2\rangle|P_1\rangle$ or $|R_2\rangle|P_1\rangle-|R_1\rangle|P_2\rangle$ in the $\calQ$ representation might be assigned, corresponding to $|R_1\rangle|P_2\rangle$ or $|R_2\rangle|P_1\rangle$ respectively in the $\calR$ representation, so that an overall phase factor for the pre-augmentation collection of bi-fermions would be kept by the additional sign keeper, while the augmented collection of bi-fermions as a whole could still be fully represented by a wavefunction as prescribed in equation (\ref{GeneralMultiRebitPsi}).

In the first-quantization representation with the two identical fermions of each bi-fermion artificially labeled, an abstract quantum computational $n$-rebit state $|\psi\rangle=\sum_{s_1\cdots s_n}c_{s_1\cdots s_n}|s_1\cdots s_n\rangle$ is embodied as a superposition of $2^n$ equally weighted wavefunctions $|\psi_{\chi_1\cdots\chi_n}\rangle_{\mathsmaller{\rm BF}}$, $(\chi_1\cdots\chi_n)\in\{0,1\}^n$, each of which is homomorphic to the abstract quantum computational state $|\psi\rangle$, and each of which can be obtained from another through fermion exchanges $\{(\pi_{12})_i:i\in[1,n]\}$, with $(\pi_{12})_i$ denoting an exchange of the two identical fermions of the $i$-th bi-fermion, $\forall i\in[1,n]$. Furthermore, it is easily verified that, under any quantum gate operation corresponding to a partial Hamiltonian $h\in\calB(\calY^{\otimes n})$ that is necessarily symmetric under permutations of identical particles of the same species, the transformation of $|\psi\rangle_{\mathsmaller{\rm BF}}$ can be interpreted as having each wavefunction $|\psi_{\chi_1\cdots\chi_n}\rangle_{\mathsmaller{\rm BF}}$, $(\chi_1\cdots\chi_n)\in\{0,1\}^n$, which behaves much like a boltzmannonic wavefunction of distinguishable particles, subject to a boltzmannonic quantum gate operation $\Base(h)$, simultaneously and in parallel. Intuitively, the homophysical mapping between an abstract $n$-rebit quantum register and a system of $n$ bi-fermions is as if the bi-fermionic system created a $2^n$-fold ``multiverse'', where in each of the $2^n$ ``parallel universes'', there were a quantum register homomorphic to the abstract $n$-rebit register, albeit such a fictitious quantum register comprises artificially labeled identical fermions, and the states of all of the fictitious quantum registers must be superposed with equal weights and specific signs to constitute a $|\psi\rangle_{\mathsmaller{\rm BF}}$ that is antisymmetric under exchanges of identical fermions. Clearly, there exists a $2^n$-fold redundancy.

Such redundancy can be handled effectively via an orbit partition under a group action \cite{Rotman99,Kerber99,Lin94,Bekka00,Jacobson09,Rose09}. When acting on the configuration space of the spatial coordinates of $n$ bi-fermions $\mathbb{X}^{2n}$ or on the Hilbert spaces $\calQ^{\otimes n}$ of $n$-bi-fermion wavefunctions, the fermion-exchange operators $(\pi_{12})_i$ and $(\pi_{12})_j$ of two different bi-fermions indexed by $i,j\in[1,n]$ commute, so $\bfS_n\defeq\prod_{i=1}^n\{I,(\pi_{12})_i,*\}$ as a product group is abelian, called the {\em exchange symmetry group of $n$ bi-fermions}. All even-parity operator products of the form $\prod_{i=1}^n(\pi_{12})_i^{\xi_i}$ such that $\{\xi_i\}_{i\in[1,n]}\!\subset\!\{0,1\}^n$, $\sum_{i=1}^n\xi_i\!\equiv\!0\;(\bmod\; 2)$ form a subgroup, which shall be referred to as the {\em alternating group of $n$ bi-fermions} and denoted by
\begin{equation}
\bfA_n \,\defeq\, \mathlarger{\mathlarger{\{}}\,\textstyle{\prod_{i=1}^n}(\pi_{12})_i^{\xi_i}\!:\xi_i\in\{0,1\},\,\forall i\in[1,n],\,\textstyle{\sum_{i=1}^n}\xi_i\,\equiv\,0\;(\bmod\; 2)\,\mathlarger{\mathlarger{\}}}. \label{AltGroupOfNBiFermions}
\end{equation}
There are exactly two cosets of $\bfA_n$ in $\bfS_n$, namely, $\forall i\in[1,n]$, it holds that $(\pi_{12})_i\bfA_n\cap\bfA_n=\emptyset$ and $(\pi_{12})_i\bfA_n\cup\bfA_n=\bfS_n$. Under the group action of $\bfA_n$ on the set $\mathbb{X}^{2n}$, the orbit $\bfA_nq\defeq\{r\in\mathbb{X}^{2n}\!:\exists\,\pi\in\bfA_n\;\mbox{such that}\;r=\pi q\}$ of any point $q\in\mathbb{X}^{2n}$ is a subset of equivalent coordinates from the standpoint of quantum physics of $n$ bi-fermions, as the fermion-exchange symmetries associated with bi-fermions dictate that all of the points in an orbit $\bfA_nq$, $q\in\mathbb{X}^{2n}$ must be assigned exactly the same wave amplitude, including specifically the $\pm$ signs, for any legitimate quantum state of a system of $n$ quantum bi-fermions, or equivalently, for a classical random walk simulating the corresponding system of $n$ rectified bi-fermions, all of the points in an orbit $\bfA_nq$, $q\in\mathbb{X}^{2n}$ always assume exactly the same probability distribution at the stationary state of the Markov chain. Indeed, when Monte Carlo-simulating a system of $n$ bi-fermions, the Markov chain is best interpreted as a random walk over the quotient space $\mathbb{X}^{2n}\!/\!\bfA_n\defeq\{\bfA_nq:q\in\mathbb{X}^{2n}\}$, where a visit of a representative point $q\in\mathbb{X}^{2n}$ by the walker should be considered as all of the $2^{n-1}$ points in the orbit $\bfA_nq$ being visited simultaneously, by $2^{n-1}$ walker replicas if one will imagine, and a walk from a representative point $q\in\mathbb{X}^{2n}$ to another representative point $r\in\mathbb{X}^{2n}$ should be understood as representing a Markov state transition from the orbit subset $\bfA_nq$ to the orbit subset $\bfA_nr$, that may be pictured as $2^{n-1}\times 2^{n-1}$ walkers moving simultaneously and independently, each of which escapes from a representative point $q'\in\bfA_nq$ and arrives at a representative point $r'\in\bfA_nr$. Still further, it is perfectly practical and a good alternative to do Monte Carlo with a Markov chain over the original configuration space $\mathbb{X}^{2n}$ with redundancy, but incorporate a step/subroutine of random walk to sample from the alternating group $\bfA_n$, which picks a random permutation uniformly from the set $\bfA_n$ and makes a transition from a given configuration point $q\in\mathbb{X}^{2n}$ to any $r\in\bfA_nq$ with equal probability. Such random sampling from a permutation group is rather similar to techniques used in conventional quantum Monte Carlo simulations for modeling exchange effects of identical particles, either bosons or fermions \cite{Ceperley96,Bernu02,Militzer00}.

Consider the set of computational basis states $\{|s_1\cdots s_n\rangle\in L^2(\{0,1\}^n):s_i\in\{0,1\},\forall i\in[1,n]\}$ associated with $n$ computational rebits of a quantum computer, which are henceforth referred to as the {\em computational $n$-rebit basis states}, and the set of homophysical images $\{\mathfrak{M}(|s_1\cdots s_n\rangle)\in\calY^{\otimes n}:s_i\in\{0,1\},\forall i\in[1,n]\}$, $\calY\in\{\calQ,\calR\}$ that are associated with an homophysical $n$-bi-fermion system in the first quantization representation with the pair of identical fermions of each bi-fermion artificially labeled as the first and the second respectively, which are henceforth referred to as the {\em quantum $n$-bi-fermion basis states}. Except for a negligible subset with a small measure $O\!\left(\gamma_0^{{-}2}\right)$ in quantum probability ({\it i.e.}, quantum amplitude squared), for most configuration point $q\in\mathbb{X}^{2n}$, there exists a unique $n$-bit array or string $s\defeq(s_1\cdots s_i\cdots s_n)\in\{0,1\}^n$, signified as $s(q)\defeq(s_1\cdots s_i\cdots s_n)_q\defeq(s_1(q)\cdots s_i(q)\cdots s_n(q))$, with which the computational basis state $|s\rangle\defeq|s_1\cdots s_i\cdots s_n\rangle$, signified as $|s(q)\rangle\defeq|s_1\cdots s_i\cdots s_n\rangle_q$, is homophysically mapped to a unique $n$-bi-fermion basis state $|\mathfrak{M}(s(q))\rangle\in\calY^{\otimes n}$ such that the point amplitude $\langle q|\mathfrak{M}(s(q))\rangle\neq 0$, {\em in which case $q$ is said to be signed}, with $\sign(q)\defeq{+}1$ when $\langle q|\mathfrak{M}(s(q))\rangle>0$, and $\sign(q)\defeq{-}1$ when $\langle q|\mathfrak{M}(s(q))\rangle<0$. Otherwise, $\sign(q)\defeq 0$ if $q\in\mathbb{X}^{2n}$ cannot be signed. For any $q\in\mathbb{X}^{2n}$ that is signed, let $\abs(q)$ denote the unique equivalent point in the orbit ${\bf A}_nq$, called {\em the absolute value of $q$}, which has every bi-fermion so configured as to have the first identical fermion located in a logic well and the second localized in the Pauli well. Otherwise, $\abs(q)\defeq q\in\mathbb{X}^{2n}$ if $\sign(q)=0$.
For any signed $q\in\mathbb{X}^2$, the state vectors $\sign(q)|s(q)\rangle\in L^2(\{0,1\}^n)$, $\sign(q)|\mathfrak{M}(s(q))\rangle\in\calQ^{\otimes n}$, and $\pi^{(1-\sign(q))/2}|\mathfrak{M}(s(q))\rangle\in\calR^{\otimes n}$ with any odd permutation $\pi\in\bfS_n\setminus\bfA_n$ are respectively called {\em the computational $n$-rebit basis state, the quantum $n$-bi-fermion basis state, and the rectified $n$-bi-fermion basis state consistent with the configuration point $q$}. For any signed $q\in\mathbb{X}^{2n}$ that corresponds to an $n$-bit array $s(q)=(s_1(q)\cdots s_i(q)\cdots s_n(q))\in\{0,1\}^n$, define $\calI_0(q)\defeq\{i\in[1,n]:s_i(q)=0\}$ and $\calI_1(q)\defeq\{j\in[1,n]:s_j(q)=1\}$. Otherwise, let $\calI_0(q)=\calI_1(q)=\emptyset$ if $q$ cannot be signed. Let
\begin{equation}
\Pi_{\bm{0}}\defeq {\textstyle{\prod_{i\,\in\,\calI_0(q)}}} (\pi_{12})_i, \; \Pi_{\bm{1}}\defeq {\textstyle{\prod_{j\,\in\,\calI_1(q)}}} (\pi_{12})_j, \; \Pi_{\bm{*}}\defeq\Pi_{\bm{0}}\Pi_{\bm{1}} = {\textstyle{\prod_{i\,\in\,[1,n]}}} (\pi_{12})_i \label{defiPi0Pi1Pi*}
\end{equation}
denote the operations of exchanging the pairs of identical fermions simultaneously for, respectively, all of the bi-fermions that register a bit $0$, all of the bi-fermions that register a bit $1$, and all of the bi-fermions indiscriminately.

\begin{definition}{(Sentinel Configuration Points, Configuration Space, and Representation)}\label{defiSentinelStuff}\\
Given a signed $q\in\mathbb{X}^{2n}$ corresponding to a bit array $(s_1\cdots s_n)_q\in\{0,1\}^n$, let $(\rho_0,\rho_1)\in\{0,1\}^2$ be such that $\rho_i \equiv \calI_i \pmod*{2}$, $\forall i\in\{0,1\}$. Such a signed $q\in\mathbb{X}^{2n}$ is called a sentinel configuration point if $q=\Pi_{\bm{0}}^{\,\delta_0\rho_0\,}\Pi_{\bm{1}}^{\,\delta_1\rho_1\,}\abs(q)$ holds true for some choices of $(\delta_0,\delta_1)\in\{0,1\}^2$. For any general $q\in\mathbb{X}^{2n}$ that is signed, let $\sent(q)$ denote the unique sentinel configuration point within the equivalence class ${\bf A}_nq$. The subset $\mathbb{S}^n\defeq\{\sent(q):q\in\mathbb{X}^{2n},\,\sign(q)\neq 0\}\subseteq\mathbb{X}^{2n}$ is called the sentinel configuration space. For any given quantum or rectified state $\psi \in \calQ^{\otimes n} \defeq [\calQ(\mathbb{X}^2)]^{\otimes n} \subseteq L_{\!\mathsmaller{F}}^2(\mathbb{X}^{2n})$ or $\psi \in \calR^{\otimes n} \defeq [\calR(\mathbb{X}_*^2)]^{\otimes n} \subseteq L^2(\mathbb{X}_*^{2n})$, the domain-restricted wavefunction $P_{\scalebox{0.6}{$\,\mathbb{S}^n$}}\psi\defeq\psi|_{\scalebox{0.6}{$\mathbb{S}^n$}}\subseteq L^2(\mathbb{S}^n)$ is called the sentinel representation of $\psi$. In particular, the sentinel representation of an either quantum or rectified $n$-bi-fermion basis state is called a sentinel $n$-bi-fermion basis state.
\vspace{-1.5ex}
\end{definition}

In the above definition, $P_{\scalebox{0.6}{$\,\mathbb{S}^n$}}:\calY^{\otimes n}\mapsto L^2(\mathbb{S}^n)\subseteq L^2(\mathbb{X}^{2n})$, $\calY\in\{\calQ,\calR\}$ denotes a projection operator such that $P_{\scalebox{0.6}{$\,\mathbb{S}^n$}}\psi\defeq\psi\times\mathbbm{1}_{\scalebox{0.6}{$\mathbb{S}^n$}}$, $\forall\psi\in\calY^{\otimes n}$, with $\mathbbm{1}_{\scalebox{0.6}{$\mathbb{S}^n$}}$ being the indicator function for the set $\mathbb{S}^n$, and $L^2(\mathbb{S}^n)$ represents a Hilbert space of square integrable functions having $\mathbb{S}^n$ as their domain, over either the conventional field $\mathbb{R}$ of real numbers or the isomorphic field $\mathbb{R}'$ as specified by equations (\ref{defiRprimeMul}) and (\ref{defiRprimeAdd}), depending upon the interested wavefunction is interpreted as $\psi\in\calQ^{\otimes n}$ or $\psi\in\calR^{\otimes n}$. In the latter case, the Hilbert subspace over the field $\mathbb{R}'$ and supported by the sentinel configuration space $\mathbb{S}^n$ is denoted specifically as $\calS^n\defeq L^2(\mathbb{S}^n)$, and called the {\em sentinel state space of $n$ bi-fermions}. For any $|\psi'\rangle=\int_{\raisebox{0.3\height}{\tiny $q\!\in\!\mathbb{S}^n$}}\psi'(q)\,|q\rangle\,dV_g(q)\in\calS^n$, the coefficient $\psi'(q)$ is $\mathbb{R}'$-valued, $\forall q\in\mathbb{S}^n$, which if negative would carry a $\pi\in\bfS_n\setminus\bfA_n$ operator representing a negative sign. If $n$ is odd, then for each signed configuration point $q\in\mathbb{X}^{2n}$, there is precisely one index set $\calI_{\bm{i_1}}(q)$ whose cardinality $|\calI_{\bm{i_1}}(q)|$ is odd, with $i_1\in\{0,1\}$, consequently, a negative sign in the coefficient $\psi'(q)$ is equivalent to the operator $\Pi_{\bm{i_1}}$ and can be, and indeed will always be assumed to be, absorbed into the state vector by turning $|q\rangle$ into $|\Pi_{\bm{i_1}}q\rangle$, so to leave only the absolute value $|\psi'(q)|$ in the coefficient, while maintaining the closedness and completeness of the sentinel configuration space $\mathbb{S}^n$.

{\em From this point on, as often being reminded explicitly, it is convenient and without loss of generality to assume or make the total number of bi-fermions to be odd in any physics or system of concern, because, if necessary, an auxiliary and dummy bi-fermion can always be added to any interested system of bi-fermions, with said auxiliary and dummy bi-fermion not interacting with any bi-fermion or rebit, but just remaining at a fixed ground state and helping to book-keep the negative sign.}

Let $\mathbb{Y}^n$ with $\mathbb{Y}\in\{\mathbb{X}^2,\mathbb{X}_*^2\}$ be a configuration space and $\calY^{\otimes n}$ with $\calY\in\{\calQ,\calR\}$ be a Hilbert space representing either the quantum or the rectified state space of $n$ bi-fermions, $\calB(\calY^{\otimes n})$ be a Banach algebra of bounded linear operators on $\calY^{\otimes n}$, so that $(\mathbb{Y}^n,\calY^{\otimes n},\calB(\calY^{\otimes n}))$ constitutes a quantum physics/system of $n$ bi-fermions, where the $n$ bi-fermions are distinguishable and can interact with each other, while no interaction is allowed between the two identical fermions constituting each bi-fermion. Construct a mapping between physical systems $\mathfrak{M}:(\mathbb{Y}^n,\calY^{\otimes n},\calB(\calY^{\otimes n}))\mapsto(\mathbb{S}^n,\calS^n,\calB(\calS^n))$ such that 1) $\mathfrak{M}(U)=\{\sent(q):q\in U\}$, for any subset $U\subseteq\mathbb{Y}^n$, $\mathfrak{M}^{{-}1}(V)=\bigcup_{q\in V}(\bfA_nq\cap\mathbb{Y}^n)$ conversely, for any subset $V\subseteq\mathbb{S}^n$; 2) $\mathfrak{M}(\psi) \defeq P_{\!\mathsmaller{F}}^{{-}1}\psi \defeq \psi|_{\scalebox{0.6}{$\mathbb{S}^n$}}\in\calS^n$, $\forall\psi\in\calY^{\otimes n}$, and conversely, $\mathfrak{M}^{-1}(\phi)=P_{\!\mathsmaller{F}}\phi\in\calY^{\otimes n}$, $\forall\phi\in\calS^n$, where $P_{\!\mathsmaller{F}}$ is the familiar fermionic antisymmetrization operator defined as $P_{\!\mathsmaller{F}}\phi\defeq const \times \sum_{\pi\in\bfS_n}[\sign(\pi)\times\phi\circ\pi]$, $\forall\phi\in\calS^n\subseteq L^2(\mathbb{X}^{2n})$ (note the peculiar definition regarding the operator $P_{\!\mathsmaller{F}}^{{-}1}$); 3) $\mathfrak{M}(O)=P_{\!\mathsmaller{F}}^{{-}1}OP_{\!\mathsmaller{F}}\in\calB(\calS^n)\subseteq\calL(\calS^n)$ for all $O\in\calB(\calY^{\otimes n})$, and conversely, $\mathfrak{M}^{{-}1}(O_*)=P_{\!\mathsmaller{F}}O_*P_{\!\mathsmaller{F}}^{{-}1}\in\calB(\calY^{\otimes n})\subseteq\calL(\calY^{\otimes n})$, $\forall O_*\in\calB(\calS^n)$, which define the rules of correspondence and simultaneously the domains and ranges of the mappings between operators. Note that any operator $O\in\calB(\calY^{\otimes n})$ is symmetric and invariant under any exchange permutation $\pi\in\bfS_n$, thus the mapping $\mathfrak{M}:\calB(\calY^{\otimes n})\mapsto\calB(\calS^n)$ is one-to-one and invertible.

\begin{lemma}{(Isophysicality of Sentinel Representations)}\label{SentinelIsophysics}\\
The mapping $\mathfrak{M}:(\mathbb{Y}^n,\calY^{\otimes n},\calB(\calY^{\otimes n}))\mapsto(\mathbb{S}^n,\calS^n,\calB(\calS^n))$ just defined is an isophysics.
\end{lemma}
\vspace{-4.0ex}
\begin{proof}[\iftoggle{ForUSPTO} {Demonstration} {Proof}]
It follows straightforwardly from the definitions.
\vspace{-2.0ex}
\end{proof}

Note that the sentinel configuration space $\mathbb{S}^n$ and the sentinel state space $\calS^n$ can be considered rectified spaces, just like the SIR configuration subspace $\mathbb{X}_*^{2n}$ and the rectified state space $\calR^{\otimes n}$, where the Hilbert spaces $\calR^{\otimes n}$ and $\calS^n$ are over the field of rectified numbers $\mathbb{R}'$, with the negative sign being implemented by exchanging pairs of identical fermions comprising the bi-fermions, so to avoid the use of any conventional negative real numbers in $\mathbb{R}$, thus afford a pure probabilistic interpretation of any wavefunction in $\calR^{\otimes n}$ or $\calS^n$, and facilitate an isophysical simulation of the physics $(\mathbb{X}_*^{2n},\calR^{\otimes n},\calB(\calR^{\otimes n}))$ or $(\mathbb{S}^n,\calS^n,\calB(\calS^n))$ using Monte Carlo on a classical probabilistic machine. In relation to the original (quantum) physics $(\mathbb{X}^{2n},\calQ^{\otimes n},\calB(\calQ^{\otimes n}))$, for any interested quantum state in the $\calQ^{\otimes n}$-representation, denoted by $\psi\in\calQ^{\otimes n}\subseteq L_{\!\mathsmaller{F}}^2(\mathbb{X}^{2n})$, the rectified physics $(\mathbb{X}_*^{2n},\calR^{\otimes n},\calB(\calR^{\otimes n}))$ or $(\mathbb{S}^n,\calS^n,\calB(\calS^n))$ works with a wavefunction $\psi'\defeq\psi|_{\mathbb{U}^n}\in\calR^{\otimes n}$ or $\calS^n$, which is a portion of the positive part of the $\psi\in L_{\!\mathsmaller{F}}^2(\mathbb{X}^{2n})$ wavefunction, restricted to a configuration subspace $\mathbb{U}^n\subset\mathbb{X}^{2n}$, $\mathbb{U}\in\{\mathbb{X}_*^2,\mathbb{S}\}$, which does not intersect any of the negative nodal cells of $\psi$, namely, $\mathbb{U}^n\cap\{q\in\mathbb{X}^{2n}:\psi(q)<0\}=\emptyset$. Nevertheless, the rectified wavefunction $\psi'=\psi|_{\mathbb{U}^n}$ still contains complete and sufficient information for the concerned quantum physics and facilitates isophysical simulations of it by a BPP machine, as the entire quantum state can always be fully recovered through the formula $\psi=P_{\!\mathsmaller{F}}\psi'=const \times \sum_{\pi\in\bfS_n}[\sign(\pi)\times\psi'\circ\pi]$.

For the sake of mathematical formality and rigor, let us spell out the construction of tensor product Hilbert spaces as well as an explicit ground state homophysics between a system of many quantum bi-fermions and a system of many rectified bi-fermions. From the two-dimensional Hilbert spaces $\{\calQ_i\}_{i=1}^n$ of $n\in\mathbb{N}$ distinguishable quantum bi-fermions,
\begin{equation}
\calQ_i \,\defeq\, \left\{\textstyle{\frac{1}{\sqrt{2}}}\left(c_{i0}|L_1\rangle_i|P_2\rangle_i+c_{i1}|R_1\rangle_i|P_2\rangle_i\right)+
{\rm E.S.T.}:c_{i0},c_{i1}\in\mathbb{R}\right\},\;\forall i\in[1,n],
\end{equation}
the tensor product Hilbert space ${\mathsmaller{\bigotimes}}_{i=1}^n\calQ_i\defeq F(\prod_{i=1}^n\!\calQ_i)/{\raisebox{-.2\height}{$\simX{Q}$}}$ is constructed to describe the whole quantum system of $n$ bi-fermions, which is the quotient of the free vector space $F(\prod_{i=1}^n\!\calQ_i)$ generated by the Cartesian product set $\prod_{i=1}^n\!\calQ_i$ \cite{Lang02} over the tensor equivalence relation ``{\raisebox{-.2\height}{$\simX{Q}$}}'' that is characterized by the following equations \cite{Lang02}, $\forall(u_1,\cdots\!,u_n)\in F(\prod_{i=1}^n\!\calQ_i)$, $\forall u'_i\in\calQ_i$,
\begin{align}
& \mbox{\em Component-wise Linearity \!\rm :} \nonumber \\[0.75ex]
& (u_1,\cdots\!,u_i+u'_i,\cdots\!,u_n)\;{\raisebox{-.2\height}{$\simX{Q}$}}\;(u_1,\cdots\!,u_i,\cdots\!,u_n)+(u_1,\cdots\!,u'_i,\cdots\!,u_n),\;\forall i\in[1,n], \\[0.75ex]
& \mbox{\em Commutativity and Distributivity of Scalar Multiplication \!\rm :} \nonumber \\[0.75ex]
& r(u_1,\cdots\!,u_i,\cdots\!,u_j,\cdots\!,u_n)\;{\raisebox{-.2\height}{$\simX{Q}$}}\;(u_1,\cdots\!,ru_i,\cdots\!,u_j,\cdots\!,u_n) \\[0.75ex]
& \;\;\;\;\;\;\;\;\;\;\;\;\;\;\;\;\;\;\;\;\;\;\;\;\;\;\;\;\;\;\;\;\;\;\;\;\;\;\;\;\;\;{\raisebox{-.2\height}{$\simX{Q}$}}\;(u_1,\cdots\!,u_i,\cdots\!,ru_j,\cdots\!,u_n),
\;\forall r\in\mathbb{R},\;\forall i,j\in[1,n]. \nonumber
\end{align}
The equivalence class containing the element $(u_1,\cdots\!,u_i,\cdots\!,u_n)\in F(\prod_{i=1}^n\!\calQ_i)$,
\begin{equation}
[(u_1,\cdots\!,u_i,\cdots\!,u_n)]{\raisebox{.2\height}{$_{\simX{Q}}$}}\defeq u_1\otimes\cdots\otimes u_i\otimes\cdots\otimes u_n
\defeq|u_1\rangle\cdots|u_i\rangle\cdots|u_n\rangle\in\textstyle{{\mathsmaller{\bigotimes}}_{i=1}^n\calQ_i},
\end{equation}
is called the tensor product of the vectors $\{u_i\!\in\!\calQ_i\}_{i=1}^n$. A fundamental tenet of many-body quantum mechanics is that all possible quantum states of a composite system consisting of multiple parts or particles are represented completely and exactly by the tensor product space generated by the individual Hilbert spaces describing the constituent parts or particles, subject to further constraints of bosonic or fermionic exchange symmetry when indistinguishable particles are involved. In particular, for the concerned case of multiple bi-fermions each consisting of two identical fermions, the ``commutativity of scalar multiplication'' property of tensor product asserts that fermion-exchange operators associated with any two bi-fermions commute and invert each other, namely, if $(\pi_{12})_i$ denotes the operation of exchanging the two fermions for the $i$-th bi-fermion, and $(\pi_{12})_j$ for the $j$-th bi-fermion, $i,j\in[1,n]$, then the equivalence identity $(\pi_{12})_i(\pi_{12})_j\,{\raisebox{-.2\height}{$\simX{Q}$}}\,(\pi_{12})_j(\pi_{12})_i\,{\raisebox{-.2\height}{$\simX{Q}$}}\,I$ holds, $\forall i\in[1,n]$, $\forall j\in[1,n]$, over the tensor product Hilbert space ${\mathsmaller{\bigotimes}}_{i=1}^n\calQ_i$.

Similarly, with the state spaces of $n\!\in\!\mathbb{N}$ distinguishable rectified bi-fermions
\begin{equation}
\calR_i \,\defeq\, \left\{c_{i0}(-\pi_{12})_i^{(c_{i0}<0)}|L_1\rangle_i|P_2\rangle_i+c_{i1}(-\pi_{12})_i^{(c_{i1}<0)}|R_1\rangle_i|P_2\rangle_i:c_{i0},c_{i1}\in\mathbb{R}\right\},\;\forall i\in[1,n],
\end{equation}
the Cartesian product set $\prod_{i=1}^n\!\calR_i$ generates a free vector space $F(\prod_{i=1}^n\!\calR_i)$ over the field of scalars $(\mathbb{R}',+,*)$, which reduces to a tensor product Hilbert space ${\mathsmaller{\bigotimes}}_{i=1}^n\calR_i\defeq F(\prod_{i=1}^n\!\calR_i)/{\raisebox{-.1\height}{$\simX{R}$}}$ modulo a tensor equivalence relation ``${\raisebox{-.1\height}{$\simX{R}$}}$'' that is characterized by the following equations, $\forall(v_1,\cdots\!,v_n)\in F(\prod_{i=1}^n\!\calR_i)$, $\forall v'_i\in\calR_i$,
\begin{align}
& \mbox{\em Component-wise Linearity \!\rm :} \nonumber \\[0.75ex]
& (v_1,\cdots\!,v_i+v'_i,\cdots\!,v_n)\;{\raisebox{-.1\height}{$\simX{R}$}}\;(v_1,\cdots\!,v_i,\cdots\!,v_n)+(v_1,\cdots\!,v'_i,\cdots\!,v_n),\;\forall i\in[1,n], \\[0.75ex]
& \mbox{\em Commutativity and Distributivity of Scalar Multiplication \!\rm :} \nonumber \\[0.75ex]
& r'(v_1,\cdots\!,v_i,\cdots\!,v_j,\cdots\!,v_n)\;{\raisebox{-.1\height}{$\simX{R}$}}\;(v_1,\cdots\!,r'v_i,\cdots\!,v_j,\cdots\!,v_n) \label{r'ScalarMul} \\[0.75ex]
& \,\;\;\;\;\;\;\;\;\;\;\;\;\;\;\;\;\;\;\;\;\;\;\;\;\;\;\;\;\;\;\;\;\;\;\;\;\;\;\;\;\;{\raisebox{-.1\height}{$\simX{R}$}}\;(v_1,\cdots\!,v_i,\cdots\!,r'v_j,\cdots\!,v_n),
\;\forall r'\in(\mathbb{R}',+,*),\;\forall i,j\in[1,n]. \nonumber
\end{align}
The equivalence class containing the element $(v_1,\cdots\!,v_i,\cdots\!,v_n)\in F(\prod_{i=1}^n\!\calR_i)$,
\begin{equation}
[(v_1,\cdots\!,v_i,\cdots\!,v_n)]{\raisebox{.1\height}{$_{\simX{R}}$}}\defeq v_1\otimes\cdots\otimes v_i\otimes\cdots\otimes v_n
\defeq|v_1\rangle\cdots|v_i\rangle\cdots|v_n\rangle\in\textstyle{{\mathsmaller{\bigotimes}}_{i=1}^n\calR_i},
\end{equation}
is called the tensor product of the vectors $\{v_i\!\in\!\calR_i\}_{i=1}^n$. It should be noted that equation (\ref{r'ScalarMul}) makes sense only under the premise that the fermion-exchange operators of any two rectified bi-fermions commute, and invert each other, namely, $\forall i\in[1,n]$, $\forall j\in[1,n]$, the equivalence identity $(\pi_{12})_i(\pi_{12})_j\,{\raisebox{-.1\height}{$\simX{R}$}}\,(\pi_{12})_j(\pi_{12})_i\,{\raisebox{-.1\height}{$\simX{R}$}}\,I$ holds over the free vector space $F(\prod_{i=1}^n\!\calR_i)$, in which case, the subscript ``$i$'' in a fermion-exchange operator $(\pi_{12})_i$ becomes insignificant and can be omitted safely, such that $\forall r'=|r|\pi_{12}{\raisebox{-.05\height}{$^{(r<0)}$}}$, $r\in\mathbb{R}$, $\forall v_i\in\calR_i$, $i\in[1,n]$, the expression $r'v_i$ is always well-defined and equated with $|r|(\pi_{12})_i{\raisebox{-.05\height}{$^{(r<0)}$}}v_i$. Consequently, there is a large amount of redundancy among the $4^n$ computational basis states for a system of $n$ rectified bi-fermions
\begin{equation}
\left\{|W_{\chi_1}P_{3-\chi_1}\rangle_1\cdots|W_{\chi_i}P_{3-\chi_i}\rangle_i\cdots|W_{\chi_n}P_{3-\chi_n}\rangle_n:\;W\in\{L,R\},\;\chi_i\in\{1,2\},\;\forall i\in[1,n]\right\},
\end{equation}
many of which must be regarded as equivalent and the same. Specifically, the equivalence relation
\begin{align}
\textstyle{\left[\prod_{i=1}^n(\pi_{12})_i^{\xi_i}\right]}\;& |W_{\chi_1}P_{3-\chi_1}\rangle_1\cdots|W_{\chi_i}P_{3-\chi_i}\rangle_i\cdots|W_{\chi_n}P_{3-\chi_n}\rangle_n \nonumber \\[0.75ex]
{\raisebox{-.1\height}{$\simX{R}$}}\;& |W_{\chi_1}P_{3-\chi_1}\rangle_1\cdots|W_{\chi_i}P_{3-\chi_i}\rangle_i\cdots|W_{\chi_n}P_{3-\chi_n}\rangle_n
\end{align}
holds true for any computational basis state $|W_{\chi_1}P_{3-\chi_1}\rangle_1\cdots|W_{\chi_i}P_{3-\chi_i}\rangle_i\cdots|W_{\chi_n}P_{3-\chi_n}\rangle_n$, so long as $\sum_{i=1}^n\xi_i\,\equiv\,0\;(\bmod\; 2)$, where $\xi_i\in\{0,1\}$, $\forall i\in[1,n]$. Obviously, the equivalence classes of the computational basis states are divided into two groups, which are labeled even and odd respectively, such that each even equivalence class has a representative computational basis state that has no tensor factor of the form $|W_2P_1\rangle_i$, $W\in\{L,R\}$ associated with any rectified bi-fermion $i\in[1,n]$, while each odd equivalence class has a representative computational basis state that has one and only one tensor factor of the form $|W_2P_1\rangle_i$, $W\in\{L,R\}$ associated with a certain rectified bi-fermion indexed by $i\in[1,n]$.

The tensor product Hilbert space ${\mathsmaller{\bigotimes}}_{i=1}^n\calQ_i$ or ${\mathsmaller{\bigotimes}}_{i=1}^n\calR_i$ still contains much redundancy, because the state of a system of $n$ quantum or rectified bi-fermions is sufficiently and uniquely represented by an element, called a ray, from the projective Hilbert space $\mathbb{P}({\mathsmaller{\bigotimes}}_{i=1}^n\calQ_i)$ or $\mathbb{P}({\mathsmaller{\bigotimes}}_{i=1}^n\calR_i)$, that ignores any nonzero scaling in the respective base Hilbert space,
\begin{align}
\mathbb{P}(\textstyle{{\mathsmaller{\bigotimes}}_{i=1}^n\calQ_i}) \,\defeq\,& (\textstyle{{\mathsmaller{\bigotimes}}_{i=1}^n\calQ_i})/{\raisebox{-.2\height}{$\simX{\bf Q}$}}, \\[0.75ex]
\mathbb{P}(\textstyle{{\mathsmaller{\bigotimes}}_{i=1}^n\calR_i}) \,\defeq\,& (\textstyle{{\mathsmaller{\bigotimes}}_{i=1}^n\calR_i})/{\raisebox{-.1\height}{$\simX{\bf R}$}},
\end{align}
with the projective equivalence relations ${\raisebox{-.2\height}{$\simX{\bf Q}$}}$ and ${\raisebox{-.1\height}{$\simX{\bf R}$}}$ defined as
\begin{align}
\forall u,u'\in\textstyle{{\mathsmaller{\bigotimes}}_{i=1}^n\calQ_i},\;u\;{\raisebox{-.2\height}{$\simX{\bf Q}$}}\;u' \;& \mbox{if}\;\exists r\in(\mathbb{R}_{\neq0 },*)\;\mbox{such that}\;u'=ru, \\[0.75ex]
\forall v,v'\in\textstyle{{\mathsmaller{\bigotimes}}_{i=1}^n\calR_i},\;v\;{\raisebox{-.1\height}{$\simX{\bf R}$}}\;v' \;& \mbox{if}\;\exists s\in(\mathbb{R}'_{\neq0 },*)\;\mbox{such that}\;v'=sv.
\end{align}
In particular, any state $[\psi]{\raisebox{.2\height}{$_{\simX{\bf Q}}$}}\in\mathbb{P}({\mathsmaller{\bigotimes}}_{i=1}^n\calQ_i)$ describing a system of $n$ quantum bi-fermions is a ${\raisebox{-.2\height}{$\simX{\bf Q}$}}$-equivalence class of a representative vector $\psi\in{\mathsmaller{\bigotimes}}_{i=1}^n\calQ_i$ such that
\begin{equation}
\psi=\sum_{(s_1\cdots s_i\cdots s_n)\in\{0,1\}^n}r_{s_1\cdots s_i\cdots s_n}|W_1(s_1)P_2\rangle_1\cdots|W_1(s_i)P_2\rangle_i\cdots|W_1(s_n)P_2\rangle_n\,+\,{\rm E.S.T.}, \label{psiQGeneral}
\vspace{-0.75ex}
\end{equation}
in which $W_1:\{0,1\}\mapsto\{L_1,R_1\}$ is a mapping such that $W_1(0)=L_1$, $W_1(1)=R_1$, and the first nonzero element in the sequence $\{r_{s_1\cdots s_i\cdots s_n}:(s_1\cdots s_i\cdots s_n)\in[0,2^n)\}\subset\mathbb{R}^n$ is positive, where the sequence is so arranged with the index $(s_1\cdots s_i\cdots s_n)$ as a binary integer increasing from $0$ to $2^n-1$. Similarly, any state $[\psi']{\raisebox{.1\height}{$_{\simX{\bf R}}$}}\in\mathbb{P}({\mathsmaller{\bigotimes}}_{i=1}^n\calR_i)$ describing a system of $n$ rectified bi-fermions is a ${\raisebox{-.1\height}{$\simX{\bf R}$}}$-equivalence class of a representative vector $\psi'\in{\mathsmaller{\bigotimes}}_{i=1}^n\calR_i$ such that
\begin{equation}
\psi'=\sum_{(s_1\cdots s_i\cdots s_n)\in\{0,1\}^n}r'_{s_1\cdots s_i\cdots s_n}|W_1(s_1)P_2\rangle_1\cdots|W_1(s_i)P_2\rangle_i\cdots|W_1(s_n)P_2\rangle_n, \label{psi'RGeneral}
\end{equation}
where $\forall(s_1\cdots s_i\cdots s_n)\in\{0,1\}^n$, the coefficient $r'_{s_1\cdots s_i\cdots s_n}\in\mathbb{R}'$, and the first nonzero element of the sequence $\{r'_{s_1\cdots s_i\cdots s_n}:(s_1\cdots s_i\cdots s_n)\in[0,2^n)\}\subset\mathbb{R}'^n$, denoted by $r'_*$, is positive, in the sense that $r'_*=(r_*,I)\in\mathbb{R}'$ for an $r_*\in\mathbb{R}_{>0}$, when again the sequence $\{r'_{s_1\cdots s_i\cdots s_n}:(s_1\cdots s_i\cdots s_n)\in[0,2^n)\}$ is so arranged to have the index $(s_1\cdots s_i\cdots s_n)$ as a binary integer increasing from $0$ to $2^n-1$. It is noted that $\forall(s_1\cdots s_i\cdots s_n)\in\{0,1\}^n$, the $\mathbb{R}'$-coefficient $r'_{s_1\cdots s_i\cdots s_n}$ for the basis state $|W_1(s_1)P_2\rangle_1\cdots|W_1(s_i)P_2\rangle_i\cdots|W_1(s_n)P_2\rangle_n$ needs in principle just one operator $\pi_{12}$ of fermion-exchange, which can be equivalently applied onto any of the $n$ rectified bi-fermions, to represent a sign flip for the coefficients of the $n$-rebit computational basis states, although in the previously discussed scheme of SIR representation, it is chosen deliberately to always apply simultaneously an odd number $n_0\in 2\mathbb{N}+1$ of $\pi_{12}$ operators in total, each of which to one of $n_0$ bi-fermions in a cluster. For a subsystem that collects all of the bi-fermions for ordinary rebits, it is possible, although not absolutely necessary, to always keep all bi-fermions rectified except at most one, called the {\em single unrectified delegate} (SUD), to carry an overall $\pi_{12}$ operator/sign if present, then any quantum gate $U$ that moves an $i$-th bi-fermion of ordinary rebit, $i\in[1,n]$ can be effected by firstly transferring the if-present $\pi_{12}$ operator from the old SUD to the $U$-moved $i$-th bi-fermion, so that the old SUD becomes rectified and the $U$-moved $i$-th bi-fermion becomes the new SUD, then applying the gate $U$, such that, any $\mathbb{R}'$-coefficient in the expansion of the quantum state of the subsystem of ordinary rebits, in terms of the computational basis states, similar to what has been illustrated in equation (\ref{psi'RGeneral}), never has more than one operator $\pi_{12}$ of fermion-exchange, both before and after any gate operation on a single bi-fermion controlled by multiple bi-fermions.

It is noted that for each $i\in[1,n]$, the Hilbert spaces $\calQ_i$ for a quantum bi-fermion and $\calR_i$ for a rectified bi-fermion can both be regarded as subspaces of a larger Hilbert space
\begin{equation}
\calX_i \,\defeq\, \left\{c_{i0}|L_1\rangle_i|P_2\rangle_i+c_{i1}|R_1\rangle_i|P_2\rangle_i+c_{i2}|L_2\rangle_i|P_1\rangle_i+c_{i3}|R_2\rangle_i|P_1\rangle_i:c_{i0},c_{i1},c_{i2},c_{i3}\in\mathbb{R}\right\}\!.
\end{equation}
Naturally, the tensor product Hilbert spaces ${\mathsmaller{\bigotimes}}_{i=1}^n\calQ_i$ and ${\mathsmaller{\bigotimes}}_{i=1}^n\calR_i$ can be considered as subspaces of the tensor product Hilbert space ${\mathsmaller{\bigotimes}}_{i=1}^n\calX_i\defeq F(\prod_{i=1}^n\!\calX_i)/{\raisebox{-.1\height}{$\simX{S}$}}$, which is the quotient space of the free vector space $F(\prod_{i=1}^n\!\calX_i)$ modulo a tensor equivalence relation ${\raisebox{-.1\height}{$\simX{S}$}}$ characterized by
\begin{align}
& \mbox{\em Component-wise Linearity \!\rm :} \nonumber \\[0.75ex]
& (w_1,\cdots\!,w_i+w'_i,\cdots\!,w_n)\;{\raisebox{-.1\height}{$\simX{S}$}}\;(w_1,\cdots\!,w_i,\cdots\!,w_n)+(w_1,\cdots\!,w'_i,\cdots\!,w_n),\;\forall i\in[1,n], \\[0.75ex]
& \mbox{\em Commutativity and Distributivity of Scalar Multiplication \!\rm :} \nonumber \\[0.75ex]
& r(w_1,\cdots\!,w_i,\cdots\!,w_j,\cdots\!,w_n)\;{\raisebox{-.1\height}{$\simX{S}$}}\;(w_1,\cdots\!,rw_i,\cdots\!,w_j,\cdots\!,w_n) \\[0.75ex]
& \;\;\;\;\;\;\;\;\;\;\;\;\;\;\;\;\;\;\;\;\;\;\;\;\;\;\;\;\;\;\;\;\;\;\;\;\;\;\;\;\;\;\;\;{\raisebox{-.1\height}{$\simX{S}$}}\;(w_1,\cdots\!,w_i,\cdots\!,rw_j,\cdots\!,w_n),
\;\forall r\in\mathbb{R},\;\forall i,j\in[1,n], \nonumber
\end{align}
$\forall(w_1,\cdots\!,w_n)\in F(\prod_{i=1}^n\!\calX_i)$, $\forall w'_i\in\calX_i$. Each wavefunction in ${\mathsmaller{\bigotimes}}_{i=1}^n\calR_i$ of the form
\begin{equation}
\sum_{(s_1\cdots s_i\cdots s_n)\in\{0,1\}^n}\!\!\mathlarger{\mathlarger{(}}r_{s_1\cdots s_i\cdots s_n},\pi_{12}^{t_{s_1\cdots s_i\cdots s_n}}\mathlarger{\mathlarger{)}}|W_1(s_1)P_2\rangle_1\cdots|W_1(s_i)P_2\rangle_i\cdots|W_1(s_n)P_2\rangle_n,
\end{equation}
with $r_{s_1\cdots s_i\cdots s_n}\in\mathbb{R}_{\ge 0}$, $t_{s_1\cdots s_i\cdots s_n}\in\{0,1\}$, $\forall(s_1\cdots s_i\cdots s_n)\in\{0,1\}^n$, can be interpreted as a vector
\begin{equation}
\sum_{(s_1\cdots s_i\cdots s_n)\in\{0,1\}^n}\!\!r_{s_1\cdots s_i\cdots s_n}\mathlarger{\mathlarger{(}}\pi_{12}^{t_{s_1\cdots s_i\cdots s_n}}|W_1(s_1)P_2\rangle_1\cdots|W_1(s_i)P_2\rangle_i\cdots|W_1(s_n)P_2\rangle_n\mathlarger{\mathlarger{)}}
\end{equation}
in the conical hull $\coni[({\mathsmaller{\bigotimes}}_{i=1}^n\calX_i)_{\star}]$ spanned by
\begin{align}
\textstyle{({\mathsmaller{\bigotimes}}_{i=1}^n\calX_i)_{\star}} \,\defeq\, \mathlarger{\mathlarger{\{}}\pi_{12}^{t_{s_1\cdots s_i\cdots s_n}}|W_1(s_1)P_2\rangle_1\cdots|W_1(s_i)P_2\rangle_i\cdots|W_1(s_n)P_2\rangle_n: \nonumber \\[0.75ex]
t_{s_1\cdots s_i\cdots s_n}\in\{0,1\},\;(s_1\cdots s_i\cdots s_n)\in\{0,1\}^n\mathlarger{\mathlarger{\}}}
\end{align}
with scalar coefficients from the set $\mathbb{R}_{\ge 0}$, where each wavefunction is non-negative-valued at any location in the bi-fermion configuration space $\mathbb{X}^{2n}$. In this way of interpretation, ${\mathsmaller{\bigotimes}}_{i=1}^n\calR_i$ is a subset of $\coni[({\mathsmaller{\bigotimes}}_{i=1}^n\calX_i)_{\star}]$, which is in turn a convex cone in the Hilbert space ${\mathsmaller{\bigotimes}}_{i=1}^n\calX_i$. Similarly, a typical dyad operator in $\calB({\mathsmaller{\bigotimes}}_{i=1}^n\calR_i)$ of the form $|\psi'\rangle\langle\phi'|$, $\psi',\phi'\in{\mathsmaller{\bigotimes}}_{i=1}^n\calR_i$ can be interpreted as an element in the convex cone $\coni[({\mathsmaller{\bigotimes}}_{i=1}^n\calX_i)_{\star}({\mathsmaller{\bigotimes}}_{i=1}^n\calX_i)_{\star}^+]\subset\calB({\mathsmaller{\bigotimes}}_{i=1}^n\calX_i)$, with
\begin{equation}
\textstyle{({\mathsmaller{\bigotimes}}_{i=1}^n\calX_i)_{\star}({\mathsmaller{\bigotimes}}_{i=1}^n\calX_i)_{\star}^+}\,\defeq\,\{|u'\rangle\langle v'|:u'\in\textstyle{({\mathsmaller{\bigotimes}}_{i=1}^n\calX_i)_{\star}},\,v'\in\textstyle{({\mathsmaller{\bigotimes}}_{i=1}^n\calX_i)_{\star}}\}.
\end{equation}
Therefore, by virtue of fermion sign rectification and as the basis of Monte Carlo quantum computing, a collection of multiple quantum bi-fermions with a Hilbert space $\calH\defeq{\mathsmaller{\bigotimes}}_{i=1}^n\calQ_i$ and the Banach algebra $\calB(\calH)$ can be isophysically represented by a system of multiple rectified bi-fermions with a Hilbert space $\calH'\defeq{\mathsmaller{\bigotimes}}_{i=1}^n\calR_i$ and the Banach algebra $\calB(\calH')$, through an isophysics $\mathfrak{M}:(\mathbb{X}^{2n},\calH,\calB(\calH))\mapsto(\mathbb{X}_*^{2n},\calH',\calB(\calH'))$, such that any wavefunction $\psi'\in\calH'$, interpreted as an element in $\coni[({\mathsmaller{\bigotimes}}_{i=1}^n\calX_i)_{\star}]\subset{\mathsmaller{\bigotimes}}_{i=1}^n\calX_i$, is just the positive part of the vector $\psi=\mathfrak{M}^{\mathsmaller{-}1}(\psi')\in\calH\subset{\mathsmaller{\bigotimes}}_{i=1}^n\calX_i$, with $\psi$ and $\psi'$ mutually and uniquely determining each other through the isophysics $\mathfrak{M}$. Having all wavefunctions and bounded linear operators positively valued in the first-quantized configuration space representation, the quantum physics $(\mathbb{X}_*^{2n},\calH',\calB(\calH'))$ becomes amenable to Monte Carlo simulations on a classical probabilistic computer.

A ground state computational physics is basically concerned with the $C_0$-semigroups, also known as strongly continuous one-parameter semigroups \cite{Engel00}, of the Gibbs operators $\{G(\tau H_k)\}_{k=0}^{\mathsmaller{K}}$, $\tau\in[0,\infty)$ that are associated with the set of self-adjoint FBM tensor monomials $\{H_k\}_{k=1}^{\mathsmaller{K}}$ and the Hamiltonian of total energy $H_0\defeq\sum_{k=1}^{\mathsmaller{K}}H_k$ as the infinitesimal generators, such that $G(\tau H_k)\defeq\exp({-}\tau H_k)$, $\forall k\in[0,K]$, $\forall\tau\in[0,\infty)$. Of particular interest is to drive a GSQC system to its ground state $\psi_0(H)$ by applying the Gibbs operators, and to sample from the probability distribution $|\psi_0(H;q\in\mathbb{X}^{2n})|^2dV_g$. If all of the operators $H_k$, $k\in[0,K]$ are either of a finite rank or unbounded from above, then all $G(\tau H_k)$, $k\in[0,K]$ are compact operators, $\forall\tau\in(0,\infty)$. Furthermore, if all $H_k$, $k\in[0,K]$ are lower-bounded by zero, then $\{G(\tau H_k)\}_{\tau\in[0,\infty)}$, $k\in[0,K]$ are all contraction semigroups. GSQC arises naturally from and is widely useful in the contexts of computational physics and quantum chemistry, as well as non-physics-related computational problems of optimizations and searches, where an objective function is encoded in the form of an energy functional resembling the Hamiltonian of a fictitious physical system. However, the power of GSQC far transcends the evaluation and optimization of energy functionals. As having been well established previously \cite{Feynman85,Kitaev02,Mizel04,Kempe06}, as well as will be detailed in the next section, GSQC is actually universal for quantum computing in that, any quantum dynamics as a series of (a polynomial number of) unitary operations on a quantum system that is either executing or being simulated by a quantum algorithm, can be encoded into the ground state stationary properties of a larger quantum system with a polynomially bounded overhead, whose configuration space has an extra spatial dimension encoding the time variable of the concerned quantum dynamics. In a sense, GSQC is all that is needed as far as quantum computability and quantum computational complexity are concerned.

Nevertheless, it is still inspiring to note, even just in passing, that a quantum dynamics which either describes a quantum system or is implemented on a quantum computer based on rebits, where the effects of unitary state transformations are emphasized, can be described by a quantum physics $(\mathbb{X}^{2(n{+}1)},\calH\otimes\calQ,\calB(\calH\otimes\calQ))$, $n\in\mathbb{N}$ in conjunction with a self-adjoint operator $H\in\calL_0(\calH\otimes\calQ)$ designated as the Hamiltonian representing the total energy of the system, where the Hilbert spaces $\calH\defeq\calH(\mathbb{X}^{2n})$ and $\calQ=\{c_0|0\rangle+c_1|1\rangle:c_0,c_1\in\mathbb{R}\}$, as well as the Banach algebra $\calB(\calH\otimes\calQ)=\calB(\calH)\otimes\calB(\calQ)$ and the Banach space $\calL_0(\calH\otimes\calQ)$, are all defined with respect to the field of real numbers. By Stone's theorem \cite{Engel00}, even though itself may not lie in $\calB(\calH\otimes\calQ)$, $H$ can generate a strongly continuous one-parameter unitary group $\{U(\tau H)\defeq\exp(\tau H\otimes\bfi)\}_{\tau\ge 0}$, which acts on the tensor product Hilbert space $\calH\otimes\calQ$ and is contained in the Banach algebra $\calB(\calH\otimes\calQ)$, with $\calQ$ being the state space associated with an auxiliary rebit to flag and represent the real and imaginary parts of complex-valued amplitudes, and $\bfi\defeq[0, -1; 1, 0]\in GL(2,\mathbb{R})$ representing an operator in $\calB(\calQ)$ that constitutes a matrix embodiment of the imaginary unit, such that $\bfi^2=-{\bf 1}$, with ${\bf 1}\defeq[1,0;0,1]\in GL(2,\mathbb{R})$ representing the identity operator in $\calB(\calQ)$. The matrix $\bfi$ generates a Banach algebra $\mathbb{R}[\bfi]$, that consists of matrices with a non-negative definite determinant, the square root of which serves as the norm. $\mathbb{R}[\bfi]$ is also a unital and commutative division algebra, constituting an isomorphic representation  $\mathbb{R}[\bfi]\cong\mathbb{C}$ of the space of complex numbers, both as an algebraic field and as a topological group. Now $\calH\otimes\calQ$ can be considered as a Hilbert space over the field $\mathbb{R}[\bfi]$ which is isomorphic to the field of complex numbers $\mathbb{C}$, and $\calB(\calH\otimes\calQ)$ can be regarded as a Banach algebra also over the field $\mathbb{R}[\bfi]\cong\mathbb{C}$. Therefore, $(\mathbb{X}^{2(n{+}1)},\calH\otimes\calQ,\calB(\calH\otimes\calQ))$ constitutes an isophysical representation for any general quantum physics $(\mathbb{X}^{2n},\calH^*(\mathbb{X}^{2n}),\calB^*(\calH^*))$ involving complex-valued state vectors, with a Hilbert space $\calH^*\!\defeq\calH^*(\mathbb{X}^{2n})$ and the Banach algebra $\calB^*\!\defeq\calB^*(\calH^*)$ being defined with respect to the field of complex numbers $\mathbb{C}$. Through another homophysics $\mathfrak{M}$ using our field of rectified numbers $\mathbb{R}'$ and an imaginary unit $\bfi'\defeq[0,\pi_{12};1,0]\in GL(2,\mathbb{R}')$, the universal rebit-based quantum physics $(\mathbb{X}^{2(n{+}1)},\calH\otimes\calQ,\calB(\calH\otimes\calQ))$ can be homophysically mapped into a quantum system $\scalebox{1.1}{(}\mathbb{X}_*^{2n{+}2},\calH'\otimes\calR,\calB(\calH'\otimes\calR)\scalebox{1.1}{)}$ of rectified bi-fermions, with $\calH'\!\defeq\mathfrak{M}(\calH)$, $\calR=\mathfrak{M}(\calQ)$ representing the rectified state space of a single bi-fermion as defined in (\ref{DefHilbertR}). Thus, even a quantum dynamics involving complex-valued wavefunctions being transformed by a one-parameter unitary group can be sign-rectified and implemented in a system of many rectified bi-fermions, where both the wavefunctions and the unitary groups are represented without involving any negative real number, consequently, the quantum dynamics is homophysically mapped into a standard Markovian process in classical probability theory, which is amenable to efficient Monte Carlo simulations on a classical probabilistic machine.

In spirit, our method of fermion sign rectification using bi-fermions and representing the numerical ``$-$'' sign by the fermion-exchange operator $\pi_{12}$ is analogous to the Janzing-Wocjan representation \cite{Janzing06} and its application in the Jordan-Gosset-Love technique of stoquastization \cite{Jordan10}, where the multiplicative group $(\{1,-1\},*)$ for the sign of real numbers is represented by the multiplicative group $(\{I,\sigma^x\},*)$ of operators on an ancilla qubit, so that any designer Hamiltonians for QMA-complete and universal adiabatic computations using excited states can be converted into stoquastic matrices. $I$ and $\sigma^x$ are respectively the $2\times 2$ identity matrix and the Pauli matrix measuring spin along the $x$ axis. Also similarly, our method of fermion sign rectification and the  Jordan-Gosset-Love technique both entail a multi-fold degeneracy/redundancy of sorts, when encoding the solution of a general BQP problem into an eigenstate of the correspondingly constructed stoquastic Hamiltonian. The  Jordan-Gosset-Love technique introduces an ancilla qubit whose $\sigma^x$ eigenstates $|\pm\rangle=(|0\rangle\pm|1\rangle)/\sqrt{2}$ create two computational manifolds, where the $|-\rangle$-associated manifold encodes the general BQP problem of interest, while the $|+\rangle$-associated manifold corresponds to a stoquastic Hamiltonian that has little to do with the original BQP problem. Our method of fermion sign rectification, using a system of bi-fermions and the fermion-exchange operator $\pi_{12}$, encodes the solution of a general BQP problem into the ground state of a many-bi-fermion system consisting of two classes of probability amplitude distributions that are isomorphic and symmetric with respect to each other, with one class being amplitude distributions in positive nodal cells, and the other class being amplitude distributions in negative nodal cells. There is a crucial difference though. While the  Jordan-Gosset-Love technique either encodes the desired solution in a certain excited state associated with the $|-\rangle$ state of the ancilla, or alternatively biases the $|-\rangle$ state to a lower energy than the $|+\rangle$ state of the ancilla, thus unfortunately spoils the stoquasticity of the total Hamiltonian, our method of fermion sign rectification has the desired solution encoded exactly in the ground state, whose nodal structure can always be determined locally in the configuration space, thanks to the property of local node-determinacy afforded by the employed SFF Hamiltonian. That is fundamentally why our method of fermion sign rectification overcomes the infamous sign problem and enables efficient Monte Carlo simulations of SFF-FS or SFF-EB systems.

More specifically, the essential property of a designer Hamiltonian being separately frustration-free endows our method of fermion sign rectification with distributive and efficient solutions of nodal surfaces associated with FBM tensor monomials, which are all consistent with and collectively determine the global nodal structure of the unique ground state of the designer Hamiltonian. Localization of nodal surfaces afforded by the property of a designer Hamiltonian being separately frustration-free is the central pillar for solving the sign problem, which can even do without using bi-fermions. Indeed, a general bi-fermion supported by a continuous or discrete configuration space $\mathbb{X}^2$ can be equivalently encoded into a pair of ordinary rebits, with the first ordinary rebit bearing the $0$ or $1$ logic value, and the second indicating a positive or negative sign for the quantum amplitude, such that a bi-fermion state $a|\!\downarrow\rangle+b|\!\uparrow\rangle$ with $a,b\in\mathbb{R}$ is encoded, namely, isophysically implemented, as $\phi\defeq a|\!\downarrow\rangle+b|\!\uparrow\rangle\IsoPhysLR\half(|a|+a)|00\rangle+\half(|a|-a)|10\rangle+\half(|b|+b)|01\rangle+\half(|b|-b)|11\rangle\defeq\mathfrak{M}(\phi)$, with the representation of $\mathfrak{M}(\phi)$ involving no negative real numbers. Essentially, the $\pi_{12}$ operator that acts on a bi-fermion and represents a negative sign is in turn mapped isophysically to the $\sigma^x=|1\rangle\langle 0|+|0\rangle\langle 1|$ operator acting on the sign-indicating ordinary rebit. Then, any single-bi-fermion-moving GSQC gate $h\defeq I-|\phi\rangle\langle\phi|$ is isophysically implemented as $\mathfrak{M}(h)=I-|\mathfrak{M}(\phi)\rangle\langle\mathfrak{M}(\phi)|$, which is stoquastic, and consequently, the associated Gibbs operator $\mathfrak{M}[G(\tau h)]\defeq\exp[{-}\tau\mathfrak{M}(h)]=e^{{-}\tau}I+(1-e^{{-}\tau})|\phi\rangle\langle\phi|$ for any $\tau\ge 0$ has all non-negative matrix elements when represented in terms of the four computational basis states $|00\rangle,|01\rangle,|10\rangle,|11\rangle$ of the conventional two-rebit system, so that the Markovian transition matrix $[\phi]\,G(\tau h)\,[\phi]^{{-}1}$ becomes a bona fide stochastic matrix in terms of the same basis states. Furthermore, and clearly, it is even possible to employ just one sign-indicating ordinary rebit shared by many ordinary rebits representing a collection of many bi-fermions, because, as being noted before, it is sufficient to restrict all computational basis states of a many-bi-fermion system to contain no more than one negatively configured bi-fermion at any given time, and it is permissible to transfer any $\pi_{12}$ operator freely from one bi-fermion to any other bi-fermion.

\iftoggle{ForUSPTO} {
\subsection*{\bf Universal Quantum Circuits Using Bi-Fermion Rebits} \label{UQCUBFR}
} {
\section{Universal Quantum Circuits Using Bi-Fermion Rebits} \label{UQCUBFR}
}
Having rigorously established the BPP solvability of SFF-FS and/or SFF-EB systems, also thoroughly discussed bi-fermions and interactions among them as examples of building blocks, it only remains to demonstrate how an SFF-FS or SFF-EB system with a designer Hamiltonian for universal GSQC can be constructed out of such bi-fermion building blocks, which is able to encode any BQP computation in the designer Hamiltonian, more specifically, its ground state. In terms of a general quantum system consisting of abstract computational rebits, such a construct has been well established and known as the Feynman-Kitaev construct of time-space circuit-Hamiltonian mapping \cite{Feynman85,Kitaev02,Mizel04,Kempe06,Biamonte08,Jordan10,Nagaj10,Breuckmann14}, which provides a recipe to compose a designer Hamiltonian $H_{\mathsmaller{\rm FK}}$, called a Feynman-Kitaev Hamiltonian, whose ground state encodes the entire computational history of any BQP circuit that may or may not be fault-tolerant, and is given as a series of self-inverse quantum gates $\{U(t)\}_{t=1}^{\mathsmaller{T}}$, where, $\forall t\in[1,T]\subseteq\mathbb{N}$, $U(t)$ is a unitary operator with real-valued entries such that $[U(t)]^{{-}1}=U(t)$, corresponding to and representing a quantum computational gate.

In the case of quantum error correction being employed in said BQP circuit, it is necessary to note that concatenated quantum fault-tolerant encoding induces a hierarchical structure consisting of multiple levels of relatively called {\em noisy} versus {\em coded} rebits/gates, where at each level of fault-tolerant encoding, multiple error-prone noisy rebits are used to encode a single coded rebit, which in turn may serve as a noisy rebit at the next level of hierarchy to encode a higher-level coded rebit. Also, at each such level of encoding hierarchy, every ``elementary'' gate operation on one or two coded rebits comprises, or is realized by, a series of gate operations on the constituent noisy rebits. Now it can be clarified that the aforementioned series of self-inverse quantum gates $\{U(t)\}_{t=1}^{\mathsmaller{T}}$ refer to quantum operations on the truly physical rebits at the lowest level of the encoding hierarchy, which are the atomic units and building blocks of all of the multi-leveled circuitry of quantum computing, state preparation/measurement, and error-correction encoding/decoding to realize said fault-tolerant BQP circuit.

In a general Feynman-Kitaev construct, the computational history of a possibly fault-tolerant BQP circuit as a series of self-inverse quantum gates $\{U_{\mathsmaller{L}}(t)\}_{t=1}^{\mathsmaller{T}}$ is ground-state-encoded using a so-called {\em clock register} consisting of $T+1$ {\em clock rebits} and a so-called {\em logic register} comprising $N$ {\em logic rebits}, $T\in\mathbb{N}$, $N\in\mathbb{N}$, with the combined system of two quantum registers living in a Hilbert space spanned by effective computational basis states of the form
$$|c_0c_1\cdots c_{\mathsmaller{T}}\rangle_{\mathsmaller{C}}|l_1\cdots l_{\mathsmaller{N}}\rangle_{\mathsmaller{L}} , \; c_t\in\{0,1\} , \; \forall t\in[0,T] , \; l_i\in\{0,1\} , \; \forall i\in[1,N] ,$$
and being governed by a Feynman-Kitaev Hamiltonian
\begin{equation}
H_{\mathsmaller{\rm FK}} \,\defeq\, H_{\rm clock} \,+\, H_{\rm init} \,+\, H_{\rm prop}^{\rm odd} \,+\, H_{\rm prop}^{\rm even} \,, \label{HFKtotal}
\end{equation}
with each of the four partial Hamiltonians on the right hand side being DFF as
\begin{align}
H_{\rm clock} \,\defeq\;& \itPi_{\mathsmaller{C},\,0}^+ \,+\, {\textstyle{ \scalebox{1.15}{$\sum$}_{t=1}^{\mathsmaller{T}} }}
\itPi^+_{\mathsmaller{C},\,t{-}1}\otimes \itPi^-_{\mathsmaller{C},\,t} \,, \label{HFKclock} \\[0.75ex]
H_{\rm init} \,\defeq\;& \itPi^+_{\mathsmaller{C},\,1}\otimes \scalebox{1.25}{(} \,{\textstyle{ \scalebox{1.15}{$\sum$}_{i\,:\,l'_i=0} }}\,Z^-_{\mathsmaller{L},\,i} \,+\, {\textstyle{ \scalebox{1.15}{$\sum$}_{j\,:\,l'_j=1} }}\,Z^+_{\mathsmaller{L},\,j} \scalebox{1.25}{)} \,, \label{HFKinit} \\[0.75ex]
H_{\rm prop}^{\rm odd} \,\defeq\;& {\textstyle{ \scalebox{1.15}{$\sum$} }}_{t=1}^{\lfloor(\mathsmaller{T}+1)/2\rfloor}\,H_{{\rm prop},\,2t-1} \,, \label{HFKpropOdd} \\[0.75ex]
H_{\rm prop}^{\rm even} \,\defeq\;& {\textstyle{ \scalebox{1.15}{$\sum$} }}_{t=1}^{\lfloor\mathsmaller{T}/2\rfloor}\,H_{{\rm prop},\,2t} \,, \label{HFKpropEven} \\[0.75ex]
H_{{\rm prop},\,t} \,\defeq\;& \itPi_{\mathsmaller{C},\,t{-}1}^-\otimes\itPi_{\mathsmaller{C},\,t{+}1}^+\otimes \scalebox{1.15}{[} I-\itGamma_{\mathsmaller{C},\,t}\otimes U_{\mathsmaller{L}}(t) \scalebox{1.15}{]} \,, \; \forall t\in[1,T] \,, \label{HFKpropt}
\end{align}
where $\forall x \in \mathbb{R}$, $\lfloor x\rfloor$ denotes the largest integer that is smaller than or equal to $x$, $\itPi^{\pm}_{\mathsmaller{C},\,t}\defeq\half(I\pm\itPi_{\mathsmaller{C},\,t})$, $\itPi_{\mathsmaller{C},\,t}$ and $\itGamma_{\mathsmaller{C},\,t}$ are single-rebit operators acting on the $t$-th clock rebit such that $\itPi_{\mathsmaller{C},\,t}|c_t\rangle_{\mathsmaller{C},\,t}=(1-2c_t)\,|c_t\rangle_{\mathsmaller{C},\,t}$, $\itGamma_t|c_t\rangle_{\mathsmaller{C},\,t}=|(1-c_t)\rangle_{\mathsmaller{C},\,t}$, $\forall c_t\in\{0,1\}$, $\forall t\in[0,T]$, with $\itPi^+_{\mathsmaller{C},\,t=\mathsmaller{T}+1}$ reducing to the identity when there is no rebit to operate upon, $Z^{\pm}_{\mathsmaller{L},\,i}\defeq\half(I\pm Z_{\mathsmaller{L},\,i})$ is a single-rebit operator acting on the $i$-th clock rebit such that $Z_{\mathsmaller{L},\,i}|l_i\rangle_{\mathsmaller{L},\,i}=(1-2l_i)\,|l_i\rangle_{\mathsmaller{L},\,i}$, $\forall l_i\in\{0,1\}$, $\forall i\in[1,N]$, while $U_{\mathsmaller{L}}(t)$ denotes the $t$-th quantum computational gate that operates on no more than two logic rebits, $H_{{\rm prop},t}$ as defined in equation (\ref{HFKpropt}) is called the $t$-th {\em Feynman-Kitaev propagator (FK-propagator)}, $\forall t\in[1,T]$. Here and after in this \iftoggle{ForUSPTO} {specification} {presentation}, a subscript ``$\mathsmaller{C}$'' indicates a clock rebit or the clock register, a state vector of a clock rebit or the clock register, or an operator acting on one or more clock rebits. By the same token, a subscript ``$\mathsmaller{L}$'' signifies a logic rebit or the logic register, or a state vector or operator associated with logic rebits or the logic register. Such subscript ``$\mathsmaller{C}$'' or ``$\mathsmaller{L}$'' may be omitted when there is no ambiguity as to whether a clock or logic rebit is referred to, especially when a universal index is employed to address all of the rebits uniquely. The partial Hamiltonian $H_{\rm clock}$ ensures that the so-called {\em domain wall clock} \cite{Nagaj10,Breuckmann14} states $\scalebox{1.1}{\{} |t\rangle_{\mathsmaller{C}}\defeq |1\rangle_{\mathsmaller{C}}^{\otimes(t{+}1)}|0\rangle_{\mathsmaller{C}}^{\otimes(\mathsmaller{T}{-}t)} \scalebox{1.1}{\}}_{t=0}^{\mathsmaller{T}}$ lie in and span the manifold of the lowest and zero-valued energy, $H_{\rm init}$ initializes the logic register to a prescribed initial state $|\phi_0\rangle_{\mathsmaller{L}}\defeq|l'_1\cdots l'_{\mathsmaller{N}}\rangle_{\mathsmaller{L}}$, $(l'_1,\cdots\!,l'_{\mathsmaller{N}})\in\{0,1\}^{\mathsmaller{N}}$ at ``time'' $t=0$ associated with and represented by the clock state $|(t=0)\rangle_{\mathsmaller{C}}\defeq|10\cdots 0\rangle_{\mathsmaller{C}}$, while the partial Hamiltonian $H_{{\rm prop},t}$, $\forall t\in[1,T]$ facilitates a transition of the whole Feynman-Kitaev construct between the state $|(t-1)\rangle_{\mathsmaller{C}}|\phi_{t{-}1}\rangle_{\mathsmaller{L}}$ at time $t{-}1$ and the state $|t\rangle_{\mathsmaller{C}}U_{\mathsmaller{L}}(t)|\phi_{t{-}1}\rangle_{\mathsmaller{C}}\defeq|t\rangle_{\mathsmaller{C}}|U_{\mathsmaller{L}}(t)\phi_{t{-}1}\rangle_{\mathsmaller{L}}$ at time $t$, through which, the quantum gate $U_{\mathsmaller{L}}(t)$ is effected on the logic register, such that, as $t\in\mathbb{Z}$ increases from $0$ to $T$, the sequence of unitarily time-evolving quantum states $\{|\phi_t\rangle_{\mathsmaller{L}}:t\in[1,T]\}\subseteq\calH(\{0,1\}^{\mathsmaller{N}})$ are defined and generated as $|\phi_t\rangle_{\mathsmaller{L}}\defeq U_{\mathsmaller{L}}(t)|\phi_{t{-}1}\rangle_{\mathsmaller{L}}$ for all $t\in[1,T]$ recursively.

For convenience in referring to the partial Hamiltonians, let ${\sf H}_1 \defeq H_{\rm clock}$, ${\sf H}_2 \defeq H_{\rm init}$, ${\sf H}_3 \defeq H_{\rm prop}^{\rm odd}$, ${\sf H}_4 \defeq H_{\rm prop}^{\rm even}$, therefore, $H_{\mathsmaller{\rm FK}} = \sum_{k=1}^4 {\sf H}_k$. The partial Hamiltonians ${\sf H}_1$ and ${\sf H}_2$ are straightforwardly DFF. It is precisely in order to render both ${\sf H}_3$ and ${\sf H}_4$ DFF that the set of FK-propagators $\{H_{\rm prop,\,t} : t \in [1,T]\}$ are split into the two of them. Define $\size(H_{\mathsmaller{\rm FK}}) \defeq T+N+1$. As will be demonstrated below, when $\tau \in (0,\infty)$, $\tau = \Theta(\poly(\size(H_{\mathsmaller{\rm FK}})))$ is sufficiently large, the Gibbs operators $\exp(-\tau H_{\mathsmaller{\rm FK}})$ and $\{{\sf G}_k \defeq \exp(-\tau{\sf H}_k)\}_{k \in [1,4]}$ are all essentially ground state projectors, and a process of iterating the four Gibbs operators $\{{\sf G}_k\}_{k \in [1,4]}$, or the four similarity transformed Markov counterparts $\{{\sf M}_k \defeq [\psi_0({\sf H}_k)]\exp(-\tau{\sf H}_k)[\psi_0({\sf H}_k)]^{-1}\}_{k \in [1,4]}$, mixes rapidly and converges to the ground state projector $\exp(-\tau H_{\mathsmaller{\rm FK}})$ or its Markov counterpart. All combined, the Feynman-Kitaev Hamiltonian $H_{\mathsmaller{\rm FK}} = \sum_{k=1}^4 {\sf H}_k$ is SFF, which confines the system of two quantum registers to within a ground space spanned by the states $\left\{|t\rangle_{\mathsmaller{C}}|\phi_t\rangle_{\mathsmaller{L}}:t\in[0,T]\right\}$, called the {\em Feynman-Kitaev history states}, and links these Feynman-Kitaev history states into a one-dimensional {\em Feynman-Kitaev lattice}, with each lattice site being indexed by a $t \in [0,T]$ and called the $t$-th Feynman-Kitaev site, or the Feynman-Kitaev time $t$. On the Feynman-Kitaev lattice, the Hamiltonian $H_{\mathsmaller{\rm FK}}$ is unitarily similar and equivalent to a tridiagonal matrix \cite{Kitaev02,Mizel04,Kempe06,Biamonte08,Jordan10,Nagaj10,Breuckmann14}
\begin{align}
& {\textstyle{\sum_{t\,\in\,\{0,\,\mathsmaller{T}\}}}}\,|t\rangle\langle t| \,+\, 2\,{\textstyle{\sum_{t\,=\,1}^{\mathsmaller{T}-1}}}\,|t\rangle\langle t| \,-\, {\textstyle{\sum_{t\,=\,1}^{\mathsmaller{T}}}}\,|t\rangle\langle(t-1)| \,-\, {\textstyle{\sum_{t\,=\,1}^{\mathsmaller{T}}}}\,|(t-1)\rangle\langle t| \nonumber \\[0.75ex]
& \;\;\;\;\;\;\;\;\;\;\;\;=\;\left[\begin{array}{rrrrr}
 1 &     -1 &        &        &    \\
-1 &      2 &     -1 &        &    \\
   & \ddots & \ddots & \ddots &    \\
   &        &     -1 &      2 & -1 \\
   &        &        &     -1 &  1
\end{array}\right]_{(\mathsmaller{T}+1)\times(\mathsmaller{T}+1)} \,, \label{FeynmanKitaevLattice}
\end{align}
which has
\begin{equation}
\left\{\lambda_k \,\defeq\, 2-2\cos\!\left(\frac{\pi k}{T+1}\right)\right\}_{k=0}^{\mathsmaller{T}}\;\;\;\mbox{and}\;\;\;\left\{u_k(t) \,\defeq\, \cos\!\left[\frac{\pi k(2t+1)}{2(T+1)}\right]\right\}_{k=0}^{\mathsmaller{T}} \label{OneDimLatticeSolutions}
\end{equation}
as eigenvalues and the corresponding eigenvectors respectively \cite{Kitaev02}.
Obviously, the Hamiltonian $H_{\mathsmaller{\rm FK}}$ is polynomially gapped, having $0$ and $2-2\cos[\pi/(T{+}1)]=\Omega\!\left(T^{\mathsmaller{-}2}\right)$ as the lowest and second lowest eigenvalues. The unique ground state of $H_{\mathsmaller{\rm FK}}$ is given by
\begin{align}
\sqrt{T{+}1}\,|\psi_0(H_{\mathsmaller{\rm FK}})\rangle
\,=\,\;& |10\cdots 00\rangle_{\mathsmaller{C}} |l_1\cdots l_{\mathsmaller{N}}\rangle_{\mathsmaller{L}} \nonumber \\[0.75ex]
\,+\,\;& |11\cdots 00\rangle_{\mathsmaller{C}} U_{\mathsmaller{L}}(1)|l_1\cdots l_{\mathsmaller{N}}\rangle_{\mathsmaller{L}} \nonumber \\[0.75ex]
\,+\,\;& \cdots \label{Psi0HFKlong} \\[0.75ex]
\,+\,\;& |11\cdots 10\rangle_{\mathsmaller{C}} U_{\mathsmaller{L}}(T{-}1)\cdots U_{\mathsmaller{L}}(1)|l_1\cdots l_{\mathsmaller{N}}\rangle_{\mathsmaller{L}} \nonumber \\[0.75ex]
\,+\,\;& |11\cdots 11\rangle_{\mathsmaller{C}} U_{\mathsmaller{L}}(T)U_{\mathsmaller{L}}(T{-}1)\cdots U_{\mathsmaller{L}}(1)|l_1\cdots l_{\mathsmaller{N}}\rangle_{\mathsmaller{L}} \,, \nonumber
\end{align}
or in a more compact form using shorthand notations,
\begin{equation}
\sqrt{T{+}1}\,|\psi_0(H_{\mathsmaller{\rm FK}})\rangle \,=\, {\textstyle{ \scalebox{1.15}{$\sum$}_{t=0}^{\mathsmaller{T}} }}\, |t\rangle_{\mathsmaller{C}} \scalebox{1.1}{[}\,{\textstyle{ \prod_{\tau=1}^t }}U_{\mathsmaller{L}}(\tau)\scalebox{1.1}{]}|l_1\cdots l_{\mathsmaller{N}}\rangle_{\mathsmaller{L}} \,. \label{Psi0HFKshort}
\end{equation}
Being able to sample from $|\psi_0(H_{\mathsmaller{\rm FK}})\rangle$ implies that the {\em result state} $|\phi_*\rangle_{\mathsmaller{L}}\defeq U_{\mathsmaller{L}}(T)\cdots U_{\mathsmaller{L}}(1)|l_1\cdots l_{\mathsmaller{N}}\rangle_{\mathsmaller{L}}$ $=|\phi_{\mathsmaller{T}}\rangle_{\mathsmaller{L}}$ is amenable to efficient sampling, with the overhead factor $T+1$ being polynomially bounded. Better yet, a quantum circuit can be so arranged to have a number $T'=\Theta(T)$ of identity gates padded \cite{Jordan10,Nagaj10,Breuckmann14} after the last nontrivial gate $U_{\mathsmaller{L}}(T)$ for a required state transformation to generate $|\phi_*\rangle_{\mathsmaller{L}}=|\phi_{\mathsmaller{T}}\rangle_{\mathsmaller{L}}$ as the result of the desired quantum computation, so that the ground state of the Feynman-Kitaev construct contains $T'+1$ identical copies of the result state $|\phi_*\rangle_{\mathsmaller{L}}$. With such padding of identity gates, there is a boosted probability of $(T'+1)/(T+T'+1)=\Omega(1)$ to obtain a copy of the result state $|\phi_*\rangle_{\mathsmaller{L}}$ by sampling from the ground state of the Feynman-Kitaev construct. Depending on how the computational complexity $\Cost(H_{\mathsmaller{\rm FK}},T,T')$ of generating a sample from $|\psi_0(H_{\mathsmaller{\rm FK}})\rangle$ scales as the numbers of gates $T$ and $T'$ increase, the parameter $T'$ can be optimized so that the cost $\Cost(H_{\mathsmaller{\rm FK}},T,T')\times(T+T'+1)/(T'+1)$ of producing a sample from the result state $|\phi_*\rangle_{\mathsmaller{L}}$ is minimized. Throughout this \iftoggle{ForUSPTO} {specification} {presentation},
unless explicitly stated otherwise, it is always assumed that a Feynman-Kitaev construct is identity-gate-padded properly, and the variable $T$ is redefined to represent the total number of gates in the identity-gate-padded Feynman-Kitaev construct, having an optimal portion of the quantum gates being the identity operator to duplicate the computational result.

As usual, let $X_i\defeq\sigma^x_i$, $X_i^{\pm}\defeq\half(I\pm X_i)$, $Z_i\defeq\sigma^z_i$, $Z_i^{\pm}\defeq\half(I\pm Z_i)$, and $R_i(\theta)\defeq X_i\sin\theta+Z_i\cos\theta$, $\theta\in[-\pi,\pi)$, $R_i^{\pm}(\theta)\defeq\half[I\pm R_i(\theta)]$ denote the Pauli matrices as single-rebit operators acting on a general computational rebit addressed by a universal index $i\in[0,T+N]$, which may be either an $i$-th clock rebit when $i\in[0,T]$, or an $(i-T)$-th logic rebit when $i\in[T+1,T+N]$. Define shorthand notations for the {\em single-rebit-controlled single-rebit-transforming} controlled-$R/X/Z$ gates as
\begin{align}
R_{ij}(\theta) \,\defeq\,& Z^+_i+Z^-_i\otimes R_j(\theta) \,=\, Z^+_i+Z^-_i\otimes(X_j\sin\theta+Z_j\cos\theta), \\[0.75ex]
X_{ij} \,\defeq\,& Z^+_i+Z^-_i\otimes X_j \,=\, R_{ij}(\pi/2), \\[0.75ex]
Z_{ij} \,\defeq\,& Z^+_j+Z^-_i\otimes Z_j \,=\, R_{ij}(0),
\end{align}
with $i,j\in[0,T+N]$ indexing the clock or logic rebits. It is easily verified that $R_j(\theta)$, $X_j$, $Z_j$, $R_{ij}(\theta)$, $X_{ij}$, and $Z_{ij}$ are all self-inverse operators. It is also known that matrix-entry-wise real-valued quantum gates operating on real-valued wavefunctions can implement any quantum circuits \cite{Bernstein97,Rudolph02,Biamonte08} in the sense of homophysics, and the composite gate $R_{ij}(\theta)X_{ij}$ or $R_{ij}(\theta)Z_{ij}$, $i,j\in[T+1,T+N]$ is already universal for any single fixed $\theta\notin\pi\mathbb{Q}$ \cite{Rudolph02}. Therefore, it is more than sufficient to consider an arbitrary BQP algorithm implemented as a fault-tolerant quantum circuit using self-inverse operators \cite{Biamonte08} from a set of at-most-two-rebit gates ${\bf 2RG}\defeq{\bf 1RG}\cup\{R_{ij}(\theta),X_{ij},Z_{ij}:\theta\in\pi\itTheta,\,i,j\in[T+1,T+N]\}$, with ${\bf 1RG}\defeq\{I\}\cup\{R_j(\theta),X_j,Z_j:\theta\in\pi\itTheta,\,j\in[T+1,T+N]\}$, and $\itTheta\subseteq\mathbb{R}$ being any set of real number(s) that contains at least one irrational number, where each gate operation does not have to be perfect, but is allowed and assumed to err independently, so long as the probability of error is below a constant threshold set by an employed scheme of quantum fault-tolerant encoding \cite{Shor95,Steane96,Steane97,Aharonov96,Kitaev97a,Kitaev97b,Knill98,Aliferis06}, or without the use of quantum error correction, as long as the rate of error per gate is lower than a polynomial bound, such that the cumulative probability of errors being treated as hard failures will not become unacceptably too large even after the entire sequence of a polynomial number of quantum gates implementing the BQP algorithm.

Let $(\calC,\calH,\calB)$ be a quantum physics/system representing the Feynman-Kitaev construct of abstract computational rebits, where $\calC\defeq\{0,1\}^{\mathsmaller{T}+\mathsmaller{N}+1}$ is a discrete configuration space, $\calH\defeq\calH(\calC)\subseteq\mathbb{R}^{2^{\mathsmaller{T}+\mathsmaller{N}+1}}$ is a Hilbert space supported by $\calC$, and $\calB\defeq\calB(\calH)$ is a Banach algebra of operators over the field $\mathbb{R}$, which contains all of the semigroups $\{\exp({-}\tau{\sf H}_k):\tau\in[0,\infty),k\in[1,4]\}$ generated by the partial Hamiltonians of $\{{\sf H}_k\}_{k \in [1,4]}$, as well as the semigroup $\{\exp({-}\tau H_{\mathsmaller{\rm FK}}):\tau\in[0,\infty)\}$ generated by $H_{\mathsmaller{\rm FK}}=\sum_{k=1}^4{\sf H}_k$. It is straightforward to have each abstract computational rebit implemented into, or represented by, a concrete bi-fermion as being constructed and described in the previous section with great details, so to effect a homophysics $\mathfrak{M}:(\calC,\calH,\calB)\mapsto(\calC',\calH',\calB')$ from the quantum system $(\calC,\calH,\calB)$ into a many-fermion system $(\calC',\calH',\calB')$, where each of the image operators $\{\mathfrak{M}({\sf H}_k)\}_{k \in [1,4]}$ is fermionic Schr\"odinger, or essentially bounded, or both fermionic Schr\"odinger and essentially bounded, mostly importantly, as well as DFF, because each of the preimages $\{{\sf H}_k\}_{k \in [1,4]}$ is DFF by design. The above-mentioned process of iterating the Gibbs operators $\{{\sf G}_k\}_{k \in [1,4]}$ or the four Markov operators $\{{\sf M}_k\}_{k \in [1,4]}$ that mixes rapidly and converges to the ground state projector $\exp(-\tau H_{\mathsmaller{\rm FK}})$ or its Markov counterpart is homophysically mapped to a process of iterating the four Gibbs operators $\{\mathfrak{M}({\sf G}_k)\}_{k \in [1,4]}$ or the four Markov operators $\{\mathfrak{M}({\sf M}_k)\}_{k \in [1,4]}$ that mixes rapidly and converges to the ground state projector $\exp(-\tau\mathfrak{M}(H_{\mathsmaller{\rm FK}}))$ or its Markov counterpart. Therefore, the homophysical Feynman-Kitaev Hamiltonian $\mathfrak{M}(H_{\mathsmaller{\rm FK}})$ is SFF-FS, or SFF-EB, or both SFF-FS and SFF-EB thus SFF-DU. Such $\mathfrak{M}(H_{\mathsmaller{\rm FK}})$ is always amenable to efficient simulations via Monte Carlo on a classical computer.

It is important and interesting to note that, as building blocks in homophysical implementations of GSQC systems, particularly the Feynman-Kitaev constructs, the fermionic Schr\"odinger-oriented bi-fermions and the quantum gate operations on them, may be inherently imperfect and prone to errors and deviations, which necessitate considerations to control and correct them so to prevent their accumulation into a disastrous failure, similar to the need of minimizing qubit and gate error rates and employing quantum error correction in traditional experimental realizations of quantum computing devices. However, unlike the enterprise of experimentally building an actual and physical device for quantum computing, here for MCQC with theoretical and numerical implementations of a quantum circuit using theoretical and numerically simulated bi-fermions, the rebit and gate errors may be by design as a matter of principle, which can be made arbitrarily small easily and cost-effectively, to the limit of completely vanishing, also as a matter of principle. Specifically, the rate of leakage error $\Tr(\sfE_{\rm leak}(\gamma_0))$ per bi-fermion either can be diminished completely by using an essentially bounded only implementation employing $\mathbb{L}_3$ as the substrate space to support bi-fermions, or if the fermionic Schr\"odinger property is desired, then bi-fermions supported by the substrate spaces $\mathbb{X}=\mathbb{T}$ or $\mathbb{X}=\mathbb{L}^{2n}$, $n\ge 2$ can be employed, for which the rate of leakage error $\Tr(\sfE_{\rm leak}(\gamma_0))$ can be made as small as $\Tr(\sfE_{\rm leak}(\gamma_0))=O\!\left(\gamma_0^{{-}3}\log^3\gamma_0\right)$ for $\mathbb{X}=\mathbb{T}$ or $\mathbb{X}=\mathbb{L}^{2n}$ when $n$ is large, or as small as $\Tr(\sfE_{\rm leak}(\gamma_0))=O\!\left(\gamma_0^{{-}2}\right)$ for $\mathbb{X}=\mathbb{L}^{2n}$ in the limit of $n\rightarrow 2$, namely, the rate of leakage error can be rendered polynomially close to zero,, by choosing the parameter $\gamma_0>0$ sufficiently large but bounded by a polynomial of $\size(H_{\mathsmaller{\rm FK}})$, at the price of a polynomially increased time complexity to simulate such a system via a classical Monte Carlo.

As already been mentioned previously, there are two broadly categorized strategies to deal with the $\Tr(\sfE_{\rm leak}(\gamma_0))$ probability of error associated with each application of a Feynman-Kitaev gate during MCQC. The first strategy is brute-force, treating any instance of such error as a hard failure, upon which an ongoing Monte Carlo simulation is simply aborted, and a new simulation is restarted from the very beginning. At the most it requires a $\gamma_0 = O(\poly(\size(H_{\mathsmaller{\rm FK}})))$ to make sure that the total probability of such a hard failure taking place is less than $1/3$, over the course of an entire MCQC procedure consisting of an $O(\poly(\size(H_{\mathsmaller{\rm FK}})))$ number of Feynman-Kitaev gates, so that said MCQC procedure has a larger than $2/3$ probability to succeed, at the price of an $O(\poly(\size(H_{\mathsmaller{\rm FK}})))$ complexity overhead for simulating a single Feynman-Kitaev gate. The second strategy takes advantage of the well developed methods of quantum error correction, either employing techniques known as subsystem-, operator-, or Hamiltonian-encoding \cite{Kribs05a,Kribs05b,Bacon06,Jordan06} that suppress qubit errors using energy penalties, or more fundamentally and systematically, building the Feynman-Kitaev construct from a quantum circuit $\{U_{\mathsmaller{L}}(t)\}_{t=1}^{\mathsmaller{T}}$ that has quantum error correction and fault tolerance via concatenated encoding \cite{Shor95,Steane96,Steane97,Aharonov96,Kitaev97a,Kitaev97b,Knill98,Aliferis06} built in as being explained previously in this section. Then, by the celebrated quantum threshold theorem (also known as the quantum fault-tolerance theorem), the built-in quantum fault tolerance is able to correct and clean up all of the errors and detrimental effects that take the form of an inaccurate or wrong state in a logic rebit, or a total loss of a rebit (namely, leaking out of its computational basis space), or an inaccurate or wrong logic gate operation leading to such a logic rebit error, at a computational cost of no more than a polylogarithmic overhead factor to the time and space complexities, provided that the error probability $\Tr(\sfE_{\rm leak}(\gamma_0))$ is below a certain constant threshold \cite{Shor95,Steane96,Steane97,Aharonov96,Kitaev97a,Kitaev97b,Knill98,Aliferis06}.

Using any of the Algorithms \ref{MCSimuDisc} through \ref{MCSimuNoCoPathInt} with the polynomial Lie-Trotter-Kato decomposition of (\ref{LieTrotterKatoBounds}), each of the four ground state projectors $\{{\sf G}_k\}_{k \in [1,4]}$ or $\{{\sf M}_k\}_{k \in [1,4]}$ can be Monte Carlo simulated at the cost of a polynomial-bounded computational complexity. Then using Algorithm \ref{MCSimuSFF}, the ground state projector $\exp(-\tau\mathfrak{M}(H_{\mathsmaller{\rm FK}}))$ or $[\psi_0(\mathfrak{M}(H_{\mathsmaller{\rm FK}}))] \exp(-\tau\mathfrak{M}(H_{\mathsmaller{\rm FK}})) [\psi_0(\mathfrak{M}(H_{\mathsmaller{\rm FK}}))]^{-1}$, $\tau \in (0,\infty)$, $\tau = \Theta(\poly(\size(H_{\mathsmaller{\rm FK}})))$ can be simulated by simulating the sequence $\{{\sf G}_k\}_{k \in [1,4]}$ or $\{{\sf M}_k\}_{k \in [1,4]}$ repeatedly for a polynomial-bounded number of times.

\begin{theorem}{(A Second {\myTheorem} of Monte Carlo Quantum Computing)}\\
Any BQP algorithm can be polynomially reduced to the computational problem of simulating an SFF-DU Hamiltonian, with the polynomial reduction also guaranteeing warm starts.
\label{SecondTheorem}
\end{theorem}
\vspace{-3.0ex}
\begin{proof}[\iftoggle{ForUSPTO} {Demonstration} {Proof}]
Firstly, if desired, any BQP circuit that executes a prescribed BQP algorithm using a number $T_0$ of quantum gates from a fixed universal set of few-body-moving operators on $N_0$ rebits can be transformed into a fault-tolerant quantum circuit comprising a number $T=O((T_0{+}N_0)\poly(\log(T_0{+}N_0)))$ of {\em multi-rebit-controlled single-rebit-transforming} gates $\{U(t)\}_{t=1}^{\mathsmaller{T}}$ on $N=O((T_0{+}N_0)\poly(\log(T_0{+}N_0)))$ rebits, with the built-in quantum fault-tolerance ensuring successful completion of quantum computation even under gate and rebit errors, provided that the rate of errors per gate and rebit is below a threshold. Or if no quantum error correction is built into the transformed quantum circuit $\{U(t)\}_{t=1}^{\mathsmaller{T}}$, then it is only necessary to set $N=N_0$ and $T=O(T_0)$, where a constant overhead $T/T_0=O(1)$ may be incurred to decompose each of the quantum gates used by the original BQP circuit into an $O(1)$ number of the multi-rebit-controlled single-rebit-transforming gates in $\{U(t)\}_{t=1}^{\mathsmaller{T}}$, as well to pad identity gates at the end of the Feynman-Kitaev construct for boosting the probability of measuring the computational result.

Next, a Feynman-Kitaev construct can be generated for the transformed quantum circuit $\{U(t)\}_{t=1}^{\mathsmaller{T}}$, with a Hamiltonian $H_{\mathsmaller{\rm FK}}$ as defined in equations (\ref{HFKtotal}) through (\ref{HFKpropt}), such that the ground state $\psi_0(H_{\mathsmaller{\rm FK}})$ as a superposition of Feynman-Kitaev history states contains both a component that represents the initial state to start the BQP algorithm and one or more components which encodes the solution to the original computational problem, with each of such solution-encoding components carrying an $\Omega((T{+}1)^{{-}1/2})$-sized probability amplitude. When the Feynman-Kitaev construct is implemented into a GSQC system of bi-fermions through a homophysics $\mathfrak{M}$, the ground state of the homophysical image $\mathfrak{M}(H_{\mathsmaller{\rm FK}})$ can be obtained by iterating the four Gibbs operators $\{\mathfrak{M}({\sf G}_k)\}_{k \in [1,4]}$ or the four Markov operators $\{\mathfrak{M}({\sf M}_k)\}_{k \in [1,4]}$ in sequence for an $O(\poly(\size(H_{\mathsmaller{\rm FK}})))$ times. Each partial Hamiltonian $\mathfrak{M}({\sf H}_k)$, $k \in [1,4]$ is DFF as a polynomial Lie-Trotter-Kato decomposition into a number $O(T_0{+}N_0)$ of $O(1/\poly(T_0{+}N_0))$-almost node-determinate FBM tensor monomials. A warm start is readily obtained by sampling a configuration point from the initial Feynman-Kitaev history state with $t=0$.

It only remains to show that iterating the ground state projectors $\{{\sf G}_k\}_{k \in [1,4]}$ or $\{{\sf M}_k\}_{k \in [1,4]}$ in sequence per Algorithm \ref{MCSimuSFF} does mix rapidly and converge to the ground state projector $\exp(-\tau H_{\mathsmaller{\rm FK}})$ or its Markov counterpart polynomially fast, therefore the same holds true for and with the homophysical images. To that end, it is noted that repeated applications of the $\{{\sf G}_k\}_{k \in [1,4]}$ or $\{{\sf M}_k\}_{k \in [1,4]}$ operators ensure that the Feynman-Kitaev construct remains in the low-energy manifold spanned by the Feynman-Kitaev history states, using which as basis states each of the DFF operators ${\sf H}_3$ and ${\sf H}_4$ is represented as a particular form of tridiagonal matrix, that is a direct sum of $2\times 2$ matrices, each $2\times 2$ matrix mixing the quantum probability amplitudes between two adjacent Feynman-Kitaev sites. Iterating the sequence of operators $\{{\sf G}_k\}_{k \in [1,4]}$ or $\{{\sf M}_k\}_{k \in [1,4]}$ constitutes a time-inhomogeneous Markov chain \cite{Seneta81,Winkler03,Stroock05}, where an effective transition matrix ${\sf T}_1$ per each iteration, defined as ${\sf T}_1 \defeq \prod_{=1}^4{\sf G}_k$, or ${\sf T}_1 \defeq \prod_{=1}^4{\sf M}_k$, is obviously an irreducible matrix, so that the time-inhomogeneous Markov chain is bound to converge to the unique steady state in which all Feynman-Kitaev sites have an equal probability amplitude, corresponding to the Feynman-Kitaev ground state of equation (\ref{Psi0HFKlong}). Further, the ${\sf T}_1$ matrix corresponds to a linear graph on the Feynman-Kitaev lattice, where the Feynman-Kitaev sites are the vertices, and between each pair of adjacent Feynman-Kitaev sites there is an edge associated with a matrix element that is lower-bounded by $\Omega(T^{-1})$. In other words, the graph has an inverse polynomial-bounded {\em conductance} or {\em Cheeger constant}, which ensures an inverse polynomial-bounded spectral gap for the ${\sf T}_1$ matrix by the celebrated Cheeger inequalities \cite{Cheeger70,Jerrum89,Chung97,Levin08}, which in turn guarantees that the time-inhomogeneous Markov chain of interest converges polynomially fast.
\end{proof}

\begin{theorem}{(A Third {\myTheorem} of Monte Carlo Quantum Computing)}\\
{\rm BQP}$\,\subseteq\,${\rm BPP}, therefore {\rm BPP}$\,=\,${\rm BQP}.
\label{ThirdTheorem}
\end{theorem}
\vspace{-3.0ex}
\begin{proof}[\iftoggle{ForUSPTO} {Demonstration} {Proof}]
By {\myTheorem} \ref{SecondTheorem}, any BQP algorithm executed on a suitable fault-tolerant BQP circuit can be encoded into the ground state of an SFF-DU Hamiltonian, that is polynomially gapped and polynomial Lie-Trotter-Kato decomposable into a polynomial number of FBM interactions, each of which is polynomially computable and polynomially gapped. Furthermore, a warm start is guaranteed by such ground state encoding. Then BQP$\,\subseteq\,$BPP follows straightforwardly from {\myTheorem}  \ref{FirstTheorem}, while BPP$\,\subseteq\,$BQP is trivial as well known, thus BPP$\,=\,$BQP.
\vspace{-1.5ex}
\end{proof}

The significance of BPP$\,=\,$BQP can be hardly missed. It not only answers the long outstanding question of whether the laws of quantum mechanics endow more computational power, but also provides constructive and practical methods for simulating the all-important many-body quantum systems efficiently on classical machines. Furthermore, it opens up new avenues for developing and identifying efficient probabilistic algorithms from the vantage point of quantum computing. Any quantum based or inspired solution to a computational problem translates automatically into an efficient classical probabilistic algorithm. For instance, it is now certain that integer factorization is in BPP, thanks to Shor's celebrated quantum discovery \cite{Shor97} that started the long line of thoughts on quantum computing. Beyond the decision/promise problems in BPP$\,=\,$BQP, as well as function/optimization/search problems that are polynomially reducible or equivalent to them, there are potentially harder computational problems captured by the complexity classes of MA/QMA \cite{Arora09,Watrous09}), which characterize decision/promise problems whose instances can be decided by Arthur the verifier running a bounded-error probabilistic/quantum polynomial time Turing machine, after receiving from Merlin the prover a polynomial-sized classical/quantum proof message, that is, respectively, a string of classical bits whose length is polynomially bounded, or a quantum state of a polynomially bounded number of qubits. There is also the class MQA (Merlin Quantum Arthur  \cite{Watrous09}), alternatively known as QCMA (Quantum Classical Merlin Arthur \cite{Aharonov02}), which represents a subset of QMA by restricting the proof message from Merlin to Arthur to be a string of classical bits. Other than the obvious MA$\,\subseteq\,$MQA$\,\subseteq\,$QMA, it has been an open question to refine the set relationship. A straightforward \iftoggle{ForUSPTO} {Derived Utility} {corollary} of {\myTheorem} \ref{ThirdTheorem} is the following:
\begin{corollary}
{\rm MA}$\,=\,${\rm MQA}.
\end{corollary}
\vspace{-1.5ex}
Below are two other interesting \iftoggle{ForUSPTO} {Derived Utilities} {corollaries} regarding complexity classes.
\begin{corollary}
{\rm BQP}$\,\subseteq\,${\rm MA}$\,\subseteq \Sigma_2\cap\Pi_2\subseteq\,${\rm PH}.
\end{corollary}
\vspace{-4.0ex}
\begin{proof}[\iftoggle{ForUSPTO} {Demonstration} {Proof}]
This follows straightforwardly from {\myTheorem} \ref{ThirdTheorem} and the Sipser-Lautemann theorem \cite{Sipser83,Lautemann83,Canetti96}.
\end{proof}
\begin{corollary}
If {\rm NP}$\,\subseteq\,${\rm BQP}, then {\rm P}$\,=\,${\rm BPP}$\,=\,${\rm NP}$\,=\,${\rm PH}.
\end{corollary}
\vspace{-4.0ex}
\begin{proof}[\iftoggle{ForUSPTO} {Demonstration} {Proof}]
{\rm (NP}$\,\subseteq\,${\rm BQP)}$\,\wedge\,${\rm (BQP}$\,\subseteq\,${\rm BPP)}$\implies${\rm NP}$\,\subseteq\,${\rm BPP}$\implies${\rm P}$\,=\,${\rm BPP}$\,=\,${\rm NP}$\,=\,${\rm PH}, where the latter implication has been well known for some time \cite{Lautemann83,Zachos88}.
\end{proof}

\iftoggle{ForUSPTO} {
\subsection*{\bf Further Considerations of Monte Carlo Quantum Computing in Practice}
} {
\section{Further Considerations of Monte Carlo Quantum Computing in Practice}
}
For demonstrating the BPP membership of quantum algorithms, it has been sufficient to construct a quantum circuit using a polynomially bounded number $T$ of self-inverse real-valued gates $\{U(t):t\in[1,T]\}$ each of which moving no more than two rebits, with a portion of the gates being identities for state copying, then to employ a Feynman-Kitaev construct for GSQC that time-space maps the quantum circuit into a Hamiltonian $H_\mathsmaller{\rm FK}$, whose non-degenerate ground state and excited states are separated by an energy gap that scales inverse quadratically against the number $T$ of quantum gates, which is sufficiently polynomial, although not terribly efficient. Practical Monte Carlo quantum computing applications can no doubt benefit from further optimizations to achieve the most efficient probabilistic simulations.

One type of efficiency consideration tries to minimize the number $T$ of Feynman-Kitaev sites. As having been aforementioned, the many-bodiness of an interaction among rebits/bi-fermions is not necessarily a concern in Monte Carlo quantum computing, when such an interaction is a controlled operation on a number of rebits/bi-fermions conditioned on the $Z^{\pm}$ projections of a number of other rebits/bi-fermions. While actual physical interactions among three and more particles are rare in nature and difficult to realize in material implementations of physical devices, simulations on classical computers have no problem to model a {\em multi-rebit-controlled multi-rebit-transforming} gate effecting a conditional operation on multiple rebits/bi-fermions controlled by the $Z^{\pm}$ projections of any number of other rebits/bi-fermions, that includes a {\em controlled classical reversible gate array} (controlled CRGA), namely, a CRGA operating on a first plurality of rebits/bi-fermions controlled by the $Z^{\pm}$ projections of a second plurality of rebits/bi-fermions. This possibility can be taken advantage of and used to reduce the number $T$ of Feynman-Kitaev sites, by employing computationally more powerful many-body quantum gates in the forms of controlled CRGAs, as opposed to having always to break down such many-body quantum gates into pieces of few-body interactions.

Indeed, many quantum algorithms often employ classical reversible computations (such as the modular exponentiation function used in the famed Shor's algorithm of factorization \cite{Shor97}) that are CRGAs of the form
\begin{equation}
F:|s_1s_2\cdots s_k\rangle|0\rangle^{\otimes l} \mapsto |s_1s_2\cdots s_k\rangle|F(s_1s_2\cdots s_k)\rangle = |s_1s_2\cdots s_k\rangle|f_1f_2\cdots f_l\rangle, \; k,l\in\mathbb{N}, \label{FGate1}
\end{equation}
where $(s_1s_2\cdots s_k)\in\{0,1\}^l$, and $\forall i\in[1,l]$, $f_j=f_j(s_1s_2\cdots s_k)\in\{0,1\}$ denotes the $j$-th bit or component of the function value $F(s_1s_2\cdots s_k)$, as either a binary integer or a Boolean tuple. Using node-determinate controlled gates $\{U_j(f_j)\}_{j=1}^l$ of the form $\scalebox{1.2}{\{} \poly \scalebox{1.15}{(} \scalebox{1.1}{\{} Z_i^{\pm} : i\in[1,k] \scalebox{1.1}{\}} \scalebox{1.15}{)} \raisebox{.025\height}{$_j$} \otimes X_j \raisebox{-.025\height}{$^{f_j}$} \scalebox{1.2}{\}} \raisebox{.025\height}{$_{j\,\in\,[1,l]}$}$ with $\scalebox{1.2}{\{} \poly \scalebox{1.15}{(} \scalebox{1.1}{\{} Z_i^{\pm} : i\in[1,k] \scalebox{1.1}{\}} \scalebox{1.15}{)} \raisebox{.025\height}{$_j$} \raisebox{-.025\height}{$^{f_j}$} \scalebox{1.2}{\}} \raisebox{.025\height}{$_{j\,\in\,[1,l]}$}$ being tensor polynomials of the projection operators $ \scalebox{1.1}{\{} Z_i^{\pm} : i\in[1,k] \scalebox{1.1}{\}}$, the $F$-computation of equation (\ref{FGate1}) can be implemented in a sequence of just $k$ Feynman-Kitaev propagators,
\begin{equation}
F \,=\, {\textstyle{ \scalebox{1.15}{$\prod$}_{j=1}^l }} U_j[f_j(s_1s_2\cdots s_k)] \,\defeq\, {\textstyle{ \scalebox{1.3}{$\prod$}_{j=1}^l }} \scalebox{1.5}{\{} \scalebox{1.4}{\{} {\textstyle{ \scalebox{1.2}{$\bigotimes$}_{i=1}^k }} \scalebox{1.2}{[} (1-s_i)Z_i^+ + s_iZ_i^- \scalebox{1.2}{]} \scalebox{1.4}{\}} {\textstyle{ \scalebox{1.2}{$\bigotimes$} }} X_j^{f_j(s_1s_2\cdots s_k)} \scalebox{1.5}{\}}, \label{FGate2}
\end{equation}
where $\forall j\in[1,l]$, the $j$-th Feynman-Kitaev propagator applies the gate $U_j[f_j(s_1s_2\cdots s_k)]$, which conditionally flips the $j$-th bit of the register on the right, as controlled by the state $|s_1s_2\cdots s_k\rangle$ of the register on the left and in accordance with the Boolean value of the function $f_j(s_1s_2\cdots s_k)$.

Furthermore, it is possible to implement the entire computation $F$ directly within a single clock step through a multi-rebit interaction $H_t(F) \defeq \itPi^-_{\mathsmaller{C},\,t{-}1} \otimes \itPi^+_{\mathsmaller{C},\,t{+}1} \otimes (I - X_{\mathsmaller{C},\,t} \otimes F_{\mathsmaller{L},\,t} )$ without using any auxiliary rebit, where
\begin{equation}
{\textstyle{ F_{\mathsmaller{L},\,t} \,\defeq\, \scalebox{1.4}{\{} \scalebox{1.1}{$\bigotimes$}_{i=1}^k \scalebox{1.2}{[} (1-s_i)Z_i^++s_iZ_i^- \scalebox{1.2}{]} \scalebox{1.4}{\}} \,\scalebox{1.1}{$\bigotimes$}_{j=1}^l X_j^{f_j(s_1s_2\cdots s_k)} }} \, \label{FGate3}
\end{equation}
is treated as a single composite logic gate controlled by the $(s_1s_2\cdots s_k)$ configuration and in conjunction with the single $t$-th clock rebit that is implemented as an ordinary rebit, where, as a result of $X_{\mathsmaller{C},\,t} \otimes F_{\mathsmaller{L},\,t}$ being a multi-rebit-simultaneously-flipping operator, the Feynman-Kitaev propagator $H_t(F)$ is already node-determinate and efficiently computable at the abstract level of interactions among computational rebits, which guarantees node-determinacy and efficient computability for any homophysical image $\mathfrak{M}(H_t(F))$ downstream as implemented using bi-fermions. As explained in the paragraphs around equation (\ref{NodeDeterOfCRGA}), the operator $F_{\mathsmaller{L},\,t}$ is doubly node-determinate, therefore, the operators $H_t(F)$ and $\mathfrak{M}(H_t(F))$ are guaranteed to be note-determinate. It is also interesting to note that, {\em the operators $F_{\mathsmaller{L},\,t}$, $H_t(F)$, and $\mathfrak{M}(H_t(F))$ remain efficiently computable so long as the total number of rebits moved by them is polynomially bounded}, since they involve only or mostly projection operators as their tensor factors.

Still further, any {\em essentially idempotent} ({\it i.e.}, idempotent up to a scaling constant) GSQC operator of the form $H_{\psi} \defeq c(I-|\psi\rangle\langle\psi|)$ with $c>0$, $|\psi\rangle=\cos\theta|0\rangle|\psi_0\rangle+\sin\theta|1\rangle|\psi_1\rangle$, $|\psi_0\rangle$ and $|\psi_1\rangle$ being normalized non-negative wavefunctions associated with a fixed set of computational rebits, $\{|0\rangle,|1\rangle\}$ being the computational basis of a single rebit, $\theta\in[-\pi/2,\pi/2]$, can be homophysically implemented as an interaction among a system of quantum bi-fermions, which is further fermion sign-rectified and isophysically mapped into a stoquastic operator $H_{\psi'} \defeq \mathfrak{M}(H_{\psi}) = I - |\psi'\rangle\langle\psi'|$, where
\begin{equation}
|\psi'\rangle \,\defeq\, \mathfrak{M}(|\psi\rangle)\,=\,\cos\theta\,\mathfrak{M}(|0\rangle|\psi_0\rangle)+(-\pi_{12})^{(\sin\theta<0)}\sin\theta\,\mathfrak{M}(|1\rangle|\psi_1\rangle)
\end{equation}
is a non-negative-definite wavefunction for a system of rectified bi-fermions, whose support is just one specific nodal cell of the associated system of quantum bi-fermions. The isomorphic images of $|\psi'\rangle$ extending to other nodal cells are obtained straightforwardly by exchanging constituent identical fermions, so that the  overall quantum ground state wavefunction $|\psi''\rangle$ for the associated  system of quantum bi-fermions can be easily reconstructed as $|\psi''\rangle \defeq const \times \sum_{\pi\in G_{\rm ex}} \pi |\psi'\rangle$, with $G_{\rm ex}$ being the exchange symmetry group of the system of bi-fermions. Then a Markov operator $[\psi]\exp({-}\tau H_{\psi})[\psi]^{\mathsmaller{-}1}$, $\tau>0$ associated with a Gibbs operator $\exp({-}\tau H_{\psi})$ of the original system of computational rebits can be homophysically implemented as
\begin{align}
[\psi] \exp\{{-}\tau(I-|\psi\rangle\langle\psi|)\} [\psi]^{\mathsmaller{-}1} &\,\HomoPhysR\,
[\psi''] \exp\{{-}\tau(I-|\psi''\rangle\langle\psi''|)\} [\psi'']^{\mathsmaller{-}1} \nonumber \\[0.75ex]
&\,\,=\, {\textstyle{ \scalebox{1.2}{$\sum$}_{\pi\in G_{\rm ex}} }}
[\pi\psi'] \exp\{{-}\tau(I-\pi|\psi'\rangle\langle\psi'|\pi)\} [\pi\psi']^{\mathsmaller{-}1},
\end{align}
which is efficiently computable and directly implemented in a single step of random walk during Monte Carlo simulations, without having to decompose the operator $H_{\psi}$, or $H_{\psi'}$, or $H_{\psi''} \defeq I - |\psi''\rangle\langle\psi''|$, or an associated Gibbs operator into a combination of elementary interactions. Namely, the Gibbs transition amplitude $\langle r|\exp({-}\tau H_{\psi''})|q\rangle \defeq \psi''(r) \langle r | \exp\{{-}\tau(I-|\psi''\rangle\langle\psi''|)\} | q \rangle \psi''(q)^{\mathsmaller{-}1}$ can be efficiently computed for any given pair of configuration points $q$ and $r$, by firstly identifying a $\pi\in G_{\rm ex}$ such that both $q' \defeq \pi(q)$ and $r' \defeq \pi(r)$ are within the support of the wavefunction $|\psi'\rangle$, then using the identity $\langle r|\exp({-}\tau H_{\psi''})|q\rangle = \langle r'|\exp({-}\tau H_{\psi'})|q'\rangle$ to compute the desired Gibbs transition amplitude. In case no $\pi\in G_{\rm ex}$ exists such that $(r',q')\subseteq\supp(|\psi'\rangle)$, then set $\langle r|\exp({-}\tau H_{\psi''})|q\rangle = 0$, which enforces the fixed-node diffusion strategy for Monte Carlo simulations.

Another category of efficiency consideration seeks to improve the $\Omega(T^{\mathsmaller{-}2})$ scaling of the Feynman-Kitaev energy gap. In the literature, there have been techniques reported to obtain polynomially improved energy gaps for adiabatic quantum computations and efficient quantum simulations of classical Monte Carlo methods \cite{Ambainis04,Szegedy04,Somma08,Somma10}, culminating in a systematic formulation and solution of ``spectral gap amplification'' \cite{Somma13}, which transforms any Hamiltonian in the form of a sum of frustration-free essentially idempotent self-adjoint operators into a new Hamiltonian with an energy gap that enjoys a polynomially improved scaling. Reference \cite{Mizel04} has also alluded to a GSQC Hamiltonian with an energy gap that scales as $\Omega(T^{\mathsmaller{-}1})$. It is also interesting to note another ``operator square-rooting'' type of technique which rewrites a given Hamiltonian $H=\sum_tH_t$ with $H_t\defeq|t\rangle\langle t|+|(t{+}1)\rangle\langle(t{+}1)|+|(t{+}1)\rangle\langle t|\otimes U_t+|t\rangle\langle(t{+}1)|\otimes U_t^+$, $t\in\mathbb{N}$ into $H=\mathlarger{\mathlarger{(}}\sum_tA_t\mathlarger{\mathlarger{)}}\mathlarger{\mathlarger{(}}\sum_tA_t^+\mathlarger{\mathlarger{)}}$ with $A_t\defeq|t\rangle\langle t|+|(t{+}1)\rangle\langle t|\otimes U_t$, $t\in\mathbb{N}$, then constructs a new Hamiltonian $H^{1/2}\defeq\sum_tA_t\otimes|1\rangle\langle 0|+\sum_tA_t^+\otimes|0\rangle\langle 1|$ with $|1\rangle\langle 0|$ and $|0\rangle\langle 1|$ operating on an auxiliary rebit, such that the eigenvalues of $H^{1/2}$ are $\mathlarger{\mathlarger{\{}}\!\pm[\lambda_n(H)]^{1/2}:n\ge 0\mathlarger{\mathlarger{\}}}$. However, these spectral gap-improving techniques may not be readily amenable to MCQC, as they either create a degenerate ground state space, or need to operate quantum simulations and adiabatic evolutions on an excited energy level, or produce a new Hamiltonian that is no longer SFF-FS or SFF-EB.

Here I propose a method that adapts the technique of {\em lifting Markov chains} \cite{Chen99,Diaconis00,Vucelja16} and designs a {\em lifted Feynman-Kitaev construct} which is Monte Carlo-simulated via a {\em lifted Markov chain}, so to achieve an $O(T)$ mixing time. A lifted Markov chain \cite{Chen99,Diaconis00,Vucelja16} may be time-inhomogeneous and violate the condition of detailed balance that is usually obeyed by the conventional reversible Markov chains, yet it could still be designed to relax to a desired stationary distribution over time, and do so at a polynomially faster rate of convergence. The proposal is to lift a ring-shaped graph of $T\in\mathbb{N}$ vertices for a discrete time Markov chain or its transition matrix into two isomorphic copies at a lower and an upper elevations respectively, similar to what has been exemplified and illustrated in reference \cite{Vucelja16}, where for each state index $t\in[0,T-1]$, there is a $t$-th vertex at the lower elevation in one-to-one correspondence with a $t$-th vertex at the upper elevation, both being derived from the $t$-th vertex of the ring-shaped graph for the original Markov chain before lifting. In each discrete time step indexed by a $\sft\in\mathbb{N}$, on the ring at the lower elevation, a walker moves predominantly clockwise, hopping from a vertex to its clockwise next neighbor with probability $1-\epsilon(T)$, $\epsilon(T)=O(T^{{-}1})$, $\epsilon(T)>0$, and on the ring at the upper elevation, a walker moves predominantly counterclockwise, transitioning from a vertex to its counterclockwise next neighbor with probability $1-\epsilon(T)$. The two rings are weakly coupled with each pair of one-to-one corresponding vertices at the two elevations to swap with probability $\epsilon(T)$ in each time step $\sft\in\mathbb{N}$. Such a Markov chain will be referred to as a {\em two-rings-at-two-elevations lifted Markov chain}. It is the rapid and mostly unidirectional motion of walkers traversing the state space in $O(T)$ time that enables the lifted Markov chain to reach equilibrium in $O(T)$ time, which is a remarkable improvement over the $O(T^2)$ mixing time associated with the usual ``diffusive'' behavior of a reversible Markov chain.

Given a possibly fault-tolerant quantum circuit as a series of self-inverse quantum gates $\{U(t)\}_{t=1}^{\mathsmaller{T}_*}$, $T_*\in\mathbb{N}$ to implement a BQP algorithm that computes a {\em result state} $|\phi_*\rangle \defeq \scalebox{1.1}{[} \prod_{t=1}^{\mathsmaller{T}_*}U(t) \scalebox{1.1}{]} |\phi_0\rangle$ from an initial state $|\phi_0\rangle \defeq |l'_1\cdots l'_{\mathsmaller{N}}\rangle$, $(l'_1,\cdots\!,l'_{\mathsmaller{N}})\in\{0,1\}^{\mathsmaller{N}}$, $N\in\mathbb{N}$, a {\em self-reversing quantum circuit} as a sequence of self-inverse quantum gates $\{U_{\mathsmaller{L}}(t)\}_{t=1}^{6\mathsmaller{T}}$, $T\in\mathbb{N}$, $T > 2T_*/3$, $T - 2T_*/3 = O(T_*)$ is constructed accordingly, such that $U_{\mathsmaller{L}}(2t) = I_{\mathsmaller{L}}$ the identity operator for all $t\in[1,3T]$, and $U_{\mathsmaller{L}}(2t-1) = U(t) = U_{\mathsmaller{L}}(6T-2t+1)$ for all $t\in[1,T_*]$, whereas $U_{\mathsmaller{L}}(2t-1) = I_{\mathsmaller{L}}$ for all $t\in[T_*+1,3T-T_*]$ by intentionally padding identity gates to duplicate the result state. Starting from an initial state $|\phi_0\rangle_{\mathsmaller{L}} \defeq |\phi_0\rangle$, the self-reversing quantum circuit $\{U_{\mathsmaller{L}}(t)\}_{t=1}^{6\mathsmaller{T}}$ has a series of quantum states $\{|\phi_t\rangle_{\mathsmaller{L}}:t\in[1,6T]\}\subseteq\calH(\{0,1\}^{\mathsmaller{N}})$ defined and generated recursively through $|\phi_t\rangle_{\mathsmaller{L}}\defeq U_{\mathsmaller{L}}(t)|\phi_{t{-}1}\rangle_{\mathsmaller{L}}$ for all $t\in[1,6T]$, with $|\phi_{2t}\rangle_{\mathsmaller{L}} = |\phi_{2t{-}1}\rangle_{\mathsmaller{L}}$ for all $t\in[1,3T]$, and $|\phi_t\rangle_{\mathsmaller{L}} = |\phi_*\rangle_{\mathsmaller{L}} \defeq |\phi_{2\mathsmaller{T}_*{-}1}\rangle_{\mathsmaller{L}}$ for all $t\in[2T_*-1,6T-2T_*]$, due to the intentional state-copying by identity gates. The quantum circuit $\{U_{\mathsmaller{L}}(t)\}_{t=1}^{6\mathsmaller{T}}$ is called self-reversing because of the reflective symmetry $U_{\mathsmaller{L}}(2t-1) = U_{\mathsmaller{L}}(6T-2t+1)$ for all $t\in[1,3T]$ and $U_{\mathsmaller{L}}(2t) = I_{\mathsmaller{L}}$ for all $t\in[1,3T]$, such that, $|\phi_{2t{-}1}\rangle_{\mathsmaller{L}} = |\phi_{2t}\rangle_{\mathsmaller{L}} = |\phi_{6T-2t{-}1}\rangle_{\mathsmaller{L}} = |\phi_{6T-2t}\rangle_{\mathsmaller{L}}$, $\forall t\in[1,3T]$, in particular, $\prod_{t=1}^{6\mathsmaller{T}}U_{\mathsmaller{L}}(t) = I_{\mathsmaller{L}}$, and the result states $|\phi_t\rangle_{\mathsmaller{L}} = |\phi_*\rangle_{\mathsmaller{L}}$, $t\in[2T_*-1,6T-2T_*]$ around the mid-point of the circuit at $t=3T$ are reversely computed back to the initial state $|\phi_{6\mathsmaller{T}{-}1}\rangle_{\mathsmaller{L}} = |\phi_{6\mathsmaller{T}}\rangle_{\mathsmaller{L}} = |\phi_0\rangle_{\mathsmaller{L}}$ at the end. In other words, the graph of states and transitions associated with the self-reversing quantum circuit $\{U_{\mathsmaller{L}}(t)\}_{t=1}^{6\mathsmaller{T}}$ is a ring-shaped loop, where each vertex is indexed by a unique $t\in[0,6T-1]$, as an element of the cyclic additive group $\mathbb{Z}/6T\mathbb{Z}$. Any index and arithmetic operations on it are interpreted in the sense of modulo-$6T$ arithmetic, with any resulted index $t\in\mathbb{N}$ out of the set $[0,6T-1]$ rolled back into an equivalence representative modulo $6T$. Fig.\;\ref{SixVertexSegment} depicts a $6$-vertex segment of such a ring-shaped graph, where each circle with an index $t'=6t+i$, $t\in[0,T-1]$, $i\in[0,5]$ inside denotes a graph vertex that is associated with a state $|\phi_{t'}\rangle_{\mathsmaller{L}}$, a double-line ``\,\begin{tikzpicture} \draw [black, line width=1.25] (0,0) -- (.75,0); \draw [black, line width=1.25] (0,0.15) -- (.75,0.15); \end{tikzpicture}\,'' between adjacent vertices represents a graph edge associated with a general quantum gate that may or may not reduce to an identity, while a single-line ``\,\begin{tikzpicture} \draw [black, line width=1.25] (0,0) -- (.75,0); \end{tikzpicture}\,'' between adjacent vertices represents a graph edge that is definitively associated with an identity gate.

\iftoggle{ForUSPTO} {
} {
\vspace{2.0ex}
\begin{figure}[ht]
\centering
\begin{tikzpicture} [x=2cm,y=2cm]
\def\CircRad{0.25}
\def\TwoBarY{0.06}
\def\TwoBarX{{sqrt(\CircRad*\CircRad-\TwoBarY*\TwoBarY)}}
\filldraw [gray!25,fill=gray!25] ($(-3,-\CircRad)-1.5*(0,\TwoBarY)$) rectangle ($(3,\CircRad)+1.5*(0,\TwoBarY)$) ;
\draw [black, very thick] ($(-3,0)-6*(0.075,0)$) -- ($(-3,0)-5*(0.075,0)$) ;
\draw [black, very thick] ($(-3,0)-4*(0.075,0)$) -- ($(-3,0)-3*(0.075,0)$) ;
\draw [black, very thick] ($(-3,0)-2*(0.075,0)$) -- ($(-3,0)-1*(0.075,0)$) ;
\draw [black, very thick] (-3,0) -- ($(-2.5,0)-(\CircRad,0)$) ;
\draw [black, very thick] (-2.5,0) circle [radius=\CircRad] ;
\draw (-2.5,-0.01) node {$\mathsmaller{6t}$} ;
\draw [black, very thick] ($(-2.5,\TwoBarY)+(\TwoBarX,0)$) -- ($(-1.5,\TwoBarY)-(\TwoBarX,0)$) ;
\draw [black, very thick] ($(-2.5,-\TwoBarY)+(\TwoBarX,0)$) -- ($(-1.5,-\TwoBarY)-(\TwoBarX,0)$) ;
\draw [black, very thick] (-1.5,0) circle [radius=\CircRad] ;
\draw (-1.5,-0.01) node {$\mathsmaller{6t+1}$} ;
\draw [black, very thick] ($(-1.5,0)+(\CircRad,0)$) -- ($(-0.5,0)-(\CircRad,0)$) ;
\draw [black, very thick] (-0.5,0) circle [radius=\CircRad] ;
\draw (-0.5,-0.01) node {$\mathsmaller{6t+2}$} ;
\draw [black, very thick] ($(-0.5,\TwoBarY)+(\TwoBarX,0)$) -- ($(0.5,\TwoBarY)-(\TwoBarX,0)$) ;
\draw [black, very thick] ($(-0.5,-\TwoBarY)+(\TwoBarX,0)$) -- ($(0.5,-\TwoBarY)-(\TwoBarX,0)$) ;
\draw [black, very thick] (0.5,0) circle [radius=\CircRad] ;
\draw (0.5,-0.01) node {$\mathsmaller{6t+3}$} ;
\draw [black, very thick] ($(0.5,0)+(\CircRad,0)$) -- ($(1.5,0)-(\CircRad,0)$) ;
\draw [black, very thick] (1.5,0) circle [radius=\CircRad] ;
\draw (1.5,-0.01) node {$\mathsmaller{6t+4}$} ;
\draw [black, very thick] ($(1.5,\TwoBarY)+(\TwoBarX,0)$) -- ($(2.5,\TwoBarY)-(\TwoBarX,0)$) ;
\draw [black, very thick] ($(1.5,-\TwoBarY)+(\TwoBarX,0)$) -- ($(2.5,-\TwoBarY)-(\TwoBarX,0)$) ;
\draw [black, very thick] (2.5,0) circle [radius=\CircRad] ;
\draw (2.5,-0.01) node {$\mathsmaller{6t+5}$} ;
\draw [black, very thick] ($(2.5,0)+(\CircRad,0)$) -- (3,0) ;
\draw [black, very thick] ($(3,0)+1*(0.075,0)$) -- ($(3,0)+2*(0.075,0)$) ;
\draw [black, very thick] ($(3,0)+3*(0.075,0)$) -- ($(3,0)+4*(0.075,0)$) ;
\draw [black, very thick] ($(3,0)+5*(0.075,0)$) -- ($(3,0)+6*(0.075,0)$) ;
\end{tikzpicture}
\vspace{0.5ex}
\caption{A $6$-vertex segment of a ring-shaped graph.}
\label{SixVertexSegment}
\end{figure}
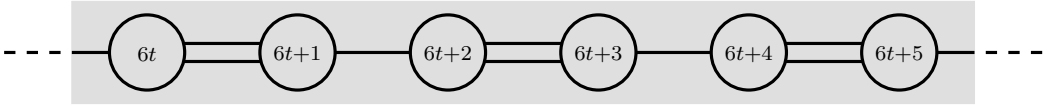
}

To generate a lifted Markov chain of two rings at two elevations, a Feynman-Kitaev construct for a self-reversing quantum circuit as a sequence of self-inverse quantum gates $\{U_{\mathsmaller{L}}(t)\}_{t=1}^{6\mathsmaller{T}}$, $T\in\mathbb{N}$ can be constructed using a GSQC system comprising a clock register of $6T$ clock rebits, a logic register of $N$ logic rebits, and a single elevation rebit, whose Hilbert space is spanned by effective computational basis states of the form
$$|c_0c_1\cdots c_{6\mathsmaller{T}{-}1}\rangle_{\mathsmaller{C}} |l_1\cdots l_{\mathsmaller{N}}\rangle_{\mathsmaller{L}} |z\rangle_{\mathsmaller{E}}, \; c_t\in\{0,1\}, \; \forall t\in[0,T], \; l_i\in\{0,1\}, \; \forall i\in[1,N], \; z\in\{0,1\},$$
where $z=0$ is associated with the lower elevation, and $z=1$ the upper, each of the clock rebits may be implemented as a bi-fermion, while the elevation and the logic rebits may all be embodied as ordinary rebits of single bi-fermions or any conventional non-bi-fermion rebits. The Feynman-Kitaev construct shall employ the so-called {\em pulse clock} \cite{Nagaj10,Breuckmann14} and be governed by a real (or Markovian) time $\sft$-dependent FBM frustration-free Hamiltonian
\begin{align}
H_{\mathsmaller{\rm FK}}(\sft) \,\defeq\,& H_{\rm clock} \,+\, H_{\rm init} \,+\, {\textstyle{ \scalebox{1.15}{$\sum$} }}_{t=1}^{6\mathsmaller{T}}\,H_{{\rm prop},\,t}(\sft\scalebox{0.85}{$\,\%\,$}4) \,+\, H_{\rm swap}(\sft\scalebox{0.85}{$\,\%\,$}4), \; \forall\sft\in\mathbb{N}, \label{LiftHFKtotal} \\[0.75ex]
H_{\rm clock} \,\defeq\,&  \scalebox{1.25}{(} I \,-\, {\textstyle{ \scalebox{1.15}{$\sum$}_{t=0}^{6\mathsmaller{T}{-}1} }}
\itPi^-_{\mathsmaller{C},\,t}  \scalebox{1.25}{)}^2, \label{LiftHFKclock} \\[0.75ex]
H_{\rm init} \,\defeq\,& \itPi^-_{\mathsmaller{C},\,0} \otimes \scalebox{1.25}{(} \,{\textstyle{ \scalebox{1.15}{$\sum$}_{i\,:\,l'_i=0} }}\,Z^-_{\mathsmaller{L},\,i} \,+\, {\textstyle{ \scalebox{1.15}{$\sum$}_{j\,:\,l'_j=1} }}\,Z^+_{\mathsmaller{L},\,j} \scalebox{1.25}{)}, \label{LiftHFKinit} \\[0.75ex]
H_{{\rm prop},\,t}(k) \,\defeq\,& W_{{\rm ctrl},\,t} \otimes W_{{\rm prop},\,t}(k) \Iver{k \neq 0}, \; \forall t\in[1,6T], \; \forall k\in[0,3], \label{LiftHFKpropt} \\[0.75ex]
H_{\rm swap}(k) \,\defeq\,& \scalebox{1.15}{(} I - X_{\mathsmaller{E}} \scalebox{1.15}{)} \Iver{k = 0}, \; \forall k\in[0,3], \label{LiftHFKswapt}
\end{align}
where $\forall\sft\in\mathbb{N}$, $\sft\scalebox{0.85}{$\,\%\,$}4$ denotes the remainder of $\sft$ modulo $4$, $\Iver{\,\cdot\,}$ is the Iverson bracket, $\forall t\in[1,6T]$, $W_{{\rm ctrl},\,t}$ and $W_{{\rm prop},\,t}(k)$, $k\in[0,3]$ are called the $t$-th control switch and the $t$-th Feynman-Kitaev propagator depending on $k\in[0,3]$ respectively, which are defined as
\begin{align}
W_{{\rm ctrl},\,t} \,\defeq\,& \itPi^+_{\mathsmaller{C},\,t{-}2} \otimes (\itPi^-_{\mathsmaller{C},\,t{-}1} \otimes \itPi^+_{\mathsmaller{C},\,t} + \itPi^+_{\mathsmaller{C},\,t{-}1} \otimes \itPi^-_{\mathsmaller{C},\,t}) \otimes \itPi^+_{\mathsmaller{C},\,t{+}1}, \label{WctrlDef} \\[0.75ex]
W_{{\rm prop},\,t}(k) \,\defeq\,& \itPi^-_{\mathsmaller{C},\,t{-}1} \otimes \itPi^+_{\mathsmaller{C},\,t} \otimes \exp[((t+1)\scalebox{0.85}{$\,\%\,$}2)v(Z^{-}_{\mathsmaller{E}},t,k)] \nonumber \\[0.75ex]
-\;\,& \itGamma_{\mathsmaller{C},\,t{-}1} \otimes \itGamma_{\mathsmaller{C},\,t} \otimes U_{\mathsmaller{L}}(t) \otimes I_{\mathsmaller{E}} \label{WpropDef} \\[0.75ex]
+\;\,& \itPi^+_{\mathsmaller{C},\,t{-}1} \otimes \itPi^-_{\mathsmaller{C},\,t} \otimes \exp[{-}((t+1)\scalebox{0.85}{$\,\%\,$}2) v(Z^{-}_{\mathsmaller{E}},t,k)]\,, \nonumber
\end{align}
where $(t+1)\scalebox{0.85}{$\,\%\,$}2$ evaluates to $1$ if and only if $t$ is even, otherwise it is $0$, the operator $v(Z^{-}_{\mathsmaller{E}},t,k)$ reduces to a scalar-valued potential $v(z,t,k)\in\mathbb{R}$, when projected onto each eigenstate $|z\rangle_{\mathsmaller{E}}$, $z\in\{0,1\}$ of the operator $Z^{-}_{\mathsmaller{E}} \defeq \half(I-Z_{\mathsmaller{E}})$, $\forall t\in[1,6T]$, $\forall k\in[0,3]$. More details will be given below on the compositions and operations of the partial Hamiltonians $H_{\rm clock}$, $H_{\rm init}$, $H_{{\rm prop},\,t}(k)$, $t\in[1,6T]$, $k\in[0,3]$, and $H_{\rm swap}(k)$, $k\in[0,3]$, which are called the clock Hamiltonian, the initialization Hamiltonian, the Feynman-Kitaev propagators, and the elevation swap Hamiltonian respectively. 

An effective computational basis state $|c_0c_1\cdots c_{6\mathsmaller{T}{-}1}\rangle_{\mathsmaller{C}}$ of the clock register is called a {\em pulse clock state} when a $t\in[0,6T-1]$ exists such that $c_t=1$ while $c_{t'}=0$ for all $t'\in[0,6T-1]\setminus\{t\}$, in which case $|t\rangle_{\mathsmaller{C}} \defeq |c_0c_1\cdots c_{6\mathsmaller{T}{-}1}\rangle_{\mathsmaller{C}}$ is used as a shorthand notation for the basis state, which represents the Feynman-Kitaev time $t\in[0,6T-1]$. That only the pulse clock states are allowed for the clock register at the lowest energy is guaranteed by the clock Hamiltonian $H_{\rm clock}$ of equation (\ref{LiftHFKclock}), which is easily computable for any given configuration point $q\in\calC'$ during a Monte Carlo simulation over a configuration space $\calC'=\mathfrak{M}(\calC)$ as a homophysical image of the computational rebit configuration space $\calC\defeq\{0,1\}^{6\mathsmaller{T}}$, by simply parsing $q$ to determine the computational $6T$-rebit basis state $|s(q)\rangle_{\mathsmaller{C}} \defeq |s_0(q)\cdots s_{6\mathsmaller{T}{-}1}(q)\rangle_{\mathsmaller{C}} \in \calH(\calC)$ that is consistent with $q\in\calC'$, then reducing the $\calC$-diagonal clock Hamiltonian to $H_{\rm clock} \sim \scalebox{1.15}{[} 1 - \sum_{t=0}^{6\mathsmaller{T}{-}1}s_t(q) \scalebox{1.15}{]}^2 \times |s(q)\rangle_{\mathsmaller{C}}\langle s(q)|_{\mathsmaller{C}}$, when computing any Gibbs transition amplitude from $q$ to another configuration point $r\in\calC'$. The initialization Hamiltonian $H_{\rm init}$ sets up the predetermined initial state $|\phi_0\rangle_{\mathsmaller{L}} \defeq |l'_1\cdots l'_{\mathsmaller{N}}\rangle_{\mathsmaller{L}}$, $(l'_1,\cdots\!,l'_{\mathsmaller{N}})\in\{0,1\}^{\mathsmaller{N}}$ for the logic register at the Feynman-Kitaev time $|(t=0)\rangle_{\mathsmaller{C}}\defeq|10\cdots 0\rangle_{\mathsmaller{C}}$. Then the Feynman-Kitaev propagators $\{H_{{\rm prop},\,t}(k):t\in[1,T],\,k\in[0,3]\}$ will generate the series of quantum states $\{|\phi_t\rangle_{\mathsmaller{L}}:t\in[1,6T]\}\subseteq\calH(\{0,1\}^{\mathsmaller{N}})$ recursively through $|\phi_t\rangle_{\mathsmaller{L}} = U_{\mathsmaller{L}}(t)|\phi_{t{-}1}\rangle_{\mathsmaller{L}}$ for all $t\in[1,6T]$. At any fixed instant of the real (or Markovian) time $\sft\in\mathbb{N}$, or at any finer-resolved time instant within each discrete $\sft$-step, the instantaneous Hamiltonian $\sft\in\mathbb{N}$ is frustration-free, whose ground state(s) can only live in the linear space that is spanned by the so-called Feynman-Kitaev history states $\{|t\rangle_{\mathsmaller{C}}|\phi_t\rangle_{\mathsmaller{L}}|z\rangle_{\mathsmaller{E}}:t\in[0,6T-1],\,z\in\{0,1\}\}$.

\begin{definition}{(SFF Group of Hamiltonians, SFF Periodic Sequence of Hamiltonians)}\label{SFFPeriodicSequence}\\
A group of Hamiltonians $\{H_k:k\in[0,K-1]\}$ supported by a configuration space $\calC$ is said to generate a periodic sequence of Hamiltonians $\{H(\sft):\sft\in\mathbb{N}\}$ with period $K$, when $H(\sft) = H_{\sft\mathsmaller{\,\%\,K}}$, $\forall\sft\in\mathbb{N}$, where $\sft\scalebox{0.85}{$\,\%\,$}K$ denotes the remainder of $\sft$ modulo $K$. Such periodic sequence of Hamiltonians in turn generates a time-inhomogeneous Markov chain \cite{Seneta81,Winkler03,Stroock05} with a periodic sequence of Markov operators or transition matrices $\{M(\sft):\sft\in\mathbb{N}\}$, where $M(\sft) \defeq [\psi_0(H(\sft))] \exp\{{-}\tau(\sft)H(\sft)\} [\psi_0(H(\sft))]^{{-}1}\!$ as specified in \myLemma\;\ref{QuasiStochasticOper}, $\tau(\sft) > 0$ is a $\sft$-dependent constant, $\forall\sft\in\mathbb{N}$.

Starting from an initial probability distribution $\phi_0^2(q\in\calC)$, such $\{M(\sft):\sft\in\mathbb{N}\}$, being called a periodic Markov chain, generates a sequence of probability distributions $\{\phi_{\sft}^2(q\in\calC):\sft\in\{0\}\cup\mathbb{N}\}$ such that $\phi_{\sft}^2 \defeq M(\sft) \phi_{\sft{-}1}^2$, $\forall\sft\in\mathbb{N}$. A period-averaged probability distribution $\scalebox{1.1}{$\langle$} \phi_{\sft}^2 \scalebox{1.1}{$\rangle$}$ is defined as $\scalebox{1.1}{$\langle$} \phi_{\sft}^2 \scalebox{1.1}{$\rangle$} \defeq K^{{-}1}\sum_{k=0}^{\mathsmaller{K}{-}1}\phi_{\sft+k}^2$, $\forall\sft\in\{0\}\cup\mathbb{N}$. When $\lim_{\,\sft\,\rightarrow\,\infty\!} \scalebox{1.1}{$\langle$} \phi_{\sft}^2 \scalebox{1.1}{$\rangle$} \defeq \scalebox{1.1}{$\langle$} \phi_{\infty}^2 \scalebox{1.1}{$\rangle$}\!$ exists in the topology of total variation distance and is independent of the initial $\phi_0^2$, then the generating periodic Markov chain $\{M(\sft):\sft\in\mathbb{N}\}$ is said to converge, and $\scalebox{1.1}{$\langle$} \phi_{\infty}^2 \scalebox{1.1}{$\rangle$}\!$ is called its stationary distribution, which generalizes the like-named notion for a time-homogeneous, irreducible, and reversible Markov chain.

A group of Hamiltonians $\{H_k:k\in[0,K-1]\}$, and the periodic sequence of Hamiltonians $\{H(\sft)=H_{(\sft{-}1)\mathsmaller{\,\%\,K}}:\sft\in\mathbb{N}\}$ generated by it, are both called separately frustration-free (SFF), when every Hamiltonian in the group $\{H_k:k\in[0,K-1]\}$ satisfies all the conditions as specified in Definition \ref{defiSFFH} for it to be SFF, and further, the periodic sequence of Hamiltonians $\{H(\sft):\sft\in\mathbb{N}\}$ generates a periodic Markov chain that converges to a stationary distribution $\scalebox{1.1}{$\langle$} \phi_{\infty}^2 \scalebox{1.1}{$\rangle$}$.
\vspace{-1.5ex}
\end{definition}

When homophysically implemented into a system of bi-fermions, the homophysical image of the time-dependent Feynman-Kitaev Hamiltonian $\{H_{\mathsmaller{\rm FK}}(\sft):\sft\in\mathbb{N}\}$ as defined in equations (\ref{LiftHFKtotal}) through (\ref{WpropDef}), still denoted by the same $\{H_{\mathsmaller{\rm FK}}(\sft):\sft\in\mathbb{N}\}$ for brevity of the mathematical formulas, is an SFF periodic sequence of Hamiltonians with period $4$, which is generated by a group of Hamiltonians $\{H_{\mathsmaller{\rm FK},\,k}\}_{k\in[0,3]}$, with $H_{\mathsmaller{\rm FK},\,k} \defeq H_{\rm clock} \,+\, H_{\rm init} \,+\, {\textstyle{ \scalebox{1.15}{$\sum$} }}_{t=1}^{6\mathsmaller{T}}\,H_{{\rm prop},\,t}(k) \,+\, H_{\rm swap}(k)$, $\forall k\in[0,3]$. Exactly to which of the two homophysical Feynman-Kitaev Hamiltonians is referred to by the same notation $\{H_{\mathsmaller{\rm FK}}(\sft):\sft\in\mathbb{N}\}$ should always be easily inferred from the context without ambiguity. Moreover, when applicable, the adjective ``SFF'' is often used before a reference of Feynman-Kitaev Hamiltonian $\{H_{\mathsmaller{\rm FK}}(\sft):\sft\in\mathbb{N}\}$ to clearly indicate that it refers to the homophysical implementation in a system of bi-fermions which does possess the SFF property.

Let $0 < v_0 = O(\poly(T+N))$ and $0 < \tau_0 = O(\poly(T+N))$ be two sufficiently large but still polynomially bounded constants. In order for the SFF Feynman-Kitaev Hamiltonian $\{H_{\mathsmaller{\rm FK}}(\sft):\sft\in\mathbb{N}\}$ to generate a useful and efficient lifted Markov chain, the $\mathbb{R}$-valued potential $v(z,t,k)$ as a function of $(z,t,k)\in\{0,1\}\times[0,6T-1]\times[0,3]$ is so designed that a periodic Markov chain
\begin{equation}
\{M(\sft) \defeq [\psi_0(H_{\mathsmaller{\rm FK}}(\sft))] \exp\{{-}\tau(\sft)H_{\mathsmaller{\rm FK}}(\sft)\} [\psi_0(H_{\mathsmaller{\rm FK}}(\sft))]^{{-}1}\!:\tau(\sft)>0,\,\sft\in\mathbb{N}\} \label{PeriodicMarkovChainDefi}
\end{equation}
works by repeating 4-phase cycles (or called 4-stroke cycles) on each $6$-vertex segment of the ring-shaped graph as depicted in Fig.\;\ref{SixVertexSegment}, where

{\em Phase 1}\,: in the first phase/stroke with $\sft\scalebox{0.85}{$\,\%\,$}4 = k = 1$, the potential values are set as $v(0,6t+2,1) = {-}v(0,6t+4,1) = v_0$, and $v(0,6t,1) = 0$ for the lower elevation, while $v(1,6t+4,1) = {-}v(1,6t,1) = v_0$, and $v(1,6t+2,1) = 0$ for the higher elevation, $\forall t\in[0,T-1]$, and the Markov transition matrix $M(\sft)$ is applied with $\tau(\sft) = \tau_0$, such that an equilibrium distribution is reached, which has the probability amplitude concentrated in the Feynman-Kitaev history states $\{|(6t+i)\rangle_{\mathsmaller{C}}|\phi_{6t+i}\rangle_{\mathsmaller{L}}|0\rangle_{\mathsmaller{E}}:i\in\{2,3\}\}$ at the lower elevation, and the Feynman-Kitaev history states $\{|(6t+j)\rangle_{\mathsmaller{C}}|\phi_{6t+j}\rangle_{\mathsmaller{L}}|1\rangle_{\mathsmaller{E}}:j\in\{4,5\}\}$ at the higher elevation, moreover, the probability amplitude is substantially equally distributed between each pair of adjacent and connected Feynman-Kitaev history states at both the lower and the higher elevations;

{\em Phase 2}\,: in the second phase/stroke with $\sft\scalebox{0.85}{$\,\%\,$}4 = k = 2$, the potential values are set as $v(0,6t+4,2) = {-}v(0,6t,2) = v_0$, and $v(0,6t+2,2) = 0$ for the lower elevation, while $v(1,6t+2,2) = {-}v(1,6t+4,2) = v_0$, and $v(1,6t,2) = 0$ for the higher elevation, $\forall t\in[0,T-1]$, and the Markov transition matrix $M(\sft)$ is applied with $\tau(\sft) = \tau_0$, such that an equilibrium distribution is reached, which has the probability amplitude concentrated in the Feynman-Kitaev history states $\{|(6t+i)\rangle_{\mathsmaller{C}}|\phi_{6t+i}\rangle_{\mathsmaller{L}}|0\rangle_{\mathsmaller{E}}:i\in\{4,5\}\}$ at the lower elevation, and the Feynman-Kitaev history states $\{|(6t+j)\rangle_{\mathsmaller{C}}|\phi_{6t+j}\rangle_{\mathsmaller{L}}|1\rangle_{\mathsmaller{E}}:j\in\{2,3\}\}$ at the higher elevation, moreover, the probability amplitude is substantially equally distributed between each pair of adjacent and connected Feynman-Kitaev history states at both the lower and the higher elevations;

{\em Phase 3}\,: in the third phase/stroke with $\sft\scalebox{0.85}{$\,\%\,$}4 = k = 3$, the potential values are set as $v(z,6t,3) = {-}v(z,6t+2,3) = v_0$, while $v(z,6t+4,3) = 0$, $\forall t\in[0,T-1]$, $\forall z\in\{0,1\}$, and the Markov transition matrix $M(\sft)$ is applied with $\tau(\sft) = \tau_0$, to effect transfer of probability amplitude from the Feynman-Kitaev history states $\{|(6t+i)\rangle_{\mathsmaller{C}}|\phi_{6t+i}\rangle_{\mathsmaller{L}}|z\rangle_{\mathsmaller{E}}:i\in[2,5]\}$ to the adjacent Feynman-Kitaev history states $\{|(6t+j)\rangle_{\mathsmaller{C}}|\phi_{6t+j}\rangle_{\mathsmaller{L}}|z\rangle_{\mathsmaller{E}}:j\in\{0,1\}\}$, such that when an equilibrium distribution is reached, each of the Feynman-Kitaev history states $\{|(6t+i)\rangle_{\mathsmaller{C}}|\phi_{6t+i}\rangle_{\mathsmaller{L}}|z\rangle_{\mathsmaller{E}}:i\in[2,5]\}$ retains merely an exponentially small thus negligible probability amplitude, while each of the Feynman-Kitaev history states $\{|(6t+j)\rangle_{\mathsmaller{C}}|\phi_{6t+j}\rangle_{\mathsmaller{L}}|z\rangle_{\mathsmaller{E}}:j\in\{0,1\}\}$ gets substantially all of the available probability amplitude, $\forall t\in[0,T-1]$, $\forall z\in\{0,1\}$;

{\em Phase 4}\,: in the fourth phase/stroke with $\sft\scalebox{0.85}{$\,\%\,$}4 = k = 0$, all of the Feynman-Kitaev propagators are switched off, while the elevation swap Hamiltonian $H_{\rm swap}(k=0)$ is turned on, so that the Markov operator $[\psi_0(H_{\rm swap}(k))] \exp\{{-}\tau(\sft)H_{\rm swap}(k)\} [\psi_0(H_{\rm swap}(k))]^{{-}1}$ is applied with $\tau(\sft) = O(\epsilon(T)) = O(T^{{-}1})$, to effect a cross-elevation transition with probability $\epsilon(T) = O(T^{{-}1})$.

It is now easily seen that the repeated 4-phase or 4-stroke cycles of the periodic Markov chain $\{M(\sft):\sft\in\mathbb{N}\}$ amount to a shift register-type of operation for the probability amplitude in each of the ring-shaped graphs at an elevation $z\in\{0,1\}$, in conjunction with a brief and weak coupling periodically between the two elevations. At any instant of real time $\sft\in\mathbb{N}$, essentially all of the probability amplitude is localized at a single pair of consecutive vertices in each $6$-vertex segment as depicted in Fig.\;\ref{SixVertexSegment} at any elevation $z\in\{0,1\}$. In the first three phases of each 4-phase or 4-stroke cycle with $\sft \equiv 1,2,3 \pmod*{4}$, at the lower elevation with $z=0$, the localized probability amplitude transfers clockwise, firstly from a vertex pair $(6t,6t+1)$ to a vertex pair $(6t+2,6t+3)$, then from $(6t+2,6t+3)$ to $(6t+4,6t+5)$, and lastly from $(6t+4,6t+5)$ to a vertex pair $(6(t+1),6(t+1)+1)$ in the next $6$-vertex segment, much like a cloud of electric charges transferring in a charge coupled device (CCD) \cite{Boyle70,Amelio70,Boyle74}. The similar transfer of probability amplitude takes place at the higher elevation with $z=1$, only that the probability amplitude moves counterclockwise. In the fourth phase of each 4-phase or 4-stroke cycle with $\sft \equiv 0 \pmod*{4}$, the two rings at different elevations are allowed to briefly exchange probability amplitude, with a transition rate $\epsilon(T)=O(T^{{-}1})$.

It can be easily verified that the periodic Markov chain $\{M(\sft):\sft\in\mathbb{N}\}$ of equation (\ref{PeriodicMarkovChainDefi}) implements precisely the two-rings-at-two-elevations lifted Markov chain as desired, where, at the lower elevation, a walker moves predominantly clockwise from a vertex to its neighbor with probability $1-\epsilon(T)$, while at the upper elevation, a walker moves predominantly counterclockwise from a vertex to its neighbor with probability $1-\epsilon(T)$, and the two rings are weakly coupled with a transition probability $\epsilon(T)=O(T^{{-}1})$. For this reason, $\{M(\sft):\sft\in\mathbb{N}\}$ of equation (\ref{PeriodicMarkovChainDefi}) is called a {\em lifted periodic Markov chain}. It follows straightforwardly from well established results \cite{Chen99,Diaconis00,Vucelja16} that the lifted periodic Markov chain $\{M(\sft):\sft\in\mathbb{N}\}$ of equation (\ref{PeriodicMarkovChainDefi}) converges in $O(T)$ time to its stationary distribution, which is an equilibrium distribution with all Feynman-Kitaev history states equally populated. On the other hand, it is without difficulty to repeat the familiar analysis to demonstrate that the periodic Feynman-Kitaev Hamiltonian $\{H_{\mathsmaller{\rm FK}}(\sft):\sft\in\mathbb{N}\}$ as defined in equations (\ref{LiftHFKtotal}) through (\ref{WpropDef}) can be made PLTKD, with each additive term in the Lie-Trotter-Kato decomposition being homophysically implementable as a node-determinate and efficiently computable tensor monomials on bi-fermions. In summary, the periodic Feynman-Kitaev Hamiltonian $\{H_{\mathsmaller{\rm FK}}(\sft):\sft\in\mathbb{N}\}$ and the lifted periodic Markov chain $\{M(\sft):\sft\in\mathbb{N}\}$ generated by it facilitate highly efficient methods of MCQC.

By way of example and not limitation, the following algorithm describes an overall process for solving a general classical or quantum computational problem via Monte Carlo quantum computing, which firstly designs a quantum algorithm that solves the given computational problem, then synthesizes a quantum circuit and creates a Feynman-Kitaev construct whose bi-fermion implementation has either a time-independent SFF Hamiltonian or an SFF periodic sequence of Hamiltonians, to generate either a time-homogenous or a lifted periodic Markov chain, finally runs a Monte Carlo simulation using said Markov chain and derives from its stationary distribution a solution to the computational problem.

\renewcommand{\labelenumi}{\thealgorithm.\arabic{enumi}}
\begin{algorithm}{(Monte Carlo  Quantum Computing)} \label{MCQCOverall}
\vspace{-1.5ex}
\begin{enumerate}[start=1]
\item Design a quantum algorithm that solves a given computational problem; \label{OverallDesignQuanAlg}
\item Construct a quantum circuit processing computational rebits, {\it i.e.}, a sequence of quantum gates that implements the designed quantum algorithm using real-valued and self-inverse quantum gates, with the total number of such quantum gates minimized; \label{OverallConstructQuanCirc}
\item Create a Feynman-Kitaev construct for the quantum circuit, using either a time-independent Hamiltonian or a periodic sequence of Hamiltonians; \label{OverallCreateFeynmanKitaev}
\item Homophysically map said time-independent Hamiltonian or periodic sequence of Hamiltonians to either a time-independent SFF Hamiltonian or an SFF periodic sequence of Hamiltonians governing a physical system of bi-fermions; \label{OverallHomoPhysMap}
\item Run Monte Carlo simulations using either a time-homogeneous Markov chain generated by the time-independent SFF Hamiltonian, or a lifted periodic Markov chain generated by the SFF periodic sequence of Hamiltonians, to sample from the stationary distribution and derive a solution to the original computational problem. \label{OverallRunMonteCarlo}
\end{enumerate}
\end{algorithm}
\vspace{-1.5ex}

As being discussed and specified previously, both multi-rebit-controlled single-rebit-transforming and multi-rebit-controlled multi-rebit-transforming quantum gates can be used to construct the quantum circuit in step \ref{MCQCOverall}.\ref{OverallConstructQuanCirc}, although a typical multi-rebit-transforming quantum gate might entail a different and likely higher computational cost than a single-rebit-transforming quantum gate for MCQC implementations and simulations. The quantum circuit should then be optimized so that the sum of the computational costs of all multi-rebit-transforming or single-rebit-transforming quantum gates is minimized.

It is only trivial a point to note that quantum mechanics and computing are not limited to isolated quantum systems described by a pure quantum state. Rather, they can be, as have already been, easily generalized to open quantum systems in mixed quantum states described by density matrices (also known as density operators) \cite{Feynman72,Davis76,Gruska99,Nielsen10}. A quantum computation executes a BQP algorithm as a sequence of {\em primitive gate operations} to a general quantum system associated with a {\em description size} $N\in\mathbb{R}$, where the description size characterizes the minimally required number of classical bits to encode the setup and parameters of the quantum system. In one \iftoggle{ForUSPTO} {preferred embodiment} {example}, $N$ is the size of a configuration space associated with the quantum system. In another \iftoggle{ForUSPTO} {preferred embodiment} {example}, $N$ counts the number of particles or modes of field constituting the quantum system. Each primitive gate operation is a so-called {\em superoperator effecting a transformation on density matrices} representing states of the physical system \cite{Davis76,Gruska99}, with said transformation involving multiplications and taking traces of tensor-factor matrices, and being implemented either numerically or physically up to an $O(1/\poly(N))$-bounded relative error. Said tensor-factor matrices are associated with low-dimensional subsystems and have dimensions upper-bounded by a constant or $O(\log N)$. The quantum system is initially in either a {\em simple classical state} with each constituent particle or mode of field individually set to a predetermined basis state, or a {\em simple quantum state} that can be generated from a simple classical state followed by an application of an $O(\poly(N))$ number of primitive gates. The density matrices representing the possible initial states and result states of a quantum system after an application of an $O(\poly(N))$ number of primitive gates constitute a space of the so-called {\em BQP computable states}, which captures all of the relevant states to quantum computing. Quantum states out of the space of BQP computable (states) are too much entangled to be accessible and useful \cite{Mora05,Gross09,Bremner09}.

\begin{definition}{(BQP Computable Density Matrices)}\label{defiPolyDensMatr}\\
A density matrix representing a state of a quantum system with a description size $N\in\mathbb{R}$ is BQP computable when the represented state of the quantum system can be generated by applying an $O(\poly(N))$ number of primitive gated implemented numerically or physically up to an $O(1/\poly(N))$-bounded relative error, said quantum system being initially in a simple classical state having each constituent particle or mode of field individually set to a predetermined basis state.
\vspace{-1.5ex}
\end{definition}

Quantum computing is just to generate or manipulate BQP computable density matrices, while classical computing, as a special case of quantum computing, is just to generate and manipulate a particular kind of BQP computable states that are necessarily diagonal and non-negative. The proposition of {\myTheorem} \ref{ThirdTheorem}, {\rm BPP}$\,=\,${\rm BQP}, simply says that the computational power of generating and manipulating only diagonal and non-negative density matrices is polynomially equivalent to that of generating and manipulating all BQP computable density matrices.

\begin{definition}{(Classical and Quantum Computational Procedures)}\label{defiCompProc}\\
A quantum computational procedure is a computational process that employs a BQP algorithm to generate and manipulate a plurality of BQP computable density matrices, said density matrices representing pure or mixed quantum states of a physical system. A classical computational procedure is a special case of a quantum computational procedure with respect to the same physical system, which employs a BPP algorithm to generate and manipulate a plurality of BQP computable density matrices that are necessarily diagonal and non-negative, said density matrices representing classical probabilistic states of said physical system.
\end{definition}

\begin{definition}{(Minimally Entangled Density Matrices, Substantially Entangled Density Matrices, Practically Substantially Entangled Density Matrices)}\label{defiEntaDensMatr}\\
Consider a quantum system supported by a configuration space $\calC$, and a density matrix $M$ representing a state of the quantum system. The density matrix $M$ is said to be minimally entangled, or substantially entangle, or practically substantially entangled, when its configuration coordinate representation $\{\langle r|M|q\rangle:(r,q)\in\calC^2\}$ as a density on the measurable space $(\calC^2,\Bor(\calC^2))$ is minimally entangled, or substantially entangle, or practically substantially entangled, respectively, where $\Bor(\calC^2)$ denotes the smallest $\sigma$-algebra consisting all of the Borel subsets of $\calC^2$.
\vspace{-1.5ex}
\end{definition}

The methods and theory presented in this \iftoggle{ForUSPTO} {specification} {documentation} have shown that MCQC is able to implement quantum computational procedures on classical computers, so to generate and simulate many of the densities as well as density matrices which have hitherto been found or thought to be practically substantially entangled.

Lastly, taking cue from the hierarchically leveled specialization, organization, and cooperation of programming languages and software in the tremendously successful industry of classical computing, where there are low-level programming languages such as machine codes and assembly languages, that are strongly coupled to specific hardware architectures and machine instruction sets, which can run very efficiently but are difficult to use and hardly portable, and by contrast, high-level programming languages that provide strong abstraction from the hardware and architecture details of the machines being used, which are much more friendly to programmers and easily portable, as well as importantly, compilers that translate source codes from high-level programming languages to low-level programming languages and often perform code optimizations during the course, it is easily envisioned that a large scale adoption of Monte Carlo quantum computing will benefit from an MCQC compiler that lies in between and bridges two specialized domains of quantum computing and applications, where in the higher-level domain, users and applications of quantum computing can be oblivious of the underlying MCQC mechanism, particularly the MCMC details, but focus on constructing/programming general and abstract quantum algorithms/circuits that can be run on a bona fide, yet-to-be-built, materially quantum piece of hardware, as well as on a Monte Carlo-simulated quantum process using an ordinary classical computer in conjunction with the presently specified method of Monte Carlo quantum computing, while in the lower-level domain, an MCQC compiler and programming and software utilities take care automatically, and free users of higher-level quantum programming and applications from considerations of\,:

\vspace{-1.5ex}
\begin{itemize}[itemsep=0.75ex]
\item[] {\em 1) Creating a Feynman-Kitaev construct that turns a quantum algorithm/circuit into either a time-independent Hamiltonian or a periodic sequence of Hamiltonians, which generates either a time-homogeneous or a lifted periodic Markov chain, whose stationary distribution encodes the result state of said quantum algorithm/circuit;}
\item[] {\em 2) Homophysically mapping said either time-independent Hamiltonian or a periodic sequence of Hamiltonians to an either time-independent SFF Hamiltonian or an SFF periodic sequence of Hamiltonians that governs a physical system of bi-fermions;}
\item[] {\em 3) Running Monte Carlo simulations using either a time-homogeneous Markov chain generated by the time-independent SFF Hamiltonian or a lifted periodic Markov chain generated by the SFF periodic sequence of Hamiltonians, to sample from the stationary distribution, and to derive a solution to the original computational problem.}
\end{itemize}
\vspace{-1.5ex}

\iftoggle{ForUSPTO} {
According to a preferred embodiment, Fig. \ref{MethodSimuSyst} illustrates one process 700 of simulating a multivariable system on a classical computer, which can be incorporated into a software system and related computer program products. Process 700 also involves computer-readable source and machine codes as well as numerical data, which are stored, processed, and executed by a hardware system and related physical device. Process 700 comprises a step 710 of receiving data or signals describing a multivariable system, followed by a step 720 of generating a plurality of result sampling points.

At step 710, data or signals describing a multivariable system are received, which comprise data or signals that represent a configuration space 711, an SFF partial Hamiltonian 712, and an objective operator 713. In one preferred embodiment, said data or signals are computer-readable numerical data stored in memory devices, which are read, written, processed, and transferred by computing devices. In another preferred embodiment, data or signals are physical signals including but not limited to electric charges, electric currents, voltages, magnetic fields, magnetic moments, electromagnetic fields, electromagnetic waves, masses of matter, number of particles, strengths of forces, physical locations, velocities of motion, mechanical energies, mechanical waves, chemical and material properties.

Configuration space 711 in turn comprises a plurality of configuration points, which are associated with a plurality of canonical variables. In one preferred embodiment, configuration space 711 is the collection of spatial coordinates of particles constituting a many-body quantum system with particle position operators being the canonical variables. In another preferred embodiment, the configuration points are strings of $0$ or $1$ valued binary variables representing configurations of spins, or qubits that comprise a quantum computing device. In yet another preferred embodiment, the configuration points are lattice points of a lattice model, namely, a physical model that is defined on a lattice, such as a quantum Ising model, a Heisenberg model of magnetism, a Hubbard model, a $t$-$J$ model, or a lattice field theory. In yet another preferred embodiment, the configuration points are vertices of a graph on which a classical random walk or a quantum walk takes place. In still another preferred embodiment, every configuration point in configuration space 711 is a pair of two coordinates, with each of the two coordinates representing values of canonical variables associated with particles constituting a many-body quantum system. In yet another preferred embodiment, every configuration point in configuration space 711 is a Feynman path, which is a tuple of a plurality of coordinates, with each of the plurality of coordinates representing values of canonical variables associated with particles constituting a many-body quantum system. Each preferred embodiment is given by way of example but no means of limitation. Configuration space 711 is endowed with a topology and comprising $2^{D_1}$ connected components, $D_1\in\mathbb{R}$, where each pair of different such connected components are topologically disconnected, each individual connected component is topologically connected and has a dimension smaller than or equal to $D_2$, $D_2\in\mathbb{R}$, $D_2\ge 0$, while at least one of the connected components contains at least one configuration point requiring a minimum number $D_2$ of numerical values to label and distinguish from other configuration points.

SFF partial Hamiltonian 712 is a function or operator involving the canonical variables, which is a combination of a plurality of DFF or GFF partial Hamiltonians, each of which involves a substantially small number of the canonical variables. In one preferred embodiment, SFF partial Hamiltonian 712 is a linear sum of DFF or GFF partial Hamiltonians. In another preferred embodiment, SFF partial Hamiltonian 712 is Lie-Trotter-Kato decomposable into a sum of DFF or GFF partial Hamiltonians. The SFF partial Hamiltonian is associated with an SFF ground state, which is a practically substantially entangled density on said configuration space. In one preferred embodiment, the SFF ground state is a element in a Hilbert space of functions supported by the configuration space. In another preferred embodiment, the SFF ground state is a density matrix supported by the configuration space comprising configuration points as a pair of two coordinates, with each of the two coordinates representing values of canonical variables associated with particles constituting a many-body quantum system. Each of the DFF or GFF partial Hamiltonians in turn is associated with one or more efficiently computable partial ground states, one of which as a density on the configuration space is substantially the same as the SFF ground state. That a partial ground state denoted by $\psi'_0$ is substantially the same as the SFF ground state denoted by $\psi_0$ means, for substantially any pair of configuration points $q$ and $r$ on which both $\psi'_0$ and $\psi_0$ are defined, that the fraction $\psi'_0(r)/\psi'_0(q)$ is substantially the same as the fraction $\psi_0(r)/\psi_0(q)$ in one preferred embodiment, or that the plus or minus numerical sign of the fraction $\psi'_0(r)/\psi'_0(q)$ is substantially the same as the plus or minus numerical sign of the fraction $\psi_0(r)/\psi_0(q)$ in another preferred embodiment.

Importantly, the partial ground states associated with the same GFF partial Hamiltonian are substantially mutually non-overlapping in the configuration space, in the sense of a substantially negligible upper bound for the set of node-uncertain configuration points, that are the configuration points on which two different partial ground states of the same partial Hamiltonian exist, whose numerical values make a product that is substantially non-negligible comparing to the sum of the norms of the same two numerical values. Also importantly, each DFF partial Hamiltonian is a direct sum of FBM partial Hamiltonians, where each FBM partial Hamiltonian moves a configuration subspace that is annihilated by any other different FBM partial Hamiltonian, in other words or alternatively, any pair of different FBM partial Hamiltonians in the same direct sum of a DFF partial Hamiltonian multiply into a null operator. In the construction of a Feynman-Kitaev Hamiltonian, it is crucial to split the Feynman-Kitaev propagators into two groups, with each group forming a direct sum and constituting a DFF partial Hamiltonian.

In one preferred embodiment, SFF partial Hamiltonian 712 is the total energy operator of a many-body quantum system. In another preferred embodiment, SFF partial Hamiltonian 712 is a finite- although high-dimensional matrix in association with a graph, a grid or mesh, a lattice of points, or a connected network. In yet another preferred embodiment, SFF partial Hamiltonian 712 is a lower-bounded self-adjoint operator acting on square-integrable functions supported by configuration space 711. In yet another preferred embodiment, SFF partial Hamiltonian 712 is an operator generating a strongly continuous one-parameter semigroup of operators acting on square-integrable functions supported by configuration space 711. Each preferred embodiment is given by way of example but no means of limitation.

Objective operator 713 is an operator depending on the canonical variables, which enters a linear, quadratic, or another form of nonlinear functional involving the SFF ground state to produce an expectation value. In one preferred embodiment, objective operator 713 is the same as SFF partial Hamiltonian 712 that measures the total energy of a physical system. In another preferred embodiment, objective operator 713 is a configuration coordinate projection operator, or called a position measurement operator, that produces an expectation value indicating the probability of finding the multivariable system at a configuration point within a predetermined subset of configuration space 711. Each preferred embodiment is given by way of example but no means of limitation.

Process step 720 further comprises step 721 and step 722 to generate a plurality of result sampling points, which can be used to compute a numerical estimate of the expectation value produced by objective operator 713. In step 721, a current sampling point (CSP) variable holding a CSP value is initialized to a configuration point, on which the SFF ground state is valued substantially different from $0$. Step 722 repeats an iterative procedure for a predetermined number $N$ of times, where the iterative procedure further comprises sub-step 723 and sub-step 724. Sub-step 723 executes a transition from the CSP value held by the CSP variable to a new sampling point in the configuration space, according to a transition probability matrix that depends on one or more of the partial Hamiltonian(s) and uses the partial ground states associated with the one or more partial Hamiltonian(s), where the transition probability matrix has a probability density function as a stationary distribution, which is related to the partial ground states of the one or more DFF or GFF partial Hamiltonian(s). In one preferred embodiment, said probability density function is point-wise the square of the absolute value of a partial ground state. Sub-step 724 transfers the CSP value held by the CSP variable to one of one or more result container(s) as a result sampling point, then resets the CSP value held by the CSP variable to the new sampling point. In another preferred embodiment, the very first result sampling point is labeled by $+1$, and for the subsequent result sampling points denoted by $q$ and $r$ that are obtained before and after a sub-step 723 respectively, the label for $r$ will be that of $q$ multiplied by $+1$ or $-1$ depending upon whether the same or oppositely signed numerical values are assumed by the partial ground state used by the transition probability matrix for the execution of sub-step 723, then accordingly, each of the result sampling points will be transferred to a positive result container or a negative result container depending upon whether the result sampling point is labeled by $+1$ or $-1$ respectively. In another preferred embodiment, the negative result container in the previous preferred embodiment defaults to void, and all of the would-be $(-1)$-labeled result sampling points are discarded. In yet another preferred embodiments, more than two result containers are used. Each preferred embodiment is given by way of example but no means of limitation.

The result sampling points obtained in process step 720 enable a numerical estimate for the expectation value associated with objective operator 713, with the numerical estimate having a relative error smaller than or equal to a formula $a(D_1+D_2)^{\,b}N^{-c}$ when both $D_1+D_2$ and $N$ becoming substantially large, with $a$, $b$, and $c$ being constants, both $a$ and $c$ further being positive.

According to a preferred embodiment, Fig. \ref{MethodSimuDens} illustrates another process 800 of simulating a multivariable system, which can be incorporated into a software system and related computer program products. Process 800 also involves computer-readable source and machine codes as well as numerical data, which are stored, processed, and executed by a hardware system and related physical devices. Process 800 comprises a step 810 of generating data or signals describing a multivariable system, followed by a step 820 of generating a plurality of result sampling points.

At step 810, data or signals describing a multivariable system are received, which comprise data or signals that represent a multi-dimensional space (MDS) of the first kind 811, a plurality of groups of low-dimensional densities (LDDs) 812, and an objective operator 813. In one preferred embodiment, said data or signals are computer-readable numerical data stored in memory devices, which are read, written, processed, and transferred by computing devices. In another preferred embodiment, data or signals are physical signals including but not limited to electric charges, electric currents, voltages, magnetic fields, magnetic moments, electromagnetic fields, electromagnetic waves, masses of matter, number of particles, strengths of forces, physical locations, velocities of motion, mechanical energies, mechanical waves, chemical and material properties.

MDS of the first kind 811 in turn comprises a plurality of multi-dimensional points (MDPs) of a first kind. In one preferred embodiment, MDS of the first kind 811 is the collection of spatial coordinates of particles constituting a many-body quantum system with particle position operators being the canonical variables. In another preferred embodiment, the MDPs of the first kind are strings of $0$ or $1$ valued binary variables representing configurations of spins, or qubits that comprise a quantum computing device. In yet another preferred embodiment, the MDPs of the first kind are lattice points of a lattice model, namely, a physical model that is defined on a lattice, such as a quantum Ising model, a Heisenberg model of magnetism, a Hubbard model, a $t$-$J$ model, or a lattice field theory. In yet another preferred embodiment, the MDPs of the first kind are vertices of a graph on which a classical random walk or a quantum walk takes place. In still another preferred embodiment, every MDP in MDS of the first kind 811 is a pair of two coordinates, with each of the two coordinates representing values of canonical variables associated with particles constituting a many-body quantum system. In yet another preferred embodiment, every MDP in MDS of the first kind 811 is a Feynman path, which is a tuple of a plurality of coordinates, with each of the plurality of coordinates representing values of canonical variables associated with particles constituting a many-body quantum system. Each preferred embodiment is given by way of example but no means of limitation. MDS of the first kind 811 is endowed with a topology and comprising $2^{D_1}$ connected components, $D_1\in\mathbb{R}$, where each pair of different such connected components are topologically disconnected, each individual connected component is topologically connected and has a dimension smaller than or equal to $D_2$, $D_2\in\mathbb{R}$, $D_2\ge 0$, while at least one of the connected components contains at least one MDP requiring a minimum number $D_2$ of numerical values to label and distinguish from other MDPs. MDS of the first kind 811 is associated with a plurality of means of factoring itself into a low-dimensional space (LDS) comprising a plurality of low-dimensional points (LDPs) and an MDS of a second kind comprising a plurality of MDPs of a second kind, such that each of the MDPs of the first kind in MDS of the first kind 811 is substantially the same as one and only ordered pair of an LDP and an MDP of the second kind, whereas said LDS is endowed with an LDS topology and comprises $2^{d_1}$ connected LDS components, $d_1\in\mathbb{R}$, where each pair of different such connected LDS components are topologically disconnected, each of such connected LDS components is topologically connected and has a dimension smaller than or equal to $d_2$, $d_2\in\mathbb{R}$, $d_2\ge 0$, at least one of the connected LDS components comprises one or more of said LDPs that require(s) a minimum number $d_2$ of numerical values to label and distinguish, wherein $2^{d_1}$ is substantially smaller than $2^{D_1}$ when $2^{D_1}$ is a substantially large number, and $d_2$  is substantially smaller than $D_2$ when $D_2$ is a substantially large number.

All LDDs in the plurality of groups of LDDs 812 are efficiently computable. Each said group of LDDs is associated with an LDS and an MDS of the second kind according to one of said plurality of means of factoring MDS of the first kind 811, with each said LDD included in such group of LDDs being associated with a portion of the MDS of the second kind associated with the group of LDDs, thus each of the LDDs is associated with a portion of MDS of the first kind 811, such that each MDP included in said portion of MDS of the first kind 811 is substantially the same as one and only one ordered pair of an LDP included in the LDS and an MDP of the second kind associated with the group of LDDs, whereas each pair of said portions of MDS of the first kind 811 associated respectively with two different LDDs included in the same group of LDDs are substantially non-overlapping with respect to a volumetric measure, such that the volumetric measure of the intersection set of said pair of said portions of said MDS of the first kind is substantially negligible comparing to the volumetric measure of each of the pair of said portions of said MDS of the first kind, wherein the plurality of groups of LDDs 812 collectively determine a practically substantially entangled high-dimensional density (HDD) that is uniquely defined over MDS of the first kind 811, such that the HDD is substantially the same as each LDD in the corresponding portion of MDS of the first kind 811 associated with that LDD. In one preferred embodiments, a group of LDDs consists of the eigenvectors of a Gibbs operator associated with the largest eigenvalue, while in another preferred embodiment, a group of LDDs comprises general signed densities. Furthermore, in one preferred embodiment, said volumetric measure is a Lebesgue volume measure with respect to a Riemannian metric on a Riemannian manifold, while in another preferred embodiment, said volumetric measure is a probability measure with a probability density function defined in relation to the HDD and the LDDs. Still further, in one preferred embodiment, the HDD is a ground state wavefunction of a physical system, while in another preferred embodiment, the HDD is a general signed density over a multi-dimensional space. In yet another preferred embodiment, the HDD is a element in a Hilbert space of functions supported by the MDS. In still another preferred embodiment, the HDD is a density matrix supported by the MDS comprising MDPs as a pair of two coordinates, with each of the two coordinates representing values of canonical variables associated with particles constituting a many-body quantum system. Each preferred embodiment is given by way of example but no means of limitation.

Objective operator 813 is an operator that enters a linear, quadratic, or another form of nonlinear functional involving the HDD to produce an expectation value. In one preferred embodiment, objective operator 813 is a partial Hamiltonian that measures the total energy of a physical system. In another preferred embodiment, objective operator 813 is an MDP projection operator, or called a position measurement operator, that produces an expectation value indicating the probability of finding the multivariable system at an MDP within a predetermined subset of MDS of the first kind 811. Each preferred embodiment is given by way of example but no means of limitation.

Process step 820 further comprises step 821 and step 822 to generate a plurality of result sampling points, which can be used to compute a numerical estimate of the expectation value produced by objective operator 813. In step 821, a current sampling point (CSP) variable holding a CSP value is initialized to an MDP, on which the HDD is valued substantially different from $0$. Step 822 repeats an iterative procedure for a predetermined number $N$ of times, where the iterative procedure further comprises sub-step 823 and sub-step 824. Sub-step 823 executes a transition from the CSP value held by the CSP variable to a new sampling point in MDS of the first kind 811, according to a transition probability matrix that depends on one or more groups of the plurality of groups of LDDs 812 and uses the LDDs in the one or more groups of LDDs, where the transition probability matrix has a probability density function as a stationary distribution, which is related to one or more groups of the plurality of groups of LDDs 812. In one preferred embodiment, said probability density function is point-wise the square or sum of squares of the absolute value(s) of one or more LDD(s). Sub-step 824 transfers the CSP value held by the CSP variable to one of one or more result container(s) as a result sampling point, then resets the CSP value held by the CSP variable to the new sampling point. In one preferred embodiment, the very first result sampling point is labeled by $+1$, and for the subsequent result sampling points denoted by $q$ and $r$ that are obtained before and after a sub-step 823 respectively, the label for $r$ will be that of $q$ multiplied by $+1$ or $-1$ depending upon whether the same or oppositely signed numerical values are assumed by an LDD used by the transition probability matrix for the execution of sub-step 823, then accordingly, each of the result sampling points will be transferred to a positive result container or a negative result container depending upon whether the result sampling point is labeled by $+1$ or $-1$ respectively. In another preferred embodiment, the negative result container in the previous preferred embodiment defaults to void, and all of the would-be $(-1)$-labeled result sampling points are discarded. In yet another preferred embodiments, more than two result containers are used. Each preferred embodiment is given by way of example but no means of limitation.

The result sampling points obtained in process step 820 enable a numerical estimate for the expectation value associated with objective operator 813, with the numerical estimate having a relative error smaller than or equal to a formula $a(D_1+D_2)^{\,b}N^{-c}$ when both $D_1+D_2$ and $N$ becoming substantially large, with $a$, $b$, and $c$ being constants, both $a$ and $c$ further being positive.

According to a preferred embodiment, Fig. \ref{MethodConsFK} illustrates one process 900 of constructing a multivariable system, which comprises three steps: a step 910 of receiving a description of a quantum algorithm, followed by a step 920 of constructing a Feynman-Kitaev (FK) register, further followed by a step 930 of constructing a separately frustration-free (SFF) partial Hamiltonian. Process 900 can be incorporated into a software system and related computer program products. Process 900 also involves computer-readable source and machine codes as well as numerical data, which are stored, processed, and executed by a hardware system and related physical devices.

At step 910, a description of a quantum algorithm is received, which comprises: a list 911 of a plurality qubits, an initial state 912, a linearly ordered sequence 913 of a plurality of quantum gates, and a quantum measurement operator 914. Each qubit in list 911 
is associated with a single-qubit configuration space, a single-qubit Hilbert space that comprising state vectors called single-qubit states, as well as single-qubit dynamical variables as operators acting on the single-qubit states. A Cartesian product of multiple single-qubit configuration spaces becomes a multi-qubit configuration space. In one preferred embodiment, the multi-qubit configuration space is continuous. In another preferred embodiment, the multi-qubit configuration space is a topological manifold. In yet another preferred embodiment, the multi-qubit configuration space is a discrete graph or a discrete space of lattice points. In yet another preferred embodiment, the multi-qubit configuration space is a continuous-discrete mixed space, namely, a space with both continuous and discrete coordinates or degrees of freedom. In still another preferred embodiment, every configuration point in the configuration space is a pair of two coordinates, with each of the two coordinates representing values of canonical variables associated with particles constituting a many-body quantum system. In yet another preferred embodiment, every configuration point in the configuration space is a Feynman path, which is a tuple of a plurality of coordinates, with each of the plurality of coordinates representing values of canonical variables associated with particles constituting a many-body quantum system. Each preferred embodiment is given by way of example but no means of limitation. A tensor product of multiple single-qubit Hilbert spaces becomes a multi-qubit Hilbert space supported by the multi-qubit configuration space. Each element in the multi-qubit Hilbert space is called a multi-qubit state, or state vector. Initial state 912 is a multi-qubit state. In one preferred embodiment, initial state 912 is a element in a Hilbert space of functions supported by the configuration space. In another preferred embodiment, initial state 912 is a density matrix supported by the configuration space comprising configuration points as a pair of two coordinates, with each of the two coordinates representing values of canonical variables associated with particles constituting a many-body quantum system. The multi-qubit Hilbert space has an inner product defined, which takes any pair of multi-qubit states as input and produces a numerical value as output, with the output numerical value also referred to as the inner product between the two input multi-qubit states. Two multi-qubit states are said to be orthogonal, or one multi-qubit state is said to be orthogonal to the other, when the inner product between them is substantially $0$. Each quantum gate in linearly ordered sequence 913 of a plurality of quantum gates involves a number of the single-qubit dynamical variables. Each such quantum gate acts on and transforms the multi-qubit states. The plurality of quantum gates in 913 are linearly ordered in the sense that, between any pair of such quantum gates, there is always one called earlier, and the other called later. In one preferred embodiment, some of the quantum gates are unitary matrices or operators. In another preferred embodiment, some of the quantum gates are non-unitary matrices or operators. In yet another preferred embodiment, some of the quantum gates are superoperators that map density matrices to density matrices \cite{Davis76,Gruska99,Nielsen10}. Quantum measurement operator 914 also involves a number of single-qubit dynamical variables, as well as acts on and transforms the multi-qubit states. In one preferred embodiment, quantum measurement operator 914 is a partial Hamiltonian that measures energy. In another preferred embodiment, quantum measurement operator 914 is a position operator that measures the coordinates of one or more qubit(s) in their respective single-qubit configuration space(s). Each preferred embodiment is given by way of example but no means of limitation. In the described quantum algorithm, a multi-qubit system comprising multiple qubits is initially at initial state 912. Then the quantum gates in linearly ordered sequence 913 are applied to the multi-qubit system to produce a result multi-qubit state that is a practically substantially entangled density supported by the multi-qubit configuration space. Finally, a quantum measurement is performed using quantum measurement operator 914 against the result multi-qubit state, so to yield an interested expectation value of quantum measurement operator 914.

Step 920 constructs a Feynman-Kitaev (FK) register, which further comprises a sub-step 921 of allocating a plurality of FK logic bits associated with a plurality of logic states and a plurality of FK clock bits associated with a clock configuration space consisting of clock configuration points, a sub-step 922 of creating a plurality of FK time projectors with each said FK time projector fixing said clock configuration points, and a sub-step 923 of creating a plurality of FK time propagators, where said plurality of FK time propagators are divided into a first group of first FK time propagators and a second group of second FK time propagators. Each said first FK time propagator is able to connect different clock configuration points in a first subset of said clock configuration space. Each said second FK time propagator is able to connect different clock configuration points in a second subset of said clock configuration space. A first graph is formed having said configuration points in the first subset of said clock configuration space as vertices and said first FK time propagators as edges. A second graph is formed having said configuration points in the second subset of said clock configuration space as vertices and said second FK time propagators as edges. It is crucial to make sure that the first graph and the second graph are disjoint.

Step 930 constructs an SFF partial Hamiltonian that is a combination of a plurality of directly frustration-free (DFF) or ground state frustration-free (GFF) partial Hamiltonians, which further comprises a sub-step 931 of creating an FK state initializer as a tensor product between an FK time projector and a logic state initializer, a sub-step 932 of creating a first FK state operator as a combination of a plurality of first FK state propagators, and a sub-step 933 of creating a second FK state operator as a combination of a plurality of second FK state propagators. In sub-step 931, the FK state initializer is a DFF or GFF partial Hamiltonian, and the logic state initializer produces an initial logic state, while the initial logic state is related to said initial state. In sub-step 932, the first FK state operator is a DFF or GFF partial Hamiltonian, while each said first FK state propagator is a tensor product between a first FK time propagator and one of said quantum gates transforming said logic states. In sub-step 933, the second FK state operator is a DFF or GFF partial Hamiltonian, while each said second FK state propagator is a tensor product between a second FK time propagator and one of said quantum gates transforming said logic states. Importantly, both the first FK state operator and the second FK state operator are a direct sum of FBM partial Hamiltonians, where each FBM partial Hamiltonian moves a configuration subspace that is annihilated by any other different FBM partial Hamiltonian, in other words or alternatively, any pair of different FBM partial Hamiltonians in the same direct sum of a first FK state operator or a second FK state operator multiply into a null operator. An SFF partial Hamiltonian so constructed in step 930 has a non-degenerate ground state, which in association with said quantum measurement operator 914 produces a numerical estimate for the expectation value of quantum measurement operator 914.

According to a preferred embodiment, Fig. \ref{MethodSolvSyst} illustrates one process 1000 of of solving a computational problem on a classical computer, which can be incorporated into a software system and related computer program products. Process 1000 also involves computer-readable source and machine codes as well as numerical data, which are stored, processed, and executed by a hardware system and related physical devices. Process 1000 comprises a step 1010 of generating data or signals describing a multivariable system, followed by a step 720 of generating a plurality of result sampling points in substantially the same manner as that being specified for process 700. In one preferred embodiment, step 1010 uses a method according to process 900. In another preferred embodiment, during step 722 of process 700, the plurality of FK time projectors included in sub-step 931 and the plurality of FK state propagators generated in sub-step 933 are applied one by one, where each such application of an FK time projector or an FK state propagator parses the current sampling point (CSP) then exerts the intended operation. In yet another preferred embodiment, during step 722 of process 700, the CSP is firstly parsed to select one or just a few of the plurality of FK time projectors included in sub-step 931 and the plurality of FK state propagators generated in sub-step 933 whose intended operation(s) produce(s) non-trivial result(s), then the selected one or just a few of the FK time projectors and the FK state propagators are applied to exert the intended operation(s).

Step 1010 generates data or signals describing a multivariable system, which further comprises a sub-step 1011 of allocating a configuration space, a sub-step 1012 of creating an SFF partial Hamiltonian, and a sub-step 1013 of generating an objective operator. In one preferred embodiment, said data or signals are computer-readable numerical data stored in memory devices, which are read, written, processed, and transferred by computing devices. In another preferred embodiment, data or signals are physical signals including but not limited to electric charges, electric currents, voltages, magnetic fields, magnetic moments, electromagnetic fields, electromagnetic waves, masses of matter, number of particles, strengths of forces, physical locations, velocities of motion, mechanical energies, mechanical waves, chemical and material properties.

The configuration space allocated in sub-step 1011 comprises a plurality of configuration points, which are associated with a plurality of canonical variables. In one preferred embodiment, the configuration space is the collection of spatial coordinates of particles constituting a many-body quantum system with particle position operators being the canonical variables. In another preferred embodiment, the configuration points are strings of $0$ or $1$ valued binary variables representing configurations of spins, or qubits that comprise a quantum computing device. In yet another preferred embodiment, the configuration points are lattice points of a lattice model, namely, a physical model that is defined on a lattice, such as a quantum Ising model, a Heisenberg model of magnetism, a Hubbard model, a $t$-$J$ model, or a lattice field theory. In yet another preferred embodiment, the configuration points are vertices of a graph on which a classical random walk or a quantum walk takes place. In still another preferred embodiment, every configuration point in the configuration space is a pair of two coordinates, with each of the two coordinates representing values of canonical variables associated with particles constituting a many-body quantum system. In yet another preferred embodiment, every configuration point in the configuration space is a Feynman path, which is a tuple of a plurality of coordinates, with each of the plurality of coordinates representing values of canonical variables associated with particles constituting a many-body quantum system. Each preferred embodiment is given by way of example but no means of limitation. The configuration space is endowed with a topology and comprising $2^{D_1}$ connected components, $D_1\in\mathbb{R}$, where each pair of different such connected components are topologically disconnected, each individual connected component is topologically connected and has a dimension smaller than or equal to $D_2$, $D_2\in\mathbb{R}$, $D_2\ge 0$, while at least one of the connected components contains at least one configuration point requiring a minimum number $D_2$ of numerical values to label and distinguish from other configuration points.

The SFF partial Hamiltonian created in sub-step 1012 is a function or operator involving the canonical variables, which is a combination of a plurality of partial Hamiltonians, each of which involves a substantially small number of the canonical variables. In one preferred embodiment, the SFF partial Hamiltonian is a linear sum of partial Hamiltonians. In another preferred embodiment, the SFF partial Hamiltonian is Lie-Trotter-Kato decomposable into a sum of partial Hamiltonians. The SFF partial Hamiltonian is associated with an SFF ground state, which is a practically substantially entangled density on said configuration space. In one preferred embodiment, the SFF ground state is a element in a Hilbert space of functions supported by the configuration space. In another preferred embodiment, the SFF ground state is a density matrix supported by the configuration space comprising configuration points as a pair of two coordinates, with each of the two coordinates representing values of canonical variables associated with particles constituting a many-body quantum system. Each of the partial Hamiltonians in turn is associated with one or more efficiently computable partial ground states, one of which as a density on the configuration space is substantially the same as the SFF ground state. That a partial ground state denoted by $\psi'_0$ is substantially the same as the SFF ground state denoted by $\psi_0$ means, for substantially any pair of configuration points $q$ and $r$ on which both $\psi'_0$ and $\psi_0$ are defined, that the fraction $\psi'_0(r)/\psi'_0(q)$ is substantially the same as the fraction $\psi_0(r)/\psi_0(q)$ in one preferred embodiment, or that the plus or minus numerical sign of the fraction $\psi'_0(r)/\psi'_0(q)$ is substantially the same as the plus or minus numerical sign of the fraction $\psi_0(r)/\psi_0(q)$ in another preferred embodiment. Importantly, the partial ground states associated with the same partial Hamiltonian are substantially mutually non-overlapping in the configuration space, in the sense of a substantially negligible upper bound for the set of node-uncertain configuration points, that are the configuration points on which two different partial ground states of the same partial Hamiltonian exist, whose numerical values make a product that is substantially non-negligible comparing to the sum of the norms of the same two numerical values. In one preferred embodiment, the SFF partial Hamiltonian is the total energy operator of a many-body quantum system. In another preferred embodiment, the SFF partial Hamiltonian is a finite- although high-dimensional matrix in association with a graph, a grid or mesh, a lattice of points, or a connected network. In yet another preferred embodiment, the SFF partial Hamiltonian is a lower-bounded self-adjoint operator acting on square-integrable functions supported by the configuration space allocated in sub-step 1011. In yet another preferred embodiment, the SFF partial Hamiltonian is an operator generating a strongly continuous one-parameter semigroup of operators acting on square-integrable functions supported by the configuration space allocated in sub-step 1011. Each preferred embodiment is given by way of example but no means of limitation.

The objective operator generated in sub-step 1013 is an operator depending on the canonical variables, which enters a linear, quadratic, or another form of nonlinear functional involving the SFF ground state to produce an expectation value. In one preferred embodiment, the objective operator is the same as the SFF partial Hamiltonian created in sub-step 1012 that measures the total energy of a physical system. In another preferred embodiment, the objective operator is a configuration coordinate projection operator, or called a position measurement operator, that produces an expectation value indicating the probability of finding the multivariable system at a configuration point within a predetermined subset of the configuration space allocated in sub-step 1011. Each preferred embodiment is given by way of example but no means of limitation.

According to a preferred embodiment, Fig. \ref{MethodSolvDens} illustrates one process 1100 of of solving a computational problem on a classical computer, which can be incorporated into a software system and related computer program products. Process 1100 also involves computer-readable source and machine codes as well as numerical data, which are stored, processed, and executed by a hardware system and related physical devices. Process 1100 comprises a step 1110 of generating data or signals describing a multivariable system, followed by a step 820 of generating a plurality of result sampling points in substantially the same manner as that being specified for process 800. In one preferred embodiment, step 1110 uses a method according to process 900. In another preferred embodiment, during step 822 of process 800, the plurality of FK time projectors included in sub-step 931 and the plurality of FK state propagators generated in sub-step 933 are applied one by one, where each such application of an FK time projector or an FK state propagator parses the current sampling point (CSP) then exerts the intended operation. In yet another preferred embodiment, during step 822 of process 800, the CSP is firstly parsed to select one or just a few of the plurality of FK time projectors included in sub-step 931 and the plurality of FK state propagators generated in sub-step 933 whose intended operation(s) produce(s) non-trivial result(s), then the selected one or just a few of the FK time projectors and the FK state propagators are applied to exert the intended operation(s).

At step 1110, data or signals describing a multivariable system are generated, which comprise data or signals that represent a multi-dimensional space (MDS) of the first kind 1111, a plurality of groups of low-dimensional densities (LDDs) 1112, and an objective operator 1113. In one preferred embodiment, said data or signals are computer-readable numerical data stored in memory devices, which are read, written, processed, and transferred by computing devices. In another preferred embodiment, data or signals are physical signals including but not limited to electric charges, electric currents, voltages, magnetic fields, magnetic moments, electromagnetic fields, electromagnetic waves, masses of matter, number of particles, strengths of forces, physical locations, velocities of motion, mechanical energies, mechanical waves, chemical and material properties.

MDS of the first kind 1111 in turn comprises a plurality of multi-dimensional points (MDPs) of a first kind. In one preferred embodiment, MDS of the first kind 1111 is the collection of spatial coordinates of particles constituting a many-body quantum system with particle position operators being the canonical variables. In another preferred embodiment, the MDPs of the first kind are strings of $0$ or $1$ valued binary variables representing configurations of spins, or qubits that comprise a quantum computing device. In yet another preferred embodiment, the MDPs of the first kind are lattice points of a lattice model, namely, a physical model that is defined on a lattice, such as a quantum Ising model, a Heisenberg model of magnetism, a Hubbard model, a $t$-$J$ model, or a lattice field theory. In yet another preferred embodiment, the MDPs of the first kind are vertices of a graph on which a classical random walk or a quantum walk takes place. Each preferred embodiment is given by way of example but no means of limitation. MDS of the first kind 1111 is endowed with a topology and comprising $2^{D_1}$ connected components, where each pair of different such connected components are topologically disconnected, each individual connected component is topologically connected and has a dimension smaller than or equal to $D_2$, while at least one of the connected components contains at least one MDP requiring a minimum number $D_2$ of numerical values to label and distinguish from other MDPs. MDS of the first kind 1111 is associated with a plurality of means of factoring itself into a low-dimensional space (LDS) comprising a plurality of low-dimensional points (LDPs) and an MDS of a second kind comprising a plurality of MDPs of a second kind, such that each of the MDPs of the first kind in MDS of the first kind 1111 is substantially the same as one and only ordered pair of an LDP and an MDP of the second kind, whereas said LDS is endowed with an LDS topology and comprises $2^{d_1}$ connected LDS components, where each pair of different such connected LDS components are topologically disconnected, each of such connected LDS components is topologically connected and has a dimension smaller than or equal to $d_2$, at least one of the connected LDS components comprises one or more of said LDPs that require(s) a minimum number $d_2$ of numerical values to label and distinguish, wherein $2^{d_1}$ is substantially smaller than $2^{D_1}$ when $2^{D_1}$ is a substantially large number, and $d_2$  is substantially smaller than $D_2$ when $D_2$ is a substantially large number.

All LDDs in the plurality of groups of LDDs 1112 are efficiently computable. Each said group of LDDs is associated with an LDS and an MDS of the second kind according to one of said plurality of means of factoring MDS of the first kind 1111, with each said LDD included in such group of LDDs being associated with a portion of the MDS of the second kind associated with the group of LDDs, thus each of the LDDs is associated with a portion of MDS of the first kind 1111, such that each MDP included in said portion of MDS of the first kind 1111 is substantially the same as one and only one ordered pair of an LDP included in the LDS and an MDP of the second kind associated with the group of LDDs, whereas each pair of said portions of MDS of the first kind 1111 associated respectively with two different LDDs included in the same group of LDDs are substantially non-overlapping with respect to a volumetric measure, such that the volumetric measure of the intersection set of said pair of said portions of said MDS of the first kind is substantially negligible comparing to the volumetric measure of each of the pair of said portions of said MDS of the first kind, wherein the plurality of groups of LDDs 1112 collectively determine a practically substantially entangled high-dimensional density (HDD) that is uniquely defined over MDS of the first kind 1111, such that the HDD is substantially the same as each LDD in the corresponding portion of MDS of the first kind 1111 associated with that LDD. In one preferred embodiments, a group of LDDs consists of the eigenvectors of a Gibbs operator associated with the largest eigenvalue, while in another preferred embodiment, a group of LDDs comprises general signed densities. Furthermore, in one preferred embodiment, said volumetric measure is a Lebesgue volume measure with respect to a Riemannian metric on a Riemannian manifold, while in another preferred embodiment, said volumetric measure is a probability measure with a probability density function defined in relation to the HDD and the LDDs. Still further, in one preferred embodiment, the HDD is a ground state wavefunction of a physical system, while in another preferred embodiment, the HDD is a general signed density over a multi-dimensional space. Each preferred embodiment is given by way of example but no means of limitation.

Objective operator 1113 is an operator that enters a linear, quadratic, or another form of nonlinear functional involving the HDD to produce an expectation value. In one preferred embodiment, objective operator 1113 is a partial Hamiltonian that measures the total energy of a physical system. In another preferred embodiment, objective operator 1113 is an MDP projection operator, or called a position measurement operator, that produces an expectation value indicating the probability of finding the multivariable system at an MDP within a predetermined subset of MDS of the first kind 1111. Each preferred embodiment is given by way of example but no means of limitation.
} {}

\iftoggle{ForUSPTO} {
\clearpage
\section*{CLAIMS}
\newcounter{claim}[section]
\refstepcounter{claim}

What is claimed is:

\vspace{-2.0ex}
\begin{itemize}[leftmargin=0.0ex]

\item[\theclaim.] \label{ClaimSimuSyst} A method of simulating a multivariable system, comprising:
\begin{itemize}[leftmargin=0.0ex]
\item[] receiving data or signals describing a multivariable system, said data or signals comprising:
    \begin{itemize}[itemsep=0.75ex]
    \item[] a description of a configuration space comprising in turn a plurality of configuration points, said configuration points being associated with a plurality of canonical variables, said configuration space being endowed with a topology and comprising $2^{D_1}$ connected components, each pair of different said connected components being topologically disconnected, each said connected component being topologically connected and having a dimension smaller than or equal to $D_2$, at least one of said connected component comprising one or more of said configuration points requiring a minimum number $D_2$ of numerical values to label and distinguish;
    \item[] a description of a separately frustration-free (SFF) partial Hamiltonian involving said canonical variables, said SFF partial Hamiltonian being a combination of a plurality of directly frustration-free (DFF) or ground state frustration-free (GFF) partial Hamiltonians, said SFF partial Hamiltonian being associated with an SFF ground state, said SFF ground state as a density on said configuration space being practically substantially entangled;
    \item[] each said DFF partial Hamiltonian being a direct sum of DFF-few-body-moving (DFF-FBM) partial Hamiltonians and being associated with a DFF ground state, each said DFF-FBM partial Hamiltonian involving a substantially small number of said canonical variables, different DFF-FBM partial Hamiltonians in said direct sum of each said DFF partial Hamiltonian involving disjoint sets of said canonical variables;
    \item[] each said GFF partial Hamiltonian being a combination of GFF-few-body-moving (GFF-FBM) partial Hamiltonians and being associated with a GFF ground state, each said GFF-FBM partial Hamiltonian involving a substantially small number of said canonical variables and being associated with said GFF ground state as a ground state;
    \item[] said multivariable system still further comprising an objective variable, said objective variable involving said canonical variables and being associated with an expectation value with respect to said SFF ground state;
    \end{itemize}
\item[] generating a plurality of result sampling points, comprising:
    \begin{itemize}
    \item[] initializing a current sampling point (CSP) value held by a CSP variable to a first sampling point as a member included in said configuration space, said SFF ground state being valued substantially different from zero at said first sampling point;
    \item[] repeating an iterative procedure for a predetermined number $N$ of times, said iterative procedure comprising:
        \begin{itemize}[itemsep=0.75ex]
        \item[] a first iterative step executing a transition from the CSP value held by said CSP variable to a second sampling point as a member included in said configuration space according to a transition probability matrix, said transition probability matrix depending on one or more of said DFF or GFF partial Hamiltonian(s) at the same time using said DFF or GFF ground states associated with said one or more of said DFF or GFF partial Hamiltonian(s), wherein said transition probability matrix has a probability density function as a stationary distribution, said probability density function is related to said DFF or GFF ground states associated with said one or more of said DFF or GFF partial Hamiltonian(s);
        \item[] a second iterative step transferring the CSP value held by said CSP variable to one of one or more result container(s) as a result sampling point and resetting the CSP value held by said CSP variable to said second sampling point;
        \end{itemize}
    \item[] whereby said result sampling points enable a numerical estimate for said objective variable, the numerical estimate having a relative error smaller than or equal to a formula $a(D_1+D_2)^{\,b}N^{-c}$ when both $D_1+D_2$ and $N$ becoming substantially large, with $a$, $b$, and $c$ being constants, both $a$ and $c$ further being positive.
    \end{itemize}
\end{itemize}
\refstepcounter{claim}

\item[\theclaim.] \label{ClaimConsFK} A method of constructing a multivariable system, comprising:
\begin{itemize}[leftmargin=0.0ex]
\item[] receiving a description of a quantum algorithm, said description comprising:
    \begin{itemize}[itemsep=0.75ex]
    \item[] a list of a plurality of qubits, each said qubit being associated with a first Hilbert space comprising states of a first kind supported by a first configuration space, each said qubit further being associated with dynamical variables of a first kind transforming said states of the first kind, whereby a Cartesian product of said first configuration spaces associated with said plurality of qubits becoming a second configuration space, a tensor product of said first Hilbert spaces associated with said plurality of qubits becoming a second Hilbert space comprising states of a second kind supported by the second configuration space, said second Hilbert space having an inner product function defined, said inner product function taking a pair of said states of the second kind as input and producing a numerical value as output, with said numerical value also referred to as the inner product between the two states of the second kind, whereby two states of the second kind are said to be orthogonal, or one state of the second kind is said to be orthogonal to the other, when the inner product between them is substantially $0$;
    \item[] an initial state, said initial state being a state of the second kind;
    \item[] an ordered sequence of quantum gates, each said quantum gate involving a number of said dynamical variables of the first kind and transforming said sates of the second kind;
    \item[] a quantum measurement operator involving said dynamical variables of the first kind;
    \item[] whereby said ordered sequence of quantum gates applied in order produce a result state, said result state is a state of the second kind and practically substantially entangled as a density on the second configuration space, said quantum measurement operator in association with said result state produces a first expectation value;
    \end{itemize}
\item[] constructing a Feynman-Kitaev (FK) register, comprising:
    \begin{itemize}[itemsep=0.75ex]
    \item[] allocating a plurality of FK logic bits associated with a plurality of logic states and a plurality of FK clock bits associated with a clock configuration space consisting of clock configuration points;
    \item[] creating a plurality of FK time projectors, each said FK time projector fixing said clock configuration points;
    \item[] creating a first group of first FK time propagators and a second group of second FK time propagators, each said first FK time propagator being able to connect different clock configuration points in a first subset of said clock configuration space, each said second FK time propagator being able to connect different clock configuration points in a second subset of said clock configuration space;
    \item[] whereby a first graph is formed having said configuration points in the first subset of said clock configuration space as vertices and said first FK time propagators as edges, a second graph is formed having said configuration points in the second subset of said clock configuration space as vertices and said second FK time propagators as edges, said first graph and said second graph are disjoint;
    \end{itemize}
\item[] constructing a separately frustration-free (SFF) partial Hamiltonian as a combination of a plurality of directly frustration-free (DFF) or ground state frustration-free (GFF) partial Hamiltonians, comprising:
    \begin{itemize}[itemsep=0.75ex]
    \item[] creating an FK state initializer as a tensor product between an FK time projector and a logic state initializer, said FK state initializer being a DFF or GFF partial Hamiltonian, said logic state initializer producing an initial logic state, said initial logic state being related to said initial state;
    \item[] creating a first FK state operator as a combination of a plurality of first FK state propagators, said first FK state operator being a DFF or GFF partial Hamiltonian, each said first FK state propagator being a tensor product between a first FK time propagator and one of said quantum gates transforming said logic states;
    \item[] creating a second FK state operator as a combination of a plurality of second FK state propagators, said second FK state operator being a DFF or GFF partial Hamiltonian, each said second FK state propagator being a tensor product between a second FK time propagator and one of said quantum gates transforming said logic states;
    \end{itemize}
\item[] whereby the ground state of said SFF partial Hamiltonian is non-degenerate and in association with said quantum measurement operator produces a second expectation value that is substantially the same as said first expectation value.
\end{itemize}
\refstepcounter{claim}

\item[\theclaim.] \label{ClaimSolvSyst} A method of solving a computational problem, comprising:
\vspace{-1.0ex}
\begin{itemize}[leftmargin=0.0ex]
\item[] generating data or signals describing a multivariable system, said data or signals comprising:
    \begin{itemize}[itemsep=0.75ex]
    \item[] creating a description of a configuration space comprising in turn a plurality of configuration points, said configuration points being associated with a plurality of canonical variables, said configuration space being endowed with a topology and comprising $2^{D_1}$ connected components, each pair of different said connected components being topologically disconnected, each said connected component being topologically connected and having a dimension smaller than or equal to $D_2$, at least one of said connected component comprising one or more of said configuration points requiring a minimum number $D_2$ of numerical values to label and distinguish;
    \item[] creating a separately frustration-free (SFF) partial Hamiltonian involving said canonical variables:
        \begin{itemize}[itemsep=0.75ex]
        \item[] said SFF partial Hamiltonian being a combination of a plurality of directly frustration-free (DFF) or ground state frustration-free (GFF) partial Hamiltonians, said SFF partial Hamiltonian being associated with an SFF ground state, said SFF ground state as a density on said configuration space being practically substantially entangled;
        \item[] each said DFF partial Hamiltonian being a direct sum of DFF-few-body-moving (DFF-FBM) partial Hamiltonians and being associated with a DFF ground state, each said DFF-FBM partial Hamiltonian involving a substantially small number of said canonical variables, different DFF-FBM partial Hamiltonians in said direct sum of each said DFF partial Hamiltonian involving disjoint sets of said canonical variables;
        \item[] each said GFF partial Hamiltonian being a combination of GFF-few-body-moving (GFF-FBM) partial Hamiltonians and being associated with a GFF ground state, each said GFF-FBM partial Hamiltonian involving a substantially small number of said canonical variables and being associated with said GFF ground state as a ground state;
        \end{itemize}
    \item[] generating an objective variable, said objective variable involving said canonical variables and being associated with an expectation value with respect to said SFF ground state;
    \end{itemize}
\item[] generating a plurality of result sampling points, comprising:
    \begin{itemize}
    \item[] initializing a current sampling point (CSP) value held by a CSP variable to a first sampling point as a member included in said configuration space, said SFF ground state being valued substantially different from zero at said first sampling point;
    \item[] repeating an iterative procedure for a predetermined number $N$ of times, said iterative procedure comprising:
        \begin{itemize}[itemsep=0.75ex]
        \item[] a first iterative step executing a transition from the CSP value held by said CSP variable to a second sampling point as a member included in said configuration space according to a transition probability matrix, said transition probability matrix depending on one or more of said DFF or GFF partial Hamiltonian(s) at the same time using said DFF or GFF ground states associated with said one or more of said DFF or GFF partial Hamiltonian(s), wherein said transition probability matrix has a probability density function as a stationary distribution, said probability density function is related to said DFF or GFF ground states associated with said one or more of said DFF or GFF partial Hamiltonian(s);
        \item[] a second iterative step transferring the CSP value held by said CSP variable to one of one or more result container(s) as a result sampling point and resetting the CSP value held by said CSP variable to said second sampling point;
        \end{itemize}
    \item[] whereby said result sampling points enable a numerical estimate for said objective variable, the numerical estimate having a relative error smaller than or equal to a formula $a(D_1+D_2)^{\,b}N^{-c}$ when both $D_1+D_2$ and $N$ becoming substantially large, with $a$, $b$, and $c$ being constants, both $a$ and $c$ further being positive.
    \end{itemize}
\end{itemize}
\refstepcounter{claim}

\end{itemize}

\clearpage
\section*{ABSTRACT}
\PatentAbstract

\clearpage
\setcounter{page}{1}
\pagestyle{empty}
\section*{DRAWINGS}

\begin{figure}[ht]
\centering
\includegraphics[width=0.8\textwidth]{figs/BlackWhiteDiracDeltaPotential.png}
\caption{A three-well potential on a circle.}
\label{ThreeWellVx}
\end{figure}

\clearpage
\begin{figure}[ht]
\centering
\includegraphics[width=0.8\textwidth]{figs/BlackWhiteThreeEigenStatesEvenOdd.png}
\caption{The eigenstates $\psi_{\mathsmaller{P}}$, $\psi_+$, and $\psi_-$ in dotted, solid, and dashed lines respectively.}
\label{ThreeWellPsiEvenOdd}
\end{figure}

\clearpage
\begin{figure}[ht]
\centering
\includegraphics[width=0.8\textwidth]{figs/BlackWhiteThreeEigenStatesPsiLPsiR.png}
\caption{The eigenstates $\psi_{\mathsmaller{P}}$, $\psi_{\mathsmaller{L}}$, and $\psi_{\mathsmaller{R}}$ in dotted, solid, and dashed lines respectively.}
\label{ThreeWellPsiPsiLPsiR}
\end{figure}

\clearpage
\begin{figure}[ht]
\centering
\begin{tabular}{ccc}
\subfloat[Subfigure 1 list of figures text][$\Phi_+(x_1,x_2)$]{\includegraphics[width=0.495\textwidth]
{figs/ContourBiFermionNodalCellsEven.png}\label{fig:FourPhis1}} & \hspace{-2.85em} &
\subfloat[Subfigure 2 list of figures text][$\Phi_-(x_1,x_2)$]{\includegraphics[width=0.495\textwidth]
{figs/ContourBiFermionNodalCellsOdd.png}\label{fig:FourPhis2}} \\
\subfloat[Subfigure 3 list of figures text][$\Phi_{\mathsmaller{L}}(x_1,x_2)$]{\includegraphics[width=0.495\textwidth]
{figs/ContourBiFermionNodalCellsPsiL.png}\label{fig:FourPhis3}} & \hspace{-2.85em} &
\subfloat[Subfigure 4 list of figures text][$\Phi_{\mathsmaller{R}}(x_1,x_2)$]{\includegraphics[width=0.495\textwidth]
{figs/ContourBiFermionNodalCellsPsiR.png}\label{fig:FourPhis4}}
\end{tabular}
\caption{The bi-fermion wavefunctions $\Phi_+(x_1,x_2)$, $\Phi_-(x_1,x_2)$, $\Phi_{\mathsmaller{L}}(x_1,x_2)$, $\Phi_{\mathsmaller{R}}(x_1,x_2)$.}
\label{FourPhisAndNodalCurves}
\end{figure}

\clearpage
\begin{figure}[ht]
\centering
\begin{tabular}{ccc}
\subfloat[Subfigure 1 list of figures text][$\theta=-3\pi/4$]{\includegraphics[width=0.495\textwidth]
{figs/ContourBiFermionNodalCellsNeg3PiDiv4.png}\label{fig:FourXsinZcos1}} & \hspace{-2.85em} &
\subfloat[Subfigure 2 list of figures text][$\theta=-\pi/4$]{\includegraphics[width=0.495\textwidth]
{figs/ContourBiFermionNodalCellsNeg1PiDiv4.png}\label{fig:FourXsinZcos2}} \\
\subfloat[Subfigure 3 list of figures text][$\theta=\pi/4$]{\includegraphics[width=0.495\textwidth]
{figs/ContourBiFermionNodalCellsPos1PiDiv4.png}\label{fig:FourXsinZcos3}} & \hspace{-2.85em} &
\subfloat[Subfigure 4 list of figures text][$\theta=3\pi/4$]{\includegraphics[width=0.495\textwidth]
{figs/ContourBiFermionNodalCellsPos3PiDiv4.png}\label{fig:FourXsinZcos4}}
\end{tabular}
\caption{The ground states of $H_{\mathsmaller{\rm BF},0}+V_{\mathsmaller{X},\eta\sin\theta}+V_{\mathsmaller{Z},\eta\cos\theta}$ for $\theta=-3\pi/4,-\pi/4,\pi/4,3\pi/4$.}
\label{FourXsinZcosNodalCurves}
\end{figure}

\clearpage
\begin{sidewaysfigure}[ht]
\centering
\begin{tikzpicture} [x=3cm,y=3cm]
\def\CircRad{0.25}
\def\TwoBarY{0.06}
\def\TwoBarX{{sqrt(\CircRad*\CircRad-\TwoBarY*\TwoBarY)}}
\filldraw [gray!25,fill=gray!25] ($(-3,-\CircRad)-1.5*(0,\TwoBarY)$) rectangle ($(3,\CircRad)+1.5*(0,\TwoBarY)$) ;
\draw [black, very thick] ($(-3,0)-6*(0.075,0)$) -- ($(-3,0)-5*(0.075,0)$) ;
\draw [black, very thick] ($(-3,0)-4*(0.075,0)$) -- ($(-3,0)-3*(0.075,0)$) ;
\draw [black, very thick] ($(-3,0)-2*(0.075,0)$) -- ($(-3,0)-1*(0.075,0)$) ;
\draw [black, very thick] (-3,0) -- ($(-2.5,0)-(\CircRad,0)$) ;
\draw [black, very thick] (-2.5,0) circle [radius=\CircRad] ;
\draw (-2.5,-0.01) node {$\mathsmaller{6t}$} ;
\draw [black, very thick] ($(-2.5,\TwoBarY)+(\TwoBarX,0)$) -- ($(-1.5,\TwoBarY)-(\TwoBarX,0)$) ;
\draw [black, very thick] ($(-2.5,-\TwoBarY)+(\TwoBarX,0)$) -- ($(-1.5,-\TwoBarY)-(\TwoBarX,0)$) ;
\draw [black, very thick] (-1.5,0) circle [radius=\CircRad] ;
\draw (-1.5,-0.01) node {$\mathsmaller{6t+1}$} ;
\draw [black, very thick] ($(-1.5,0)+(\CircRad,0)$) -- ($(-0.5,0)-(\CircRad,0)$) ;
\draw [black, very thick] (-0.5,0) circle [radius=\CircRad] ;
\draw (-0.5,-0.01) node {$\mathsmaller{6t+2}$} ;
\draw [black, very thick] ($(-0.5,\TwoBarY)+(\TwoBarX,0)$) -- ($(0.5,\TwoBarY)-(\TwoBarX,0)$) ;
\draw [black, very thick] ($(-0.5,-\TwoBarY)+(\TwoBarX,0)$) -- ($(0.5,-\TwoBarY)-(\TwoBarX,0)$) ;
\draw [black, very thick] (0.5,0) circle [radius=\CircRad] ;
\draw (0.5,-0.01) node {$\mathsmaller{6t+3}$} ;
\draw [black, very thick] ($(0.5,0)+(\CircRad,0)$) -- ($(1.5,0)-(\CircRad,0)$) ;
\draw [black, very thick] (1.5,0) circle [radius=\CircRad] ;
\draw (1.5,-0.01) node {$\mathsmaller{6t+4}$} ;
\draw [black, very thick] ($(1.5,\TwoBarY)+(\TwoBarX,0)$) -- ($(2.5,\TwoBarY)-(\TwoBarX,0)$) ;
\draw [black, very thick] ($(1.5,-\TwoBarY)+(\TwoBarX,0)$) -- ($(2.5,-\TwoBarY)-(\TwoBarX,0)$) ;
\draw [black, very thick] (2.5,0) circle [radius=\CircRad] ;
\draw (2.5,-0.01) node {$\mathsmaller{6t+5}$} ;
\draw [black, very thick] ($(2.5,0)+(\CircRad,0)$) -- (3,0) ;
\draw [black, very thick] ($(3,0)+1*(0.075,0)$) -- ($(3,0)+2*(0.075,0)$) ;
\draw [black, very thick] ($(3,0)+3*(0.075,0)$) -- ($(3,0)+4*(0.075,0)$) ;
\draw [black, very thick] ($(3,0)+5*(0.075,0)$) -- ($(3,0)+6*(0.075,0)$) ;
\end{tikzpicture}
\vspace{3.0ex}
\caption{A $6$-vertex segment of a ring-shaped graph.}
\label{SixVertexSegment}
\end{sidewaysfigure}

\clearpage
\begin{sidewaysfigure}[ht]
\centering
\includegraphics[width=1.0\textwidth]{figs/MethodSimuSyst.png}
\caption{One process of simulating a multivariable system.}
\label{MethodSimuSyst}
\end{sidewaysfigure}

\clearpage
\begin{sidewaysfigure}[ht]
\centering
\includegraphics[width=1.0\textwidth]{figs/MethodSimuDens.png}
\caption{Another process of simulating a multivariable system.}
\label{MethodSimuDens}
\end{sidewaysfigure}

\clearpage
\begin{figure}[ht]
\centering
\includegraphics[width=1.0\textwidth]{figs/MethodConsFK.png}
\vspace{-20ex}
\caption{A process of constructing a multivariable system.}
\label{MethodConsFK}
\end{figure}

\clearpage
\begin{sidewaysfigure}[ht]
\centering
\includegraphics[width=1.0\textwidth]{figs/MethodSolvSyst.png}
\caption{One process of solving a computational problem.}
\label{MethodSolvSyst}
\end{sidewaysfigure}

\clearpage
\begin{sidewaysfigure}[ht]
\centering
\includegraphics[width=1.0\textwidth]{figs/MethodSolvDens.png}
\caption{Another process of solving a computational problem.}
\label{MethodSolvDens}
\end{sidewaysfigure}

\clearpage
\section*{LIST OF THE DRAWINGS}
\ListOfDraw
} {
} 

\end{document}